\pdfoutput=1 
\documentclass[10pt,oneside,openany,a4paper]{book}
\usepackage{parskip}
\usepackage{amssymb}
\usepackage{amsmath}
\usepackage{fancyhdr}
\usepackage{mathrsfs}
\usepackage{mathtools}
\usepackage{bbm}
\usepackage[latin1]{inputenc}
\usepackage{chngcntr} 
\counterwithout{footnote}{chapter}
\usepackage{enumerate}
\usepackage{euscript}
\usepackage{ushort}
\usepackage{lmodern}
\usepackage{longtable}
\usepackage{ccicons}

\PassOptionsToPackage{hyphens}{url} 
\usepackage[breaklinks=true]{hyperref} 
\usepackage{aliascnt} 
\usepackage{natbib}
\usepackage{har2nat}

\makeatletter
\newlength\@quotemargin
\newenvironment{xquote}[1]{
       \def\xquote@author{#1}%
       \par%
       \parshape 1 \@quotemargin \dimexpr\textwidth-2 \@quotemargin \relax%
       \noindent\ignorespaces%
   }
   {
       \par%
       \parshape 1 \@quotemargin \dimexpr\textwidth-2 \@quotemargin \relax%
       \flushright{\normalfont--- \xquote@author}%
       \par\bigskip%
   }
\renewenvironment{quote}
               {\list{}{\rightmargin=\@quotemargin \leftmargin=\@quotemargin}%
                \item\relax}
               {\endlist}
\makeatother
\newtheorem{theorem}{Theorem}[chapter]

\newcommand{\theoremalias}[2]{
	\newaliascnt{#1}{theorem}
	\newtheorem{#1}[#1]{#2}
	\aliascntresetthe{#1}
	\expandafter\def\csname #1autorefname\endcsname{#2}
}

\theoremalias{conjecture}{Conjecture}
\theoremalias{lemma}{Lemma}
\theoremalias{definition}{Definition}
\theoremalias{corollary}{Corollary}
\theoremalias{proposition}{Proposition}
\theoremalias{example}{Example} 
\let\oldexample\example
\renewcommand{\example}{\oldexample\normalfont}

\newenvironment{proof}[1][Proof]{\emph{#1. }}{\hfill$\square$}
\newenvironment{bottompar}{\par\vspace*{\fill}\noindent}{\clearpage}
   
\DeclareMathAlphabet{\mathscr}{OT1}{pzc}{m}{it}
\newcommand{\Cont}{\mathcal{C}}
\newcommand{\X}{\EuScript{X}}
\newcommand{\V}{\EuScript{V}}
\newcommand{\U}{\EuScript{U}}
\newcommand{\Vee}{V_{\mathrm{ee}}}
\newcommand{\Vext}{V_{\mathrm{ext}}}
\newcommand{\VH}{V_{\mathrm{H}}}
\newcommand{\Vx}{V_{\mathrm{x}}}
\newcommand{\SO}{\mathrm{SO}}
\newcommand{\Bounded}{\EuScript{B}} 
\newcommand{\lilo}{o}
\newcommand{\bigO}{\EuScript{O}}
\DeclareRobustCommand{\onehalf}{\textstyle \frac{1}{2}}
\newcommand{\overarrow}[1]{{\buildrel{#1}\over\longrightarrow\;}}
\newcommand{\mtext}[1]{\mathop{\quad\mbox{#1}\quad}}
\newcommand{\1}{\mathbf{1}}

\newcommand{\R}{\mathbb{R}}
\newcommand{\C}{\mathbb{C}}
\newcommand{\N}{\mathbb{N}}
\renewcommand{\d}{\,\mathrm{d}} 
\newcommand{\e}{\mathrm{e}}
\renewcommand{\i}{\mathrm{i}}
\renewcommand{\H}{\mathcal{H}}
\newcommand{\id}{\mathrm{id}}

\newcommand{\esssup}[1]{\mathop{\mathrm{ess}\,\sup}_{#1}}

\newcommand{\sign}{\mathop{\mathrm{sign}}}
\newcommand{\q}[1]{``#1''}
\newcommand{\vKS}{v_{\mathrm{KS}}}
\newcommand{\vH}{v_{\mathrm{H}}}
\newcommand{\vxc}{v_{\mathrm{xc}}}

\newcommand{\Lloc}{L^1_{\mathrm{loc}}}
\newcommand{\AC}{\mathcal{AC}}
\newcommand{\ran}{\operatorname{ran}}
\newcommand{\supp}{\operatorname{supp}}
\newcommand{\Tr}{\operatorname{Tr}}
\newcommand{\Lip}{\mathrm{Lip}}

\DeclareRobustCommand{\onehalf}
	{\textstyle \frac{1}{2}}

\begin{document}
\frontmatter 
\begin{titlepage}

{\Huge\noindent
The Density-Potential Mapping in\\[0.1em]
Quantum Dynamics}\\[0.5em]
{\Large\q{There and Back Again}}

\vspace{0.7cm}
{\Large\noindent Markus Penz}\\
Basic Research Community for Physics\\

\noindent Dissertation submitted to the Faculty of Mathematics, Computer Science and Physics of the University of Innsbruck in partial fulfillment of the requirements for the degree of doctor of science.\\
Advisor: Gebhard Grübl, University Innsbruck\\
Secondary advisors: Michael Ruggenthaler, Max Planck Institute for the Structure and Dynamics of Matter, and Robert van Leeuwen, University of Jyväskylä

\noindent 
Digital Edition\\ 
Innsbruck, 18\textsuperscript{th} October 2016

\begin{bottompar}
\textbf{Abstract.} 
This work studies in detail the possibility of defining a one-to-one mapping from charge densities as obtained by the time-dependent Schrödinger equation to external potentials. Such a mapping is provided by the Runge--Gross theorem and lies at the very core of time-dependent density functional theory. After introducing the necessary mathematical concepts, the usual mapping \q{there}---from potentials to wave functions as solutions to the Schrödinger equation---is revisited paying special attention to Sobolev regularity. This is scrutinised further when the question of functional differentiability of the solution with respect to the potential arises, a concept related to linear response theory. Finally, after a brief introduction to general density functional theory, the mapping \q{back again}---from densities to potentials thereby inverting the Schrödinger equation for a fixed initial state---is defined. Apart from utilising the original Runge--Gross proof this is achieved through a fixed-point procedure. Both approaches give rise to mathematical issues, previously unresolved, which however could be dealt with to some extent within the framework at hand.
\\



\ccLogo \; \ccZero \quad To the extent possible under law, the author has waived all copyright and related or neighbouring rights to this work. It thereby belongs to public domain. This work is published from Austria.
%
\end{bottompar}

\end{titlepage}

\tableofcontents
\mainmatter 


\chapter{Easy Reading}

\begin{xquote}{Marcus Aurelius, \emph{Meditations}}
Everything we hear is an opinion, not a fact. Everything we see is a perspective, not the truth.
\end{xquote}

\section{Preamble}


The belief in the continuous progress of knowledge along with a steady or even exponential rise in economical wealth is contrasted by obvious limits to growth, outbreak of social struggles, or profound cultural changes, but still persists.  Science especially is frequently thought as gradually improving our view towards an objective world, revealing what is \emph{true} and providing us with ever expanding models to predict the future. This view is contrasted by the historic accounts of \citeasnoun{fleck} and \citeasnoun{kuhn} who collected evidence for revolutionary changes in the science world that created radical new thought collectives using completely different concepts and language. That a body of knowledge is in itself not a linearly progressing set of statements also became apparent in the writing of this thesis. Parts that were once in the final chapter suddenly found their way to the far front just to be later moved to an appendix (that does not exist anymore) or  utterly removed. It became obvious that a naturally built presentation is hard to achieve, and the demanded form only allows for a treatise page by page. A more suitable format for the rhizomatic structure of knowledge \cite{deleuze-guattari} without definite center, no full hierarchy, and many connections into all directions, which also marks the typical working progress, would be a hypertext document. Maybe this will find consideration in future academic standards.

A linearly oriented thought style also built linear structures as the basis of many scientific models, quantum theory is to a large part no exception. Only the sudden rupture of wave-function collapse introduces a discontinuous yet epistemologically highly problematic element. The following study of potential-density mappings in quantum mechanics is fully immersed in this orthodox framework. It does not even need any reference to measurements and thus avoids all speculation about interpretations. It rather searches to give a precise relation between the physical formulation and mathematical objects, tries to embed the structures into mathematical branches, especially relating to functional analysis and the study of partial differential equations. Its aim is to open the field of time-dependent density functional theory (TDDFT) more for mathematical investigation, to fortify its formal foundations or---on the contrary---to criticise weaknesses in terms of rigour. In the case of the time-independent variant of density functional theory (DFT) much effort has already been laid into this task but we missed a comparable scrutiny in TDDFT.

Of course it is no secret that non-linear phenomena play an essential role in physics. But it is due to the greater availability of mathematical methods that one usually concentrates on linearisations, like the small perturbation of a ground state. \citeasnoun[p.~VII]{dautray-lions-1} express it this way:

\begin{quote}
It has been observed for a long time that the majority of the phenomena of mathematical physics are \emph{non-linear} [...] However, having the possibility of using in a systematic -- and almost \q{commonplace} -- way the procedures for calculating \emph{approximate solutions} of the state of the system, precise results can generally only be obtained in \emph{the linear cases}.
\end{quote} 

Especially in quantum many-body theory, the breeding grounds of TDDFT, it seems to be a typical approach to force as much classical language on the quantum picture as it can sustain to still reproduce some notable non-classical effects. Models are the representations of such effects and experiments are designed to make them visible in nature.

\section{Acknowledgements}

First of all I extend my gratitude to all my friends and family who constitute the tight social web around me that actually makes me who I am, and that enables me to live and work joyfully. You are the best.

Only for the people who also accompanied me in my scientific activities I will also give names. Josef Rothleitner who sadly passed away in 2011 was the one who guided me through theoretical physics in my undergraduate years and sparked primal curiosity in the subject. Gebhard Grübl, who was already my diploma thesis advisor and is now my official doctoral advisor, was the main source of knowledge and advice in all the years that came after. He is surely my most important teacher and I am also happy to call him a friend. With Michael Ruggenthaler I enjoyed lots of exciting bits of knowledge, but we also sat through the dull parts of university life together; this makes him a true companion. Further he served as a principal advisor for this work which was conducted in close cooperation with him. The second main collaborator to name is Robert van Leeuwen who helped Michael and me a lot in getting a grip on mathematical issues in TDDFT and with whom it has always been a great pleasure to work.

Of course there are many other people who helped me to accomplish this work by giving advice or answering questions. My thanks go out to Iva Brezinova, Klaas Giesbertz, Franz Gmeineder, Paul Lammert, Neepa Maitra, S{\o}ren Nielsen, Mikko Salo, and Barry Simon among others.

I have to name a few close friends that were not directly involved in my research activity but helped to create the intellectual surrounding for it to prosper. Most notably I want to mention Alexander Steinicke who achieved this at our very home. I would like to name Lior Aermark as a very special and dear friend I got to know via science. And finally I should not forget the members of the \emph{Basic Research Community for Physics}; it is our collective mission for the future to foster a cooperative, respectful, and open-minded atmosphere for scientific research also outside traditional institutions.

\section{Epistemological disclaimer}
\label{sect-epistem-disclaimer}

\begin{xquote}{\citeasnoun{fleck-1929}}
Natural science is the art of shaping a democratic reality and to be directed by it -- thus being reshaped by it. It is an eternal, synthetic rather than analytic, never-ending labour -- eternal because it resembles that of a river that is cutting its own bed. This is the true and living natural science.
\end{xquote}

Although as a doctoral thesis this piece of work is expected to represent an individual intellectual endeavour, it clearly depends on a whole history of scientific development. As a highly specialised product of science it is closely embedded in a cultural environment, uses the proper language of this community and---more tragically---can actually only be read by those circles. \cite{fleck} Maybe this is a reason why many scientific publications use \q{we} when it is actually a single author writing. This habit shall also be used here.

Of course the scientific working style has been greatly influenced by the advent of online communities. We can not only turn towards portals such as Wikipedia for instant information but also share questions and answers in more interactive spaces that unite thousands of minds. The most useful such site for our purposes proved to be the forums of \emph{StackExchange}, most notably \emph{MathOverflow} (for professional mathematicians) as well as \emph{Mathematics} (for graduate level questions). References to such postings will be given in footnotes instead of the bibliography, while other online resources can be found in the bibliography.

Even though the main aim of this work is to supplement the mathematical framework of TDDFT, the amount of real mathematical rigour varies greatly throughout it. This is mainly due to the fact that this topic stretches from mathematics to physics and into chemistry, thus touching communities with tremendously different vocabularies. As a consequence some sections are presented in a strict definition--theorem--proof style while others tend to be of rather colloquial form including mathematically non-strict theorems and proofs from literature, then marked as \q{conjectures} and \q{formal proofs} respectively. Such a typical `physical' approach to theory may also be justified by a certain fear of `over-mathematisation' of physics that would only slow down progress, expressed in the following quote.

\begin{xquote}{\citeasnoun{schwartz-1962}}
It is a continual result of the fact that science tries to deal with reality that even the most precise sciences normally work with more or less ill-understood approximations toward which the scientist must maintain an appropriate skepticism. Thus, for instance, it may come as a shock to the mathematician to learn that the Schrodinger equation for the hydrogen atom, which he is able to solve only after a considerable effort of functional analysis and special function theory, is not a literally correct description of this atom, but only an approximation to a somewhat more correct equation taking account of spin, magnetic dipole, and relativistic effects; that this
corrected equation is itself only an ill-understood approximation to an infinite set of quantum field-theoretical equations; and finally that the quantum field theory, besides diverging, neglects a myriad of strange-particle interactions whose strength and form are largely unknown. The physicist, looking at the original Schrodinger equation, learns to sense in it the presence of many invisible terms, integral, integrodifferential, perhaps even more complicated types of operators, in addition to the differential terms visible, and this sense inspires an entirely appropriate disregard for the purely technical features of the equation which he sees. This very healthy self-skepticism
is foreign to the mathematical approach.
\end{xquote}

Yet what we do is essentially making a halt in physical progress, taking a break in mathematical rigour, and see what fruits this bears. Such a theoretical introspection also gives rise to some doubts about long-established facts reigning the community.

There seem to be two different viewpoints towards quantum mechanics. One is more algebraic and fitted to the bra-ket notation of Dirac mostly used in quantum optics and quantum information theory, where quantum states and operators acting on them are considered independent of a chosen Hilbert space basis and completely irrespective of an underlying configuration space like $\R^3$. The most natural choice of a basis is then in terms of eigenfunctions of the Hamiltonian that determines the dynamics of the system. Here we consider mostly analytical questions regarding the PDEs that result from an embedding of the system into the usual $\R^3$ space, attributing special significance to wave functions $\psi(\ushort{x})$ and the corresponding spatial densities $n(x)$. Moreover the Hamiltonian often is explicitly time-dependent and thus does not allow for a fixed reference basis. This leads to the other viewpoint that attributes special prominence to a spatial representation of states which will be used all over this work.

A similar aberration from puristic quantum theories seems to occur in quantum many-body theory when it comes to the interpretation of the one-particle density
\[
n(x) = N \int |\psi(x, x_2, \ldots, x_N)|^2 \, \d x_2 \ldots
\d x_N
\]
which reduces to $n(x) = |\psi(x)|^2$ in the one-particle case. Instead of the usual Born rule\footnote{Interestingly the square rule for probabilities was originally established in a footnote resulting from proof corrections in \citeasnoun{born-1926} not for spatial distributions but for scattering amplitudes of electrons.} applied to the wave function in spatial representation and interpreting $n(x)/N$ as a probability density, one usually acts as if $n(x)$ is a genuine charge distribution of the some `electronic fluid'. Indeed Tokatly in the book edited by \citeasnoun{marques} (see the quote preceding \autoref{sect-tddft}) compares TDDFT to hydrodynamics and similar talk is quite common in quantum chemistry. To accommodate those two conflicting views, \citeasnoun[p.~152]{cartwright} identifies the whole theory as a \emph{simulacrum} explanation: \q{To explain a phenomenon is to find a model that fits it into the basic framework of the theory and that thus allows us to derive analogues for the messy and complicated phenomenological laws which are true of it.} This means such general and far-reaching theories as quantum mechanics are essentially `anti-realistic' while specialised phenomenological laws come closer to `truth' which just means they could be brought into accordance with experience. This aspect is again stressed in \autoref{sect-quantum-chem}. The simulacrum account also includes the famous Duhem--Quine thesis that foundational theories are substantially non-unique in that they are always under-determined by observations. No model is `real', yet they might still be judged pragmatically or aesthetically, and thus our judgement over theories is always purely \emph{epistemic}. In the words of \citeasnoun{glasersfeld}: \q{The image of the scientist gradually unveiling the mysteries of a world that is and forever remains what it is, does not seem appropriate.}

We want to add a quote by \citeasnoun{kato-1951} on the `realness' of Coulombic singular potentials as they play a fundamental role in our derivations. A considerable effort will be devoted in \autoref{sect-kato-peturbations} to include them and other singular potentials to the mathematical theory: \q{In this sense the Coulomb potential can be regarded as a faithful and convenient approximation to the real potential in atomic systems.} So Kato seems to share an attitude that assigns `reality' only to something beyond the fundamental laws of physics which  are \emph{simulacra} themselves. Still the physics community is abuzz with people that demand the Coulombic singular potential to be included in the theory not because it is beautiful, simple, or practical but because it is `real'.

\section{A note on non-analyticity}
\label{sect-note-non-analyticity}

The classical proofs of TDDFT work mostly in the context of analytic functions with respect to wave functions, densities, and potentials. The use of only analytic functions is reminiscent of the days when solutions of differential equations were mostly analytically\footnote{Do not confuse \q{analytic}, functions that can be given by convergent power series in any neighbourhood, and \q{analytical}, a term sometimes used in opposition to \q{numerical} that signifies that some closed mathematical expression can be written down for a function. Of course a power series solution would be such a closed expression, so the two terms might correspond.} derived from a power-series ansatz. But many questions of PDE theory cannot be answered sufficiently in the analytic class. Much of our effort was devoted to the task of lifting such constraints and we want to share some thoughts about why considering only analytic functions in physics is \emph{not} enough.
The \citeasnoun{encyclopedia-analytic-function} says

\begin{quote}
Finally, an important property of an analytic function is its uniqueness: Each analytic function is an \q{organically connected whole}, which represents a \q{unique} function throughout its natural domain of existence. This property, which in the 18th century was considered as inseparable from the very notion of a function, became of fundamental significance after a function had come to be regarded, in the first half of the 19th century, as an arbitrary correspondence.
\end{quote}

That an analytic function is globally uniquely defined by its values in any small open subset of its domain (\q{identity principle}) may seem unnatural if it ought to represent an observational quantity from a conceptual point of view. But also the frequent argument that any function can be approximated arbitrarily well by an analytic function like linear combinations of Hermite functions $H_n(x)\exp(-x^2)$ or cut-off Fourier series in the case of bounded domains may prove problematic. This view was expressed by \citeasnoun[p.~33]{hadamard} in the following quote.\footnote{This quote and further inspirations for this discussion were found in a thread on \emph{StackExchange MathOverflow}:
\url{http://mathoverflow.net/questions/114555/does-physics-need-non-analytic-smooth-functions} (2012).}

\begin{quote}
I have often maintained, against different geometers, the importance of this distinction [between analytic and other functions]. Some of them indeed argued that you may always consider any functions as analytic, as, in the contrary case, they can be approximated with any required precision by analytic ones. But, in my opinion, this objection would not apply, the question not being whether such an approximation would alter the data very little, but whether it would alter the solution very little.
\end{quote}

He goes on to show that the Laplace equation in two dimensions interpreted as an initial value problem (Cauchy problem) $\partial_t^2 u = -\partial_x^2 u$ is not well-posed, which means that in this case an arbitrary small variation of the initial data can produce large differences for the respective solutions. A simple analytical example would be
\[
u(t,x) = \varepsilon \e^{\lambda (t+\i x)} \mtext{with} u(0,x) = \varepsilon \e^{\i \lambda x}, \partial_t u(0,x) = \varepsilon\lambda \e^{\i \lambda x},
\]
or its real part respectively if we favour real functions. It is as small as we want at $t=0$ but can be made arbitrarily large for any $t>0$. This is clearly undesirable when describing physical processes that are not supposed to describe chaotic systems, because any small error in the data would destroy the possibility of predictions. In the words of \citeasnoun[p.~38]{hadamard} again: \q{Everything takes place, physically speaking, as if the knowledge of Cauchy's data would \emph{not} determine the unknown function.} He concludes further that no elliptic PDE can lead to a well-posed initial value problem thus such equations are ruled out for describing evolution processes in physics. This is in stark contrast to the formulation of the Laplace equation as a boundary value problem in which case it is the archetypical example of a well-posed, yet static, problem. This also motivates the classification of PDEs into elliptic, hyperbolic and parabolic as explained by \citeasnoun{klainerman} who notes further:

\begin{quote}
The issue of well-posedness comes about when we distinguish between analytic and smooth solutions. This is far from being an academic subtlety, without smooth, non-analytic solutions we cannot talk about \emph{finite speed of propagation}, the distinctive mark of relativistic phy\-sics.
\end{quote}

In the case of the parabolic heat equation where the speed of propagation is indeed infinite, non-analytic solutions can even arise for analytic initial conditions. An example was already given by \citeasnoun[p.~22ff]{kovalevskaya} and it gets adapted for the Schrödinger equation in \autoref{sect-kovalevskaya}.


This process of extending the working space from analytic functions to smooth ones and further also relates to the battle-cry of category theory attributed to Alexander Grothendieck\footnote{This lore apparently comes from John Baez and was found via \emph{StackExchange MathOverflow}: \url{http://mathoverflow.net/questions/35840/the-role-of-completeness-in-hilbert-spaces} (Aug 18, 2010).}:

\begin{quote}
It is better to work in a nice category with nasty objects than in a nasty category with nice objects.
\end{quote}

Just consider a very nice class of functions like the test functions $\Cont^\infty_0(\R)$ or Schwartz functions that have perfect properties but cannot even be equipped with a natural norm.\footnote{Yet there are constructions possible that make it isomorphic to a Banach space of dimension $2^{\aleph_0}$ thus inducing some kind of topology linked to a norm. This wisdom is from Daniel Fischer on \emph{StackExchange Mathematics}:\\
\url{http://math.stackexchange.com/questions/918318/there-is-no-norm-in-c-infty-a-b-which-makes-it-a-banach-space} (Sep 3, 2014).} We can only define a countable collection of increasing seminorms like $\|f\|_{K,M} = \sup\{ |f^{(k)}(x)| \mid k\leq K ,|x| \leq M \}$ on $\Cont^\infty_0(\R)$ that gives them the structure of Fréchet spaces. Yet these classes are included in $L^2(\R)$, a space including all kinds of horrendous functions but with the very rich structure of a Hilbert space that underpins the whole analysis in quantum mechanics. But even when working with Hilbert spaces one may resolve to use sets of nice functions, especially if they are dense and make arbitrarily good approximations possible, see for example \autoref{sect-ess-sa}. Metaphorically spoken this has been expressed in the same online posting where the Grothendieck quote was found:

\begin{quote}
Sometimes it's possible to go to a party hosted by the Square Integrables in their posh mansion, but spend the whole time hanging out with the Schwartz family.
\end{quote}

And we do indeed work with beautifully elaborate function spaces which include very nasty elements, such as Sobolev spaces and derivatives thereof.

\section[The conception of molecules in quantum chemistry]{The conception of molecules in\\quantum chemistry}
\label{sect-quantum-chem}

\begin{xquote}{\citeasnoun[p.~159]{cartwright}}
Schroedinger's equation, even coupled with principles which tell what Hamiltonians to use for square-well potentials, two-body Coulomb interaction, and the like, does not constitute a theory \emph{of} anything. To have a theory of the ruby laser, or of bonding in a benzene molecule, one must have models for those phenomena which tie them to descriptions in the mathematical theory.
\end{xquote}

This following section is taken from the evening talk titled \q{Where do theories come from?} at the \q{Rethinking Foundations of Physics} Workshop in Dorfgastein (Austria) that was held on April 3rd, 2015 by the author. It serves as a first preliminary introduction to quantum chemistry methods and more specifically into the field of DFT which will be resumed in \autoref{ch-dft}.

The generally accepted foundational law for all effects of chemistry is the Schrö\-dinger equation. To fully account for bond and ionization energies its time-independent version is considered sufficient. All relevant information about the chemical structure is thought to be contained in the lowest eigenstate of the Hamiltonian. As the precise determination of it still poses a formidable problem one is bound to several layers of approximation, each one physically well founded of course, where the full molecular Hamiltonian is reduced to a form where calculations become numerically feasible. This usually starts with the Born--Oppenheimer approximation requiring that the electron wave function is treated first with fixed values for the nuclear degrees of freedom (like vibrational and rotational modes). The total energy gets split into
\[
E = E_{\mathrm{el}} + E_{\mathrm{vibr}} + E_{\mathrm{rot}}
\]
This kind of summation of energies is clearly reminiscent of its use as a universal currency of different physical effects. The calculated electronic energy is then inserted into Schrödinger's equation dealing only with the nuclear wave function or combined with a classical treatment of nuclei. Already at this stage the nuclear geometry is designed after known structural formulae, a representation employed since the 1860s.

Because the number of electronic degrees of freedom is still much too damn high in any space grid of reasonable resolution one has to further reduce them. This usually involves an expansion of the full electronic wave function into so-called molecular orbitals that account for covalent chemical bonds by encompassing multiple atoms. Those molecular orbitals come from linear combinations of a chosen basis set like the hydrogen orbitals where the coefficients are calculated with the Hartree--Fock method such that again minimal energy is achieved. A different and computationally less costly approach is DFT where the degrees of freedom are essentially reduced to the overall charge density. Needless to say, this involves further approximations.

DFT stems itself from the Hohenberg--Kohn theorem (\autoref{hk-th}) which roughly states that the charge density of any (non-degenerate) ground state uniquely fixes the effective external potential for interacting and also for fictitious non-interacting electrons. Thus in principle it is possible to substitute any system of interacting electrons with one of non-interacting electrons but with an added auxiliary potential that accommodates for all inter-electron effects. The new system is now much easier to solve because the single-electron Schrödinger equations decouple and the resulting orbitals can just be filled up one after another. That way a different electronic structure is produced but one still gets the same charge density by virtue of the Hohenberg--Kohn theorem. This serves as the remaining functional variable of all other properties of interest, particularly the energy, thus the name \q{density functional theory}. Such an approximation for the energy is applied to a further partition of the energy functional.
\[
E_{\mathrm{el}} = E_{\mathrm{kin}} + E_{\mathrm{ext}} + E_{\mathrm{H}} + E_{\mathrm{xc}}
\]
The kinetic energy and the energy of the electrons in any external potential are straightforward. The inter-electron effects are approximated by the Hartree term $E_{\mathrm{H}}$ \eqref{hartree-term} describing the mean Coulombic repulsion given by the charge density. All hope rests on the last term, the exchange-correlation energy that needs to account for all quantum effects, see \autoref{sect-tfdt-lda} for a detailed account. \citeasnoun{feynman-stat-mech} in his book on statistical mechanics notes that the $E_{\mathrm{xc}}$ which remains without any idea for an exact expression is sometimes called the \q{stupidity energy}. It gets approximated by quite arcane methods, sometimes involving numerous parameters that are fitted to test scenarios or are taken from experience. Dozens of different such \q{functionals} exist, each performing good for one type of problem (like metals, organics, different bond-types etc.) and worse for others. Yet this heavily approximative theory is most successfully applied, impressively displayed by \citeasnoun{redner} in his \q{Citation Statistics From More Than a Century of Physical Review} where the three most cited papers are from the field of DFT. The theory lies at the core of a whole industry of computer-aided chemical computation, enabling the calculation of properties of manually manipulated molecules. The product description of a popular software package called \citeasnoun{gaussian} states:

\begin{quote}
With GaussView, you can import or build the molecular structures that interest you, set up, launch, monitor and control Gaussian calculations, and retrieve and view the results, all without ever leaving the application. GaussView 5 includes many new features designed to make working with large systems of chemical interest convenient and straightforward. [...] We invite you to try the techniques described here with your own molecules.
\end{quote}

The long way from Schrödinger's equation to useful approximations casts the method far into the domain of phenomenological laws. Those are laws that really `connect to nature' in that one gathers experimental evidence for them in practice. This is also the reason why \citeasnoun{cartwright} attributes `truth' only to such phenomenological laws, while the theoretical, fundamental laws like Schrödinger's serve as explanations and rules to guide our physical intuition. This split in the realm of scientific laws also finds expression in \citeasnoun{weyl-1925}:

\begin{quote}
If phenomenal insight is referred to as knowledge, then the theoretical one is based on belief -- the belief in the reality of the own I and that of others, or belief in reality of the external world, or belief in the reality of God. If the organ of the former is \q{seeing} in the widest sense, so the organ of theory is \q{creativity}.
\end{quote}

\q{Science does not discover; rather it creates}, writes \citeasnoun{sousa-santos-1992}. Such a science can be seen as an economical pursuit, a method to save thought and to provide laws as mnemonic tricks. \cite[p.~55]{mach}
But again, this view towards science is no new conception at all and has already been nicely expressed in the introduction of one of the classical masterpieces of science, Copernicus' \emph{De revolutionibus orbium coelestium} (1543), in an unsigned letter by Andreas Osiander meant to calm down hostile reaction from the church, added without Copernicus' permission, and sometimes viewed as a betrayal on the realist program of natural science. It might tell something about where theories come from if they are not supposed to be divinely revealed to us.

\begin{quote}
For it is the duty of an astronomer to compose the history of the
celestial motions through careful and expert study. Then he must conceive and devise the causes of these motions or hypotheses about them. Since he cannot in any way attain to the true causes, he will adopt whatever suppositions enable the motions to be computed correctly from the principles of geometry for the future as well as for the past. The present author has performed both these duties excellently. For these hypotheses need not be true nor even probable. On the contrary, if they provide a calculus consistent with the observations, that alone is enough. [...] For this art, it is quite clear, is completely and absolutely ignorant of the causes of the apparent nonuniform motions. And if any causes are devised by the imagination, as indeed very many are, they are not put forward to convince anyone that are true, but merely to provide a reliable basis for computation. However, since different hypotheses are sometimes offered for one and the same motion (for example, eccentricity and an epicycle for the sun's motion), the astronomer will take as his first choice that hypothesis which is the easiest to grasp. The philosopher will perhaps rather seek the semblance of the truth. But neither of them will understand or state anything certain, unless it has been divinely revealed to him.
\end{quote}

This concludes the introductive epistemic remarks and after some further notes on notation and terminology we delve right into some serious mathematics, setting the stage for the investigations that follow.

\section{Remarks on notation and terminology}

Even though we follow standard mathematical notation a few aberrations are inevitable plus we use some special terminology that will be explained here. Please refer to the following table of symbols or the additional explanations below if any notation seems unclear.

\begin{longtable}{p{.25\textwidth}p{.65\textwidth}}
$\emptyset$ & empty set\\
$\N$ & natural numbers, $\{1,2,3,\ldots\}$ \\
$\N_0$ & natural numbers including zero, $\{0,1,2,3,\ldots\}$\\
$\Re, \Im$ & real and imaginary part of a complex expression\\
$B_r(x)$ & open ball of radius $r$ centred at $x$\\
$x,y,x_i,y_i$ & usually particle positions in $\R^d$, mostly $d=3$; sometimes $x_i$ can also be the Cartesian coordinates of a vector instead of particle positions\\
$\ushort{x} = (x_1,\ldots,x_N)$ & ordered collection of $N$ particle positions\\
$\bar{x} = (x_2,\ldots,x_N)$ & ordered collection of $N-1$ particle positions, missing the first one\\
$x s, x_i s_i, \ushort{x}\ushort{s}$ & particle positions together with spin coordinates\\
$\partial_x$ & partial derivative $\partial/\partial x$\\
$\partial_x^\nu$ & $\nu$-th order partial derivative\\
$\partial_k$ & partial derivative in the $k$-th coordinate direction\\
$D^\alpha$ & multi-index notation for the (weak) partial derivatives $D^\alpha = D^{(\alpha_1,\ldots,\alpha_d)} = \partial_1^{\alpha_1}\ldots \partial_d^{\alpha_d}$ \\
$\dot f$ & short notation for time derivative $\partial_t f$ \\
$f \circ g$ & composition of functions\\
$\hat f$ or $\mathcal{F}f$ & Fourier transformation\\
$f[g]$ & square bracket notation reserved for 
functional dependency\\
$\1_X$ & characteristic function for the set $X$\\
$\lfloor\cdot\rfloor$ and $\lceil\cdot\rceil$ & floor and ceiling functions for real numbers\\
$\|\cdot\|$ & operator norm or canonical norm of other normed vector spaces, usually the Hilbert space under consideration\\
$\|\cdot\|_X$ & norm of the normed vector space $X$\\
$z^*$ & complex conjugate for a $z \in \C$\\
$A^*$ & adjoint of an operator (\autoref{def-symmetric-selfadjoint})\\
$\overline{A}$ & closure of an operator (\autoref{def-closure})\\
$\hat A, \tilde A$ & transformations of mathematical objects, the $\,\hat{}\,$ for operators in quantum mechanics is used only when indicated\\
$\sigma(A)$ & spectrum of an operator\\
$\langle A \rangle_\psi$ & expectation value of an operator in state $\psi$ \\
$\Cont^k(I,X)$ & $k$-times continuously differentiable functions\\
$\Lip(I,X)$ & Lipschitz continuous functions\\
$\mathcal{B}(X,Y)$ & bounded linear maps\\
$\ker$ & kernel of a map\\
$\ran$ & range of a map\\
$c.c.$ & stands for the complex conjugate of the previous term\\
$c,c_i,C,C_i$ & positive constants in the widest sense (yet constant in space and time), possible dependencies are indicated by indices or by functional arguments\\
$\cdot$ & inner product of Euclidean $\R^n$, but sometimes only the ordinary multiplication of numbers used for better legibility\\
$SO(n)$ & group of rotations on the Euclidean $\R^n$
\end{longtable}

If a domain $\Omega^N$ is taken as an $N$-particle configuration space with particle positions $x_1,\ldots,x_N$ then $\bar{\Omega}$ is the reduced configuration space only involving particle positions $x_2,\ldots,x_N$. This must not be confused with $\overline\Omega$, the closure of $\Omega$ in the Euclidean space $\R^n$.

A smooth function is generally considered to be of type $\Cont^\infty$ though the term `smoothing' may refer to less and just means an increase in regularity of any kind.

A few abbreviations are used throughout the text:

\begin{longtable}{p{.25\textwidth}p{.65\textwidth}}
AGM & inequality of arithmetic and geometric means\\
CSB & inequality of Cauchy--Schwartz--Bunyakowsy\\
DFT & time-independent density functional theory\\
HK & Hohenberg--Kohn, see \autoref{sect-hk-theorem}\\
KS & Kohn--Sham, see \autoref{sect-KS}\\
LDA & local-density approximation, see \autoref{sect-tfdt-lda}\\
ODE & ordinary differential equation\\
PDE & partial differential equation\\
RDM & reduced density matrix, see \autoref{sect-rdm}\\
TDDFT & time-dependent density functional theory\\
UCP & unique continuation property, see \autoref{sect-ucp}
\end{longtable}

For brevity we usually use the convention that all physical quantities are thought to be given in Hartree atomic units, i.e.,
\[
e = \hbar = m_e = \frac{1}{4 \pi \varepsilon_0} = 1.
\]

Finally we will introduce the following notation for quantities that can be mutually estimated, plus a special notation for one-sided estimates including a constant.

\begin{definition}\label{def-sim}
For two expressions $x,y$ that allow for scalar multiplication and have a total order defined on them we write $x \sim y$ if there are constants $C_2 \geq C_1 > 0$ such that $C_1 x \leq y \leq C_2 x$.
\end{definition}

Note that this is a real equivalence relation and thus symmetric because the inequality can easily be transformed to hold with the reciprocal constants as $1/C_2 y \leq x \leq 1/C_1 y$. We also define a one-sided version of this relation.

\begin{definition}\label{def-sim-onesided}
For two expressions $x,y$ like above we write $x \lesssim y$ if there is a constant $C>0$ such that $x \leq C y$.
\end{definition}

At first sight this does not seem a very useful relation because it is always true for any two positive numbers due to the Archimedean property. Yet we also consider expressions such as functions $f,g$ defined on some set $M$ and then $f \sim g$ means $f(x) \sim g(x)$ has to hold with the same constants for all $x \in M$. Thus for example $f \lesssim 1$ is a function with arbitrary upper bound. This notation is useful because it frees us from successively choosing constants in estimates for expressions like special norms. The involved constants can depend on external parameters such as dimensionality and when explicitly noted may also rely on other parameters such as the involved potential. This notation is mainly employed in \autoref{sect-stepwise-static}, \autoref{sect-regularity-sobolev}, and \autoref{sect-variation-trajectories}.


\chapter{Function Spaces and Duality}

\begin{xquote}{Mos Def, \emph{Mathematics}}
	\q{It's simple mathematics.}\\
	\q{Check it out!}\\
	\q{I revolve around science.}\\
	\q{What are we talking about here?}
\end{xquote}

\begin{xquote}{\citeasnoun{reed-simon-2}}
[...] almost all deep ideas in functional analysis have their \emph{immediate} roots in ``applications'', either to classical areas of analysis such as harmonic analysis or partial differential equations, or to another science, primarily physics. [...]\\
More deleterious than historical ignorance is the fact that students are too often misled into believing that the most profitable directions for research in functional analysis are the abstract ones. In our opinion, exactly the opposite is 
true.
\end{xquote}

\section{Sobolev spaces}

\begin{xquote}{\citeasnoun{fichera-1977}}
These spaces, at least in the particular case $p = 2$, were known since the very beginning of this century, to the Italian mathematicians Beppo Levi and Guido Fubini who investigated the Dirichlet minimum principle for elliptic equations. Later on many mathematicians have used these spaces in their work. Some French mathematicians, at the beginning of the fifties, decided to invent a name for such spaces as, very often, French mathematicians like to do. They proposed the name Beppo Levi spaces. Although this name is not very exciting in the Italian language and it sounds because of the name \q{Beppo}, somewhat peasant, the outcome in French must be gorgeous since the special French pronunciation of the names makes it to sound very impressive. Unfortunately this choice was deeply disliked by Beppo Levi, who at that time was still alive, and -- as many elderly people -- was strongly against the modern way of viewing mathematics. In a review of a paper of an Italian mathematician, who, imitating the Frenchman, had written something on \q{Beppo Levi spaces}, he practically said that he did not want to leave his name mixed up with this kind of things. Thus the name had to be changed.  A good choice was to name the spaces after S.~L.~Sobolev. Sobolev did not object and the name Sobolev spaces is nowadays universally accepted.
\end{xquote}

\subsection{Classes of domains}

A \emph{domain}, usually denoted $\Omega$, is always considered to be an open connected subset of $\R^n$, bounded or not. Yet in the theorems below sometimes domains of higher regularity class are demanded. The following definitions are after \citeasnoun[4.4ff]{adams}.

\begin{definition}\label{def-cone-condition}
$\Omega$ satisfies the \textbf{cone condition} if there exists a finite cone $C = \{ x \in \R^n \mid 0\leq |x|\leq\rho, \angle(x,v) \leq \alpha/2 \}$ of height $\rho \in \R_{>0}$, axis direction $0 \neq v \in \R^n$, and aperture angle $0<\alpha\leq\pi$, such that each $x \in \Omega$ is the vertex of a finite cone $C_x$ contained in $\Omega$ and obtained from $C$ by rigid motion.
\end{definition}

\begin{definition}\label{def-strong-lipschitz-condition}
$\Omega$ bounded satisfies the \textbf{strong Lipschitz condition}, also called being of class $\Cont^{0,1}$, if it has a locally Lipschitz boundary, that is, that each point $x \in \partial \Omega$ should have a neighbourhood $U_x \subset \R^n$ such that $\partial\Omega \cap U_x$ is the graph of a Lipschitz continuous function from $\R^{n-1}$ into $\R$.
\end{definition}

A corresponding definition is also available for unbounded domains and then called the strong \emph{local} Lipschitz condition \cite[4.9]{adams}, but it is much more involved and will thus be omitted here. The strong Lipschitz condition certainly holds for rectangular boxes, the only doubt arising in the corners which can be envisioned as the graph of the function $x \mapsto |x|$ if tilted by 45°. Thus problems only arise if corners are like cusps with infinite slope. Domains fulfilling the strong Lipschitz condition always have the cone property.

Two further more basic domain classes shall be given.

\begin{definition}\label{def-convex-domain}
$\Omega$ is called \textbf{convex} if for all $x,y \in \Omega$ the straight line connecting those points is also completely included in $\Omega$.
\end{definition}

\begin{definition}\label{def-quasiconvex-domain}
$\Omega$ is called \textbf{quasiconvex} if there is a constant $M>0$ such that all $x,y \in \Omega$ can be joined by a curve of length less than $M |x-y|$. Any convex domain has this property with $M=1$.
\end{definition}

\subsection{From Lebesgue to Sobolev spaces}
\label{sect-lebesgue-sobolev-spaces}

The Lebesgue function spaces $L^p(\Omega)$, $1 \leq p \leq \infty$, not only form the basic space for wave functions, where the norm defined on them is the basis for an interpretation in terms of probabilities, but become relevant in this work also as the domains of the density-potential mapping. They consist of all Lebesgue-measurable functions $f : \Omega \rightarrow \R$ or $\C$ with finite $L^p$-norm, i.e.,
\[
\|f\|_p = \left(\int_\Omega |f(x)|^p d x\right)^{1/p} < \infty.
\]
The special case $p=\infty$ represents all functions that are bounded up to a set of measure zero and the associated norm $\|f\|_\infty$ is given by the smallest such bound. If we take $f$ roughly with amplitude $A$ and non-zero on a volume $V$ then the $L^p$-norm measures the quantity $AV^{1/p}$. This means that lower $L^p$-spaces allow more singularity while higher ones are more forgiving towards spreading, also expressed by the (continuous) embedding $L^p(\Omega) \subset L^q(\Omega)$ if $p > q$ on bounded domains $\Omega$ where the spread is not an issue. Similarly $f \in L_\mathrm{loc}^p(\Omega)$ just demands local integrability, meaning that $f \in L^p(K)$ for all compact $K \subset \Omega$. As all those spaces are normed vector spaces in which every Cauchy sequence converges (completeness) they form proper Banach spaces. In the case $p=2$ the norm is directly linked to the usual inner product by $\langle f,f \rangle = \|f\|^2$ (note that we sometimes omit the index $2$ of the norm in this important case) and $L^2(\Omega)$ has the structure of a Hilbert space that has risen to eminent prominence within quantum theory.

\begin{definition}\label{weak-derivative}
Given $f,g \in \Lloc(\Omega)$ we define $g$ as the \textbf{weak partial derivative} of $f$ with respect to the $k$-th coordinate, written here just as usual $g = \partial_k f$, if for all test functions $\varphi \in \Cont^\infty_0(\Omega)$ it holds $\langle \varphi, g \rangle = -\langle \partial_k\varphi, f \rangle$.
\end{definition}

The related class of Sobolev spaces includes the weak derivatives of several orders into its definition. Thus not only amplitude $A$ and volume $V$ are measured but also frequency $N$ with the sensitivity controlled by a parameter $m \in \N$ defining up to what order derivatives get included into the Sobolev norm. This norm will be concerned with the quantity $AV^{1/p}N^m$.\footnote{This idea is taken from a very nice and intuitive answer of Terence Tao on the collaborative website \emph{StackExchange MathOverflow}:
\url{http://mathoverflow.net/questions/17736/way-to-memorize-relations-between-the-sobolev-spaces} (Mar 11, 2010).} This is already a strong indication towards the Sobolev embedding theorems mentioned below. To define the $W^{m,p}$-norm \cite[3.1]{adams} we use the multi-index notation $D^\alpha = D^{(\alpha_1,\ldots,\alpha_d)} = \partial_1^{\alpha_1}\ldots \partial_d^{\alpha_d}$ for the weak $\alpha$-th partial derivative. We also note another equivalent and shorter definition right away (cf.~\autoref{lemma-equiv-norms} below).
\begin{equation}\label{eq-sobolev-norm}
\begin{aligned}
\|u\|_{m,p} &= \left( \sum_{|\alpha| \leq m} \|D^\alpha u\|_p^p \right)^{1/p} \sim \sum_{|\alpha| \leq m} \|D^\alpha u\|_p \mtext{if} 1 \leq p < \infty\\[0.5em]
\|u\|_{m,\infty} &= \max_{0\leq |\alpha|\leq m} \|D^\alpha u\|_\infty.
\end{aligned}
\end{equation}

Originally there were actually two definitions of spaces related to the Sobolev norm defined above. $W^{m,p}(\Omega)$ is the space of function where all derivatives $D^\alpha u$ up to order $|\alpha| = m$ are in $L^p(\Omega)$, whereas $H^{m,p}(\Omega)$ denoted the closure of $\Cont^m(\Omega)$ under that norm. The embedding $H^{m,p}(\Omega) \subset W^{m,p}(\Omega)$ \cite[3.4]{adams} seems quite natural but that in fact identity $H^{m,p}(\Omega) = W^{m,p}(\Omega)$ holds for arbitrary open domains $\Omega$ has been proven relatively late by \citeasnoun{meyers-serrin} in a short note simplisticly titled \q{$H=W$}, see also \citeasnoun[3.17]{adams}. Identity is not achieved for $p=\infty$ though, demonstrated by the counterexample $u(x) = |x|$ that is given just after the theorem in \citeasnoun{adams}. We will use the notation $W^{m,p}$ from here on but the notation $H^{m,p}$ will live on for the special case of $p=2$ where we have Hilbert spaces written as $H^{m}=W^{m,2}$.

If the considered Sobolev space has enough regularity, that is $\Omega \subseteq \R^n$ satisfies the cone condition and $mp>n$ or $p=1$ and $m \geq n$, then it can be endowed with pointwise multiplication to yield a Banach algebra. \cite[4.39]{adams} This means for $u,v \in W^{m,p}(\Omega)$ one has also $u \cdot v \in W^{m,p}(\Omega)$ guaranteed from the embedding theorems mentioned in the following section. Note that this algebra has an identity if and only if $\Omega$ is bounded. Actually $1 \in W^{m,p}(\Omega)$ for all $\Omega$ with finite volume, but there are no unbounded domains that simultaneously satisfy the cone condition. Although interesting, this algebra feature has not been utilised in the present work.

The definitive resources on almost all topics relating to Sobolev spaces are the already frequently cited \citeasnoun{adams} and \citeasnoun{mazya}.

\subsection{Sobolev embeddings}

\begin{definition}\label{def-cont-comp-embedded}
Let $U, V$ be Banach spaces with $U \subset V$. We say that $U$ is \textbf{continuously embedded} in $V$ and write $U \hookrightarrow V$ if there is a constant $c \geq 0$ such that for all $u \in U$
\[
\|u\|_V \leq c \, \|u\|_U.
\]
We say that $U$ is \textbf{compactly embedded} in $V$ and write $U \hookrightarrow\hookrightarrow V$ if additionally every bounded sequence in $U$ has a subsequence converging in $V$.
\end{definition}

Compact embeddings are important in showing that linear elliptic PDEs defined over bounded domains have discrete spectra. \cite[6.2]{adams} An example for a compact embedding is $H^1 \hookrightarrow\hookrightarrow L^2$ on $\Omega \subset \R^3$ bounded and satisfying the cone condition. This is a version of the Rellich--Kondrachov theorem (see \autoref{embedding2}) and instead of $H^1$ any higher Sobolev space can be chosen just as well. We can show the plausibility of compactness by considering a cube with edges of unit length in the $H^1$-norm, i.e., all edges are vectors $u$ connecting two neighbouring vertices with $\|u\|_2^2 + \|\nabla u\|_2^2 = 1$. Now consider all edges $u_i$ ordered by increasing $\|\nabla u_i\|_2$ then the embeddings of $u_i$ into $L^2$ clearly have not equal length but their length $\|u_i\|_2$ decreases in the same ordering. Surely there is also an edge with $\|\nabla u\|_2 \approx 1$ and thus $\|u\|_2 \approx 0$. The original unit cube if embedded into $L^2$ gets thinner in every additional dimension and eventually a sequence inside the cube has less and less space to escape into \q{extra dimensions}. Thus every such sequence (that is clearly bounded) must have a converging subsequence and the embedding is indeed compact. A similar embedding will be proved to be compact in \autoref{embedding3}.

Other embeddings relate Sobolev spaces with classes of continuous functions and they are of particular importance regarding boundary conditions. The sought for embedding for $\Omega$ satisfying the cone condition, $j \geq 0, mp > n$ or $m=n$ and $p=1$, is \cite[4.12, I.A]{adams}
\begin{equation}\label{eq-sobolev-embedding}
W^{j+m,p}(\Omega) \hookrightarrow \Cont^j(\overline\Omega).
\end{equation}
Here the space $\Cont^j(\overline\Omega)$ consists of $j$-times continuously differentiable functions $\Cont^j(\Omega)$ that have unique, bounded, continuous extensions to $\overline\Omega$ on all orders. This is equivalent to demanding that all orders of derivatives up to order $j$ are bounded, uniformly continuous functions or by restricting all functions from $\Cont^j(\R^n)$ to the smaller domain $\Omega$. The norm on $\Cont^j(\overline\Omega)$ is thus the maximum of all $\sup_{x \in \Omega} |D^\alpha u(x)|, |\alpha| \leq j$. It is immediately clear how these spaces relate to boundary value problems, since now we can give the boundary value of a function within a Sobolev space a definite meaning. To state something like $u|_{\partial\Omega} = 0$ we only need $u \in \Cont^0(\overline\Omega)$ so we have to assume at least $u \in W^{n,1}(\Omega)$ if $p=1$ or $u \in W^{m,2}(\Omega)$ with $m = \lceil (n+1)/2 \rceil$ if $p=2$ if considering only integer indices. Consequently in such cases we can identify $W^{m,p}(\overline\Omega) = W^{m,p}(\Omega)$ because there is such a unique extension.

On the other side this embedding yields a rule for boundedness just in terms of belonging to Sobolev space with big enough indices to be embedded into the bounded, continuous functions $\Cont^0(\overline\Omega)$. In the case $\Omega \subseteq \R^3$ both $W^{3,1}(\Omega) \hookrightarrow L^\infty(\Omega)$ and $H^2(\Omega) = W^{2,2}(\Omega) \hookrightarrow L^\infty(\Omega)$. This fact will later be important because the dimensional reduction in density functional theory leads from the full $3N$-dimensional configuration space into the usual physical 3-dimensional space whereas the order of differentiability usually stays the same. Thus the degree of regularity can be high enough for the reduced quantities to be bounded. This fact is beneficially used in \autoref{lemma-q-L2-2} and \autoref{cor-strong-lm-sol}.

\subsection{Lipschitz and absolute continuity}
\label{sect-lip-abs-cont}

Regarding the space $W^{1,\infty}(\Omega)$ with norm equivalent to $\|u\|_\infty + \|\nabla u\|_\infty$ one is easily led to believe that all functions therein are Lipschitz continuous, as their first derivatives have bounded (essential) supremum. Indeed $W^{1,\infty}(\Omega) = \Cont^{0,1}(\Omega)$, the space of Lipschitz continuous functions, in any quasiconvex bounded domain $\Omega$.\footnote{This is shown in an answer by user53153 on the \emph{StackExchange Mathematics} board:
\url{http://math.stackexchange.com/questions/269526/sobolev-embedding-for-w1-infty} (Jan 3, 2013).} Since Lipschitz continuity is probably the strongest (and thus nicest) form of continuity and we have boundedness it implies $W^{1,\infty}(\Omega) = \Cont^{0,1}(\Omega) = \Cont^{0,1}(\overline\Omega) = W^{1,\infty}(\overline\Omega)$ by unique extension.

Another popular Sobolev space, now for domains on the real line $\Omega \subseteq \R$, is $W^{1,1}(\Omega) \subset \Cont^0(\overline\Omega)$ where functions and their respective derivatives are Lebesgue integrable. Clearly we have $W^{1,\infty}(\Omega) \subset W^{1,1}(\Omega)$ on bounded $\Omega$ but the space has its own classification in terms of continuity and that is \emph{absolute continuity}. One beautiful characterisation of an absolute continuous function $f$ on a compact interval $\Omega=[a,b]$ using Lebesgue integrals is that $f$ is differentiable almost everywhere, $f' \in L^1([a,b])$ and for all $x \in [a,b]$
\[
f(x)-f(a) = \int_a^x f'(t) \d t.
\]
Now as $f$ is continuous on a compact interval we also have $f \in L^1([a,b])$ and thus $f \in W^{1,1}([a,b])$. The argument can also be given in the other direction so indeed $W^{1,1}([a,b]) = \AC([a,b])$. Note that conversely to before we started out with a closed domain. Using again unique continuation admitted by absolute continuity we have $W^{1,1}([a,b])$ $= W^{1,1}((a,b))$. The possibility of unique continuous continuation to the boundary points means we can give meaningful boundary conditions.

\subsection{Zero boundary conditions}
\label{sect-zero-boundary}

To define functions with zero boundary conditions in one of the spaces above that admit unique, bounded, continuous extension to $\overline\Omega$ it is enough to state $u|_{\partial\Omega} = 0$. In more general cases like that of $L^p$ functions it makes no sense to speak of the boundary values as another function in some $L^q$-sense. But we can rely again on the machinery of Sobolev spaces and utilise the spaces $W_0^{m,p}(\Omega)$ defined as the completion of test functions $\Cont_0^\infty(\Omega)$ under the $W^{m,p}$-norm. Be cautious however because the domain $\Omega$ must be strictly open now. This is because for compact domains one has $\Cont_0^\infty(\Omega) = \Cont^\infty(\Omega)$ because the support of such a test function can only be compact and thus we get $W_0^{m,p}(\Omega)=W^{m,p}(\Omega)$ and no zero boundary conditions at all. It is also true that $W_0^{m,p}(\R^n) = W^{m,p}(\R^n)$ because clearly in such cases there can be no boundary.

A function in $W_0^{1,p}(\Omega)$ is interpreted as the space of Sobolev functions $W^{1,p}(\Omega)$ that vanish at the boundary in the Lebesgue-sense. Thus for $f\in W_0^{m,p}(\Omega)$ one has zero boundary conditions for all $D^\alpha f$ up to $|\alpha| \leq m-1$. \cite[5.37]{adams} If one extends $f$ by 0 then the derivatives of all orders of the extended function are given by 0-extension of the derivative. \cite[p.~217]{walter} In the case of a Dirichlet problem, let us say for the Laplace operator on $L^p$, the adjusted domain is $W^{2,p}(\Omega)\cap W_0^{1,p}(\Omega)$, allowing the necessary order of (weak) derivatives and guaranteeing zero boundary conditions. In the special case $p=2$ (Hilbert spaces) this also gives the domain of the Laplacian on which the operator is self-adjoint which easily follows by twice integration by parts.

To answer the question of restrictions of functions to the border of their domain more precisely, one studies the so-called \emph{trace map}, a bounded linear operator
\[
T: W^{m,p}(\Omega) \longrightarrow L^q(\partial\Omega)
\]
for $mp<n$ and $p \leq q \leq (n-1)p/(n-mp)$. It also demands special constraints on the domain $\Omega \subset \R^n$ not stated here. This tells us that through a loss of certain smoothness we can in fact give a meaning to the function restricted to the boundary. \cite[5.36]{adams}

\subsection{Equivalence of Sobolev norms}

We may use \autoref{def-sim} for equivalence of expressions right away for different norms defined on the same normed vector space. The following lemma gives a basic equivalence of norms used extensively in later chapters.

\begin{lemma}\label{lemma-equiv-norms}\cite[1.23]{adams}\\
Consider the vector space $\R^n$ with different norms.
\begin{align*}
\|x\|_p &= \left( \sum_{i=1}^n |x_i|^p \right)^{1/p} \mtext{if} 1 \leq p < \infty \\
\|x\|_\infty &= \max_{1 \leq i \leq n} |x_i|
\end{align*}
Then all those norms are equivalent.
\end{lemma}

\begin{proof}
Equivalence follows directly from the inequality $\|x\|_\infty \leq \|x\|_p \leq \|x\|_1 \leq n \|x\|_\infty$.
\end{proof}

The seminal book of \citeasnoun{mazya} on Sobolev spaces starts with a much more careful distinction of spaces with certain derivatives in $L^p(\Omega)$ than presented here that are then shown to be isomorphic if defined on domains of a certain regularity. Besides the already given Sobolev norm \eqref{eq-sobolev-norm} of $W^{m,p}(\Omega)$ this includes the norm that only incorporates derivatives of maximal order $m$ next to the ordinary $L^p$-norm.
\begin{equation}\label{eq-equiv-sobolev-norm}
\|u\|_p + \sum_{|\alpha| = m} \|D^\alpha u\|_p
\end{equation}

First we want to give a counterexample \cite[1.1.4, Ex.~2]{mazya} for a domain $\Omega \subset \R^2$ on which one can define a function $u \in L^2(\Omega)$ with $D^{\alpha}u \in L^2(\Omega), |\alpha|=2$, that has a finite norm of the kind given above but infinite $\|u\|_{2,2}$. This clearly needs a non-vanishing first order derivative that yields the infinite value. The domain is built as a \q{rooms and passages} set: one larger bounded domain where the function is set to a constant value and an infinite number of passages with exponentially decreasing length leading to infinitely many rooms of exponentially decreasing size. Now along the passages the function is allowed to have curvature but in the outer rooms it has vanishing second derivative. One can now imagine how an appropriately constructed function may lead to $\|D^\alpha f\|_p = \infty$ for $|\alpha|=1$, thus leading to inequivalent norms. Yet the considered domain is highly irregular and of fractal type so one might ask for necessary and sufficient conditions for equivalence to hold. And indeed the norms \eqref{eq-sobolev-norm}, \eqref{eq-equiv-sobolev-norm}, and even one only involving the highest derivatives
\[
\sum_{|\alpha| = m} \|D^\alpha u\|_p
\]
are equivalent on bounded domains with the cone property. This is shown in \citeasnoun[1.1.11]{mazya} in a corollary following from a generalised Poincaré inequality (cf.~\autoref{th-poincare}). Such inequalities are shown for more general domains and $u \in W^{m,p}_0(\Omega)$ in \citeasnoun[15.4]{mazya} and \citeasnoun[6.29ff]{adams}, already present in the first editions of the books, and more recently by \citeasnoun{wannebo-1994}. The conditions on the domain in the modern treatment uses the concept of \q{capacities of sets}. For a short introduction of this notion and its relation to physical quantities see \citeasnoun{gasinski-1997}.

Later we will be interested in a similar equivalence of the Sobolev norm with the graph norm of the Laplacian\footnote{Mathematical physics is obsessed with this operator, probably because it is the simplest differential operator that is invariant under the group of isometries of the Euclidean space $\R^n$.} $\Delta$ or similar norms, yet on an arbitrary (possibly unbounded) domain, so the theorems cited above will not hold in general. The graph norm of $\Delta$ is defined as $\|\cdot\|_\Delta = \sqrt{\|\cdot\|_2^2 + \|\Delta\cdot\|_2^2}$ which is equivalent to the simple sum $\|\cdot\|_2 + \|\Delta\cdot\|_2$ (see \autoref{lemma-equiv-graph-norm}). We proof the following.

\begin{theorem}\label{th-sobolev-norm-laplace}
For general domains $\Omega \subseteq \R^n$ and $u \in H^{2m}(\Omega) \cap H^{m}_0(\Omega)$, $m \in \N$, it holds that 
\[
\sum_{l=0}^m \|\Delta^l u\|_2 \sim \|u\|_{2m,2}.
\]
In the cases $m=1,2$ this can be reduced to the graph norm $\|u\|_{\Delta^m} \sim \|u\|_{2m,2}$.
\end{theorem}

\begin{proof}
First observe that for arbitrary $u,v \in L^2(\Omega)$ it holds with the  CSB and AGM inequalities that
\begin{equation}\label{eq-uv-estimate}
|\langle u,v \rangle| \leq \|u\|_2 \|v\|_2 \leq \onehalf (\|u\|_2^2 + \|v\|_2^2).
\end{equation}
The relation 
\[
\sum_{l=0}^m \|\Delta^l u\|_2 \sim \left( \sum_{l=0}^m \|\Delta^l u\|_2^2 \right)^{1/2} \leq \|u\|_{2m,2}
\]
is fairly obvious. But can we also establish an estimate
\[
\sum_{|\alpha| = k} \|D^\alpha u\|_2^2 \lesssim \sum_{l=0}^m \|\Delta^l u\|_2^2
\]
in the other direction for all $0 < k \leq 2m$? If $k$ odd we use integration by parts ($u \in H_0^{m}$ is enough such that all boundary terms vanish) to get with \eqref{eq-uv-estimate}
\[
\|D^\alpha u\|_2^2 = \langle D^\alpha u, D^\alpha u \rangle = |\langle D^{\alpha_1} u, D^{\alpha_2} u \rangle| \leq \onehalf (\|D^{\alpha_1} u\|_2^2 + \|D^{\alpha_2} u\|_2^2)
\]
where now $|\alpha_1|, |\alpha_2|$ even. For even $|\alpha|$ we proceed inductively and start with $|\alpha|=2$ and write out all partial derivatives separately. Integration by parts then yields
\begin{equation}\label{eq-sum-laplace-norm}
\begin{aligned}
\sum_{|\alpha| = 2} \|D^\alpha u\|_2^2 &= \sum_{|\alpha| = 2} \langle D^\alpha u, D^\alpha u \rangle \\
&= \sum_{i=1}^n \langle \partial_i^2 u, \partial_i^2 u \rangle + \sum_{i \neq j} \langle \partial_i \partial_j u, \partial_i \partial_j u \rangle \\
&= \sum_{i=1}^n \langle \partial_i^2 u, \partial_i^2 u \rangle + \sum_{i \neq j} \langle \partial_i^2 u, \partial_j^2 u \rangle \\
&= \sum_{i=1}^n \langle \partial_i^2 u, \partial_i^2 u + \sum_{j\neq i}\partial_j^2 u \rangle \\[0.5em]
&= \langle \Delta u,\Delta u \rangle = \|\Delta u\|_2^2.
\end{aligned}
\end{equation}
For $|\alpha|>2$ even we have to repeat the argument taking $u = \Delta v$ which means that
\[
\|\Delta^2 v\|_2^2 = \sum_{|\alpha| = 2} \|D^\alpha \Delta v\|_2^2 = \sum_{|\alpha| = 2} \|\Delta D^\alpha v\|_2^2 = \sum_{|\alpha|, |\beta| = 2} \|D^{\alpha + \beta} v\|_2^2.
\]
Now $D^{\alpha + \beta}$ with all possible $|\alpha|=|\beta|=2$ includes all derivatives of order 4, some even multiple times. So we have the estimate
\[
\sum_{|\alpha| = 4} \|D^{\alpha} v\|_2^2 \lesssim \|\Delta^2 v\|_2^2
\]
that continues likewise to higher even $|\alpha|>4$.\\
Finally we settle the special cases $m=1,2$, where $m=1$ is obvious and for $m=2$ integration by parts readily gives
\[
\|\Delta u\|_2^2 = |\langle u, \Delta^2 u \rangle| \leq \onehalf (\|u\|_2^2 + \|\Delta^2 u\|_2^2)
\]
and finishes the proof.\footnote{This proof was inspired by answers in the following two \emph{StackExchange Mathematics} threads:
\url{http://math.stackexchange.com/questions/101021/equivalent-norms-on-sobolev-spaces} (Jan 21, 2012) and
\url{http://math.stackexchange.com/questions/301404/eqiuvalent-norms-in-h-02} (Feb 12, 2013).
}
\end{proof}

The zero boundary condition $u \in H_0^m$ might also be replaced by a periodic domain where boundary terms vanish when integrating by parts. Note that in the theorem above the particular properties of the domain $\Omega$ are not of interest as it is often the case considering $W^{m,p}_0$ Sobolev spaces because the respective (test) function can just be extended to all of $\R^n$ with zero (see \citeasnoun[4.12 III and 3.27]{adams}). A similar result for $m=1$ yet on bounded domains including the graph norm of more general elliptic partial differential operators and associated weak solutions is called \q{boundary regularity} in \citeasnoun[6.3.2]{evans}. The even more general setting of elliptic partial differential operators of any order on compact manifolds is discussed in \citeasnoun[ch.~III, Th.~5.2 (iii)]{lawson-michelsohn}.

\section{Unbounded theory of operators}
\label{sect-unbounded}

\begin{xquote}{\citeasnoun{lax-2008}}
Friedrichs once told me of a chance encounter with Heisenberg in the sixties. He took the opportunity to express to Heisenberg the profound gratitude of mathematicians for his having created a subject that has led to so much beautiful mathematics. Heisenberg allowed that this was so; Friedrichs then added that mathematics has, to some extent, repaid this debt. Heisenberg was noncommittal, so Friedrichs pointed out that it was a mathematician, von Neumann, who clarified the difference between a selfadjoint operator and one that was merely symmetric.  \q{What's the difference?} said Heisenberg.
\end{xquote}

This section gives a brief introduction into the most important concepts of linear operators on Hilbert spaces with a focus on self-adjoint operators utilised in nearly all of the following chapters. This is usually considered a topic in the field of functional analysis. References heavily consulted by us include \citeasnoun{blanchard-bruening-2} and all four volumes of Reed--Simon. Even in infinite dimensional spaces everything is comparably easy, i.e., similar to linear algebra, if the operators under consideration are bounded. Everything gets way more involved if, a thing unthinkable in finite dimensionality, operators are unbounded.  Yet we have to deal with such operators because they comprise some of the most important and interesting operators used in physics as \q{a fact of life}, as \citeasnoun[VIII.1]{reed-simon-1} phrase it. This is necessarily so because without unbounded operators the cornerstone canonical commutation relation cannot be satisfied. Proposing $[x,p] = \i \hbar$ as an axiom of quantum mechanics (or defining relation of the Heisenberg Lie algebra respectively) we have by commuting $n$-times
\[
[x^n,p] = \i \hbar n x^{n-1}
\]
and thus the operator norms
\begin{align*}
&2\|p\| \|x\|^n \geq \|[x^n,p]\| = \hbar n \|x\|^{n-1}, \\
&2\|p\| \|x\| \geq \hbar n.
\end{align*}
But $n$ can be arbitrarily large and thus at least one of the operators must be unbounded. To allow for that the dimension of the Hilbert space must be infinite.\footnote{This example is from \citeasnoun{wiki-ccr}.}

Because opening the field to unbounded operators really opens a Pandora's box, at least for disciples first entering the subject, we call it an \q{unbounded theory}.

Before we start it is maybe time to lament about one of the most common loosenesses in physics when dealing with the realm of mathematics that could cause significant confusion: A map\footnote{We try to stick to a convention also for example advocated by Serge Lang that a \q{function} (as well as a \q{functional}) has a set of numbers, usually $\R$ or $\C$, as its codomain, whereas a \q{map} or \q{mapping} is more general, like into a Banach space, group or something even more fancy. \cite{wiki-map} A map between similar or identical function spaces is mostly called an \q{operator} and is usually of linear type.} is not only a rule, like \q{take the square of the argument}, but has to include domain and codomain, else we have no idea what to operate on, i.e., where a meaningful usage of the map is possible. If we deal with a differential operator natural domains include the Sobolev spaces.

\subsection{Preceding definitions}
\label{sect-preceding-def}

\begin{definition}
A linear operator $A : \H \rightarrow \H$ is called \textbf{bounded} if there exists an $M \geq 0$ such that for all $\varphi \in \H$
\[
\|A \varphi\| \leq M \|\varphi\|.
\]
The smallest of all such $M$ is called the \textbf{operator norm} of $A$ and written simply $\|A\|$.
\end{definition}

The domain of a bounded operator is always the full Hilbert space, conversely an operator that is not everywhere defined can clearly not be bounded (it may have a bounded extension though). Yet there are unbounded, linear operators that can be defined on all of $\H$. This is shown by the following example taken from \citeasnoun[Prop.~22.1]{blanchard-bruening-2}.

\begin{example}
Take $H$ as a Hamel basis of $\H$. This set is necessarily uncountable for infinite-dimensional Hilbert spaces (as a consequence of the Baire category theorem) but allows for a unique representation of all $\varphi \in \H$ as a finite linear combination of $h_i \in H$.
\[
\varphi = \sum_{i = 1}^n \varphi_i h_i
\]
Now choose a (countable) sequence $(h_j)_{j \in \N} \subset H$ and define the linear operator $A h_j = j h_j$ that gets extended linearly to all of $\H$.
\[
A\varphi = \sum_{i = 1}^n \varphi_i A h_i
\]
If one of the $h_i$ is not part of the chosen sequence then just set $A h_i = 0$. The domain of such an operator is the whole Hilbert space $\H$ but it is unbounded if we just look at $\|A h_j\| \big/ \|h_j\| \rightarrow \infty$.
\end{example}

\begin{definition}\label{def-symmetric-selfadjoint}
A linear, densely defined operator $A : \H \supseteq D(A) \rightarrow \H$ is called \textbf{symmetric} if for all $\varphi, \psi \in D(A)$ we have
\[
\langle A \varphi, \psi \rangle = \langle \varphi, A \psi \rangle.
\]
This property can be extended to larger spaces. Let $D(A^*)$ be the set of all $\psi \in \H$ for which there exist corresponding $\eta \in \H$, such that for all $\varphi \in D(A)$
\[
\langle A \varphi, \psi \rangle = \langle \varphi, \eta \rangle.
\]
This defines the \textbf{adjoint} or \textbf{dual} operator $A^* \psi = \eta$ on $D(A^*)$. For such symmetric operators it holds that $D(A^*) \supseteq D(A)$. If in case of a symmetric operator $D(A^*) = D(A)$ then it is called \textbf{self-adjoint}. As one can easily see, every bounded symmetric operator is thus also self-adjoint.
\end{definition}

\subsection{The theorem of Hellinger--Toeplitz and closed operators}

In the case of symmetric operators we have guaranteed boundedness if the operator is defined everywhere, shown by the following theorem taken from \citeasnoun[Th.~22.1]{blanchard-bruening-2}.

\begin{theorem}[Hellinger--Toeplitz]\label{th-hellinger-toeplitz}
A symmetric, everywhere defined operator is always bounded.
\end{theorem}

\begin{proof}
We assume that a symmetric $A : \H \rightarrow \H$ is unbounded. Then there is a sequence $(\varphi_n)_{n \in \N} \subset \H, \|\varphi_n\| = 1$, such that $\|A \varphi_n\| \rightarrow \infty$. Now define a sequence of linear functionals $T_n(\psi) = \langle \varphi_n, A\psi \rangle = \langle A\varphi_n, \psi \rangle$. By the CSB inequality all these functionals ought to be continuous because $|T_n(\psi)| \leq \|A\varphi_n\| \cdot \|\psi\|$ following the second representation. Also for fixed argument $\psi$ the sequence is bounded by $|T_n(\psi)| \leq \|\varphi_n\| \cdot \|A\psi\| = \|A\psi\|$ following the first representation. The uniform boundedness principle (also called Banach--Steinhaus theorem and another important consequence of the Baire category theorem) now implies that there is an $M > 0$ such that $\|T_n\| \leq M$ for all $n \in \N$. This implies $\|A \varphi_n\|^2 = T_n(A \varphi_n) \leq \|T_n\| \cdot \|A \varphi_n\| \leq M \|A \varphi_n\|$ leading to $\|A \varphi_n\| \leq M$ which stands in contradiction to the selection of $(\varphi_n)_{n \in \N}$.
\end{proof}

Yet we will frequently be concerned with unbounded operators, which makes the whole story much more difficult (and thus more interesting). Typical examples of Hamiltonians on $L^2$ spaces are usually unbounded, especially the Laplacian $H_0=-\Delta$ (see \autoref{ex-laplacian} for its associated domain). A generalisation of boundedness for such operators in order to retain nice enough properties is \emph{closedness} and we will see later that all self-adjoint operators share this property. A closed operator can be thought of as a mapping where the corresponding domain has been extended maximally in a continuous fashion. Thus they are more general than bounded operators that are necessarily defined everywhere but still retain nice properties.

\begin{definition}\label{def-closed}
A linear, densely defined operator $T : \H \supseteq D(T) \rightarrow \H$ is called \textbf{closed}, if for every sequence $\varphi_n$ in $D(T)$ with $\varphi_n \rightarrow \varphi \in \H$ and $T\varphi_n \rightarrow \xi \in \H$ it holds that $\varphi \in D(T)$ and $\xi=T\varphi$.
\end{definition}

An equivalent notion of closedness is often given in terms of the graph of the operator $\Gamma(T) = \{ (x,Tx) : x \in D(T) \} \subset \H \times \H$. $T$ is now closed if and only if its graph is closed as a subset of $\H \times \H$. The concept of graph norm is also derived from this picture, as the natural norm for elements of $\Gamma(T)$, and used in the following lemma for a further equivalence.

\begin{lemma}
An operator $T$ is closed if and only if $D(T)$ is a Banach space, i.e., a complete normed space equipped with graph norm $\|\cdot\|_T = \sqrt{\|\cdot\|^2 + \|T\cdot\|^2}$.
\end{lemma}

\begin{proof}
The proof is straightforward. If $\varphi_n \rightarrow \varphi, T\varphi_n \rightarrow \xi$ are sequences as in \autoref{def-closed} above then $\varphi_n$ converges in graph norm and thus $\varphi \in D(T), \|\varphi_n-\varphi\| \rightarrow 0$, and $\|T\varphi_n-T\varphi\| \rightarrow 0$. This also establishes $\xi=T\varphi$.

Starting from the definition of closedness we need to show that every Cauchy sequence $\varphi_n$ with respect to the graph norm converges in $D(T)$. This means $\varphi_n$ and $T\varphi_n$ are Cauchy in $\H$ and have the limits $\varphi$ and $\xi$ respectively. Because $T$ is closed this yields $\varphi \in D(T)$ as the proper limit of the sequence in $D(T)$.
\end{proof}

Because sometimes a slightly different definition of the graph norm is more practical we will show their equivalence and use it in the following without special reference to it.

\begin{lemma}\label{lemma-equiv-graph-norm}
$\|\cdot\|_T \sim \|\cdot\| + \|T\cdot\|$.
\end{lemma}

\begin{proof}
The inequality
\[
\sqrt{\|x\|^2 + \|Tx\|^2} \leq \|x\| + \|Tx\| \leq \sqrt{2}\cdot\sqrt{\|x\|^2 + \|Tx\|^2}
\]
is clear from squaring it and using the simple case of the AGM inequality, i.e., $2 \cdot \|x\| \cdot \|Tx\| \leq \|x\|^2 + \|Tx\|^2$. Note that \autoref{lemma-equiv-norms} can also be used to prove the statement directly.
\end{proof}

\begin{definition}\label{def-closure}
An operator $T$ is called \textbf{closable} if it admits a potentially non-unique closed extension. Yet there is always a smallest closed extension called the \textbf{closure}, denoted by $\overline{T}$.
\end{definition}

\begin{lemma}\label{lemma-symmetric-closable}
Every symmetric operator is closable.
\end{lemma}

\begin{proof}
As in \autoref{def-closed} take $\varphi_n$ in $D(T)$ with $\varphi_n \rightarrow \varphi \in \H$ and $T\varphi_n \rightarrow \xi \in \H$. Now for all $\psi \in D(T)$
\[
\lim_{n \rightarrow \infty} \langle T\varphi - T\varphi_n,\psi \rangle = \lim_{n \rightarrow \infty} \langle \varphi - \varphi_n,T\psi \rangle = 0.
\]
But this implies $T \varphi = \lim_{n \rightarrow \infty} T \varphi_n = \xi$ because $D(T)$ is assumed to be dense in $\H$.
\end{proof}

\begin{lemma}\label{lemma-adj-abgeschlossen}
A self-adjoint operator is always closed.
\end{lemma}

\begin{proof}
For all $\psi \in D(A)$ it holds with the same notation as in \autoref{def-closed} that $\langle A \varphi_n, \psi \rangle = \langle \varphi_n, A \psi \rangle$ and therefor in the limit $n \rightarrow \infty$ we get $\langle \xi, \psi \rangle = \langle \varphi, A\psi \rangle$. For this reason we have $\varphi \in D(A)$ and $\xi = A\varphi$.
\end{proof}

\subsection{Resolvent operators}

\begin{lemma}[von Neumann's formula]\label{lemma-neumann}
Let $T$ be a closed operator and $z \in \C$. Then it holds
\[
\ran(z-T)^\perp = \ker(\bar{z}-T^*).
\]
\end{lemma}

\begin{proof}
Firstly we show $\ran(z-T)^\perp \subseteq \ker(\bar{z}-T^*)$. Some $\psi \in \ran(z-T)^\perp$ clearly satisfies $\langle (z-T)\varphi,\psi \rangle = 0$ for all $\varphi \in D(T)$, which is the same as $\langle T\varphi,\psi \rangle = \langle \varphi,\bar{z}\psi \rangle$. Thus by definition we have $\psi \in D(T^*)$ and $\langle \varphi,(\bar{z}-T^*)\psi \rangle = 0$. As $D(T)$ is dense in $\H$, it holds that $(\bar{z}-T^*)\psi=0$ and thereby $\psi \in \ker(\bar{z}-T^*)$.

For the inclusion in the other direction one proceeds as given, just from bottom to top.
\end{proof}

\begin{lemma}[resolvent operator]\label{lemma-resolvent}
Let $A$ be a self-adjoint operator and $z \in \C$ with $\Im z \neq 0$. Then $(z - A)^{-1}$ is bounded and fulfils
\[
\| (z - A)^{-1} \| \leq |\Im z|^{-1}.
\]
Stated differently all these $z$ belong to the resolvent set of a self-adjoint operator. The spectrum of the operator must therefore be a subset of the real numbers.
\end{lemma}

\begin{proof}
For all $\psi \in D(A)$ it holds
\begin{align*}
\|(z-A)\psi\|^2 &= |z|^2 \|\psi\|^2 + \|A\psi\|^2 - 2(\Re z) \langle \psi,A\psi \rangle \\
&\geq |z|^2 \|\psi\|^2 + \|A\psi\|^2 - 2(\Re z) \, \|\psi\| \, \|A\psi\| \\
&= (\|A\psi\| - (\Re z)\,\|\psi\|)^2 + (\Im z)^2  \|\psi\|^2 \\
&\geq (\Im z)^2  \|\psi\|^2
\end{align*}
and thus
\begin{equation}\label{eq-lemma-resolvent-inequ}
\|(z-A)\psi\| \geq |\Im z| \, \|\psi\|.
\end{equation}
The same inequality \eqref{eq-lemma-resolvent-inequ} holds for $\bar{z}-A$ and leads to $\ker(\bar{z}-A) = \{0\}$. \autoref{lemma-neumann} tells us that $\ran(z-A)^\perp = \ker(\bar{z}-A)$, so $\ran(z-A)$ has to be dense in $\H$. This can be used to define an inverse operator
\[
(z-A)^{-1} : \ran(z-A) \rightarrow \H
\]
with a dense domain. Because of $\Im z \neq 0$ and \eqref{eq-lemma-resolvent-inequ} it satisfies the inequality from \autoref{lemma-resolvent} and can be extended to a bounded operator on all of $\H$.
\end{proof}

\subsection{Positive operators}

\begin{definition}
Let $T$ be a linear, densely-defined operator. We write $T \geq \lambda$ if for all $\varphi \in D(T)$ with $\|\varphi\|=1$ it holds $\langle \varphi, T \varphi \rangle$ is real and $\langle \varphi, T \varphi \rangle \geq \lambda$. The operator is called \textbf{positive} if $T \geq 0$.
\end{definition}

This could be used to introduce a partial ordering on an algebra of linear operators by $A \geq B$ if and only if $A-B$ is positive. The classical example of a positive operator is $-\Delta$ on a domain such that the boundary terms vanish (cf.~the beginning of \autoref{sect-const-potential}). This is easy to check with integration by parts.
\[
\langle \varphi, -\Delta \varphi \rangle = \langle \nabla\varphi, \nabla\varphi \rangle = \|\nabla\varphi\|^2 \geq 0
\]

\begin{lemma}\label{lemma-positive-symmetric}
A positive operator on a complex Hilbert space is always symmetric.
\end{lemma}

\begin{proof}
The sesquilinear inner product allows $\langle x,Ax \rangle = \langle Ax,x \rangle^* = \langle Ax,x \rangle$ because its a real number from positiveness. Using the polarisation identity we have for $x,y \in D(A)$
\begin{align*}
\langle x,Ay \rangle &= \frac{1}{4} \sum_{s=0}^3 \i^{-s} \langle x+\i^s y,A (x+\i^s y) \rangle \\
&= \frac{1}{4} \sum_{s=0}^3 \i^{-s} \langle A(x+\i^s y),x+\i^s y \rangle = \langle Ax,y \rangle.
\end{align*}
\end{proof}

\begin{lemma}\label{lemma-positive-inverse}
Let $A \geq \lambda > 0$ and self-adjoint. Then there is an inverse operator $A^{-1} : \H \rightarrow D(A) \subseteq \H$ bounded by $\lambda^{-1}$.
\end{lemma}

\begin{proof}
For all $\varphi \in D(A) \setminus \{ 0 \}$ with CSB
\begin{equation}\label{lemma-positive-inverse-inequ}
0 < \lambda \|\varphi\|^2 \leq \langle \varphi, A \varphi \rangle \leq \|\varphi\| \cdot \|A\varphi\|
\end{equation}
and thus $A$ is clearly injective. If $A \varphi = \xi$ we uniquely define $\varphi = A^{-1}\xi$ and get $\|A^{-1}\xi\| \leq \lambda^{-1}\|\xi\|$ from \eqref{lemma-positive-inverse-inequ} above. It is still open if $A$ is also surjective such that this inverse can be defined on the whole Hilbert space. From \autoref{lemma-adj-abgeschlossen} we know that a self-adjoint operator if always closed. Given a sequence $\{\varphi_n\}$ in $D(A)$ with $A \varphi_n \rightarrow \psi$ we know from \eqref{lemma-positive-inverse-inequ}
\[
\|\varphi_n - \varphi_m\| \leq \lambda^{-1} \|A\varphi_n - A\varphi_m\| \rightarrow 0
\]
and thus $\{\varphi_n\}$ is a Cauchy sequence with limit $\varphi$ and $A\varphi=\psi \in \ran{A}$ because of closedness. This implies the range of the operator is a closed subset of the Hilbert space. But \autoref{lemma-neumann} tells us $(\ran A)^\perp = \ker A = \{ 0 \}$ and only $\ran A = \H$ is left as a possibility.
\end{proof}

\subsection{Essentially self-adjoint operators and the Laplacian}
\label{sect-ess-sa}

On many occasions one deals with operators that are easily shown to be symmetric on a dense subset of $\H$, yet are really self-adjoint in their closure. If the full self-adjoint domain is hard to define explicitly one could reside with a weaker definition of a restricted operator that is \emph{almost} self-adjoint.

\begin{definition}\cite[VIII.2]{reed-simon-1}\label{def-ess-sa}\\
A symmetric operator $T$ is called \textbf{essentially self-adjoint} if its closure $\overline{T}$ is self-adjoint. If some $T$ is closed, any subset $D \subset D(T)$ is called a \textbf{core} for $T$ if the closure of $T$ restricted to $D$ equals again $T$.
\end{definition}

This is important because an operator is essentially self-adjoint if and only if it admits a unique self-adjoint extension, whereas an only symmetric operator may have many self-adjoint extensions or none at all. To fully specify a self-adjoint operator it thus suffices to define it on a perhaps significantly smaller core. Typically such a core can consist of the test functions $\Cont_0^\infty$ only. The choice of the specific core can then lead to the different possible self-adjoint operators with their respective boundary conditions.

\begin{example}\label{ex-laplacian}
We will show that the free quantum Hamiltonian, the Laplacian $H_0 = -\onehalf\Delta$, on the configuration space $\Omega = \R^n$ is a self-adjoint operator with domain $D(H_0)=H^2(\R^n)$. We start with $H_0$ acting on the smaller domain of test functions $\Cont_0^\infty(\R^n)$, the supposed core of an essentially self-adjoint operator. This is clearly a symmetric operator as twice integration by parts shows. Now every symmetric operator is closable by \autoref{lemma-symmetric-closable}, we thus take all $\varphi \in L^2(\R^n)$ such that a sequence $\{\varphi_n\} \subset \Cont_0^\infty(\R^n)$ that converges to $\varphi$ in $L^2$-norm also has $\{ H_0 \varphi_n \}$ converging to some $\xi \in L^2(\R^n)$ in $L^2$-norm. The closure is now constructed by setting
\[
H_0 \varphi = H_0 \lim_{n \rightarrow \infty} \varphi_n = \lim_{n \rightarrow \infty} H_0 \varphi_n = \xi.
\]
Note especially that this looks like a condition of continuity but it is less, because it is not fulfilled by all converging sequences $\{\varphi_n\}$, only those that have converging $\{ H_0 \varphi_n \}$. Closedness of linear operators is thus a weaker form of continuity (boundedness). Also this construction is unique, i.e., it does not depend on the selected sequence, because the operator is symmetric.\\
The next step consists of determining the domain of this newly constructed closed symmetric operator. For this purpose let us collect all these $\varphi$ from above. Putting the two convergences together it says we take all $\varphi$ that have corresponding sequences $\{\varphi_n\} \subset \Cont_0^\infty(\R^n)$ converging to $\varphi$ but with respect to the stronger norm $\|\varphi\|_{H_0} = \sqrt{\|\varphi\|^2_2 + \|H_0 \varphi\|^2_2}$ (graph norm). But the completion of the test functions under this graph norm is just the Sobolev space of twice weakly differentiable functions by \autoref{th-sobolev-norm-laplace}. So this is our new domain $D(H_0)=H^2(\R^n)$ and we suspect this is the right domain for a self-adjoint operator.\\
For self-adjointness we need $D(H_0) = D(H_0^*)$ so the domains of the operator and its adjoint must coincide. The domain of the adjoint is defined in \autoref{def-symmetric-selfadjoint} as
\[
D(H_0^*) = \{ \psi \in L^2(\R^n) \mid \exists \eta \in L^2(\R^n) \;\forall \varphi \in D(H_0) : \langle H_0 \varphi, \psi \rangle = \langle \varphi,\eta \rangle \}.
\]
But this is just the definition for twice weak differentiability with the additional condition that the second (weak) derivative $\eta$ is still in $L^2(\R^n)$ (because we are not using test functions $\varphi$ now). So $\varphi \in L^2(\R^n)$ and $\eta = H_0\varphi \in L^2(\R^n)$ and thus
\[
D(H_0^*) = D(H_0)
\]
as the condition for self-adjointness.\footnote{This proof has been contributed to \emph{StackExchange Mathematics}:\\
\url{http://math.stackexchange.com/questions/733625/domain-of-the-quantum-free-hamiltonian-in-1d/805776} (May 22, 2014).}
\end{example}

\section{Dual spaces}

\subsection{State--observable duality}

The conceptual basis for many physical theories is that of a fundamental dual relation between a description of what \emph{is}, the state of the system, and what we measure, the result of an observation or experiment. The little word \q{\emph{is}} already includes a huge ontological preconception, including that the chosen description somehow maps from \emph{nature} to \emph{theory} thus implying that such a map can exist, whilst we actually have no conception of \emph{nature} apart from \emph{theory}. Aware of such a critique we might still choose \q{states} $\mathcal{S}$ and \q{observables} $\mathcal{O}$ as the two general building blocks of a physical theory and give it mathematical meaning by choosing appropriate sets for both plus defining the dual pairing $\mathcal{O} \times \mathcal{S} \rightarrow \R$ to represent \q{measurements}. Adding structures of linearity and continuity to those spaces, typically we assume the setting of topological vector spaces, and demanding a linear and continuous dependence of the outcome on the state will make the space of observables part of the topological dual space $\mathcal{O} \subset \mathcal{S}'$. We might equally well go the other way, applying the state on the observable and setting $\mathcal{S} \subset \mathcal{O}'$ which is mathematically equivalent if the space is \emph{reflexive} (see \autoref{sect-lebesgue-dual}). Please note that again we put linearity at the basis of the theory, ruling out certain chaotic effects at this level. In quantum mechanics this is motivated by the superposition principle and the relation to expectation values in probability theory. See \citeasnoun{davies-lewis-1970} for a very broad such framework generalising observables to \emph{instruments}.

The usual setting of quantum mechanics is well known \cite{neumann-1932}: The state can be represented as a density operator, a positive (thus symmetric) operator on a Hilbert space $\H$ that is of trace-class (nuclear operator) $\mathcal{S}_1(\H)$ normalised to 1. The pairing with an observable $A$, also an operator on $\H$, is realised with the trace $\Tr(A\rho)$. Interestingly states can also be seen as observables, thus $\mathcal{S} \subset \mathcal{O}$, and a state paired with itself is indeed the measure of \q{purity} with $\Tr(\rho^2) = 1$ the condition for pure and $\Tr(\rho^2) < 1$ for mixed states. Note that trace-class operators are always bounded and together with \autoref{lemma-positive-symmetric} it follows that density operators are self-adjoint. This will also follow for observables because we demand $\Tr(A\rho) \in \R$, further continuity in the $\mathcal{S}_1(\H)$ topology only allows bounded observables to be applied on such states. Thus we get $\mathcal{S} = \mathcal{S}_1^+(\H)$, the space of positive, symmetric, trace-class operators on $\H$ normalised to 1, and $\mathcal{O} = (\mathcal{S}')^\mathrm{sa} = \mathcal{B}^\mathrm{sa}(\H)$, all bounded, self-adjoint operators. Note that if we change the viewpoint towards states operating on observables and thus changing the involved topology to that of operator norm the space of observables is limited to the self-adjoint compact operators $\mathcal{S}_\infty^\mathrm{sa}(\H)$, that are those that allow for a countable basis of eigenvectors. This follows from $\mathcal{S}_\infty(\H)' \simeq \mathcal{S}_1(\H)$, making $\mathcal{S}_\infty(\H)$ the predual of $\mathcal{S}_1(\H)$, while before $\mathcal{S}_1(\H)' \simeq \mathcal{B}(\H)$ \cite[Th.~IV.3]{schatten}. Thus the space of compact operators $\mathcal{S}_\infty(\H)$ is not reflexive if $\H$ is infinite-dimensional. The mathematical formalisation is that of Schatten class operators $S_p(\H)$, for which a type of Hölder inequality holds for $p,q \in [1,\infty]$ and $1/p + 1/q = 1$.
\[
\|AB\|_{\mathcal{S}_1} \leq \|A\|_{\mathcal{S}_p} \|B\|_{\mathcal{S}_q}
\]
$\mathcal{S}_1$ are the trace-class operators, $\mathcal{S}_2$ yields the Hilbert--Schmidt operators, and $\mathcal{S}_\infty$ are as already noted the compact operators measured with operator norm. Note that \citeasnoun{schatten} calls the class of compact operators \emph{completely continuous}.

It seems this setting is quite restrictive because as already noted in \autoref{sect-unbounded} we cannot expect the operators of quantum mechanics to be compact or even bounded. This is overcome in algebraic quantum field theory by identifying observables first with an \emph{abstract} (Weyl) $C^*$-algebra that form the predual of states like above. Then by using the GNS construction \cite{araki,strocchi} a suitable Hilbert space formulation (the so-called Schrödinger representation yields $\H=L^2(\R^n)$) with a representation of the $C^*$-algebra elements as linear bounded operators is recovered. By a result of von Neumann all regular and irreducible such representations are unitarily equivalent. Representations belonging to pure states (that cannot be written as convex linear combinations) are always irreducible which is why we are usually happy with the Schrödinger representation. \cite[Th.~3.2.2]{strocchi} Position and momentum are present in transformed form as unitary bounded operators
\[
U(\alpha)=\e^{\i \alpha x},\;\; V(\beta)=\e^{\i \beta p},
\]
strongly continuous in $\alpha$ and $\beta$ respectively. The $(V(\beta)\psi)(x) = \psi(x+\beta)$ defines a shift in position space and $U(\alpha)$ does so equivalently in Fourier (momentum) space. The position and momentum operators are then recovered as derivatives that give unbounded, densly defined operators, just like the Hamiltonian as the derivative of the unitary evolution operator (see \autoref{def-inf-generator}).

\subsection{Lebesgue space dual}
\label{sect-lebesgue-dual}

The \emph{algebraic} dual $X^*$ to a vector space $X$ over a field $\mathbb{K}$ is the space of all linear functionals $f : X \rightarrow \mathbb{K}$. As soon as one is in the setting of topological vector spaces, such as Banach spaces, one usually limits this space to all \emph{continuous} linear functionals, the \emph{topological} dual space $X'$. As we will always consider this space, we call it the dual space of $X$ for brevity. A norm for the dual space can always be induced by an $X$-norm by $\|f\|_{X'} = \sup\{|f(x)| \mid x \in X, \|x\|_X \leq 1 \}$ and makes it a Banach space (even if $X$ is not complete). If one considers the dual of the dual $X''$ in such a setting a natural functional $\phi$ is given by
\begin{align*}
\phi : X &\longrightarrow X'' \\
x &\longmapsto (X' \longrightarrow \mathbb{K}: f \longmapsto f(x)).
\end{align*}
Now $\phi(x)$ is just the evaluation functional, assigning every functional $f$ its value at $x$, $\phi(x)(f) = f(x)$. The map $\phi$ is always an isometric isomorphism of $X$ into a closed subspace of $X''$. If $\phi(X) = X''$ the space is called \emph{reflexive} and usually one identifies $X''=X$.

The spaces $L^p(\Omega)$ are reflexive for $1<p<\infty$ where the dual pairing is managed by Hölder's inequality, $1/p + 1/q = 1$, $f \in L^p(\Omega) \simeq L^q(\Omega)'$, $g \in L^q(\Omega) \simeq L^p(\Omega)'$,
\[
|\phi(f)(g)| = \left|\int_\Omega f(x)g(x) \d x\right| \leq \|fg\|_1 \leq \|f\|_p \|g\|_q.
\]
The case $p=2$ leads to the famous Hilbert space of quantum mechanics. The inner product structure gives an (antilinear) isomorphism to its dual space $L^2(\Omega) \simeq L^2(\Omega)'$ by the Riesz--Fréchet representation theorem (\autoref{th-riesz-frechet} below).

The spaces $L^1(\Omega)$ and $L^\infty(\Omega)$ stand out as well, but more in a negative way as they are not linked \emph{vice versa} as dual spaces. Functionals on $L^1(\Omega)$ can be represented by essentially bounded functions, so $L^1(\Omega)' = L^\infty(\Omega)$, but the functionals on $L^\infty(\Omega)$
are made up of all finitely additive signed measures that are absolutely continuous with respect to the Lebesgue measure. This is called a \emph{ba} space, see \citeasnoun{wiki-ba}.

Another way to see this is to note that $L^p(\Omega)$ is separable for $1 \leq p < \infty$ but not in the case $p=\infty$, so $L^\infty(\Omega)'$ is neither and thus $L^\infty(\Omega)' \simeq L^1(\Omega)$ cannot hold. Separability does not in general carry over to the dual space but it does in the other direction from the dual to its predual.

\subsection{Sobolev space dual}
\label{sect-sobolev-dual}

Let us also define the dual of the Hilbert space $H_0^1(\Omega)$ defined in \autoref{sect-zero-boundary}.
\[
H^{-1}(\Omega) = H_0^1(\Omega)'
\]
This is the space of linear continuous functionals on $H_0^1(\Omega)$. In this we follow standard notation, cf.~\citeasnoun[3.12 and 3.13]{adams} or \citeasnoun[5.9.1]{evans}. We will usually not identify $H_0^1$ with its dual space $H^{-1}$ as it is customary for Hilbert spaces by invoking \autoref{th-riesz-frechet} (Riesz--Fréchet representation theorem). Rather it is more instructive to define the following continuous embedding of spaces.
\[
H_0^1(\Omega) \subset L^2(\Omega) \subset H^{-1}(\Omega)
\]
With the functional norm
\begin{equation}\label{eq-sobolev-dual-norm}
\|f\|_{H^{-1}(\Omega)} = \sup\{ f(u) \mid u \in H_0^1(\Omega), \|u\|_{1,2} = 1 \}
\end{equation}
the dual becomes a reflexive Banach space. The following theorem from \citeasnoun[5.9., Th.~1]{evans} gives a characterisation of $H^{-1}$.

\begin{theorem}\label{th-sobolev-dual-element}
For all $f \in H^{-1}(\Omega)$ there are functions $f_0, f_1, \ldots, f_n \in L^2(\Omega)$ such that for all $v \in H_0^1(\Omega)$
\[
f(v) = \langle f_0,v \rangle + \sum_{i=1}^n \langle f_i,\partial_i v \rangle.
\]
Furthermore the dual norm is given by an infimum taken over all such representatives.
\[
\|f\|_{H^{-1}(\Omega)} = \inf\left\{ \sum_{i=1}^n \|f_i\|_2^2 \right\}^{1/2}
\]
\end{theorem}

The typical notation conjuring up the weak derivative will then be $f = f_0 - \sum_{i=1}^n \partial_i f_i$ for any element of $H^{-1}$. The embedding $L^2(\Omega) \subset H^{-1}(\Omega)$ is now clearly the identification of $f_0 \in L^2(\Omega)$ with $f=f_0 \in H^{-1}(\Omega)$.

A way to construct an isomorphism between $H_0^1$ and $H^{-1}$ that is actually possible as a consequence of the Riesz--Fréchet representation theorem is through the Lax--Milgram theorem (\autoref{lax-milgram}) discussed in the next section. The bilinear form $Q(u,v) = \langle \nabla u, \nabla v \rangle$ on $H_0^1(\Omega) \times H_0^1(\Omega)$ would then give the required mapping $H^{-1} \simeq H_0^1$ through solving $Q(u,v) = f(u)$ for all $u \in H_0^1$ for a given $f \in H^{-1}$.

\subsection{The Lax--Milgram theorem}

\begin{theorem}[Riesz--Fréchet representation theorem]\label{th-riesz-frechet}
For every continuous (bounded) linear functional $f$ on a Hilbert space $\H$ there is a unique $v_f \in \H$ such that $f(u) = \langle u, v_f \rangle$ for all $u \in \H$.
\end{theorem}

\begin{proof}\cite{griffel}\\
Set $N = \{ u \in \H \mid f(u) = 0 \}$. Since $f$ is continuous and linear, this null space $N$ is a closed linear subspace. If $N = \H$ then $f = 0$ and $v_f = 0$ gives $f(u) = \langle u, v_f \rangle$ for all $u$.\\
If $N \neq \H$ we show that the orthogonal complement $N^\perp$ is one-dimensional which means that every pair of vectors in $N^\perp$ is linearly dependent. Take $x,y \in N^\perp$ then $ax+by \in N^\perp$ for all numbers $a,b$. Setting $z = f(y) x - f(x) y$ we have such a linear combination in $N^\perp$ but one gets $f(z) = 0$ and therefore $z \in N$ which means $z=0$. Having showed that there are $a,b \neq 0$ yielding $ax+by=0$ for all $x,y \in N^\perp$ we have $N^\perp$ one-dimensional.\\
Choose $e \in N^\perp$ with $\|e\|=1$ and we can write every $x\in N^\perp$ as $x = \langle x,e \rangle e$. Any $u \in \H$ can be disjointed into $u = x+w$, $x \in N^\perp$ and $w \in N$.\footnote{For such a unique orthogonal decomposition to exist one actually needs $N$ to be a complete linear subspace of $\H$, which is automatically true for $N$ closed in $\H$ complete. The proof rests on the approximation of elements in $\H$ by elements in complete, convex subsets like $N$. Thus the significance of the representation theorem can be seen as a reason why the (physical) Hilbert space needs to be complete. This issue has been discussed in a thread on \emph{StackExchange Mathematics}:
\url{http://math.stackexchange.com/questions/136756/physical-quantum-mechanical-significance-of-completeness-of-hilbert-spaces} (Apr 25, 2012)
and a thread on \emph{StackExchange MathOverflow}:
\url{http://mathoverflow.net/questions/35840/the-role-of-completeness-in-hilbert-spaces} (Aug 17, 2010).} Then $f(u) = f(x) + f(w) = \langle x,e \rangle f(e) + 0 = \langle u-w,e \rangle f(e) = \langle u,e \rangle f(e)$. Thus $f(u) = \langle u, v_f \rangle$ as required, where $v_f = f(e)^* e$.\\
To prove uniqueness, suppose $f(u) = \langle u, v_f \rangle = \langle u, v_f' \rangle$ for all $u\in\H$. Taking $x = v_f - v_f'$ gives $\langle v_f - v_f',v_f \rangle = \langle v_f - v_f',v_f' \rangle$, hence $\langle v_f - v_f',v_f - v_f' \rangle = 0$ and thus finally $v_f = v_f'$.
\end{proof}

This can be carried over to bilinear (or sesquilinear) forms, which can always be represented by bounded linear operators. We first need the following.

\begin{definition}[coercive and continuous bilinear forms]\label{def-coercive-continuous}
\hfill\\
A bilinear form $Q$ on a Hilbert space $\H$ is said to be \textbf{coercive} (bounded below) if there exists a constant $c > 0$ such that $Q(u,u) \geq c \|u\|^2$ for all $u \in \H$. It is \textbf{continuous} (bounded) if we find a $C >0$ such that $Q(u,v) \leq C \|u\| \cdot \|v\|$ for all $u,v \in \H$.
\end{definition}

The notion \q{coercive} means that such forms are forced to grow with at least the same rate as the norm. Coercive and continuous bilinear forms induce an equivalent norm in the Hilbert space they are defined on. It is conversely true that any symmetric bilinear form that induces an equivalent norm is coercive as well as continuous. For continuity the following lemma shows this reversion.

\begin{lemma}
Let $Q$ be a symmetric bilinear form on a Hilbert space $\H$ and $C > 0$ such that for all $u \in \H$ it holds $Q(u,u) \leq C \|u\|^2$. Then $Q$ is continuous.
\end{lemma}

\begin{proof}
Take any $\lambda > 0$,
\begin{align*}
|Q(u,v)| &= \frac{1}{4} |Q(u/\lambda + \lambda v, u/\lambda + \lambda v) - Q(u/\lambda - \lambda v, u/\lambda - \lambda v)| \\
&\leq \frac{C}{4} \left( \|u/\lambda + \lambda v\|^2 + \|u/\lambda - \lambda v\|^2 \right) \\
&\leq \frac{C}{2} \left( \|u\|/\lambda + \lambda\|v\| \right)^2.
\end{align*}
Now choose $\lambda = \sqrt{\|u\|/\|v\|}$ to get the condition for continuity.\footnote{This proof is from Robert Israel on \emph{StackExchange Mathematics}:\\
\url{http://math.stackexchange.com/questions/157513/a-continuity-condition-for-a-bilinear-form-on-a-hilbert-space} (Jun 12, 2012).}
\end{proof}

\begin{theorem}[representation theorem for bilinear forms]\label{th-representation-biliear-forms}
\hfill\\
For every continuous (bounded) bilinear form $Q$ on a Hilbert space $\H$ there is a unique bounded linear operator $A: \H \rightarrow \H$ such that $Q(u,v) = \langle u,Av \rangle$ for all $u,v \in \H$.
\end{theorem}

\begin{proof}\cite{achieser-glasmann}\\
Uniqueness can be easily demonstrated, as from $\langle u,Av \rangle = \langle u,Bv \rangle$ for all $u,v \in \H$ it follows $\langle u,(A-B)v \rangle=0$ and thus $A=B$. To show that the representation exists, we choose $v \in \H$ as fixed. $Q(\cdot,v)$ is now a continuous linear functional and by \autoref{th-riesz-frechet} there is a unique $w \in \H$ such that $Q(u,v) = \langle u,w \rangle$ for arbitrary $u \in \H$. Next we define $w = A v$ and have established $Q(u,v) = \langle u,Av \rangle$.\\
Linearity follows directly from the linearity of $Q$ in its second argument. For boundedness we use CSB
\[
Q(u,v) = \langle u,Av \rangle \leq \|u\| \cdot \|Av\|,
\]
with which
\[
\|Q\| = \sup \frac{|Q(u,v)|}{\|u\| \cdot \|v\|} \leq \sup \frac{\|Av\|}{\|v\|} = \|A\|.
\]
On the other side, by the reduction of the scope of the supremum to $u = Av$ we have
\[
\|Q\| = \sup \frac{|\langle u,Av \rangle|}{\|u\| \cdot \|v\|} \geq \frac{|\langle Av,Av \rangle|}{\|v\| \cdot \|Av\|} = \frac{\|Av\|}{\|v\|}.
\]
Therefore $A$ is not only bounded but has the definite norm $\|A\| = \|Q\|$.
\end{proof}

Further the representing operator $A$ for $Q$ from above can be shown to be invertible, which will then be important for the Lax--Milgram theorem. The following results which are frequently useful in the theory of elliptic partial differential operators are taken from the book of \citeasnoun[6.3]{blanchard-bruening}.

\begin{corollary}\label{cor-bilinear-form-inverse}
In the case of $Q$ as in \autoref{th-representation-biliear-forms} being coercive (with some constant $c>0$) the representing operator $A$ is bijective and has a bounded linear inverse operator $A^{-1} : \H \rightarrow \H$ with norm $\|A^{-1}\| \leq c^{-1}$.
\end{corollary}

\begin{proof}
Coerciveness yields (with $A^*$ the adjoint of $A$)
\begin{align*}
&c \|u\|^2 \leq Q(u,u) = \langle u,Au \rangle \leq \|u\| \cdot \|Au\| \\
&c \|u\|^2 \leq Q(u,u) = \langle A^*u,u \rangle \leq \|A^*u\| \cdot \|u\|
\end{align*}
thus
\[
c \|u\| \leq \|Au\|,\; \|A^*u\|
\]
and in particular $\ker A = \ker A^* = \{ 0 \}$ which means both operators are injective. The relation $(\ran A)^\perp = \ker A^*$ (cf.~\autoref{lemma-neumann}) used on the closure of $\ran A$ tells us
\[
\overline{\ran A} = (\ran A)^{\perp\perp} = \{0\}^\perp = \H.
\]
Taking a convergent sequence $A u_i \rightarrow w$ in $\H$ with
\[
\|A u_i - A u_j\| \geq c \|u_i - u_j\|
\]
we see that $u_i$ is a Cauchy sequence too. If $u_i \rightarrow u$ then $A u_i \rightarrow A u$ because every bounded operator is closed. So $A u = w$ tells us that the range of $A$ is already closed and thus $\ran A=\H$, $A : \H \rightarrow \H$ is bijective. The estimate $c \|u\| \leq \|Au\|$ implies
\[
\|A^{-1}v\| \leq \frac{1}{c}\|v\|
\]
which leads to the desired estimate $\|A^{-1}\| \leq c^{-1}$ in operator norm.
\end{proof}

\begin{theorem}[Lax--Milgram]\label{lax-milgram}
Let $Q$ be a coercive (with some constant $c>0$), continuous bilinear form on a Hilbert space $\H$. Then for every continuous linear functional $f$ on $\H$ there exists a unique $v \in \H$ such that $Q(u,v) = f(u)$ holds for all $u \in \H$. Moreover, the solution depends continuously on the given datum
\[
\|v\| \leq \frac{1}{c} \|f\|.
\]
\end{theorem}

\begin{proof}
With the above the line of proof is quite obvious. Just set $v = A^{-1} v_f$, $A^{-1}$ relating to \autoref{cor-bilinear-form-inverse} and $v_f$ being the representing vector from the Riesz--Fréchet representation theorem (\autoref{th-riesz-frechet}).
\[
Q(u,A^{-1} v_f) = \langle u, A A^{-1} v_f \rangle = \langle u, v_f \rangle = f(u)
\]
The estimate for the norm of the solution $\|v\|$ follows right away.
\end{proof}

Note that if $Q$ is also symmetric, as it is usually the case here, a quicker proof is possible by making $Q(\cdot,\cdot)$ as new inner product on $\H$ and applying the Riesz--Fréchet representation theorem (\autoref{th-riesz-frechet}) directly. Consequently the Lax--Milgram theorem is primarily significant on occasions where $Q$ is not symmetric. \cite[6.2.1]{evans}

Through the representing operator of such bilinear forms we can formulate an associated eigenvalue problem of the form $Q(u,u_\lambda) = \lambda \langle u,u_\lambda \rangle_{\H_2}$ where the inner product is now that of a larger Hilbert space, usually $L^2$, than the one $Q$ is defined on. The complete solution to this so-called \emph{general eigenvalue problem} can be found in \citeasnoun[Th.~6.3.4]{blanchard-bruening} and will be briefly repeated here without proof for further use in \autoref{sect-appl-lax-milgram}. The proof relies on the spectral theorem for compact self-adjoint operators, see e.g.~\citeasnoun[Th.~6.1.1]{blanchard-bruening}.

\begin{theorem}\label{th-lax-milgram-eigenvalue}
Let $\H_1$ and $\H_2$ be two real, infinite-dimensional Hilbert spaces where $\H_1$ is dense and compactly embedded in $\H_2$. $Q$ be a symmetric, continuous, and coercive bilinear form on $\H_1$, then there is a monotone increasing sequence $(\lambda_m)_{m \in \N}$ of real eigenvalues
\[
0 < \lambda_1 \leq \lambda_2 \leq \lambda_m \stackrel{m \rightarrow \infty}{\longrightarrow} \infty
\]
and an orthonormal basis $\{e_m\}_{m \in \N} \subset \H_1$ of $\H_2$ such that for all $m \in \N$ and $u \in \H_1$
\[
Q(u,e_m) = \lambda_m \langle u,e_m \rangle_{\H_2}.
\]
The sequence $\{\lambda_m^{-1/2}e_m\}_{m \in \N}$ is then an orthonormal basis of $\H_1$ with respect to the inner product $\langle\cdot,\cdot\rangle = Q(\cdot,\cdot)$.
\end{theorem}

To allow for an even bigger domain of inhomogeneities $f$ the following non-symmetric extension of the Lax--Milgram theorem is interesting. The condition of continuity gets relaxed to include only the Hilbert space valued argument.

\begin{theorem}[Lax--Milgram--Lions]\label{lax-milgram-lions}
Let $\H$ be a Hilbert space and $\Phi$ a normed linear space, $Q : \H \times \Phi$ a bilinear form with $Q(\cdot,\varphi)$ in $\H'$ for all $\varphi \in \Phi$. Instead of coercivity we further demand there is a $c>0$ such that
\begin{equation}\label{eq-lax-milgram-lions-1}
\adjustlimits \inf_{\|\varphi\|_\Phi=1} \sup_{\|v\|_\H \leq 1} |Q(v,\varphi)| \geq c.
\end{equation}
%
Then for all $f \in \Phi'$ there is a solution $v \in \H$ such that $Q(v,\varphi)=f(\varphi)$ holds for all $\varphi \in \Phi$ that fulfils the estimate
\[
\|v\|_\H \leq \frac{1}{c} \|f\|_{\Phi'}.
\]
\end{theorem}

\begin{proof}\cite[Th.~III.2.1, Cor.~III.2.1]{showalter}\\
Define the linear (but not necessarily continuous) representing operator $A:\Phi \rightarrow \H'$ by $Q(v,\varphi) = \langle v,A\varphi \rangle$ then \eqref{eq-lax-milgram-lions-1} is equivalent to $\|A(\varphi)\|_{\H'} \geq c\|\varphi\|_\Phi$ for $\varphi \in \Phi$. This already implies that $A$ is invertible by injectivity and $A^{-1} : \ran A \rightarrow \Phi$  linear and uniformly continuous from the estimate before. This means there is always a unique linear, uniformly continuous extension to the respective closures
\[
\overline{A^{-1}} : \overline{\ran A} \rightarrow \overline{\Phi}.
\]
Now take $f \in \Phi'$, then $Q(v,\varphi)=f(\varphi)$ for all $\varphi \in \Phi$ holds exactly if $\langle v,g \rangle = f(A^{-1}g)$ for all $g \in \ran A$. If we define $P:\H' \rightarrow \overline{\ran A}$ the orthogonal projection the same follows from
\[
\langle v,g \rangle = f(\overline{A^{-1}}Pg), \quad g \in \H'.
\]
The space for $g$ is now big enough to allow for a duality argument by defining the (continuous) dual $(\overline{A^{-1}}P)^* : \Phi'=\overline{\Phi}' \rightarrow \H$ and apply it to $f$ which gives the solution $v = (\overline{A^{-1}}P)^* f$. The required estimate readily follows from an estimate for the operator $\overline{A^{-1}}P$ that is the same as for its dual.
\end{proof}

Note that $\ran A$ is not necessarily dense in $\H'$ and we had to introduce an orthogonal projection $P:\H' \rightarrow \overline{\ran A}$ to define a solution. This means other solutions outside the such defined subspace might as well be possible. So in contrast to the original Lax--Milgram theorem we only get existence and not uniqueness of a solution. Also the estimate for solutions need not hold for \emph{all} possible solutions, just for at least one.

Finally a hint should be given towards an additional theorem that could prove useful in the presented context. Céa's lemma \cite{cea} shows how far solutions to the Lax--Milgram theorem that are found with respect to a finite-dimensional subspace of the original Hilbert space can deviate from the original solutions. This is used to have an upper bound on errors due to numerical modelling of such problems.

\chapter{Schrödinger Dynamics}
\vspace{-0.4cm}

\begin{xquote}{\citeasnoun{simon-1971}}
After one has established the connection between quantum mechanics and Hilbert space objects, a host of nontrivial mathematical questions arise.
\end{xquote}

\begin{xquote}{\citeasnoun{reed-simon-4}}
Nonrelativistic quantum mechanics is often viewed by physicists as an area whose qualitative structure, especially on the level treated here, is completely known. [...] But, in our opinion, theoretical physics should be a science and not an art and, furthermore, one does not fully understand physical fact until one can derive it from first principles.
\end{xquote}

\begin{xquote}{\citeasnoun{klainerman}}
All this seems to point to the fact that the further development of the established physical theories ought to be viewed as a genuine and central goal of Mathematics itself. In view of this I think we need to reevaluate our current preconception about what subjects we consider as belonging properly within Mathematics. We may gain, consistent with Poincaré's point of view, considerably more unity by enlarging the boundaries of Mathematics to include, on equal footing with all other more traditional fields, physical theories such as Classical and Quantum Mechanics and Relativity Theory, which are expressed in clear and unambiguous mathematical language. We may then develop them, if we wish, on pure mathematical terms asking questions we consider fundamental, which may not coincide, at any given moment, with those physicists are most interested in, and providing full rigor to 
our proofs.
\end{xquote}

As we are mainly concerned with existence and uniqueness results in the density-potential mapping it is crucial to understand the underlying Cauchy problem of the standard Schrödinger equation and to prove existence and uniqueness of its solutions as well. Usually this question is left aside in quantum mechanics, the area credulously being thought of as exhaustively studied, already assuming unique solvability and continuous dependence of the results on the given data for any PDE given. But---as expressed by Reed and Simon in the quotation above---such answers cannot usually be given trivially. The methods developed here will also provide indispensable tools for the study of the inverse problem going from densities, given by solutions of the Schrödinger equation, back to external potentials determining those solutions. Such potentials will be time-dependent and are acting as a \emph{force} to control the density, hence the title \emph{Schrödinger dynamics}. This will become most visible in the pivotal \q{divergence of local forces equation} \eqref{eq-div-force-density}.

In the whole chapter the Hamiltonian $H : \H \supseteq D(H) \rightarrow \H$ is a linear, self-adjoint, and densely defined operator on a possibly multipartite Hilbert space $\H = L^2(\R^n)$ with norm $\|\cdot\|_2$. Even in the case of a time-dependent $H(t)$ we demand a common domain $D(H)$ for all $t$. This will later be guaranteed by the Kato--Rellich theorem (\autoref{th-kato-rellich}) and its numerous consequences. In the case of Schrödinger dynamics examined here, the Hamiltonian is given as $H = -\onehalf\Delta + v$ with $v$ a real, in general time-dependent function, which also amounts for all interactions between quantum particles here. We write $H_0 = -\onehalf\Delta$ for the so-called free Hamiltonian.

As an excellent reference to the topic in the more general framework of Banach spaces we refer to \citeasnoun{pazy}. Like there we also start in a Banach space setting, denoting $A$ as the generator of evolution, but soon move on to the Hilbert space case with Hamiltonian $H$.

\section{The abstract Cauchy problem}

The abstract Cauchy problem for an evolution equation on a Banach space $X$ is given as the following initial value problem.
\begin{equation}\label{cauchy-problem}
\begin{aligned}
    \dfrac{\d u(t)}{\d t} &= A u(t), \quad t > 0 \\
    u(0) &= x
\end{aligned}
\end{equation}
We want to give exact conditions for the existence of a uniquely defined trajectory in $X$ belonging to a given initial state $x \in X$ and steered by the system's properties completely represented by the operator $A$. This trajectory is supposed to be a solution to the Cauchy problem \eqref{cauchy-problem} if it lies in $D(A)$, but we will see later that by deriving a bounded evolution operator acting on the initial state, it can be extended to the whole Banach space $X$ as a  \emph{generalised solution}.

\subsection{Evolution semigroups}
\label{sect-semigroups}

A most useful notion is that of a semigroup of operators that yields the evolution of the system's state from one time to another. Applied to the initial state and evaluated at all times this gives the whole solution to the Cauchy problem. In the case of the (autonomous) Schrödinger equation with time-independent Hamiltonian these will be the familiar unitary evolution operators $U(t)$. For the time being we will stay in the more abstract context of evolution equations on Banach spaces.

\begin{definition}\label{def-semigroup}
A one parameter family $T(t), 0 \leq t < \infty$, of bounded linear operators from $X$ into $X$ is a \textbf{strongly continuous semigroup of bounded linear operators on $X$} (or simply a $\Cont^0$ semigroup)~if
\begin{enumerate}[(i)]
	\itemsep0em
	\item $T(0) = \id$,
	\item $T(t+s) = T(t) T(s)$ for every $t,s \geq 0$ and
	\item $\lim_{t \searrow 0} \|T(t) x - x\| = 0$ for all $x \in X$.
\end{enumerate}
\end{definition}

The restriction to a semigroup instead of defining a full group that would allow all $t \in \R$ is due to demands from parabolic PDEs like the heat equation which might be solvable only forward in time. Indeed it is unnecessary in the case of the Schrödinger equation but shall still be adopted here because the considered time interval is usually of the kind $[0,T]$ or $[0,\infty)$ anyway. The notion \q{strongly continuous semigroup} is justified because it can be shown  that such a $T(t)$ is bounded by $\|T(t)\| \leq M \e^{\omega t}$ with constants $\omega \geq 0, M \geq 1$ and thus because of (ii) is also strongly continuous in $t$. \cite[Th.~I.2.2]{pazy} If even $\|T(t)\| \leq 1$ it is called a \emph{contraction semigroup}.

\begin{definition}\label{def-inf-generator}
The \textbf{infinitesimal generator} $A$ of a $\Cont^0$ semigroup $T$ is the linear operator defined by
\begin{equation}\label{eq-inf-generator}
Ax = \lim_{t \searrow 0} \frac{T(t)x - x}{t}
\end{equation}
on its domain $D(A)$, i.e., all $x\in X$  where this limit exists.
\end{definition}

We noticed already in \autoref{sect-unbounded} that being densely defined is a very elementary characteristic of linear operators and indeed we show in \autoref{th-generator-closed} below that an infinitesimal generator always shares this feature and is also closed.

Because of the obvious relation of \eqref{eq-inf-generator} to the time-derivative of $T(t) x$ we might formally write the semigroup as an operator exponential.
\[
T(t) = \e^{tA}
\]
This makes sense because it can be shown \cite[Th.~I.2.6]{pazy} that \emph{vice versa} the infinitesimal generator uniquely defines the whole $\Cont^0$ semigroup. Considering the case of the Schrödinger equation we have $A = -\i H$ and $T(t) = U(t) = \e^{-\i H t}$. In the case of bounded Hamiltonians this exponential is not only formally defined but is really an exponential series that converges, see \autoref{sect-bounded-H}.

Theorems that show the existence of a $\Cont^0$ semigroup by demanding specific properties from its generator are called \q{generation theorems}. An example is \autoref{th-schro-dyn-konst} (Stone's Theorem) and the Theorems of Hille--Yoshida \cite[Th.~I.3.1]{pazy} and Lumer--Phillips \cite[Th.~I.4.3]{pazy} in the more general setting of Banach spaces. The important relation between the $\Cont^0$ semigroups and Cauchy problems is given by the following theorem which can be found for example in \citeasnoun[Th.~I.2.4]{pazy} or \citeasnoun[Lem.~12.11]{renardy-rogers}.

\begin{theorem}\label{semigroup-solution}
Let $T(t)$ be a $\Cont^0$ semigroup and $A$ its infinitesimal generator. Then it holds for $t>0$ that
\begin{enumerate}[(i)]
	\itemsep0em
	\item if $x \in X$ then
	\[
	\int_0^t T(s) x \d s \in D(A) \mtext{and} A \int_0^t T(s) x \d s = T(t) x - x,
	\]
	\item if $x \in D(A)$ then
	\[
	T(t) x \in D(A) \mtext{and} \frac{\d}{\d t} T(t)x = A T(t) x = T(t) A x.
	\]
\end{enumerate}
\end{theorem}

Part (i) will lead to the definition of a mild solution while (ii) shows the existence of a classical solution to the Cauchy problem, see \autoref{sect-types-solutions}.

\begin{proof}
(i) Take $h >0$, then
\begin{align*}
\frac{T(h)-\id}{h} \int_0^t T(s) x \d s &= \frac{1}{h} \int_0^t (T(s+h)x-T(s)x) \d s \\
&= \frac{1}{h} \int_t^{t+h} T(s) x \d s - \frac{1}{h} \int_0^h T(s) x \d s
\end{align*}
and with $h \searrow 0$ the right-hand side goes to $T(t)x-x$.\\
(ii) Now because of boundedness of $T(t)$ we have as $h \searrow 0$
\[
\frac{T(h)-\id}{h} T(t) x = T(t) \frac{T(h)-\id}{h} x \longrightarrow T(t) A x
\]
thus $T(t)x \in D(A)$ and $AT(t)x = T(t)Ax$ as well as the right derivative of $T(t)x$ fulfilling
\[
\frac{\d^+}{\d t} T(t)x = A T(t) x = T(t) A x.
\]
To conclude we have to show the same for the left derivative.
\begin{align*}
&\lim_{h \searrow 0} \frac{T(t)x-T(t-h)x}{h} - T(t)Ax \\
&= \lim_{h \searrow 0} T(t-h) \left( \frac{T(h)x-x}{h} - Ax \right) + \lim_{h \searrow 0} (T(t-h)Ax - T(t)Ax)
\end{align*}

Both limit terms vanish, the first due to $x \in D(A)$ and boundedness of $T(t-h)$, the second by strong continuity of $T(t)$.
\end{proof}

The next theorem from \citeasnoun[Cor.~I.2.5]{pazy} and \citeasnoun[Th.~12.12]{renardy-rogers} shows special properties of such generators.

\begin{theorem}\label{th-generator-closed}
Let $A$ be the infinitesimal generator of a $\Cont^0$ semigroup $T(t)$. Then $A$ is always densely defined and closed. 
\end{theorem}

\begin{proof}
Like before it holds
\[
x = \lim_{h \searrow 0}\frac{1}{h} \int_0^h T(s) x \d s
\]
and by \autoref{semigroup-solution} (i) the argument of the limit is in $D(A)$. Hence any $x \in X$ can be approached by a limit sequence and $D(A)$ is dense in $X$. Assume now such a sequence $x_n \in D(A), x_n \rightarrow x$ and $Ax_n \rightarrow y$. Then integrating out part (ii) over the time interval $[0,h]$ from \autoref{semigroup-solution} we can write
\[
T(h)x_n-x_n = \int_0^h T(s) A x_n \d s.
\]
In the limit $n \rightarrow \infty$ this yields
\[
T(h)x - x = \int_0^h T(s) y \d s.
\]
Finally we divide by $h$ and let $h \searrow 0$ and get
\[
Ax = \lim_{h \searrow 0} \frac{1}{h} \int_0^h T(s) y \d s = y.
\]
This shows closedness.
\end{proof}

\subsection{Types of solutions}
\label{sect-types-solutions}

There are several ways to approach the concept of a solution to the abstract Cauchy problem \eqref{cauchy-problem}, we will start with the most basic notion. By $\R_+$ we denote the unbounded interval $[0,\infty)$, itself an additive semigroup.

\begin{definition}\label{def-classical-sol}\cite[sect.~4.1]{pazy}\\
By a \textbf{classical} or \textbf{strong} solution to \eqref{cauchy-problem} we mean a continuous map $u: \R_+ \rightarrow X$, continuously differentiable and $u(t) \in D(A)$ for $t > 0$ such that \eqref{cauchy-problem} is satisfied.
\end{definition}

We state a small lemma for Cauchy problems with initial state $x \in D(A)$ which always leads to a classical solution.

\begin{lemma}
If $x \in D(A)$ the Cauchy problem \eqref{cauchy-problem} has a classical solution in the space $\Cont^1(\R_+,X) \cap \Cont^0(\R_+,D(A))$ where the domain of the generator $D(A)$ gets equipped with the graph norm.
\end{lemma}

\begin{proof}
Because of classicality we have $T(t)x \in D(A)$, this time for all $t \in \R_+$. By \autoref{semigroup-solution} $\d/\!\d t\, T(t)x = AT(t)x = T(t)Ax$ and $\lim_{t \searrow 0} \|T(t) Ax - Ax\| = 0$ by the $\Cont^0$ property from \autoref{def-semigroup}. Thus we get $T(\cdot)x \in \Cont^1(\R_+,X)$. To test continuity with respect to the graph norm we have to check the limit $\lim_{t \searrow 0} \|A (T(t)x - x)\| = 0$ which is just the same as above. Both results extend from $t=0$ to all $t \geq 0$ by the semigroup property $T(t+s) = T(t) T(s)$.
\end{proof}

We will see later that in case of the Schrödinger equation due to unitarity of the evolution operator a classical solution demands for the initial state $\psi_0 \in D(H)$, making the notion of a generalised solution important for all trajectories associated to a $\psi_0 \in \H \setminus D(H)$. Thus the statement \q{for $t > 0$} in \autoref{def-classical-sol} may seem strange because classicality depends on the initial state being in $D(H)$ in that case and \q{for $t \geq 0$} would seem more concise. But it stays sensible in the general theory of evolution equations, for example the heat equation, where a non-differentiable initial state evolves into a infinitely differentiable state and classical solution after an infinitesimally short time. Yet it is absolutely possible to apply the evolution semigroup on an initial state $x \in X$ even if no classical solution arises and thus create a \emph{generalised} solution $T(t)x$ of the corresponding abstract Cauchy problem.

\begin{definition}\cite[p.~105]{pazy}\\
By a \textbf{generalised} solution to \eqref{cauchy-problem} we mean a continuous map $u: \R_+ \rightarrow X$, given as $u(t) = T(t)x$ where $T(t)$ is a $\Cont^0$ semigroup with infinitesimal generator $A$.
\end{definition}

If $X$ is a Hilbert space as in the Schrödinger case the inner product structure allows such a solution to fulfil a \q{weak} version of the Cauchy problem. \cite[p.~284]{reed-simon-2}
\[
\partial_t \langle \eta,u(t) \rangle = \langle A^*\eta,u(t) \rangle \mtext{for all} \eta \in D(A^*)
\]
In \citeasnoun{pazy} this kind of solution is called a \emph{mild} solution but we will use a different terminology here which then will soon turn out to be equivalent in the $\Cont^0$ semigroup case.

\begin{definition}\label{def-mild-solution}\cite[Def.~3.1.1]{arendt}\\
A \textbf{mild} solution to \eqref{cauchy-problem} is a continuous map $u: \R_+ \rightarrow X$ satisfying the integral version of \eqref{cauchy-problem}, i.e., for all $t>0$
\[
A \int_0^t u(s) \d s = u(t) - x \mtext{while} \int_0^t u(s) \d s \in D(A).
\]
\end{definition}

It is clear from \autoref{semigroup-solution} (i) that we can construct mild solutions for all initial states $x \in X$ if the evolution is given by a $\Cont^0$ semigroup. Thus in that case the concept of a mild and a generalised solution is equivalent. Of course any classical solution is also a generalised (mild) solution and the converse is true, if the generalised (mild) solution is continuously differentiable for $t \neq 0$.

The true reason why we rather say \emph{generalised} is that mild solutions for time-dependent $H$ will later be discussed in connection with evolution systems (see \autoref{sect-full-int-pic}) but will be related especially to inhomogeneous Cauchy problems with time-independent $A$ there. \citeasnoun[Def.~12.15]{renardy-rogers} and \citeasnoun[Def.~IV.2.3]{pazy} use the term \emph{mild solution} for such cases which justifies its use in the later chapters for us. We gave the above version of mildness just for the sake of completeness.

If we have $X$ as a Hilbert space again, a truly \emph{weak} (\emph{faible}) notion of solutions is possible, where also the time derivative is taken in a weak sense.

\begin{definition}\cite[I.2]{lions-book}\\
A \textbf{weak} solution to \eqref{cauchy-problem} is a map $u: \R_+ \rightarrow X$ satisfying the weak version of \eqref{cauchy-problem}, i.e., for $x \in X$ and all $\varphi \in \Cont^1(\R_+,X)$ with $\lim_{t \rightarrow \infty}\varphi(t)=0$ and $A^*\varphi \in \Cont^0(\R_+,X)$
\[
\int_0^\infty \Big( \langle \varphi'(t),u(t) \rangle + \langle A^*\varphi(t),u(t) \rangle \Big)\d t = \langle \varphi(0),x \rangle.
\]
\end{definition}

In the following examinations we always have $X=\H, A = -\i H$, and the time interval limited from $\R_+$ to a bounded $[0,T]$. We give a short tabular overview of the types of solutions and their respective trajectory spaces that can arise then.

\begin{center}
\begin{tabular}{ll}\label{tab-solution-types}
classical / strong & $\Cont^1([0,T], \H) \cap \Cont^0([0,T], D(H))$ \\
generalised / mild & $\Cont^0([0,T], \H)$ \\ 
weak & $L^2([0,T], \H)$
\end{tabular}
\end{center}

\section{The Schrödinger initial value problem}

As we are only concerned with the Schrödinger case the talk about $\Cont^0$ semigroups will cease and we will be mainly concerned with explicitly constructing unitary one-parameter evolution semigroups $U(t)$ and later in the case of time-dependent Hamiltonians two-parameter evolution systems $U(t,s)$. Because the former are in fact special cases of contraction semigroups all that has been said remains true. The proofs for existence of solutions to the Schrödinger equation will be given in the following sections for different cases, but in firstly assuming existence we can already prove uniqueness. We start by stating the Cauchy problem again for the quantum mechanical case.

The Schrödinger initial value problem is the Cauchy problem
\begin{equation}\label{cauchy-problem-se}
\begin{aligned}
    \i \partial_t \psi(t) &= H(t) \psi(t), \quad t \in [0,T] \\
    \psi(0) &= \psi_0 \in \H
\end{aligned}
\end{equation}
with $H(t)$ self-adjoint.

In general we look for solutions in the form of a unitary evolution system $U(t,s)$ (see \autoref{def-evolution-system}) applied to the initial state $\psi_0 \in \H$ that can therefore be of generalised or classical type. In both cases we speak of a \emph{Schrödinger solution} as the solution to the Schrödinger equation above.

\subsection{Uniqueness of solutions}

\begin{lemma}\label{lemma-uniqueness}
A solution to the Schrödinger equation \eqref{cauchy-problem-se} in the classical or generalised sense is always unique and the corresponding evolution operator $U(t,s)$ is necessarily unitary.
\end{lemma}

\begin{proof}
For uniqueness of a classical solution $\psi(t) = U(t,0)\psi_0$ we have to show that there can be no different direct solution to the Cauchy problem \eqref{cauchy-problem-se} $\psi'(t) \in D(H)$ with the same initial condition $\psi_0 \in D(H)$. Define the difference $\varphi = \psi - \psi'$ which also has to solve the Schrödinger equation because it is linear, this time with initial value $\varphi(0)=0$. Inserting the Schrödinger equation we get
\[
\partial_t \|\varphi(t)\|_2^2 = 2\Re\langle \varphi(t),\partial_t \varphi(t) \rangle = 2\Re\langle \varphi(t),-\i H(t) \varphi(t) \rangle = 0.
\]
This necessarily vanishes because from self-adjointness of $H(t)$ it follows that
\[
\langle \varphi(t), H(t) \varphi(t) \rangle \in \R.
\]
Thus $\partial_t \|\varphi(t)\|_2^2 = 0$ and the difference $\varphi(t)$ stays constantly zero which means $\psi(t) = \psi'(t)$ for all $t \in [0,T]$. Showing uniqueness for a dense subset $D(H) \subset \H$ of initial values suffices for uniqueness on the whole Hilbert space $\H$ because the involved evolution operator is bounded. Thus the case is also settled for generalised solutions. The equality of norms
\[
\|\psi(t)\|_2 = \|\psi(s)\|_2 = \|\psi(0)\|_2
\]
for arbitrary solutions also tells us that the evolution operator $U(t,s)$ is unitary.
\end{proof}

We want to not only give conditions on the quantum system, that is on the Hamiltonian and therefore on the potential $v$, but derive explicit expressions for these solutions in different settings. Those will gradually increase in difficulty, from the easiest case of time-independent (static) and bounded operators to full-fledged quantum dynamics with even singular potentials.

\subsection{Conservation laws and energy spaces}

Any solution to the Schrödinger equation conserves the $L^2$-norm of the wave function (in other words the \q{charge} or \q{mass}). This was shown in the proof of \autoref{lemma-uniqueness} above.

Owing to the usual physical tradition another conserved quantity for classical solutions is the energy, but only in case of time-independent (static) Hamiltonians $H=-\Delta +v$. The energy expectation value itself is related to the norm of the $H^1$ Sobolev space.
\begin{align*}
\partial_t \left( \onehalf\|\nabla\psi(t)\|_2^2 + \langle \psi(t),v\psi(t) \rangle \right) &= \Re\langle \partial_t\psi(t), (-\Delta+v) \psi(t) \rangle \\
&= \Re\langle \partial_t\psi(t),\i \partial_t \psi(t) \rangle =0
\end{align*}
Such conserved quantities give rise to so-called \q{energy spaces} whose norms they define. Those can be used to get global existence of solutions from the local one by such \emph{a priori} estimates for the conserved quantity. \cite[3.1]{cazenave} Another tool related to the energy space is the extension of the Hamiltonian from its usual self-adjoint domain to that space, e.g.,~$H^1$. This goes under the name \q{energetic extension}.

\subsection{Spread of wave packets under free evolution}

It seems common knowledge that an initial wave packet confined to a bounded domain spreads out to all of space after an infinitesimally short time under the effect of the free Schrödinger evolution. Citing \citeasnoun{hegerfeldt-ruijsenaars}:

\begin{quote}
For a free nonrelativistic wave packet it is well known that it spreads instantaneously over all of space if it is localized in a bounded region at time $t=0$.
\end{quote}

Or after \citeasnoun{madrid}:

\begin{quote}
It is well known that if a non-relativistic particle is initially confined to a finite region of space, then it immediately develops infinite tails, as one could already expect from the lack of an upper limit for the propagation speed in non-relativistic Quantum Mechanics.
\end{quote}

We give a proof that applies the Paley--Wiener theorem of Fourier analysis. This theorem is motivated by the question for which functions $f : \R^n \rightarrow \C$ the respective Fourier transforms $\hat{f}$ are analytic.
\[
\hat{f}(k) = (2\pi)^{-n/2} \int_{\R^n} \e^{-\i k \cdot x} f(x) \d x
\]
We may easily check for analyticity with the Cauchy--Riemann conditions if differentiating under the integral sign gives a well-defined quantity. A problem arises here if for example $\Im{k_i} > 0$ and $x_i \rightarrow \infty$ because then the exponential will grow to infinity. But this is compensated if $f \in L^2$ decays fast enough for $|x| \rightarrow \infty$ (this for example is always the case if the support of $f$ is contained in a compact set). This fact is captured in the following theorem, a version of the Paley--Wiener theorem and Theorem IX.13 in \citeasnoun{reed-simon-2}.

\begin{theorem}\label{th-paley-wiener-version}
Let $f$ be in $L^2(\R^n)$. Then $\e^{b|x|}f \in L^2(\R^n)$ for all $b<a$ if and only if $\hat{f}$ has an analytic continuation to the set $\{ k \mid |\Im k|<a \}$ with the property that for each $\eta \in \R^n$ with $|\eta|<a$, $\hat{f}(\cdot + \i \eta) \in L^2(\R^n)$ and for any $b<a$
\[
\sup_{|\eta|\leq b}\|\hat{f}(\cdot + \i \eta)\|_2 < \infty.
\]
\end{theorem}

The first consequence to draw is that a wave function $f$ of the above type has unbounded frequency bandwidth. That is because its Fourier transform $\hat{f}$ is analytic in a stripe $\{ k \mid |\Im k|<a \}$ around the real axis.\footnote{Later on we sometimes prefer the equivalent notion \q{holomorphic} if analyticity on a complex domain is meant. Finally an \emph{entire} function is everywhere holomorphic.} If we now assume it is zero on some tiny interval or just on a set with accumulation point, analyticity demands it is zero everywhere. This is clearly a contradiction, thus $\hat{f}$ has no zero set with accumulation point.

Now let us again assume that an initial state $f$ is of the above type but with compact support as well. Then the Fourier transform of the free evolution after an arbitrary time interval with length $t$ is given by
\[
\hat{g}(k) = \hat{f}(k) \cdot \e^{-\i k^2 t}.
\]
Now $\hat{f}$ fulfilled the condition of bounded $L^2$-norm even if $k$ is slightly displaced from the real axis. But is it possible that $\hat{g}$ still fulfils it such that $g$ can be of compact support as well? To these ends we have to study if the following norm is finite, where the principal variable for the $L^2$-norm of the functions is always $k \in \R^n$.
\begin{equation}\label{eq-ghat-norm}
\|\hat{g}(k + \i \eta)\|_2 = \left\|\hat{f}(k + \i \eta) \cdot \e^{-\i (k + \i \eta)^2 t}\right\|_2  = \left\|\hat{f}(k + \i \eta) \cdot \e^{2 (k \cdot \eta) t}\right\|_2
\end{equation}
We notice the growing exponential factor, but maybe it can be compensated by $\hat{f}$? We rewrite it further with the help of the following lemma (which is part of the proof of the theorem above).

\begin{lemma}
It holds $\e^{b|x|}f \in L^2(\R^n)$ for all $b<a$ if and only if $\e^{\eta\cdot x}f \in L^2(\R^n)$ for all $\eta \in \R^n$ with $|\eta|<a$.
\end{lemma}

\begin{proof}
Because of $\e^{\eta\cdot x} \leq \e^{|\eta\cdot x|} \leq \e^{|\eta|\,|x|}$ one implication readily follows.
Take now $\eta = (b, 0, \ldots, 0)$ such that $\e^{b x_1} f \in L^2$, similarly for $\eta = -(b, 0, \ldots, 0)$ we get $\e^{-b x_1} f \in L^2$ thus more generally $\e^{b |x_i|} f \in L^2$ for all $i = 1, \ldots, n$. Now using all those restrictions together we get $\e^{b \max|x_i|} f \in L^2$ wich is clearly an upper bound for $\e^{b |x|}$.
\end{proof}

Assuming \eqref{eq-ghat-norm} is bounded is thus equivalent to
\[
\left\|\hat{f}(k + \i \eta) \cdot \e^{b |k|}\right\|_2 < \infty
\]
for some $a>0$ and all $b<a$. Now setting $\eta=0$ exactly produces the conditions for \autoref{th-paley-wiener-version} if we put in $\hat{f}$ instead of $f$. This means the Fourier transform $\mathcal{F}(\hat f)(x) = \mathcal{F}^2(f)(x) = f(-x)$ is again analytic which is in contradiction with the original assumption that the initial state $f$, be it mirrored or not, is of compact support. Thus any state with compact support will evolve such that it cannot fulfil the Paley--Wiener condition, particularly it cannot have compact support and thus is spread all over all $\R^n$. This seems plausible because initial states with compact support already include modes of infinite frequency.

\citeasnoun{hegerfeldt-1998} gives a similar account for much more general situations, allowing any self-adjoint Hamiltonian that is bounded from below, not only free evolution for the classical Schrödinger equation. His result says that even in relativistic and field settings an initially strictly localised particle either forever stays in its initial region \emph{or} it cannot be localised in any bounded region at any later time thus spreading to infinity instantly! In a relativistic theory this may seem as an open contradiction to Einstein causality. A simple solution to save the speed of light is just not to allow any strict initial localisation of particles. In the case of eigenfunctions of the Hamiltonian localisation is impossible anyway as a consequence of the unique continuation property (see \autoref{sect-ucp}).

\subsection{Kovalevskaya and the Schrödinger equation}
\label{sect-kovalevskaya}

This section is largely taken from our review paper on TDDFT \cite{tddft-review} and was originally worked out mainly by Robert van Leeuwen.

We showed in the previous section that a freely evolving wave packet with compact support at time $t=0$ will have tails reaching all over space at any other time. This is not difficult to understand from a physical point of view since the Fourier components of the initially localised wave function include momenta of arbitrarily high value allowing the particle to move arbitrarily fast. This also implies that a localised wave packet enclosed in a box with hard walls will feel the presence of the boundary immediately. If we would have used other boundary conditions, such as periodic ones, then immediately after $t=0$ the time-evolution must be different from that of an unbounded configuration space. 

Clearly if one takes the solution to the Schrödinger equation \eqref{cauchy-problem-se} naively as the formal exponential
\begin{equation}\label{eq-sol-formal}
\e^{-\i H t} = \sum_{n=0}^\infty \frac{(-\i H t)^n}{n!}
\end{equation}
as in \autoref{th-schro-dyn-beschr} below that holds for bounded Hamiltonians and which is just specified by the differentiation rule, one has no information on such boundary conditions as they were neither encoded in the initial state nor in the exponential form of the time evolution operator. Further application of the rule \eqref{eq-sol-formal} to a function with compact support will anyway never lead to a non-zero value outside of the support so that a spread of the wave function is unthinkable in the first place. It is therefore no surprise that \eqref{eq-sol-formal} cannot be used to predict the time-evolution correctly.

However, not all hope is lost in applying \eqref{eq-sol-formal} because we can apply \eqref{eq-sol-formal} at least to a finite linear combination of eigenfunctions $\varphi_k$ (assuming those exist) with corresponding eigenvalues $\varepsilon_k$ of the form
\begin{equation}
\psi (0,x) = \sum_{k=1}^{N} \alpha_k \, \varphi_k (x),
\label{Psi02}
\end{equation}
since for this initial state we obtain from \eqref{eq-sol-formal} that 
\begin{equation}\label{Psi-sol-eigenfunctions}
\begin{aligned}
\psi (t,x) &= \sum_{n=0}^{\infty} \sum_{k=1}^N \alpha_k  \frac{(-\i t)^n}{n!} H^{n}  \varphi_k (x) 
 \\
&= \sum_{n=0}^{\infty} \sum_{k=1}^N \alpha_k \frac{1}{n!}  \left( -\i t \varepsilon_k \right)^n\varphi_k (x)
= \sum_{k=1}^N \alpha_k \e^{-\i t \varepsilon_k} \varphi_k (x),
\end{aligned}
\end{equation}
which is a valid solution to the Schrödinger equation. Clearly the formal approach worked in this case because as an initial state we chose a finite linear combination of eigenfunctions that already include information on the boundary conditions. We note that in the case of the free Hamiltonian $H=-\onehalf \Delta$ the eigenfunctions $\varphi_k$ are entire functions (see also \autoref{sect-ucp}).
Then clearly the solution \eqref{Psi-sol-eigenfunctions} is a real analytic (even entire) function in $t$ and $x$, i.e.,
\[
\psi (t,x) = \sum_{k,l=0}^\infty c_{kl} \, (t-t_0)^k (x-x_0)^l
\]
on the simple domain $x\in \Omega = \R$.
Could it be that the formal expression of \eqref{eq-sol-formal} would work for all real analytic initial states? Since the example of a wave function with compact support is necessarily non-analytic this would then explain in another way the failure of \eqref{eq-sol-formal}.
We already know that a function that is exactly zero on an interval of the real line cannot be
analytic unless it is the zero function and therefore any initially localised wave packet fails to be analytic.
Let us give an example which shows that the requirement of analyticity is not sufficient to make \eqref{eq-sol-formal} work. The Schr{\"o}dinger dynamics shall once more be given by the free Hamiltonian $H=-\onehalf \Delta$ on $L^2(\mathbb{R})$ and as an initial state we take a Lorentzian function
\begin{equation}
\psi (0,x) = \frac{1}{1+x^2} = \frac{\i}{2} \left( \frac{1}{x+\i} - \frac{1}{x-\i} \right) .
\label{Psi0_Lorentzian}
\end{equation}
This is an analytic function on the whole real axis with a  radius of convergence of at least 1 for any Taylor expansion of $\psi (0,x)$ in powers of $x-x_0$ around $x_0$. If we insert this initial state into \eqref{eq-sol-formal} we obtain the series
\[
\psi (t,x) = \sum_{n=0}^{\infty} \left(\frac{\i t}{2} \right)^n\frac{(2n)!}{n!} \frac{\i}{2} \left( \frac{1}{(x+\i)^{2n+1}} - \frac{1}{(x-\i)^{2n+1}} \right) .
\]
From simple convergence criteria we see that this is a divergent series for any value of $x$ and $t$. Therefore real analyticity is not a sufficient criterion to be able to apply \eqref{eq-sol-formal}. To get sufficient conditions for initial states given as power series we faithfully follow the classical derivation given by \citeasnoun[p.~22ff]{kovalevskaya} where she studies convergence of formal solutions to the heat equation as an example. Define the two time-dependent functions
\[
\psi^{(0)}(t) = \psi (t,x_0)  \quad \mbox{and} \quad \psi^{(1)} (t) = \partial_x \psi(t,x_0)
\]
for which we will assume that they are analytic. Those are new initial values if we exchange the meaning of $x$ and $t$. This is necessary because \citeasnoun{kovalevskaya} is concerned with initial-value problems where the time derivative appears in highest order. The general solution can then be written as a formal power series
\begin{equation}\label{TDSE_series}
\begin{aligned}
\psi (t,x) =& \sum_{\nu=0}^{\infty} \left( (-2\i)^{\nu}\, \partial_t^\nu \psi^{(0)}(t)\, \frac{(x-x_0)^{2 \nu}}{(2 \nu)!}  \right. \\
&+ \left.  
(-2\i)^{\nu}\, \partial_t^\nu \psi^{(1)}(t)\, \frac{(x-x_0)^{2 \nu +1}}{(2 \nu +1)!} \right),
\end{aligned}
\end{equation}
as can be checked by insertion of this expression into the free Schrödinger equation.
Let now the Taylor expansions of $\psi^{(0)} (t)$ and $\psi^{(1)} (t)$ be given by
\begin{align}
\psi^{(0)} (t) = \sum_{\nu=0}^\infty c_\nu (t- t_0)^\nu, 
\label{Psi_0_series} \\
\psi^{(1)} (t) = \sum_{\nu=0}^\infty c'_\nu (t- t_0)^\nu.  
\label{Psi_1_series}
\end{align}
Then in terms of the coefficients $c_\nu$ and $c'_\nu$ the expansion \eqref{TDSE_series} at time $t_0$ attains the form
\begin{align}
\psi (t_0,x) =  \sum_{\nu=0}^{\infty} \left( b_\nu (x-x_0)^{2 \nu} +  b'_\nu (x-x_0)^{2 \nu+1} \right)
\label{TDSE_series2}
\end{align}
where we defined
\[
b_\nu = (-2\i)^{\nu} \frac{\nu!}{(2\nu)!}  c_\nu,   \quad \quad b'_\nu =  (-2\i)^{\nu} \frac{\nu!}{(2\nu+1)!}  c'_\nu  .
\]

Since we assumed $\psi^{(0)} (t)$ and $\psi^{(1)} (t)$ to be analytic functions the expansions \eqref{Psi_0_series} and \eqref{Psi_1_series} have convergence radii $R^{(0)}$ and $R^{(1)}$ respectively.
Let $R>0$ be a radius strictly smaller than $R^{(0)}$ and $R^{(1)}$. Then since both series \eqref{Psi_0_series} and \eqref{Psi_1_series} converge within $R$ we have from the Cauchy--Hadamard theorem for the radius of convergence of a power series
\[
\frac{1}{R^{(0)}} = \limsup_{\nu \rightarrow \infty} \sqrt[\nu]{|c_\nu|} < \frac{1}{R}
\]
and the analogous expression for $c'_\nu$. This means that there is an $N \in \N$ such that for all $\nu \geq N$ it follows
\[
|c_\nu| \leq \frac{1}{R^{\nu}},  \quad \quad |c'_\nu| \leq \frac{1}{R^{\nu}} .
\]
These conditions imply that for the coefficients $b_{\nu}, b'_\nu$ and $\nu \geq N$ in \eqref{TDSE_series2} it holds
\begin{equation}\label{b_constraints}
\begin{aligned}
|b_\nu| = \frac{\nu!}{(2 \nu)!}  2^\nu |c_\nu|  \leq \frac{\nu!}{(2 \nu)!}  \left( \frac{2}{R} \right)^\nu \quad \mbox{and} \\
| b'_\nu| = \frac{\nu!}{(2 \nu+1)!} 2^\nu |c'_\nu | \leq \frac{\nu!}{(2 \nu+1)!}  \left( \frac{2}{R} \right)^\nu.
\end{aligned}
\end{equation}
One sees from standard convergence criteria that this implies that $\psi (0,x)$ as a power series in $x-x_0$ as in \eqref{TDSE_series2} has an infinite radius of convergence.
Now we understand what went wrong when we chose as initial state the Lorentzian function of \eqref{Psi0_Lorentzian}.
The function is analytic but the radius of convergence is not infinite like it would be the case for an entire function. The solution to the Schrödinger equation will then not
be time-analytic, i.e., series of the form as in \eqref{Psi_0_series} and \eqref{Psi_1_series} do not exist for this initial state.

Note that not even the requirement of infinite radius of convergence is enough, one really needs to satisfy the constraints \eqref{b_constraints}. For example the function
\[
\psi (0,x) = \sum_{\nu=0}^\infty \frac{1}{(\nu!)^{1/3}} (x-x_0)^\nu
\]
also given in \citeasnoun[p.~24]{kovalevskaya} has an infinite radius of convergence but does not satisfy the constraints \eqref{b_constraints}.
A typical example in which direct exponentiation is finally allowed is that of an initial Gaussian wave packet \cite{blinder} since the conditions \eqref{b_constraints} are fulfilled.

The issue of time non-analyticity has been raised in some papers in connection with initial states that are not differentiable at cusps. \cite{yang-burke}
However, the analysis carried out here following \citeasnoun{kovalevskaya} shows that the situation is more severe. A Taylor expansion in time for $\psi (t,x)$ does not even exist for a large class of analytic initial states without anything like cusps.

The classical Cauchy--Kovalevskaya theorem formulated in \citeasnoun{kovalevskaya} states that there is always a unique solution in the analytic class to PDEs with analytic coefficients and analytic Cauchy data. But its most basic condition is not fulfilled in the case of the Schrödinger equation, that is that the time variable appears in the highest derivative. Also Cauchy--Kovalevskaya would allow for additional non-analytic solutions, a loophole closed for many typical PDEs with the help of Holmgren's theorem that will be featured again briefly in \autoref{sect-ucp}.

\section{Schrödinger dynamics with static Hamiltonians}

By \q{static} we mean a Hamiltonian that is not time-dependent, like for a system with a constant external potential.

\subsection{Bounded Hamiltonians (operator exponential)}
\label{sect-bounded-H}

\begin{definition}\label{exp}
For a bounded linear operator $A:X \rightarrow X$ we define the \textbf{operator exponential}
\[
\e^A = \sum_{n=0}^\infty \frac{A^n}{n!}.
\]
As one can easily see, the result is a bounded operator itself, satisfying the norm-inequality $\|\e^A\| \leq \e^{\|A\|}.$
\end{definition}

\begin{theorem}\label{th-schro-dyn-beschr}
In the case of $H$ self-adjoint, bounded, and static a solution to the Schrödinger equation \eqref{cauchy-problem-se} is given by $\psi(t) = \e^{-\i H t}\psi_0$.
\end{theorem}

\begin{proof}
Just using the definition of the operator exponential and differentiating term by term shows right away that $\psi(t)$ is a solution to \eqref{cauchy-problem-se}. With $U(t)=\e^{-\i H t}$ \autoref{lemma-uniqueness} holds and shows uniqueness of the solution.
\end{proof}

To guarantee boundedness by \autoref{th-hellinger-toeplitz} (Hellinger--Toeplitz) it suffices to demand $H$ symmetric and defined on the whole space $\H$, we could therefore substitute this prerequisite. Furthermore for everywhere defined operators, \q{symmetric} equals \q{self-adjoint}, so the minimal conditions in \autoref{th-schro-dyn-beschr} are $H$ actually defined everywhere on $\H$, symmetric, and static. But as the Hamiltonian of quantum mechanics usually is not bounded (or defined everywhere), a basic reason for that was given at the beginning of \autoref{sect-unbounded}, the following case is of increased relevance.

\subsection{Unbounded Hamiltonians (Stone's theorem)}

\begin{xquote}{\citeasnoun{kato-1951}}
Thus our result serves as a mathematical basis for all theoretical works concerning nonrelativistic quantum mechanics, for they always presuppose, at least tacitly, the self-adjointness of Hamiltonian operators.
\end{xquote}

\autoref{th-schro-dyn-beschr} can also be formulated for unbounded operators if one approximates $H$ by a series of bounded operators $H_\lambda$ which then define the bounded, unitary operator $U(t)=\e^{-\i H t}$. The general case for Banach spaces is treated in the aforementioned theorems of Hille--Yosida and Lumer--Phillips. The Hilbert space case covered here is given by Stone's theorem, but our proof is greatly inspired by the one sometimes given for Hille--Yosida, see \citeasnoun[Th.~II.3.5]{engel-nagel}.

\begin{theorem}[Stone's theorem]\label{th-schro-dyn-konst}
In the case of $H$ self-adjoint and static a solution to the Schrödinger equation \eqref{cauchy-problem-se} is given by $\psi(t) = \e^{-\i H t}\psi_0$ with the formal operator exponential defined by a Yosida approximation given in the proof.
\end{theorem}

\begin{proof}
We choose $\lambda > 0$ for $t > 0$ (and analogously $\lambda < 0$ for $t < 0$, which will not be executed independently) and define $H_\lambda = -\i\lambda - \lambda^2 (\i \lambda - H)^{-1}$ (the \emph{Yosida approximation}) which is bounded following \autoref{lemma-resolvent} and further show that $H_\lambda \rightarrow H$ strongly converges on $D(H)$ when $|\lambda| \rightarrow \infty$. Subsequently we define $\e^{-\i H t}$ as the strong limit of $\e^{-\i H_\lambda t}$ on the whole of $\H$.\\
In the first step a simple calculation shows $(\i \lambda - H) H_\lambda = \i\lambda H$ and thus $H_\lambda = \i\lambda (\i\lambda - H)^{-1} H$ on $D(H)$. For $\psi \in D(H)$ the family $\i\lambda (\i\lambda - H)^{-1}$ satisfies
\begin{align*}
\|(\i\lambda (\i\lambda - H)^{-1} - \id)\psi\| &= \|\i\lambda^{-1}(\lambda^2(\i\lambda - H)^{-1} +\i\lambda)\psi\| \\
&= |\lambda|^{-1} \|H_\lambda \psi\| \\
&\leq |\lambda|^{-1} \|\i\lambda (\i\lambda - H)^{-1}\| \cdot \|H \psi\| \\
&\leq |\lambda|^{-1} \|H\psi\| \rightarrow 0,
\end{align*}
when $|\lambda| \rightarrow \infty$ and therefore $\i\lambda (\i\lambda - H)^{-1} \rightarrow \id$ strongly on $D(H)$. Because $\i\lambda (\i\lambda - H)^{-1}$ is uniformly bounded by 1 as shown by \autoref{lemma-resolvent} and $D(H) \subseteq \H$ dense, we get strong convergence on all of $\H$. Thus it follows $H_\lambda \rightarrow H$ strongly on $D(H)$ when $|\lambda| \rightarrow \infty$. Now with $H_\lambda$ bounded the exponential series $\e^{-\i H_\lambda t}$ can be constructed as in \autoref{exp}. It is uniformly bounded by applying the inequality from \autoref{lemma-resolvent}.
\[
\left\| \e^{-\i H_\lambda t} \right\| = \e^{-\lambda t} \left\| \e^{\i\lambda^2 (\i \lambda - H)^{-1} t} \right\| \leq \e^{-\lambda t} \e^{\lambda^2 |t| \cdot \| (\i \lambda - H)^{-1}\|} \leq 1
\]
It remains to show that there is a strong limit for $|\lambda| \rightarrow \infty$. For that purpose take the norm of
\[
\left(\e^{-\i H_\lambda t} - \e^{-\i H_\mu t}\right) \psi = \int_0^t \frac{\d}{\d s} \left( \e^{-\i H_\lambda s} \e^{-\i H_\mu (t-s)} \psi \right) \d s.
\]
To make use of the usual properties of the exponential map, the elements of the family $\{\e^{-\i H_\lambda t}\}_{\lambda>0}$ have to commutate. This follows directly from the fact that $\{H_\lambda\}_{\lambda>0}$ is a commuting family.
\begin{align*}
\left\| \left(\e^{-\i H_\lambda t} - \e^{-\i H_\mu t}\right) \psi \right\| &\leq \int_0^t \left\| \e^{-\i H_\lambda s} \e^{-\i H_\mu (t-s)} \right\| \, \left\| (H_\lambda - H_\mu) \psi \right\| \d s \\
&\leq |t| \, \left\| (H_\lambda - H_\mu) \psi \right\|
\end{align*}
For $|\lambda|,|\mu| \rightarrow \infty$ this converges to 0 for all $\psi \in D(H)$ and the thus formed Cauchy sequence defines $\e^{-\i H t}$ uniquely on $D(H)$. As the $\{\e^{-\i H_\lambda t}\}_{\lambda}$ are uniformly bounded as shown above and $D(H) \subset \H$ is dense by assumption, the newly found $\e^{-\i H t}$ can be uniquely extended to a bounded operator on the whole Hilbert space $\H$.\\
Next we show that $\e^{-\i H t}$ is unitary for unbounded $H$ too. To this end we observe that by definition $H_\lambda^* = H_{-\lambda}$ and with the result above $H_{\pm \lambda} \rightarrow H$ holds on $D(H)$. Thus the by construction self-adjoint $H_{(\lambda)} = \onehalf (H_\lambda + H_{-\lambda}) \rightarrow H$ converge on $D(H)$ and we can define $\e^{-\i H t}$ making use of them instead of $H_\lambda$ just as good.
\[
\langle \e^{-\i H t}\varphi,\e^{-\i H t}\psi \rangle = \lim_{|\lambda| \rightarrow \infty} \langle \e^{-\i H_{(\lambda)} t}\varphi,\e^{-\i H_{(\lambda)} t}\psi \rangle = \langle \varphi,\psi \rangle
\]
Uniqueness is finally given by \autoref{lemma-uniqueness} which would also automatically yield unitarity for $\e^{-\i H t}$.
\end{proof}

Stone's theorem can also be proven as a one-to-one correspondence between self-adjoint operators and unitary one-parameter groups. For any given unitary one-parameter group on a Hilbert space there exists a unique self-adjoint generator in the role of the Hamiltonian. This operator is then defined as a limit as in \autoref{def-inf-generator} that does not necessarily converge for all $\psi \in \H$ which gives the respective domain of the Hamiltonian.
\[
H\psi = \lim_{t \searrow 0} \frac{e^{-\i H t}\psi - \psi}{t}
\]
This amounts to an interesting shift in perspective, making the whole evolution operation the fundamental ingredient defining the dynamics of a system with the Hamiltonian remaining in the position of a derived object.

Note that the same result as Stone's theorem is amenable via spectral theory of operators where a resolution of identity yields a projection valued measure $\d E_H$. The unitary evolution operator is then given as the exponential applied to the eigenvalues in the spectral representation. \cite{blanchard-bruening-2}
\[
\e^{-\i H t} = \int \e^{-\i \varepsilon t} \d E_{H}(\varepsilon)
\]
However the spectral theorem for general self-adjoint operators is also non-trivial, so here only the path via Stone's theorem was followed.

An interesting corollary can be noted in addition, actually already at hand with \autoref{semigroup-solution} (ii), according to which trajectories which are in $D(H)$ for one instant permanently stay within $D(H)$. Conversely a trajectory outside of $D(H)$ must also stay there. A nice example is given by a particle in a box with constant initial wave function $\psi_0$ as elaborated by \citeasnoun{berry}. The result after free evolution is a continuous wave function that is nowhere differentiable, a fractal curve. This example has been utilised again in \citeasnoun{gruebl-penz} as a demonstration for well-defined non-differentiable Bohmian trajectories in one spatial dimension. These states are of infinite energy, therefore clearly lying outside of $D(H)$. Note that the converse is not necessarily true, a state of finite energy $\langle H \rangle_\psi < \infty$ might still not be in $D(H)$, the latter being all states with $\|H\psi\|_2^2 = \langle H^2 \rangle_\psi < \infty$.

\begin{corollary}\label{cor-evolut-stable}
Let $\psi(t)$ be a solution to the Schrödinger equation \eqref{cauchy-problem-se} as given by \autoref{th-schro-dyn-konst}. Then it holds $\psi(t) \in D(H)$ for all times if $\psi_0 \in D(H)$ and conversely that  $\psi(t) \notin D(H)$ for all times if $\psi_0 \notin D(H)$.
\end{corollary}

\begin{proof}
We take $H \e^{-\i H t} \psi_0$ and first show that $H$ and the bounded $\e^{-\i H t}$ are commuting, thus
\[
H \e^{-\i H t} \psi_0 = \e^{-\i H t} H \psi_0
\]
and $\e^{-\i H t} \psi_0 \in D(H)$. For that purpose we have to go back to the definition of $H$ as a limit of the series $H_\lambda$, where we make use of the $H_{(\lambda)}$ definition from \autoref{th-schro-dyn-konst}, for which both $t>0$ and $t<0$ can be used.
\[
H \e^{-\i H t} \psi_0 = \lim_{|\lambda| \rightarrow \infty} H_\lambda \lim_{|\mu| \rightarrow \infty} \e^{-\i H_\mu t} \psi_0 = \lim_{|\lambda|,|\mu| \rightarrow \infty} H_\lambda \e^{-\i H_\mu t} \psi_0
\]
Exchanging the operators within the limit is secure because $\{H_\lambda\}_\lambda$ is a commuting family and thus on $D(H)$ it holds
\[
H \e^{-\i H t} = \e^{-\i H t} H.
\]
If for a $\psi_0 \notin D(H)$ at any time it would hold that $\e^{-\i H t} \psi_0 \in D(H)$, then this could be taken as a new initial value in $D(H)$ which directly leads to a contradiction. Thus the converse is also true.
\end{proof}

\section{Classes of static potentials}
\label{sect-const-potential}

\begin{xquote}{\citeasnoun{schwartz-1962}}
Give a mathematician a situation which is the least bit
ill-defined -- he will first of all make it well defined. Perhaps appropriately, but perhaps also inappropriately. The hydrogen atom illustrates this process nicely. The physicist asks: `What are the eigenfunctions of such-and-such a differential operator?' The mathematician replies: `The question as put is not well defined. First you must specify the linear space in which you wish to operate, then the precise domain of the operator as a subspace. Carrying all this out in the simplest way, we find the following result\ldots' Whereupon the physicist may answer, much to the mathematician's chagrin: `Incidentally, I am not so much interested in the operator you have just analyzed as in the following operator, which has four or five additional small terms -- how different is the analysis of this modified problem?'
\end{xquote}

\begin{xquote}{Terence Tao\footnote{The quotation is taken from section \q{Semilinear Schrodinger (NLS)} at \url{http://www.math.ucla.edu/\~tao/Dispersive/schrodinger.html}.}}
One can also add a potential term, which leads to many physically 	interesting problems, however the field of Schrodinger operators with potential is far too vast to even attempt to summarize here.
\end{xquote}

The main assumption on $H$ in the preceding section was self-adjoint\-ness and a dense domain. Which conditions must be imposed on the potential $v$ to achieve this will be discussed here. A typical example would be the Hamiltonian of an atom with $N$ electrons and point-like, fixed nucleus with charge number $Z$ located at the origin. This amounts to the usual Born--Oppenheimer approximation with only the electronic part of the wave function under consideration and thus neglecting all degrees of freedom for the nuclei.
\begin{equation}\label{atomic-hamiltonian}
H = \underbrace{-\frac{1}{2}\sum_{i=1}^N \Delta_i}_A \;\underbrace{-\sum_{i=1}^N \frac{Z e^2}{|x_i|} + \sum_{i<j} \frac{e^2}{|x_i-x_j|}}_B
\end{equation}

In the following the operator $A=H_0$ will be viewed as the basic building block of the Hamiltonian, while $B$ is envisioned as its perturbation. We want to deduce conditions on $B$, under which for $A$ self-adjoint on $D(A)$ also $A+B$ is self-adjoint on $D(A)$. The answer is given by \autoref{th-kato-rellich} (Kato--Rellich).

This leaves us with the task to determine correctly a domain $D(A)$ for a symmetric $A$ such that the operator is self-adjoint. The property that the closure of a symmetric operator is self-adjoint is called \emph{essential self-adjointness} (cf.~\autoref{def-ess-sa}) and so we can use this construction for such operators to get a hint on $D(A)$. In the case of the free Hamiltonian $H_0$ this task has been undertaken already in \autoref{ex-laplacian} and its proper self-adjoint domain can be defined with the help of Sobolev spaces. The same result as in that example would be achieved if we employ a procedure called the \emph{Friedrichs extension} of non-negative symmetric operators. \cite[Th.~X.23]{reed-simon-2} In the case of a bounded open space-domain $\Omega \subset \R^n$ we have to take the boundary terms into account to get the correct domain of the self-adjoint operator. For zero boundary conditions this yields the following domain for $H_0$, already noted in \autoref{sect-zero-boundary}.
\[
D(H_0) = H^2(\Omega) \cap H_0^1(\Omega) = W^{2,2}(\Omega) \cap W^{1,2}_0(\Omega)
\]

It is important to note that different such boundary conditions may lead to different self-adjoint extensions of the same symmetric operator and thus correspond to different physical situations. See also Example 1 in section X.1 in \citeasnoun{reed-simon-2}. A voluminous survey over such self-adjoint extensions for a big class of so-called quasi-convex domains $\Omega$ is given by \citeasnoun{gesztesy-mitrea}.

Already at this stage we want to note that $B=v \in L^2_\mathrm{loc}(\R^n)$ as a multiplication operator is a minimal requirement for the operator $H_0+v$ to make sense on a core like the test functions $\Cont_0^\infty$. The condition can be weakened to local square-integrability because of the compact support of the test functions. This will show up again in the example for larger classes of potentials that allow essential self-adjointness of the Hamiltonian discussed briefly at the end of \autoref{sect-stummel-class}.

\subsection{The Kato--Rellich theorem}

\begin{lemma}[basic criterion for self-adjointness]\label{lemma-criterion-self-adjoint}
A symmetric operator $T$ is self-adjoint if and only if $\ran(T \pm \i) = \H$.
\end{lemma}

\begin{proof}
Assume $T$ first to be self-adjoint. Then $\ker(T+\i) = \{0\}$ because else there would exist a $0 \neq \varphi \in D(T)$ with $T\varphi=-\i\varphi$ that leads to the following immediate contradiction.
\[
-\i\langle\varphi,\varphi\rangle = \langle\varphi,T\varphi\rangle = \langle T\varphi,\varphi\rangle = \i\langle\varphi,\varphi\rangle
\]
By \autoref{lemma-neumann} it holds $\ran(T-\i)^\perp = \ker(T+\i)$, thus $\ran(T - \i)$ dense in $\H$. To show that the range is the full Hilbert space $\H$ we just have to show that $\ran(T - \i)$ is closed. Take a sequence $\varphi_n \in D(T)$ with $(T-\i)\varphi_n \rightarrow \psi_0$. We need to show that there is a $\varphi_0 \in D(T)$ fulfilling $(T-\i)\varphi_0=\psi_0$. By assumption $(T-\i)\varphi_n$ is a Cauchy sequence, therefore
\begin{align*}
\|(T-\i)(\varphi_n-\varphi_m)\|^2 &= \|T(\varphi_n-\varphi_m)-\i(\varphi_n-\varphi_m)\|^2 \\ &=
\|T(\varphi_n-\varphi_m)\|^2 + \|\varphi_n-\varphi_m\|^2 \longrightarrow 0.
\end{align*}
We conclude that both $T\varphi_n$ and $\varphi_n$ are converging sequences and set $\varphi_0 = \lim_{n \rightarrow \infty}  \varphi_n$. Since $T$ is a closed operator by \autoref{lemma-adj-abgeschlossen} it follows $\psi_0 = \lim_{n \rightarrow \infty}(T-\i) \varphi_n = (T-\i) \lim_{n \rightarrow \infty} \varphi_n = (T-\i)\varphi_0$ and $\varphi_0 \in D(T)$. $\ran(T + \i) = \H$ is shown analogously.\\
The reversed line of argument is the following. Let $\varphi \in D(T^*)$ then because of $\ran(T - \i) = \H$ there is a $\psi \in D(T)$ with $(T-\i)\psi = (T^*-\i)\varphi$. But it holds $D(T) \subset D(T^*)$ thus $\varphi-\psi \in D(T^*)$ and
\[
(T^* - \i)(\varphi-\psi) = 0.
\]
Again \autoref{lemma-neumann} in connection with $\ran(T+\i) = \H$ tells us $\ker(T^*-\i) = \{0\}$ and thereby $\varphi = \psi \in D(T)$. This proves $D(T^*) = D(T)$ and thus self-adjointness.
\end{proof}

\begin{definition}
Let $A,B$  be densely defined operators then $B$ is called \textbf{$A$-bounded} if
\begin{enumerate}[(i)]
	\itemsep0em
	\item $D(A) \subseteq D(B)$ and
	\item there are $a,b \geq 0$ such that for all $\varphi \in D(A)$ it holds that $\|B\varphi\| \leq a \|A\varphi\| + b\|\varphi\|$.
\end{enumerate}
The infimum of possible values for $a$ is called the \textbf{relative bound} of $B$ with respect to $A$ and can be as low as 0.
\end{definition}

\begin{theorem}[Kato--Rellich]\label{th-kato-rellich}\cite[Th.~X.12]{reed-simon-2}\\\cite[Th.~V.4.11]{kato-book}\\
Let $A$ be a self-adjoint operator and $B$ symmetric and $A$-bounded with relative bound $a<1$, then $A+B$ in self-adjoint on $D(A)$. Further if $A$ is bounded below by $M$, then $A + B$ is bounded below by $M-\max\{b/(1-a), a|M|+b\}$ with $a,b$ as in the definition above.
\end{theorem}

\begin{proof}
Following the basic criterion for self-adjointness (\autoref{lemma-criterion-self-adjoint}) we have to show that there is a $\mu_0 > 0$ for which $\ran(A+B \pm \i \mu_0) = \H$. The proof is the same for arbitrary $\mu_0 > 0$ as above for $\mu_0=1$.\\
Now for all $\varphi \in D(A)$ and $\mu > 0$ it holds
\[
\|(A+\i\mu)\varphi\|^2 = \|A\varphi\|^2 + \mu^2 \|\varphi\|^2.
\]
Taking $\varphi = (A+\i\mu)^{-1}\psi$ similar to \autoref{lemma-resolvent} yields
\[
1 = \|A(A+\i\mu)^{-1}\|^2 + \mu^2\|(A+\i\mu)^{-1}\|^2
\]
and thus $\|A(A+\i\mu)^{-1}\| \leq 1$ and $\|(A+\i\mu)^{-1}\| \leq \mu^{-1}$. Combined with the condition that $B$ is $A$-bounded we get
\begin{equation}\label{th-kato-rellich-eq}
\|B(A+\i\mu)^{-1}\| \leq a\|A(A+\i\mu)^{-1}\| + b\|(A+\i\mu)^{-1}\| \leq a + \frac{b}{\mu}.
\end{equation}
Considering $a<1$ and taking $\mu_0$ large enough the operator $C=B(A+\i\mu_0)^{-1}$ has norm $<1$. This means $-1 \notin \sigma(C)$, $(\id + C)^{-1}$ bounded as resolvent operator and from that $\ran(\id + C) = \H$. Because of $A$ self-adjoint we have $\ran(A+\i\mu_0)=\H$ as well and so the composition $(\id + C)(A+\i\mu_0) = A+B+\i\mu_0$ defined on $D(A)$ has
\[
\ran((\id + C)(A+\i\mu_0)) = \ran(A+B+\i\mu_0) = \H.
\]
The proof with $-\i\mu_0$ can be conducted analogously.\\
Finally we calculate the lower bound. Take $d \in \R$ with $M+d > 0$. From \autoref{lemma-neumann} we know
$\ran(A+d)^\perp = \ker(A+d) = \emptyset$
and thus $\ran(A+d) = \H$. Use $\|(A+d)^{-1}\| \leq (M+d)^{-1}$ and the lower bound $M$ for the spectrum of $A$ for the following estimate.
\begin{align*}
\|B(A+d)^{-1}\| &\leq a\|A(A+d)^{-1}\| + b\|(A+d)^{-1}\| \\
&\leq a \sup_{\lambda \geq M}\frac{|\lambda|}{\lambda + d} + b \frac{1}{M + d} \\
&= a \max\left\{ 1, \frac{|M|}{M + d} \right\} + b \frac{1}{M + d}
\end{align*}
If this norm is $<1$ then by the same trick as before we have $\ran(A+B+d) = \H$ thus $-d$ and all smaller values are in the resolvent set of $A+B$. Solving for this upper bound of $-d$ exactly yields the bound $M-\max\{b/(1-a), a|M|+b\}$.
\end{proof}

\begin{corollary}\cite[Prop.~1.3]{cycon}\\
If $A$ is self-adjoint and $D(A) \subseteq D(B)$ then $B$ is $A$-bounded if and only if $B(A+\i)^{-1}$ is bounded. The relative bound is given by the limit of the operator norm
\[
a = \lim_{\mu \rightarrow \infty} \|B(A+\i \mu)^{-1}\|.
\]
\end{corollary}

\begin{proof}
Boundedness of $B(A+\i)^{-1}$ follows from $A$-boundedness of $B$ by \eqref{th-kato-rellich-eq} in the proof of the Kato--Rellich theorem above (the assumption of $B$ being symmetric is not needed here). The relative bound as the limit $\mu \rightarrow \infty$ follows right away.\\
Now take any $\varphi \in D(A)$ and define $\psi = (A+\i)\varphi \in \H$, then the condition of boundedness implies there is some $c>0$ such that
\[
\|B(A+\i)^{-1} \psi\| \leq c\|\psi\|
\]
which is the same as
\[
\|B\varphi\| \leq c\|(A+\i)\varphi\| \leq c(\|A\varphi\| + \|\varphi\|)
\]
which shows $A$-boundedness of $B$.
\end{proof}

If the condition for the relative bound in the Kato--Rellich theorem cannot be met and $a=1$ the domain for self-adjointness is not maintained. The result of the theorem is relaxed to \emph{essential} self-adjointness of $A+B$ on $D(A)$ or any core of $A$ \cite[Th.~X.14, Wüst's Theorem]{reed-simon-2}. This means a unique self-adjoint extension of $A+B$ is still possible, but its domain is just not the former $D(A)$.

\subsection{Kato perturbations}
\label{sect-kato-peturbations}

To treat the quantum mechanical case of particles in singular Coulombic potentials and other unbounded potentials we make use of the following lemma. Here the number of dimensions of the underlying space actually plays a crucial role and we are limited to $n \leq 3$ like for the usual one-particle configuration space.

\begin{lemma}\label{lemma-ab-inequ}\cite[Th.~IX.28]{reed-simon-2}\\
Let $\varphi \in H^2(\R^n)$, $n\leq 3$, and $\varphi$ bounded and continuous. Then for all $\alpha >0$ there is a $\beta >0$ independent of $\varphi$ such that
\[
\|\varphi\|_\infty \leq \alpha\|\Delta \varphi\|_2 + \beta\|\varphi\|_2
\]
and alternatively
\[
\|\varphi\|_\infty \leq (2\pi)^{-n/2}\sqrt{\pi/2} \, (\|\Delta \varphi\|_2 + \|\varphi\|_2)
\]
which by \autoref{def-sim-onesided}, \autoref{lemma-equiv-graph-norm}, and \autoref{th-sobolev-norm-laplace} can be shortened to
\[
\|\varphi\|_\infty \lesssim \|\Delta \varphi\|_2 + \|\varphi\|_2 \sim \|\varphi\|_\Delta \sim \|\varphi\|_{2,2}.
\]
\end{lemma}

\begin{proof}
Take $\varphi \in H^2(\R^n)$ then its Fourier transform $\hat\varphi$ fulfils $(1+k^2)\hat\varphi \in L^2(\R^3)$. Also $(1+k^2)^{-1} \in L^2(\R^3)$ and thus their multiplication $\hat\varphi \in L^1(\R^3)$. By the CSB inequality
\begin{equation}\label{lemma-ab-inequ-L1-estimate}
\begin{aligned}
\|\hat\varphi\|_1 = \|(1+k^2)^{-1}(1+k^2)\hat\varphi\|_1 &\leq \|(1+k^2)^{-1}\|_2 \cdot \|(1+k^2)\hat\varphi\|_2 \\
&\leq \|(1+k^2)^{-1}\|_2 \cdot (\|\hat\varphi\|_2 + \|k^2\hat\varphi\|_2)
\end{aligned}
\end{equation}
where $\|(1+k^2)^{-1}\|_2 = \sqrt{\pi/2}$. For arbitrary $\lambda>0$ let $\hat\varphi_\lambda(k) = \lambda^3 \hat\varphi(\lambda k)$. The $L^1$-norm is invariant under this transformation $\|\hat\varphi_\lambda\|_1 = \|\hat\varphi\|_1$. Further $\|\hat\varphi_\lambda\|_2 = \lambda^{3/2} \|\hat\varphi\|_2$ and  $\|k^2\hat\varphi_\lambda\|_2 = \lambda^{-1/2} \|k^2\hat\varphi\|_2$. Thus putting $\hat\varphi_\lambda$ into the inequality \eqref{lemma-ab-inequ-L1-estimate} and using those identities we get
\[
\|\hat\varphi\|_1 \leq \sqrt{\pi/2} \,\lambda^{3/2} \|\hat\varphi\|_2 + \sqrt{\pi/2} \,\lambda^{-1/2} \|k^2\hat\varphi\|_2.
\]
Considering $\|\varphi\|_\infty \leq (2\pi)^{-n/2} \|\hat\varphi\|_1$ from the Riemann--Lebesgue lemma applied to the inverse Fourier transform (this is the reason for the limitation to bounded, continuous $\varphi$, the image of the (inverse) Fourier transform on $L^1$) and the Plancherel theorem this inequality amounts to
\[
\|\varphi\|_\infty \leq \underbrace{(2\pi)^{-n/2}\sqrt{\pi/2} \,\lambda^{3/2}}_\beta \|\varphi\|_2 + \underbrace{(2\pi)^{-n/2}\sqrt{\pi/2} \,\lambda^{-1/2}}_\alpha \|\Delta\varphi\|_2.
\]
By choosing $\lambda$ large enough we get the desired proposition regarding $\alpha$ and $\beta$. The second inequality follows if one takes $\lambda = 1$.
\end{proof}

Note that this lemma also expresses the continuous embedding $H^{2}(\R^n)$ $\hookrightarrow L^\infty(\R^n)$ from \eqref{eq-sobolev-embedding}. The same is true if $\varphi \in H^2_0(\Omega)$ with $\Omega \subset \R^n$ open and zero at the boundary because such a function can be extended by zero to all of $\R^n$ without changing any of the norm values. \cite[3.27]{adams} In the following $\Omega$ is always assumed to be a domain (open and connected) in the typical $\R^3$. Yet we keep the conditions on the potential in the whole space $\R^3$ because the construction following below in \autoref{def-sum-space} includes rotations after which the potential still has to fit onto the original domain. Anyway extending the domain to $\R^3$ does not really change anything, because an $L^p(\Omega)$ function can just as well be (discontinuously) extended to all of $\R^3$ by zero while it retains its norm value.

\begin{theorem}[Kato]\label{th-kato}
Given a real potential $v \in L^2(\R^3) + L^\infty(\R^3)$ the Hamiltonian $H_0+v$ is self-adjoint on $D(H_0) \subset L^2(\Omega)$ with zero boundary conditions.
\end{theorem}

\begin{proof}
The potential $v$ is clearly self-adjoint on $D(v) = \{\varphi \mid \varphi, v\varphi \in L^2(\R^3) \}$ and shall be given by $v=v_1+v_2$ with $v_1 \in L^2(\R^3), v_2 \in L^\infty(\R^3)$. This means that by the most obvious version of the Hölder inequality
\[
\|v\varphi\|_2 \leq \|v_1\|_2 \, \|\varphi\|_\infty + \|v_2\|_\infty \|\varphi\|_2
\]
and thus $\Cont^\infty_0(\R^3) \subset D(v)$. The test functions form a common core (\autoref{def-ess-sa}) for $H_0$, $v$, and $H_0+v$ too. For any $\varphi \in \Cont^\infty_0(\R^3)$ using \autoref{lemma-ab-inequ}
\begin{equation}\label{eq-kato-ab-inequality}
\|v\varphi\|_2 \leq \underbrace{\alpha \|v_1\|_2}_a \, \|\Delta\varphi\|_2 + (\underbrace{\beta \|v_1\|_2 + \|v_2\|_\infty}_b) \|\varphi\|_2.
\end{equation}
For $\varphi \in D(H_0)$ we choose a sequence in $\Cont^\infty_0(\R^3)$ converging to $\varphi$, thereby demanding zero boundary conditions, and use the closedness of $\Delta$ from \autoref{lemma-adj-abgeschlossen} to show that the inequality above still holds.\\
Since $\alpha$ can be chosen arbitrarily small $v$ is $H_0$-bounded with relative bound 0 and by \autoref{th-kato-rellich} (Kato--Rellich) $H_0+v$ is self-adjoint on $D(H_0)$.
\end{proof}

The Stummel and Kato class discussed in \autoref{sect-stummel-class} and \autoref{sect-kato-class} below provide analogues to the  $L^2+L^\infty$ class of potentials with similar consequences.

Next we want to overcome the restriction to three or less dimensions in the theorem above to account for more than one particle. We define a space of physical potentials acting on wave functions in the full configuration space $\R^n$ with $n\geq 3$, typically $n=3N$ for $N$ particles, as a sum of functions from a Banach space $X$ each having only three coordinates as an argument. In principle an arbitrary yet finite number of such potentials can be added, each one depending on its own set of coordinates. This set of three coordinates is chosen from the total of $n$ coordinates by rotating them to the first three coordinates and then projecting onto those. Thus not only potentials depending on one of the particle coordinates $x_i \in \R^3$ are allowed but also mixtures from rotations of them such as typically $x_i - x_j \in \R^3$ (relative coordinates). The main example for the Banach space $X$ would be the space of Kato perturbations $X = L^2(\R^3) + L^\infty(\R^3)$ following from the theorem above and being equipped with a canonical norm (see \autoref{sect-banach-spaces}, first paragraph).

\begin{definition}\label{def-sum-space}
The sum space $\Sigma(X)$ on $\R^n$ with $n\geq 3$, where $X$ is a Banach space of functions on $\R^3$, is given by the set
\[
\Sigma(X) = \left\{\left. \sum_{k=1}^K v_k \circ \pi_1 \circ \rho_k \,\right\vert v_k \in X, \rho_k \in \SO(n), K \in \N \right\}
\]
equipped with the norm
\[
\|v\|_{\Sigma(X)} = \inf\left\{\left. \sum_{k=1}^K \|v_k\|_X \,\right\vert
v = \sum_{k=1}^K v_k \circ \pi_1 \circ \rho_k, v_k \in X, \rho_k \in \SO(n), K \in \N \right\}.
\]
Here $\pi_1$ is the projection on the first three coordinates, i.e., the $\R^3$ coordinates of the first particle.
\end{definition}

A similar construction is given in Example F of \citeasnoun{simon-1982} but it uses different coordinate projections instead of rotations. The following theorem using rotations is a direct extension of Kato's theorem (\autoref{th-kato}) above and is called \q{Kato's theorem} as well in \citeasnoun[Th.~X.16]{reed-simon-2} from where the proof is taken. There only essential self-adjointness is discussed for a core of test functions which amounts to zero boundary conditions on the full domain.

\begin{theorem}\label{th-sum-space}
Let $v \in \Sigma(L^2 + L^\infty)$. Then the $N$-particle Hamiltonian $H_0+v$ is self-adjoint and bounded below on $D(H_0) \subset L^2(\Omega^N)$ with zero boundary conditions.
\end{theorem}

\begin{proof}
Take any of the $v_k$ that constitute $v$ and assume $(\pi_1 \circ \rho_k)(x)$ to be the $\R^3$ coordinates $x_1$. This can be done without loss of generality because the $L^p$-norms as well as $\Delta$ are all invariant under orthogonal transformations.  We write simply $(v_k \circ \pi_1)(x) = v_k(x_1) = v_k(x)$. Let $\Delta_1$ be the Laplacian with respect to these first three coordinates. By the same estimate as in the proof of \autoref{th-kato} and using standard Fourier transform arguments for $L^2$ functions we get the following estimate.
\begin{align*}
\|v_k \varphi\|_2 &\leq a_k \|\Delta_1 \varphi\|_2 + b_k \|\varphi\|_2 \\
&= a_k \left( \int \left| \sum_{i=1}^3 k_i^2 \hat\varphi(k) \right|^2 \d k \right)^{1/2} + b_k \|\varphi\|_2 \\
&\leq a_k \left( \int \left| \sum_{i=1}^n k_i^2 \hat\varphi(k) \right|^2 \d k \right)^{1/2} + b_k \|\varphi\|_2 \\
&= a_k \|\Delta \varphi\|_2 + b_k \|\varphi\|_2
\end{align*}
We take $a,b$ to be the maximal choice for all the $v_k$ and thus have
\[
\|v \varphi\|_2 \leq \sum_{k=1}^K \| v_k \varphi\|_2 \leq Ka \|\Delta \varphi\|_2 + Kb \|\varphi\|_2.
\]
Again $a$ can be chosen as small as we like and we get $v$ as $H_0$-bounded with relative bound 0 thus $H_0+v$ self-adjoint and bounded below by \autoref{th-kato-rellich} (Kato--Rellich) because $H_0$ is known to be bounded below by 0.
\end{proof}

Originally this theorem is from \citeasnoun{kato-1951} and \citeasnoun[Notes to X.2]{reed-simon-2} write:
\begin{quote}
	This paper was a turning point in mathematical physics for two reasons. Firstly, the proof of self-adjointness was a necessary preliminary to the problems of spectral analysis and scattering theory for these operators, problems which have occupied mathematical physicists ever since. Secondly, the paper focused attention on specific systems rather than foundational questions.
\end{quote}

In the later proof of existence of solution to the Schrödinger equation with Hamiltonians incorporating such potentials we will need the following estimate. An extended version of it will there be given with \autoref{lemma-rs-0}.

\begin{lemma}\label{lemma-sum-space-inequality}
For $v \in \Sigma = \Sigma(L^2 + L^\infty)$ and $\varphi \in D(H_0)$, $n \geq 3$, the following inequality holds.
\[
\|v\varphi\|_2 \leq \sqrt{2}\,\|v\|_{\Sigma} \|\varphi\|_{2,2}
\]
This makes $v : H^2 \rightarrow L^2$ into a bounded multiplication operator.
\end{lemma}

\begin{proof}
By the same technique as in the proof above and of \autoref{th-kato} (Kato) and incorporating the second estimate of \autoref{lemma-ab-inequ}
\begin{align*}
\|v\varphi\|_2 &\leq \sum_{k=1}^K \|v_k \varphi\|_2 \leq \sum_{k=1}^K (\|v_{k,1}\|_2 \, \|\varphi\|_\infty + \|v_{k,2}\|_\infty \|\varphi\|_2) \\
&\leq \sum_{k=1}^K ((2\pi)^{-n/2}\sqrt{\pi/2} \,\|v_{k,1}\|_2 \, (\|\Delta\varphi\|_2+\|\varphi\|_2) + \|v_{k,2}\|_\infty \|\varphi\|_2) \\
&\leq \sum_{k=1}^K ((2\pi)^{-n/2}\sqrt{\pi/2} \,\|v_{k,1}\|_2 \, (\|\Delta\varphi\|_2+\|\varphi\|_2) + \|v_{k,2}\|_\infty \|\varphi\|_2) \\
&\leq \sum_{k=1}^K ( \|v_{k,1}\|_2 + \|v_{k,2}\|_\infty) \cdot (\|\Delta\varphi\|_2+\|\varphi\|_2).
\end{align*}
The factor $(2\pi)^{-n/2}\sqrt{\pi/2} \approx 0.08$ if $n=3$ was estimated really roughly by $1$ in the last step to get a common factor for both terms. The optimal decomposition of the molecular potential $v$ into the $L^2$ and $L^\infty$ parts yields the $\Sigma(L^2 + L^\infty)$-norm and the estimate $\|\varphi\|_2 + \|\Delta\varphi\|_2 \leq \sqrt{2}\sqrt{\|\varphi\|_2^2 + \|\Delta\varphi\|_2^2} \leq \sqrt{2}\|\varphi\|_{2,2}$ from \autoref{lemma-equiv-graph-norm} and \eqref{eq-sobolev-norm} concludes the proof.
\end{proof}

\begin{example}\cite[Ex.~X.2]{reed-simon-2}\label{ex-atomic-hamiltonian}\\
We want to show that the atomic Hamiltonian \eqref{atomic-hamiltonian} involves a $\Sigma(L^2 + L^\infty)$ potential and is thus self-adjoint on $D(H_0)$ by \autoref{th-sum-space}. First we show that the Coulombic particle-nucleus interactions proportional to $|x|^{-1} = r^{-1}$ is in $L^2(\R^3) + L^\infty(\R^3)$.\footnote{It gets much more involved for $v(x) = -r^{-\alpha}, \alpha \geq 3/2$, see the remark about \q{venerable physical folklore} in \citeasnoun{reed-simon-2} after Theorem X.18 and Example 4 in section X.2.} With the help of the characteristic function we split
\[
r^{-1} = r^{-1} \1_{r \leq 1} + r^{-1} \1_{r > 1}.
\]
Clearly the second part is bounded by $1$ and thus element of $L^\infty(\R^3)$. The $L^2$ property of the first part can easily be shown by integrating
\[
\int_{\R^3} \left| r^{-1} \right|^2 \1_{r \leq 1} \d x = 4\pi \int_0^1 r^{-2} r^2 \d r = 4\pi < \infty.
\]
Secondly the Coulombic particle-particle interaction is just the same type of potential but with relative coordinates $x_i-x_j$ which are just the projection on the first particle's coordinates after a suitable rotation as already stated in \autoref{def-sum-space}.
\end{example}

One can also give a \q{physical} argument for allowing singular potentials with singularities of type $r^{-\alpha}, 0 < \alpha < 2$, in quantum mechanics if one demands lower boundedness as a fundamental property of the Hamiltonian. If now the wave function is concentrated in an area of width $r$ around a negative singularity, the kinetic energy is of order $p^2 \approx r^{-2}$ whereas the potential energy amounts to $r^{-\alpha}$. Adding this up gives the total energy $c_1 r^{-2} - c_2 r^{-\alpha}$ with constants $c_1,c_2>0$ which is bounded below for arbitrarily small $r$ surely if $\alpha < 2$. \cite[Notes to X.2]{reed-simon-2}

Note that the Coulomb law is not the natural potential in domains $\Omega \subsetneq \R^3$ because it has to fulfil the Poisson equation for a point charge $-\Delta v = 4\pi \delta(x)$. To construct Green's functions for periodically aligned box domains one employs the technique of mirror charges that form an infinite series of Coulomb potentials with changing signs in that case. But still this series adds up to a function that looks reasonably similar to the original Coulomb, just having a wee bit smaller absolute value, so it makes sense to discuss Coulomb potentials in settings with bounded domains as well. A final note should be addressed towards the harmonic oscillator Hamiltonian.

\begin{example}\label{ex-harmonic-hamiltonian}
The harmonic oscillator Hamiltonian, called Hermite operator in some references, is given by $H = -\Delta + |x|^2$ and is clearly not covered by the considerations in this section. It is a famous example of free Hamiltonians perturbed by positive operators for which a separate theory exists in which the Friedrichs extension mentioned above takes an important role.\\
By following the procedure from \autoref{ex-laplacian} we get $D(H) = \{ \varphi \in H^2(\R^n) \mid |x|^2 \varphi \in L^2(\R^n) \}$ as the domain of a self-adjoint operator $H$. Note that as a special feature the Fourier transform of all elements of $D(H)$ is again $D(H)$ due to the special symmetry of $H$ regarding Fourier transforms $\widehat{H \varphi} = H \hat\varphi$.\\
The example can be extended to all potentials $v \in L^2_\mathrm{loc}(\R^n)$ bounded from below as an \q{amusing mathematical game} (\citeasnoun[Conjecture 1]{simon-1973}, already studied in \citeasnoun{kato-1972} for even more general potentials including magnetic effects), making $H_0 + v$ essentially self-adjoint on $D(H_0)$ or on any core of $H_0$. Conversely this example should highlight that the negativity of potentials as in \eqref{atomic-hamiltonian} is just the reason why this line of thought is not applicable and techniques like above had to be developed. There is also a general result for potentials with $\lim_{|x| \rightarrow \infty} v(x) = +\infty$ stating that the associated Hamiltonian $H_0+v$ has a fully discrete spectrum \cite[8.2]{pankov}.
\end{example}

\subsection{Stummel class}
\label{sect-stummel-class}

It is natural to ask for a maximal class of potentials in the context of the Kato--Rellich theorem, i.e., to define the set of all multiplication operators that are $H_0$-bounded with relative bound $<1$. An almost maximal class of such perturbations of $H_0$ is given by the class $S_n$ first considered by \citeasnoun{stummel}, with $n$ the dimension of the underlying domain $\R^n$. Note that the integral in the definition is that of a localised Riesz potential (see \autoref{def-riesz}).

\begin{definition}[Stummel class]
$S_n$ consists of all real-valued, measurable functions $v$ on $\R^n$ that fulfil
\begin{align*}
&\lim_{r \searrow 0} \left( \sup_x \int_{|x-y| \leq r} \frac{|v(y)|^2}{|x-y|^{n-4}} \d y \right) = 0 &\mtext{if} n \neq 4, \\
&\lim_{r \searrow 0} \left( \sup_x \int_{|x-y| \leq r} \ln(|x-y|^{-1}) |v(y)|^2 \d y \right) = 0 &\mtext{if} n = 4.
\end{align*} 
\end{definition}

Note that in the case $n \leq 3$ the condition simplifies to
\begin{equation}\label{eq-stummel-cond}
\sup_x \int_{|x-y| \leq r} |v(y)|^2 \d y < \infty
\end{equation}
for an arbitrary fixed $r>0$ which is a uniformly local $L^2$-condition. This is not the same as $L^2_\mathrm{loc}(\R^n)$ which includes potentials like $v(x)=|x|^2$ considered in \autoref{ex-harmonic-hamiltonian} that would yield an infinitely large supremum above. But $S_n, n \leq 3$, clearly includes $L^2(\R^n)$ and $L^\infty(\R^n)$ and thus also $L^2(\R^n)+L^\infty(\R^n)$. The following Lemma shows completeness of $S_n$ endowed with a natural norm derived from condition \eqref{eq-stummel-cond} if $n \leq 3$.

\begin{lemma}\label{lemma-stummel-banach}
$S_n$, $n\leq 3$, is a Banach space with norm
\[
\|v\|_{S_n} = \sup_x \int_{|x-y| \leq 1} |v(y)|^2 \d y.
\]
\end{lemma}

\begin{proof}
Take $v_i$ a Cauchy sequence in $S_n$ then it holds for all $\varepsilon > 0$ there are indices $i,j \in \N$ such that for all $x \in \R^n$
\[
\int_{|x-y| \leq 1} |v_i(y) - v_j(y)|^2 \d y < \varepsilon.
\]
Because $L^2(\{ y \in \R^n \mid |x-y| \leq 1 \})$, i.e., the Stummel class restricted to a unit ball with centre $x$, is a complete vector space, the Cauchy sequence converges on every such ball. The overall limit is clearly unique because the balls may overlap. Thus one has a unique limit also globally which shows that $S_n$ is a complete normed vector space (Banach space).
\end{proof}

The Stummel class $S_n$ is fully characterised by a property involving the operator norm that demands \cite[Th.~1.7]{cycon}
\[
\lim_{E\rightarrow\infty} \|(H_0 + E)^{-2} |v|^2\| = 0 \quad\text{as an operator\;} L^\infty \rightarrow L^\infty
\]
from which follows \cite[Cor.~1.8]{cycon}
\begin{equation}\label{eq-stummel-limit}
\lim_{E\rightarrow\infty} \|(H_0 + E)^{-1} v\| = 0 \quad\text{as an operator\;} L^2 \rightarrow L^2.
\end{equation}

\begin{lemma}
From \eqref{eq-stummel-limit} it follows that $v$ is $H_0$-bounded as a multiplication operator with relative bound 0.
\end{lemma}

\begin{proof}
First one shows that $(H_0 + E)^{-1}$ and $v$ can be interchanged in \eqref{eq-stummel-limit} by a duality argument. For every $\varepsilon > 0$ there is a $E>0$ such that, by writing out the operator norm,
\[
\sup_{\|\varphi\|=1} \|(H_0 + E)^{-1} v \varphi\| < \varepsilon.
\]
From this inequality it follows by the CSB inequality and using the self-ad\-joint\-ness of both involved operators
\begin{align*}
\sup_{\substack{\|\varphi\|=1 \\ \|\psi\|=1}} \langle v(H_0 + E)^{-1}\psi, \varphi \rangle &= \sup_{\substack{\|\varphi\|=1 \\ \|\psi\|=1}} \langle \psi,(H_0 + E)^{-1} v \varphi \rangle \\
&\leq \sup_{\|\psi\|=1} \|\psi\| \cdot \sup_{\|\varphi\|=1} \|(H_0 + E)^{-1} v \varphi\| < \varepsilon.
\end{align*}
Now choose $\varphi = v(H_0 + E)^{-1}\psi$ (normalised to 1) to get from the first term above
\[
\sup_{\|\psi\|=1} \|v(H_0 + E)^{-1} \psi\| < \varepsilon
\]
which means that $\|v(H_0 + E)^{-1}\psi\| < \varepsilon \|\psi\|$ for all $\psi \in \H$. For $\varphi \in D(H_0)$ set $\psi = (H_0+E)\varphi$ and it follows directly $\|v\varphi\| < \varepsilon \|(H_0+E)\varphi\| \leq \varepsilon \|H_0\varphi\| + \varepsilon E \|\varphi\|$. This shows $H_0$-boundedness of $v$ with relative bound 0 because $\varepsilon$ can be made arbitrarily small by increasing $E$.
\end{proof}

There is now also a way to conclude from $H_0$-boundedness of $v$ with special bounds that $v \in S_n$. This is the reason why the Stummel class can be considered a natural choice for self-adjoint perturbations of $H_0$ although it is not maximal. \cite[Th.~1.9 and the example after Th.~1.12]{cycon} Inequalities of the kind of \autoref{lemma-sum-space-inequality} that will become important for estimates formulated later can also be found for Stummel class potentials, see \citeasnoun[10.3]{weidmann}.

\subsection{Kato class}
\label{sect-kato-class}

A variant of the Kato--Rellich theorem (\autoref{th-kato-rellich}), the so-called KLMN\footnote{The letters actually stand for Kato, Lions, Lax--Milgram, and Nelson, so it is rather the \q{KLLMN} theorem. \cite{simon-2004}} theorem (\citeasnoun[Th.~1.5]{cycon} and \citeasnoun[Th.~X.17]{reed-simon-2}), is also available for the \q{form version} of operators, where instead of $\|A\varphi\|$ the expectation value $\langle \varphi,A\varphi \rangle$ is considered. Such operators can usually be defined on a larger domain as the example of the Laplacian $-\langle \varphi,\Delta\varphi \rangle = \langle \nabla\varphi,\nabla\varphi \rangle$ in its \q{energy extension} to $H^1$ makes clear. The corresponding class of potentials that are $H_0$-form-bounded  with relative bound 0 is very similar to the Stummel class in definition and is called the Kato class\footnote{In \citeasnoun{tddft-review} we called the $L^2+L^\infty$ type potentials the \q{Kato class} because of its relation to perturbations of linear operators studied by Kato but this is contrary to the usual terminology in the literature followed here.} $K_n$ \cite[Def.~1.10]{cycon}.

\begin{definition}[Kato class]
$K_n$ consists of all real-valued, measurable functions $v$ on $\R^n$ that fulfil
\begin{align*}
&\lim_{r \searrow 0} \left( \sup_x \int_{|x-y| \leq r} \frac{|v(y)|}{|x-y|^{n-2}} \d y \right) = 0 &\mbox{if}\quad n \neq 2, \\
&\lim_{r \searrow 0} \left( \sup_x \int_{|x-y| \leq r} \ln(|x-y|^{-1}) |v(y)| \d y \right) = 0 &\mbox{if}\quad n = 2.
\end{align*} 
\end{definition}

Similarly to the Stummel class the case $n = 1$ of the condition simplifies to
\begin{equation*}
\sup_x \int_{|x-y| \leq r} |v(y)| \d y < \infty
\end{equation*}
for an arbitrary fixed $r>0$ which is a uniformly local $L^1$-condition. There is also a local version of the Kato class $K_n^{\mathrm{loc}}$ defined by $v \varphi \in K_n$ for all $\varphi \in \Cont_0^\infty$ that will become important in results concerning the unique continuation property of eigenstates (cf.~\autoref{sect-ucp}).

An extensive study of the semigroup $\{\exp(-tH)\}_{t \geq 0}$ (without the imaginary $\i$ and thus a solution to the associated heat equation) generated by Hamiltonians involving the Kato class that is mainly concerned with properties of eigenfunctions of $H$ was conducted by \citeasnoun{simon-1982}. In this reference it is interestingly noted (after (A15) and shown in Example I) that in contrast to the Stummel class (\autoref{lemma-stummel-banach}) the Kato class is \emph{not} complete if $n = 3$, a proposition later withdrawn in an erratum to the article \cite{simon-erratum}. Another similar class of $H_0$-form-bounded potentials with relative bound 0 are those of Rollnik \cite[Th.~X.19]{reed-simon-2}. Note the relation to uniformly local $L^p$ spaces for $n \geq 2$ \cite[Ex.~E]{simon-1982}
\begin{equation}\label{eq-kato-class-Lp}
L^p_\mathrm{unif} \subset K_n \mtext{if} p > \frac{n}{2}
\end{equation}
where
\[
L^p_\mathrm{unif} = \left\{ f \left\vert\, \sup_x \int_{|x-y|\leq 1} |f(y)|^p \d y < \infty \right.\right\}.
\]
The relation \eqref{eq-kato-class-Lp} is also true for the local versions of these spaces as follows directly from their definition.
\begin{equation}\label{eq-kato-class-Lp-local}
L^p_\mathrm{loc} \subset K_n^\mathrm{loc} \mtext{if} p > \frac{n}{2}
\end{equation}

An overview of Stummel and Kato class results is given in Appendix C of the review article of \citeasnoun{braverman} that is actually concerned with the more general setting of Schrödinger type operators on smooth manifolds. Even more interestingly they added an historical section (Appendix D) covering essential self-adjointness regarding Hamiltonians with singular potentials. As pointed out by them a further improvement to the result of \citeasnoun{kato-1972} mentioned in \autoref{ex-harmonic-hamiltonian} is that of \citeasnoun{leinfelder}. Just like \citeasnoun{kato-1972} it allows $v = v_1 + v_2$ with $v_1,v_2 \in L^2_\mathrm{loc}$, $v_1(x) \geq -c|x|^2$, and $v_2$ $H_0$-bounded with relative bound $<1$ (for example of Stummel class) but in addition widens the class of possible magnetic potentials. It was not studied if the truncation method used in their proof is extendible to the composition procedure of potentials on $\R^3$ as in \autoref{th-sum-space}. The Leinfelder--Simader result is also discussed in \citeasnoun[Sect.~1.4]{cycon} for the less general $v \in L^2_\mathrm{loc}, v \geq 0$ case with magnetic potentials.

\section{Properties of eigenfunctions}
\label{sect-prop-eigenfunctions}

Although this is a topic more from equilibrium theory and thus relates to time-independent DFT, it is interesting to note some special properties of eigenfunctions of the Hamiltonian that derive from the concepts discussed above. We study the eigenvalue problem
\[
H \psi = E \psi
\]
where $E$ can be absorbed into the potential part $v$ of $H=H_0+v$ such that the equation
\begin{equation}\label{eq-ti-se}
H \psi = 0
\end{equation}
remains. Note that a real eigenfunction is always possible by just taking $\Re \psi$ which solves just the same eigenvalue equation due to self-adjointness of $H$. In the case $v=0$ and thus $H=H_0=-\onehalf \Delta$ this is the famous Laplace equation. Solutions to that are known to be as regular as one could wish, namely analytic (see \autoref{sect-ucp}), but such properties need not pass over to the associated density (see \autoref{sect-non-regularity-density}) in whose regularity we are also interested. Special features of such eigenfunctions for typical Hamiltonians are discussed in the following sections.

As a first result we repeat \citeasnoun[Th.~C.1.1]{simon-1982}, who notes that a weak solution to \eqref{eq-ti-se} with $v \in K_n^{\mathrm{loc}}$ on a bounded domain $\Omega$ is continuous almost everywhere. A version of the Harnack inequality and other estimates follow for the same class of potentials. If one divides the potential into positive and negative parts, $v=v_+ + v_-$, then under the conditions $v_+ \in K_n^{\mathrm{loc}}, v_- \in K_n$ any $L^2$-eigenfunction decays like $\psi(x) \rightarrow 0$ as $|x| \rightarrow \infty$. This is of importance because there are $L^2$-functions that do not go to zero in this sense, but we would expect it for physical states. Many more special decay results that we do not give here are known for eigenfunctions of the Schrödinger Hamiltonian, see \citeasnoun[C.3]{simon-1982}.

\subsection{Possible non-regularity of the associated density}
\label{sect-non-regularity-density}

One has to be careful not to think that from continuity of the wave function $\psi \in \Cont^0(\R^{dN})$ continuity of the associated density $n \in \Cont^0(\R^d)$ as an integrated quantity necessarily follows. Also uniform continuity as a condition for continuity for the integral over one coordinate only holds on compact domains. What is however guaranteed is that the integral of a continuous function is always a lower semi-continuous function in the remaining arguments. This is actually a classification of all lower semi-continuous function. \cite{finkelstein} The following is a counterexample to continuity of an integrated quantity where one single dimension gets integrated out.

\begin{example}
Take a Gaussian in $x_1$-direction with its width controlled by $x_2$, getting broader as one approaches $|x_2| \rightarrow 1$, and setting it 0 outside the stripe $|x_2| < 1$.
\[
f(x_1,x_2) = \left\{ 
\begin{array}{cl}
(1-|x_2|) \exp\left(-(x_1 (1-|x_2|))^2\right)/\sqrt{\pi} \quad & |x_2| < 1 \\
0 & |x_2|\geq 0
\end{array}
\right.
\]
Clearly this is a continuous function $f \in \Cont^0(\R^2)$. But integrating out the $x_1$ coordinate gives a constant area under the Gaussian for $|x_2|<1$ and 0 else.
\[
\int_{-\infty}^\infty f(x_1,x_2) \d x_1 = \left\{ 
\begin{array}{ll}
1 \quad & |x_2| < 1 \\
0 & |x_2|\geq 0
\end{array}
\right.
\]
The resulting marginal function is perfectly integrable but discontinuous. This example can also be easily generalised to an antisym\-metric, continuous $L^2$ wave function that yields a discontinuous density. Studies regarding the regularity of the density particularly for molecular Hamiltonians are discussed in \autoref{sect-analyticity-eigenstates} below.
\end{example}

We might be worried now that this fact spoils differentiability of integrated quantities like the one-particle density defined in \eqref{def-n}. If $\psi \in W^{1,2}(\Omega^N)$ we would usually expect $n \in W^{1,1}(\Omega)$ but the above example maybe makes us a bit sceptical. Everything though is safe because in this context we consider only weak derivatives and the issue highlighted in this section magically disappears. \cite{cheng-2010} It is interesting to observe how at such frontiers the `traditional' calculus (pointwise defined, continuous or analytic functions) and the `modern' calculus (Lebesgue and Sobolev spaces) diverge.

\subsection{Unique continuation property of eigenstates}
\label{sect-ucp}

Take first a solution to Laplace's equation $\Delta u = 0$ in $H^2(\Omega)$ on $\Omega \subset \R^n$ open and connected. This is the definition of a harmonic function and it is known that in $n=2$ they are the real (or imaginary) part of an everywhere holomorphic (entire) function and they are real analytic for any dimension $n$. If such a function vanishes to infinite order at any $x_0 \in \Omega$, i.e., for all $k \in \N_0$
\begin{equation}\label{eq-strong-ucp}
\lim_{r \searrow 0} \frac{1}{r^k} \int_{B_r(x_0)} |u(x)|^2 \d x =0,
\end{equation}
then $u=0$ on all of $\Omega$ because all derivatives at this point must be zero (through the usual estimates for the derivatives of harmonic functions). This is called the \emph{strong unique continuation property} (UCP). \cite{salo-2014} The \emph{weak} version states that if $u=0$ on any open set contained in $\Omega$ then $u=0$ everywhere. Clearly the weak property follows from the strong. Now if two solutions to Laplace's equation are equal on an open set or if only their difference complies with \eqref{eq-strong-ucp} at any point $x_0$ then they are equal everywhere, thus the name \q{unique continuation property}. The same holds for $n \neq 2$ because harmonic functions are still real analytic. This proposition can even be extended to all linear elliptic differential operators with real analytic coefficients because Holmgren's uniqueness theorem guarantees analytic solutions. See \citeasnoun{escauriaza} for a similar introduction into the subject that also deals with evolution equations.

The interesting problem of establishing the UCP for more general operators with Laplacian principal part like $-\Delta + v$ was solved for potentials like $v \in L^\infty_\mathrm{loc}$. Because of the local nature of the result this can be generalised to all potentials that are bounded on compact subsets of $\R^n \setminus S$ with $S$ a closed set of measure zero such that $\R^n \setminus S$ is still connected. (\citeasnoun[Th.~XIII.57]{reed-simon-4} and \citeasnoun[C.9]{simon-1982}) This settles the physical case of singular Coulombic potentials and interactions, but leaves the question unsettled for more general classes of unbounded potentials. \citeasnoun[C.9]{simon-1982} conjectured that the weak UCP should hold for the whole Kato class $v \in K_n^{\mathrm{loc}}$, a statement answered positively only for radial potentials by \citeasnoun{fabes1990}. Taking into account $L^p_\mathrm{loc} \subset K_n^\mathrm{loc}$ if $p > n/2$ from \eqref{eq-kato-class-Lp-local} it seems natural to ask the same question for potentials from the smaller class $L^p_\mathrm{loc}$ for $p > n/2$ (this was again formulated as a conjecture by \citeasnoun{schechter-simon}) or even with the critical index included, i.e., $p \geq n/2$. This case has indeed been settled shortly afterwards. That the strong UCP holds for $n\geq 3$ and $v \in L^{n/2}_\mathrm{loc}$ was shown by \citeasnoun{jerison-kenig-1985} for solutions $u \in W^{2,q}_{\mathrm{loc}}$ and by \citeasnoun{sogge-1990} for $u \in W^{1,q}_{\mathrm{loc}}$ (by resorting to the energetic extension of the differential operator and also allowing for a magnetic term). In both cases $q=2n/(n+2)<2$ converging to 2 for high dimensionality and thus always including the spaces $H^2_\mathrm{loc}$ and $H^1_\mathrm{loc}$ respectively. The optimality of these results is demonstrated by several counterexamples of nontrivial solutions with compact support. \citeasnoun{kenig-nadirashvili-2000} found such solutions for $v \in L^1, n \geq 2$, and \citeasnoun{koch-tataru-2007-a} for $v \in L^p, p<n/2$, as well as $v \in L^{n/2}_w, n \geq 3$ (the weak $L^p$-space that is identical to the Lorentz space $L^{p,\infty}$), and $v \in L^1, n=2$.

An early overview on the topic was given by \citeasnoun{kenig-1986} and recent results including a new one are summed up in \citeasnoun{koch-tataru-2007-b}. The UCP property can even be used to show uniqueness in Cauchy problems like the time-dependent Schrödinger equation, the classical result being Holmgren's theorem. \citeasnoun{tataru-2004} gives an overview on this but without a discussion of permitted classes of potentials.

Note the form of an inequality for which such results are usually derived in the literature. From a PDE of the form $\Delta u = v u$ on a domain $\Omega \subset \R^n$ it follows of course that
\[
|\Delta u| \leq |v u|.
\]
Such a differential inequality is enough to derive the unique continuation properties for certain classes of potentials. In the most regular case $v \in L^\infty$ implies
\[
|\Delta u| \leq M |u|
\]
for some $M>0$ from which follow the classical UCP results, see \citeasnoun[XIII.63]{reed-simon-4}. An interesting consequence for the time-independent Schrö\-dinger equation $(-\Delta+v)u=Eu$ is the absence of strictly positive eigenvalues for potentials that have compact support, extended to a special class of potentials that tend to 0 as $|x| \rightarrow \infty$ in the Kato--Agmon--Simon theorem \cite[XIII.58]{reed-simon-4}. The reasoning for compactly supported potentials is as follows. A solution fulfils $-\Delta u=E u$ (the Helmholtz equation) on $\{ |x| > R \}$ for $R>0$ big enough. Then a classical result by \citeasnoun{rellich-1943} shows $u=0$ in $\{ |x| > R \}$ if $E>0$ and $u \in L^2(\R^n)$ due to a growth property of non-vanishing solutions. If the potential allows for the weak UCP then clearly $u=0$ everywhere. That this holds was thought to be true in general for all potentials that tend to 0 as $|x| \rightarrow \infty$ because intuitionally one would not expect a bounded state above zero energy because then only an infinite barrier should prevent some tunnelling. Yet there is a counterexample with a bounded potential that decays to zero but still admits a positive eigenvalue called the Wigner--von Neumann potential \cite[XIII.13 Ex.~1]{reed-simon-4}.

There are UCP results widening the class of permitted potentials to a class containing $L^{n/2}_\mathrm{loc}$ that are linked to the names of Morrey and Fefferman--Phong, see \citeasnoun{ruiz-vega} and \citeasnoun{kurata-1993}, the latter also including magnetic fields and a generalised elliptic differential operator. \citeasnoun{kurata-1997} again shows the UCP for certain Kato class potentials and magnetic fields. This still leaves open an important question raised for example in \citeasnoun{lammert}. The (weak) UCP provides proof that an eigenstate cannot vanish on an open set. But for the Hohenberg--Kohn theorem (\autoref{hk-th}) a non-vanishing property on any set with non-zero Lebesgue measure is critical. \citeasnoun{regbaoui} shows that such a sought for property is actually a consequence of the \emph{strong} UCP, therein apparently following the work of \citeasnoun{defigueiredo-gossez} that rests on an early estimate for general Sobolev functions by \citeasnoun[Lemma 3.4]{ladyzhenskaya-uraltseva}. This means the $v \in L^{n/2}_\mathrm{loc}$ is still needed, but such studies do not embrace the special features of the usual physical potentials as sums like in \autoref{def-sum-space} which might turn out to be beneficial. The same property for positive measure sets has been given by \citeasnoun{garafalo-lin-1987} in a side-note citing a previous work about more general elliptic operators that relates solutions to Muckenhoupt weights, but the admitted class of potentials (magnetic as well) is rather involved.

\subsection{Analyticity of eigenstates and densities}
\label{sect-analyticity-eigenstates}

In the beginning of the previous section it has already been noted that solutions to the Laplace equation are always analytic so one might wonder how much of this property remains valid if a singular potential like that of a Coulombic many-body system is added. This was extensively studied by \citeasnoun{fournais-2009} and in previous studies by the same authors. They found that in the presence of a Coulombic singularity in the potential at $x_i = x_0$ or $x_k = x_l$ in the case of particle--particle interaction any eigenfunction is analytic away from the singularities. In a neighbourhood of the singular set it can actually be represented as
\[
\psi(\ushort x) = \psi^{(1)}(\ushort x) + |x_i - x_0| \psi^{(2)}(\ushort x)
\]
or 
\[
\psi(\ushort x) = \psi^{(1)}(\ushort x) + |x_k - x_l| \psi^{(2)}(\ushort x)
\]
respectively with $\psi^{(1)}, \psi^{(2)}$ being real analytic. A basic example is the hydrogen 1s orbital that can be broken up straightforwardly into two power series.
\[
\e^{-|x|} = \underbrace{\sum_{k=0}^\infty \frac{(|x|^2)^k}{(2k)!}}_{\psi^{(1)}(x)} - |x| \underbrace{\sum_{k=0}^\infty \frac{(|x|^2)^k}{(2k+1)!}}_{-\psi^{(2)}(x)}
\]
This form exhibits the cusp structure already noted by \citeasnoun{kato-1957} who previously proved the wave function to be Lipschitz continuous. Actually the very heart of density functional theory already lies in his study which gets nicely expressed in an anecdote told by \citeasnoun{handy}:

\begin{quote}
When the two key theorems of modern Density Functional Theory were introduced in 1965, it is said that the eminent theoretical spectroscopist E.~Bright Wilson stood up at the meeting and said that he understood the basic principles of the theory. He said that if one knew the exact electron density $\rho(\mathbf{r})$, then the cusps of $\rho(\mathbf{r})$ would occur at the positions of the nuclei. Furthermore he argues that a knowledge of $|\nabla\rho(\mathbf{r})|$ at the nuclei would give their nuclear charges. Thus he argued that the full Schrodinger Hamiltonian was known because it is completely defined once the position and charge of the nuclei are given. Hence, in principle, the wavefunction and energy are known, and thus everything is known. In conclusion, Wilson said he understood that a knowledge of the density was all that was necessary for a complete determination of all molecular properties. It is this simple argument which is behind most of the aspirations of modern Density Functional Theory.
\end{quote}

But we have already noted in \autoref{sect-non-regularity-density} that analyticity properties do not have to pass over to the one-particle density and indeed in \citeasnoun{fournais-2004} the authors write:

\begin{quote}
It is not clear \emph{a priori} that $\rho$ is real analytic away from the nuclei since in [its definition] one integrates over subsets of [the singular set] $\Sigma$ where $\psi$ is not analytic.
\end{quote}

Yet analyticity of the density away from the positions of the nuclei holds for Coulombic many-body Hamiltonians and also for more general singular potentials. A recent alternative proof of that result using a clever unitary transformation was given by \citeasnoun{jecko-2010}. Still all those results only hold for the very special class of molecular potentials or potentials derived from those, not for a general class with vector space structure which is needed for many applications such as in the next chapters.

\section[Overview of results for time-dependent Hamiltonians]{Overview of results for time-dependent\\Hamiltonians}
\label{sect-overview-td-hamiltonian}

The physically significant case of a non-constant generator potential $v(t)$ is not even touched in the above considerations and demands some additional efforts. In this section we give an overview of some milestone results from the literature regarding this problem, either for the original Cauchy problem \eqref{cauchy-problem} or the specialised case of the Schrödinger equation \eqref{cauchy-problem-se}.

According to \citeasnoun{kato1961} the first investigations of an abstract Cauchy problem with a time-dependent family of infinitesimal generators date back to \citeasnoun{phillips}. Phillips considered an evolution equation on a Banach space $X$ 
\[
\partial_t u(t) = A(t) u(t)
\]
where the generator $A(t) = A_0 + B(t)$ is the perturbation of a time-independent part $A_0$ with well-known $\Cont^0$ semigroup $\exp(A_0 t)$ by a strongly continuously differentiable map $B : [0,\infty) \rightarrow \Bounded(X)$. The result of his Theorem 6.2 is the unique existence of strong solutions and its representation in his equation (25) is just the same as the Neumann series \eqref{neumann-series} discussed later. This does not come as a big surprise since he uses the same \q{successive substitutions} scheme already developed in the 19\textsuperscript{th} century in the context of integral equations. This scheme adds up contributions from all possible free propagations plus interactions and the task is then to prove convergence of the series. It corresponds to the transformation of the Cauchy (initial value) problem into a Volterra integral equation. More general potentials appeared in the work of \citeasnoun{miyadera-1966}, but first only for the static case. Building on that \citeasnoun{rhandi-2000} showed the existence of mild solutions under time-dependent (non-autonomous) perturbations and \citeasnoun{schnaubelt-voigt-1999} included non-autonomous Kato class potentials in the Schrödinger case.

In contrast to all this one can also consecutively propagate along a direct path for short times, assuming a generator with constant potential, then making those time intervals smaller to get the desired evolution operator in the limit. This next stage was achieved by \citeasnoun{kato1953} and others with $A$ having time-independent domain and being \q{maximal dissipative} at all times, a condition as in \autoref{lemma-resolvent} and thus automatically fulfilled for self-adjoint operators. Following \autoref{th-schro-dyn-konst} (or more generally the generation theorems of Hille--Yosida and Lumer--Phillips) we have $A(t)$ as the generator of a $\Cont^0$ semigroup. The proof is conducted using a time-partitioning with small evolution steps with static generator along every subinterval, the method is accordingly called \q{stepwise static approximation} by us. Further conditions regard relations of the generator at different times and may exclude typical cases of external potentials if naively applied. A more up-to-date summary of these techniques can be found in \citeasnoun[ch.~5]{pazy} who used the notion of \q{stable families of infinitesimal generators}. The method of \citeasnoun{kato1953} was later employed in \citeasnoun[Th.~X.70]{reed-simon-2} where they also develop a theory specifically for the Schrödinger case allowing singular Coulombic potential but only for one particle. This can be extended to the $N$-particle case with the methods given for the time-independent case in the same book and here in \autoref{sect-const-potential}. We actually  carry out this line of thought in some detail in \autoref{sect-stepwise-static} with a proof confined to the Schrödinger case.

An approach using a larger Banach space of maps $[0,T] \rightarrow X$, which we will later call the \q{trajectory space}, was already pursued by \citeasnoun{visik-ladyzhenskaya-1956}. They look for weak (\emph{faible}) solutions, similar to \citeasnoun{lions1958}, who allows for a time-indexed family of Hilbert spaces. Another idea using a larger space is due to \citeasnoun{howland} and mentioned in \citeasnoun{reed-simon-2}. He introduced a trajectory Hilbert space $L^2(\R,\H)$ with a new auxiliary time variable and conserved energy where the old time variable is taken as an ordinary coordinate. But all these results apply to abstract operators $A(t)$, not beneficially taking into account any special structure of a Schrödinger Hamiltonian.

\citeasnoun{wueller} considered the special case of the Schrödinger Hamiltonian for a single quantum particle with many moving Coulombic potentials. He is able to unitarily transform the equation to get a static singular potential and then uses results from \citeasnoun{tanabe} which in turn rely on \citeasnoun{kato1953}. A comparable technique is used in \citeasnoun{lohe-2009} to get exact solutions for a range of models. But the important case of singular, moving potentials for a larger number of quantum particles is still not addressed specifically.

The study of \citeasnoun{yajima}, in contrast, resumes the \q{successive substitutions} scheme and is for arbitrary spatial dimension $n$ allowing multiple particles in three-dimensional space. It entails Kato-like (cf.~\autoref{th-kato}) conditions on the spatial part of the potential, i.e., $v \in L^p(\R^n) + L^\infty(\R^n)$ and $p > \frac{n}{2}$ if $n \geq 3$, but holds for temporally discontinuous cases. He concentrates on the Schrödinger case and is \q{taking the characteristic features of Schrödinger equations into account [to] establish a theorem [...] for a larger class of potentials than in existing abstract theories.} The most significant such feature is the availability of Strichartz-type estimates for the Laplacian part, i.e., estimates on the evolved wave function with respect to the initial state. We give a very detailed analysis of his assumptions and proofs in \autoref{sect-yajima}. Yet the consequence of his main condition on the potential will be studied right away because it is important to mention that this effectively rules out the Coulombic case for more than one particle even though he writes that \q{[t]he conditions are general enough to accommodate potentials which have moving singularities of type $|x|^{-2+\varepsilon}$ for $n \geq 4$ and $|x|^{-n/2+\varepsilon}$ for $n \leq 3, \varepsilon > 0$.}

Let $v = v_1 + v_2$ with $v_1 \in L^p(\R^n), v_2 \in L^\infty(\R^n)$ be a radially symmetric, singular potential with its only pole at the origin. We can always assume the support of $v_1$ confined in a ball $r=|x| \leq 1$ because the outer part is bounded and thus in $L^\infty$. The $L^p$ condition now reads in spherical coordinates.
\[
\int_0^1 |v_1(r)|^p \, r^{n-1} \d r < \infty
\]
A singular potential of type $v_1(r) = -\alpha r^{-s}$ must therefore fulfil $-p s + n - 1 > -1$ for a converging norm integral which is the same as $s < \frac{n}{p}$ thus $s < 2$ by Yajima's assumption on the potential space for $n \geq 3$. This reproduces exactly the quote above, but such a potential is not of Coulombic type if more than one quantum particle in three-dimensional space is considered. Remember that the general form for a centred Coulomb potential for $N$ particles would be
\begin{equation}\label{coulomb-v}
v(x_1,\ldots,x_N) = - \alpha \sum_{i=1}^N r_i^{-1}.
\end{equation}
with $r_i = |x_i|$ the individual origin-particle distances. The $L^p$ condition thus reads for one of the most singular terms and $n=3N$ dimensional configuration space
\[
\int_{[0,1]^N} r_1^{-p+2} \d r_1 r_2^2 \d r_2 \ldots r_N^2 \d r_N < \infty
\mtext{i.e.} \int_0^1 r_1^{-p+2} \d r_1 < \infty,
\]
and we need $-p+2>-1$. Thus $\frac{n}{2} < p < 3$ which is not feasible for $n\geq 6$, the case of two or more particles. The problem persists for singular interaction terms of the kind $v(x_1,x_2) = -\alpha |x_1-x_2|^{-1}$. This effectively rules out the Coulomb case for systems of more than one particle and we have to keep that in mind. Still we will follow closely the exposition of \citeasnoun{yajima} as one possible method in our later development of potential variations for Schrödinger trajectories (see \autoref{sect-func-diff-successive}). Note that the situation changes in the case of a decoupled system as in the Kohn--Sham equations of TDDFT (see \autoref{sect-KS}), where each orbital is given by a Schrödinger equation in $n=3$ with a common effective potential.

A slightly more general recent account than in \citeasnoun{yajima} was given by \citeasnoun{dancona} just relying on a fixed point approach to show existence and uniqueness of a Schrödinger solution rather than explicitly constructing the evolution operator as a Neumann series. It is repeated in \autoref{Cv-strichartz}.

In any case because of the explicit time-dependence of the Hamiltonian, the evolution from a given state $\psi(s)$ to a state at a different time $\psi(t)$ is not only depending on the time difference $t-s$ but on the initial and final time separately. Instead of a one-parameter evolution (semi-)group one thus defines a so-called \emph{evolution system}.

\begin{definition}\label{def-evolution-system}
An \textbf{evolution system} is a two-parameter family of unitary  operators $U(t,s)$ on $\H$ for $0 \leq s,t \leq T$ fulfilling the following conditions:
\begin{enumerate}[(i)]
	\itemsep0em
	\item $U(t,t) = \id$
	\item $U(t,r) U(r,s) = U(t,s)$
	\item $U(t,s)^{-1} = U(t,s)^* = U(s,t)$
	\item $\partial_t U(t,s) = -\i\,H(t) U(t,s)$ on $D(H)$, the supposed common domain of the family of Hamiltonians $H(t)$ (classical solutions). From this together with (iii) it follows that $\partial_s U(t,s) = \i\, U(t,s) H(s)$ on $D(H)$ as well.
\end{enumerate}
\end{definition}

\section{The stepwise static approximation method}
\label{sect-stepwise-static}

\citeasnoun{reed-simon-2} give their proof for time-dependent, possibly unbounded generators of contraction semigroups based on \citeasnoun{kato1953} in Theorem X.70 and an example for the one-particle quantum mechanical case in Theorem X.71. We will repeat this proof in our own style and extend it to an arbitrary number of particles and Kato-perturbation type potentials following the strategy for time-independent Kato perturbations in \autoref{sect-kato-peturbations}. The central idea is to divide the time interval $[0,T]$ into small subintervals and evolve unitarily with a static Hamiltonian on each of these, thus the term \q{stepwise static approximation}. Then the limit to infinitesimally small subintervals is taken which leaves us with a $v \in \mathrm{Lip}([0,T], \Sigma(L^2 + L^\infty))$ (see \autoref{def-sobolev-kato-lipschitz}) condition for the potential.

\subsection{Preparatory definitions and lemmas}

The first lemma here will not be used in the present proof of the main theorem of existence of Schrödinger solutions but is interesting in its own right and closely relates to it.

\begin{lemma}\label{lemma-graph-norms-AB}
Let $A$ be a self-adjoint operator and $B(t)$ a family of symmetric and $A$-bounded operators with maximal relative bound $a<1$ and common $b$.
Then on $D(A)$ all graph norms are equivalent, $\|\cdot\|_A \sim \|\cdot\|_{A+B(t)}$, due to the inequalities
\[
\frac{1}{\sqrt{2}}\frac{1-a}{1+b} \|\varphi\|_A \leq \|\varphi\|_{A+B(t)} \leq \sqrt{2}\max\{1+a,1+b\} \|\varphi\|_A.
\]
\end{lemma}

Note that the conditions exactly resemble those of \autoref{th-kato-rellich} (Kato--Rellich) and thus $A+B(t)$ are all self-adjoint on $D(A)$. Equivalent norms mean that the topologies defined on $D(A)$ by these norms are all identical. As $D(A)$ with graph norm $\|\cdot\|_A$ is also a Hilbert space with inner product $\langle \cdot,\cdot \rangle_A = \langle \cdot,\cdot \rangle + \langle A\cdot,A\cdot \rangle$ we get homeomorphic Hilbert spaces.

\begin{proof}
We show equivalence between the norm $\|\cdot\|_A$ and all ${\|\cdot\|_{A+B(t)}}$ which automatically implies equivalence between all the individual norms ${\|\cdot\|_{A+B(t)}}$. The graph norm was defined as $\|\cdot\|_A = \sqrt{\|\cdot\|^2 + \|A\cdot\|^2}$ but is again equivalent to the more convenient norm $\|\cdot\| + \|A\cdot\|$ as we have seen in \autoref{lemma-equiv-graph-norm} and which we will use here.\\
First we show for all $t$ and $\varphi \in D(A)$ that we can find a constant $c_1 > 0$ such that
\[
\|\varphi\| + \|(A+B(t))\varphi\| \leq c_1( \|\varphi\| + \|A\varphi\| ).
\]
This part of the proof follows directly from the $A$-boundedness of the $B(t)$.
\begin{align*}
\|\varphi\| + \|(A+B(t))\varphi\| &\leq \|\varphi\| + \|A\varphi\| + \|B(t)\varphi\| \\
&\leq (1+b)\|\varphi\| + (1+a)\|A\varphi\| \\
&\leq \max\{1+a,1+b\} (\|\varphi\| + \|A\varphi\|)
\end{align*}
The backwards part is to show that we can find a constant $c_2 > 0$ such that
\[
\|\varphi\| + \|A\varphi\| \leq c_2( \|\varphi\| + \|(A+B(t))\varphi\| ).
\]
We start by introducing $B(t)-B(t)$ into $\|A\varphi\|$.
\begin{align*}
\|\varphi\| + \|A\varphi\| &= \|\varphi\| + \|(A+B(t)-B(t))\varphi\| \\
&\leq \|\varphi\| + \|(A+B(t))\varphi\| + \|B(t)\varphi\|
\end{align*}
In the next step we use the $A$-boundedness of the $B(t)$ again.
\[
\|\varphi\| + \|A\varphi\| \leq (1+b)\|\varphi\| + a\|A\varphi\| + \|(A+B(t))\varphi\| \\
\]
Finally $0 < 1-a \leq 1$ yields
\[
(1-a) (\|\varphi\| + \|A\varphi\|) \leq (1+b)\|\varphi\| + \|(A+B(t))\varphi\|
\]
and thus because of $b \geq 0$
\[
\frac{1-a}{1+b} (\|\varphi\| + \|A\varphi\|) \leq \|\varphi\| + \|(A+B(t))\varphi\|.
\]
The inequalities of the graph norms follow after converting with \autoref{lemma-equiv-graph-norm}.
\end{proof}

We define a potential space for which not only $v$ but also its derivatives up to a certain order $m \in \N_0$ are of Kato perturbation type. We will not need higher derivatives in the proof for existence of Schrödinger solutions but later in a related proof that shows higher regularity of solutions.

\begin{definition}[\q{Sobolev--Kato space}]
\label{def-sobolev-kato}
\[
W^{m,\Sigma} = \{ v \mid D^{\alpha}v \in \Sigma(L^2 + L^\infty), |\alpha| \leq m \}
\]
Its canonical norm is defined with the help of the $\Sigma$-norm from \autoref{def-sum-space}.
\[
\|v\|_{m,\Sigma} = \left( \sum_{|\alpha| \leq m} \| D^\alpha v \|_\Sigma^2 \right)^{1/2} \sim \sum_{|\alpha| \leq m} \| D^\alpha v \|_\Sigma
\]
\end{definition}

We note $W^{0,\Sigma} = \Sigma = \Sigma(L^2 + L^\infty)$. Having the potentials depend on time we also need to put time-dependence into their function space and the necessary property will be Lipschitz-continuity in time with respect to the topology of the spatial part just defined, while the potential itself is in $W^{m,\Sigma}$ at all times. See \citeasnoun{johnson-1970} for a discussion of function spaces like $\Lip([0,T],X)$ for arbitrary Banach space $X$ that are Banach spaces again.

\begin{definition}[\q{Sobolev--Kato--Lipschitz space}]
\label{def-sobolev-kato-lipschitz}
\[
\Lip([0,T],W^{m,\Sigma})
\]
A norm for this space can be defined in the following way.
\[
\max_{t\in[0,T]} \|v(t)\|_{m,\Sigma} + \max_{t_2 > t_1} \frac{\|v(t_2)-v(t_1)\|_{m,\Sigma}}{t_2-t_1}
\]
\end{definition}

Note that the space of Lipschitz continuous functions is just $W^{1,\infty}$ (see \autoref{sect-lip-abs-cont}) because $[0,T]$ is clearly convex, yet another Sobolev space!
\[
\Lip([0,T],W^{m,\Sigma}) = W^{1,\infty}([0,T],W^{m,\Sigma})
\]
We have already conjectured such conditions for potentials in \citeasnoun{tddft-review} where it was also stated that physically they are quite reasonable. For $\Omega = \R^3$ and the potential given from a charge distribution $\rho \in \Lip([0,T],L^2)$ by integrating it with Green's function for the Laplace equation, the resulting potential is in $\Lip([0,T],H^2)\subset \Lip([0,T],W^{2,\Sigma})$.
\[
v(t,x) = -\frac{1}{4\pi}\int \frac{\rho(t,x')}{|x-x'|} \d x'
\]
This should not come as a surprise because this $v$ solves the Poisson equation $\Delta v = \rho$. Point sources such as delta distributions are excluded but charge distributions in molecules might also be modelled by finite functions. \cite{andrae-2000}

We start now with an extension of \autoref{lemma-sum-space-inequality} to Sobolev--Kato spaces. Because the original lemma used $D(H_0)$ for wave functions we arrive at the self-adjoint domain of $H_0^m$ (not the Sobolev space of course, the free Hamiltonian raised to the power of $m$). This domain is $H^{2m}(\R^n)$ if the configuration space is the whole $\R^n$ or if we consider periodic functions on some tessellated $\Omega \subset \R^n$, and it is $H^{2m}(\Omega)\cap H_0^{m}(\Omega)$ otherwise. In all cases we will simply write $H^{2m}$ as the corresponding domain of the self-adjoint $H_0^m$.

\begin{lemma}\label{lemma-rs-0}
For $v \in W^{2m,\Sigma}, \varphi \in H^{2(m+1)}$ it holds
\[
\|v \varphi\|_{2m,2} \lesssim \|v\|_{2m,\Sigma} \cdot \|\varphi\|_{2(m+1),2}.
\]
\end{lemma}

\begin{proof}
Take a test function $\varphi \in \Cont^\infty_0$. From \autoref{lemma-sum-space-inequality} directly follows the case for $m=0$, i.e., $\|v \varphi\|_{2} \leq \sqrt{2} \|v\|_{\Sigma} \|\varphi\|_{2,2}$, and proceeding from that we give the proof for arbitrary orders $m$. We start by writing out the involved Sobolev space norm \eqref{eq-sobolev-norm} explicitly, then employ the general Leibniz rule for multivariable calculus, and eventually use \autoref{lemma-sum-space-inequality} on every single term of the sum.
\begin{align*}
\|v \varphi\|_{2m,2} \sim \smashoperator{\sum_{|\alpha| \leq 2m}} \left\| D^\alpha (v \varphi) \right\|_2 &= \sum_{|\alpha| \leq 2m} \Big\| \sum_{\nu \leq \alpha} {\alpha \choose \nu} D^\nu v \cdot D^{\alpha-\nu} \varphi \Big\|_2 \\
&\lesssim \sum_{|\alpha| \leq 2m} \sum_{\nu \leq \alpha} \| D^\nu v \cdot D^{\alpha-\nu} \varphi \|_2 \\
&\lesssim \sum_{|\alpha| \leq 2m} \sum_{\nu \leq \alpha} \| D^\nu v \|_\Sigma \cdot \| D^{\alpha-\nu} \varphi \|_{2,2} \\
&\leq \smashoperator{\sum_{|\alpha| \leq 2m}} \| D^\alpha v \|_\Sigma \cdot \smashoperator{\sum_{|\beta| \leq 2m}} \| D^{\beta} \varphi \|_{2,2} \\
&\sim \|v\|_{2m,\Sigma} \cdot \smashoperator{\sum_{|\beta| \leq 2(m+1)}} \| D^{\beta} \varphi \|_{2} \\
&\sim \|v\|_{2m,\Sigma} \cdot \sum_{l=0}^{m+1} \|\Delta^l\varphi\|_2
\end{align*}

\autoref{th-sobolev-norm-laplace} was used for the last equivalence between Sobolev-type norms. Now finally we employ a denseness argument just like in the proof of \autoref{th-kato} to arrive at the domain $D(\Delta^{m+1})=H^{2(m+1)}$ on which all $\Delta^l$, $0 \leq l \leq m+1$, are well-defined which concludes the proof.
\end{proof}

\begin{lemma}\label{lemma-rs-2}
For $v \in W^{2(m-1),\Sigma}$ the operator $H[v]^m = (H_0 + v)^m : H^{2m} \rightarrow L^2$ is bounded by
\[
\Pi_m[v] \sim \prod_{l=0}^{m-1} (1+\|v\|_{2l,\Sigma})
\]
depending continuously on $v$.
\end{lemma}

\begin{proof}
Because $H_0 = -\onehalf \Delta$ and $v$ clearly do not commute it is hard to write out the operator explicitly. We can help ourselves using an iterative reasoning and referring repeatedly to \autoref{lemma-rs-0}.
\begin{align*}
\|(H_0 + v)^m \varphi\|_2 &= \|(H_0 + v)(H_0 + v)^{m-1} \varphi\|_2 \\
&\leq  \|(H_0 + v)^{m-1} \varphi\|_{2,2} + \|v(H_0 + v)^{m-1} \varphi\|_2 \\
&\lesssim (1+\|v\|_\Sigma) \cdot \|(H_0 + v)^{m-1} \varphi\|_{2,2} \\
&\lesssim \cdots \lesssim \prod_{l=0}^{m-1} (1+\|v\|_{2l,\Sigma}) \cdot \|\varphi\|_{2m,2}
\end{align*}
\end{proof}


\begin{lemma}\label{lemma-rs-1}
For $v \in W^{2(m-1),\Sigma}$ with $H[v] \geq 1$ the inverse operator $H[v]^{-m} = (H_0 + v)^{-m} : L^2 \rightarrow H^{2m}$ is bounded by 1.
\end{lemma}

\begin{proof}
We have $H[v]^m : H^{2m} \rightarrow L^2$ as an operator as in \autoref{lemma-rs-2} that is self-adjoint as the power of an self-adjoint operator. Now we also have to show that $H[v]^m \geq 1$ to successfully employ \autoref{lemma-positive-inverse} and have an inverse bounded by 1. We start with $m=2$ and derive the estimate usually given for the variance of an operator in the state $\varphi$ normalised to $\|\varphi\|_2=1$.
\[
0 \leq \langle (H[v]- \langle H[v] \rangle_\varphi )^2 \rangle_\varphi = \langle H[v]^2 \rangle_\varphi - \langle H[v] \rangle_\varphi^2
\]
Thus by assumption
\[
\langle H[v]^2 \rangle_\varphi \geq \langle H[v] \rangle_\varphi^2 \geq 1.
\]
We proceed further by induction from $m$ to $m+2$ spanning all natural numbers because we now have the liberty to start at $m=1$ or 2.
\[
\langle H[v]^{m+2} \rangle_\varphi = \langle H[v]\varphi, H[v]^m H[v]\varphi \rangle \geq \langle H[v]\varphi, H[v]\varphi \rangle=\langle H[v]^2 \rangle_\varphi \geq 1.
\]
\end{proof}

\begin{lemma}\label{lemma-rs-3}
Let $v \in \Lip([0,T], W^{2(m-1),\Sigma}), m\geq 1$, with Lipschitz constant $L$, and $H([v],t) \geq 1$ for all $t \in [0,T]$. Then $H([v],t)^m H([v],t')^{-m}$ $= K_m([v],t,t') + \id : L^2 \rightarrow L^2$ with $K_m([v],t,t')$ bounded by $C_m[v] L |t-t'|$, $C_m[v]$ continuous in $v$ and $C_1[v]=\sqrt{2}$.
\end{lemma}

\begin{proof}
We omit the functional argument $[v]$ for brevity and start with $m=1$. 
\begin{align*}
H(t)H(t')^{-1} &= (H(t)-H(t')) H(t')^{-1} + \id \\
&= (v(t)-v(t')) H(t')^{-1} + \id = K_1(t,t') + \id
\end{align*}
We can derive a bound with \autoref{lemma-sum-space-inequality} and \autoref{lemma-rs-1}.
\[
\|K_1(t,t') \varphi\|_2 \leq \sqrt{2} \|v(t)-v(t')\|_\Sigma \|H(t')^{-1} \varphi\|_{2,2} \leq \sqrt{2} \|v(t)-v(t')\|_\Sigma \|\varphi\|_{2}
\]
Introducing the Lipschitz constant $L$ of $v \in \Lip([0,T], W^{0,\Sigma}) = \Lip([0,T], \Sigma(L^2+L^\infty))$ we get
\[
\|K_1(t,t') \varphi\|_2 \leq \sqrt{2} |t-t'| \frac{\|v(t)-v(t')\|_\Sigma}{|t-t'|} \|\varphi\|_{2} \leq \sqrt{2} L |t-t'| \cdot \|\varphi\|_{2}
\]
which proves the proposition for $m=1$. Now assume it to hold for for $K_{m-1}$, so we want to prove it for $m>1$ by induction.
\begin{align*}
H(t)^m H(t')^{-m} &= H(t)^{m-1} \left( (v(t)-v(t')) H(t')^{-1} + \id\right) H(t')^{-m+1} \\
&= H(t)^{m-1} (v(t)-v(t'))H(t')^{-m} + H(t)^{m-1} H(t')^{-m+1} \\
&= H(t)^{m-1} (v(t)-v(t'))H(t')^{-m} + K_{m-1}(t,t') + \id
\end{align*}
We only need to find the corresponding estimate for the first term so we use \autoref{lemma-rs-2} to get rid of the preceding $H(t)^{m-1}$, then \autoref{lemma-rs-0} to get an estimate separating potential and wave function, and finally \autoref{lemma-rs-1} to get rid of $H(t')^{-m}$.
\begin{align*}
\| H(t)^{m-1} &(v(t)-v(t'))H(t')^{-m} \varphi \|_2 \\
&\leq \Pi_{m-1}[v(t)] \cdot\| (v(t)-v(t'))H(t')^{-m} \varphi \|_{2(m-1),2} \\
&\lesssim \Pi_{m-1}[v(t)] \cdot\|v(t)-v(t')\|_{2(m-1),\Sigma} \cdot \|H(t')^{-m} \varphi\|_{2m,2} \\
&\leq \Pi_{m-1}[v(t)] \cdot\|v(t)-v(t')\|_{2(m-1),\Sigma} \cdot \|\varphi\|_{2}
\end{align*}
Like before we introduce the Lipschitz constant $L$, this time corresponding to $v \in \Lip([0,T], W^{2(m-1),\Sigma})$ and surely larger than $L$ for lower $m$. The other iteratively aggregated constants are collected together with the $\Pi_{l}[v(t)]$ in a term $C_m[v]$ depending continuously on $v$ and the proposition follows.
\end{proof}

\subsection{Existence of Schrödinger solutions}
\label{sect-existence-stepwise-static}

\begin{theorem}\label{th-schro-dyn-td}
In the case of $H(t) = H_0 + v(t)$ with $v \in \Lip([0,T],$ $\Sigma(L^2 + L^\infty))$ a solution to the Schrödinger equation \eqref{cauchy-problem-se} is given by a unitary evolution system $\psi(t) = U(t,0)\psi_0$ given by the limit procedure \eqref{U-limit} below. For an initial state $\psi_0 \in D(H_0)$ there holds $\psi(t) \in D(H_0)$ (classical solution) but the domain of the evolution system can be extended to all $\psi_0 \in \H$ (generalised solution).
\end{theorem}

\begin{proof}
We start by dividing the time-interval $[0,T]$ into subintervals of length $1/k$ with some $k \in \N$. Let $t$ be within the $i$-th interval, i.e., $t \in [(i-1)/k,i/k] = [\lfloor t k \rfloor/k, (\lfloor t k \rfloor+1)/k]$ and $i \in \{1, \ldots, \lceil T k \rceil\}$. Then we construct an approximation to the evolution system by $U_k(t,s)$ composed by stepwise evolutions with constant Hamiltonians.
\begin{equation}\label{eq-schro-dyn-td-unitary-approx}
\begin{aligned}
U_k(t,s) &= \exp\left(-\i H\left(\frac{i-1}{k}\right) (t-s) \right) &\mtext{if} \frac{i-1}{k} \leq s \leq t \leq \frac{i}{k} \\[0.5em]
U_k(t,s) &= U_k(t,r) U_k(r,s) &\mtext{if} 0 \leq s \leq r \leq t \leq T
\end{aligned}
\end{equation}
The existence of each individual exponential as a unitary operator $D(H_0) \rightarrow D(H_0)$ is guaranteed by \autoref{th-schro-dyn-konst} (Stone's theorem). We now want to show that for all $\psi_0 \in D(H_0)$
\begin{equation}\label{U-limit}
U(t,0)\psi_0 = \lim_{k \rightarrow \infty} U_k(t,0)\psi_0
\end{equation}
converges uniformly in the strong $L^2$ topology and that it yields an appropriate evolution system for the Schrödinger equation. Convergence is guaranteed if the $U_k(t,0)\psi_0$ form a Cauchy sequence. We thus form
\begin{equation}\label{eq-U-trick}
\begin{aligned}
(U_k(t,0)&-U_l(t,0))\psi_0 = U_l(t,s) U_k(s,0)\psi_0\Big|_{s=0}^t \\
&= \int_0^t \partial_s (U_l(t,s) U_k(s,0)\psi_0) \d s = \\
&= \int_0^t \big((\partial_s U_l(t,s)) U_k(s,0) + U_l(t,s) (\partial_s U_k(s,0))\big)\psi_0 \d s.
\end{aligned}
\end{equation}
Now if $s \neq j/k$ (the integral is then really a sum of Riemann integrals over intervals where the derivative exists) we can strongly differentiate according to the definition \eqref{eq-schro-dyn-td-unitary-approx}.
\begin{align*}
\partial_s U_k(s,0) &= -\i H\left( \frac{\lfloor sk \rfloor}{k} \right) U_k(s,0) \\
\partial_s U_l(t,s) &= \i U_l(t,s) H\left( \frac{\lfloor sl \rfloor}{l} \right)
\end{align*}
We get the following sequence of differences.
\begin{equation}\label{eq-schro-dyn-td-cauchy}
\begin{aligned}
(U_k(t,0)&-U_l(t,0))\psi_0 \\
&= \i \int_0^t U_l(t,s) \left( H\left( \frac{\lfloor sl \rfloor}{l} \right) - H\left( \frac{\lfloor sk \rfloor}{k} \right) \right) U_k(s,0)\psi_0 \d s \\
&= \i \int_0^t U_l(t,s) \left( v\left( \frac{\lfloor sl \rfloor}{l} \right) - v\left( \frac{\lfloor sk \rfloor}{k} \right) \right) U_k(s,0)\psi_0 \d s
\end{aligned}
\end{equation}
One could think that $v \in \Cont^0([0,T], \Sigma(L^2 + L^\infty))$ is now enough to prove that this sequence is Cauchy because $\lfloor sl \rfloor/l - \lfloor sk \rfloor/k \rightarrow 0$ as $l,k \rightarrow \infty$, but we still need $\lim_{k \rightarrow \infty} U_k(s,0)\psi_0 \in D(H_0)$ strongly in $H_0$ graph norm topology. This is because we need to guarantee a value in $L^2$ after the multiplication with the potential difference in $\Sigma(L^2 + L^\infty)$ by \autoref{lemma-sum-space-inequality}. The idea is to push the Hamiltonians trough all the short-time evolutions until they get applied to the initial state. This is achieved by using the fact that we have a bounded inverse operator $H^{-1}(t) : L^2 \rightarrow D(H_0)$ for \emph{strictly positive} Hamiltonians by \autoref{lemma-positive-inverse}. The trick is now just to add a big enough constant $c>0$ to $H(t)$ bounded below (this is already guaranteed by \autoref{th-sum-space}) to get a self-adjoint operator $\tilde{H}(t) = H(t) + c \geq 1$ for all $t \in [0,T]$ such that $\tilde{H}^{-1}(t) : L^2 \rightarrow D(H_0)$ surely exists. It can be introduced into all unitary evolution operators as well and in conclusion we then simply transform back to the usual evolution system by
\[
U(t,0) = \e^{\i c t} \tilde{U}(t,0).
\]
So without loss of generality we assume $H(t) \geq 1$ for all $t \in [0,T]$ with an inverse operator bounded by 1 in $L^2 \rightarrow L^2$ operator norm by \autoref{lemma-positive-inverse}. We start now by splitting the approximated evolution system into its individual short-interval parts.
\begin{equation}\label{eq-schro-dyn-td-U-splitting}
U_k(s,0) = U_k\left(s,\frac{\lfloor sk \rfloor}{k}\right) U_k\left(\frac{\lfloor sk \rfloor}{k},\frac{\lfloor sk \rfloor-1}{k}\right) \ldots U_k\left(\frac{1}{k},0\right)
\end{equation}
Next we introduce the identity $H^{-1}(t)H(t) : D(H_0) \rightarrow D(H_0)$ in front of every such short-time evolution taking $t$ for the lower bound of the corresponding time interval. This means we can swap $H(t)$ and the evolution operator according to \autoref{cor-evolut-stable} where we showed that the evolution for a constant Hamiltonian stabilises its domain.
\begin{align*}
U_k(s,0) = &H^{-1}\left(\frac{\lfloor sk \rfloor}{k}\right) U_k\left(s,\frac{\lfloor sk \rfloor}{k}\right) H\left(\frac{\lfloor sk \rfloor}{k}\right) \\
&H^{-1}\left(\frac{\lfloor sk \rfloor-1}{k}\right) U_k\left(\frac{\lfloor sk \rfloor}{k},\frac{\lfloor sk \rfloor-1}{k}\right)  H\left(\frac{\lfloor sk \rfloor-1}{k}\right) \cdots\\
&H^{-1}(0) U_k\left(\frac{1}{k},0\right) H(0)
\end{align*}
We rewrite the $H$-$H^{-1}$ encounters as in \autoref{lemma-rs-3} omitting the index of $K_1=K$.
\begin{align*}
U_k(s,0) = &H^{-1}\left(\frac{\lfloor sk \rfloor}{k}\right) U_k\left(s, \frac{\lfloor sk \rfloor}{k} \right) \\
&\prod_{j=1}^{\lfloor sk \rfloor} \left[ K\left( \frac{j}{k},\frac{j-1}{k} \right) + \id \right] U_k\left( \frac{j}{k}, \frac{j-1}{k} \right) H(0)
\end{align*}
(Note that the time-ordering in the product is important because the involved terms do not commute and will lead right to left from small to large times.) We use \autoref{lemma-rs-1} to get rid of the preceding $H^{-1}$ when evaluating the $H^2$-norm and \autoref{lemma-rs-3} tells us that the operator $K : L^2 \rightarrow L^2$ is bounded by $\sqrt{2} L /k$. Finally $H(0)$ is estimated by \autoref{lemma-rs-2} and the limit $k\rightarrow\infty$ for the sequence of approximations to the evolution system applied to a $\psi_0 \in D(H_0)$ is taken.
\begin{align}\label{eq-schro-evolut-estimate}
\|U_k(s,0)\psi_0\|_{2,2} &\leq \left\| \prod_{j=1}^{\lfloor sk \rfloor} \left[ K\left( \frac{j}{k},\frac{j-1}{k} \right) + \id \right] U_k\left( \frac{j}{k}, \frac{j-1}{k} \right) H(0) \psi_0 \right\|_2 \nonumber\\
&\leq \prod_{j=1}^{\lfloor sk \rfloor} \left[ \frac{\sqrt{2} L}{k} +1 \right] \|H(0)\psi_0\|_2 \nonumber\\
&\leq \left[ \frac{\sqrt{2} L}{k} +1 \right]^{\lfloor sk \rfloor} (1+\|v(0)\|_\Sigma) \cdot \|\psi_0\|_{2,2} \nonumber\\
&= \left(\left[ \frac{\sqrt{2} L}{k} +1 \right]^k \right)^{\lfloor sk \rfloor/k} (1+\|v(0)\|_\Sigma) \cdot \|\psi_0\|_{2,2} \nonumber\\[0.5em]
&\longrightarrow (1+\|v(0)\|_\Sigma) \cdot \e^{\sqrt{2} L s} \|\psi_0\|_{2,2}
\end{align}
This result also proves the $H^2$-regularity of solutions (i.e., classical solutions) subject to evolutions with $\mathrm{Lip}([0,T],\Sigma(L^2 + L^\infty))$ potentials. Higher order regularity will be addressed in \autoref{sect-regularity-sobolev}. (Note a certain relatedness of this result to estimates like Gronwall's inequality for ordinary differential equations.)\\
The task was to show (uniform) convergence of the sequence of unitary $U_k(t,0) : L^2 \rightarrow L^2$ by checking the Cauchy condition. So from \eqref{eq-schro-dyn-td-cauchy} we finally have the following estimate which holds on the dense subset $\psi_0 \in D(H_0) \subset L^2$.
\begin{align*}
\|&(U_k(t,0)-U_l(t,0))\psi_0\|_2 \\
&\leq \int_0^t \left\|\left( v\left( \frac{\lfloor sl \rfloor}{l} \right) - v\left( \frac{\lfloor sk \rfloor}{k} \right) \right) U_k(s,0)\psi_0\right\|_2 \d s \\
&\leq \int_0^t \sqrt{2} (1+\|v(0)\|_\Sigma) \cdot \e^{\sqrt{2} L s} \left\| v\left( \frac{\lfloor sl \rfloor}{l} \right) - v\left( \frac{\lfloor sk \rfloor}{k}\right) \right\|_{\Sigma} \|\psi_0\|_{2,2} \d s \\
&\leq L^{-1}(1+\|v(0)\|_\Sigma) \cdot \left( \e^{\sqrt{2} L t}-1 \right) \max_{s \in [0,t]} \left\| v\left( \frac{\lfloor sl \rfloor}{l} \right) - v\left( \frac{\lfloor sk \rfloor}{k}\right) \right\|_{\Sigma} \|\psi_0\|_{2,2} \\[0.4em]
&\longrightarrow 0 \mtext{as} k,l \rightarrow \infty
\end{align*}
We used \autoref{lemma-sum-space-inequality} again and have a null sequence clearly because of the continuity argument brought forward already after \eqref{eq-schro-dyn-td-cauchy}. Uniform convergence is guaranteed because $t \in [0,T]$ is from a bounded set. Uniform boundedness as in \autoref{th-schro-dyn-konst} (Stone's theorem) shows existence of a unitary $U(t,0)$ on the whole Hilbert space $L^2$. Finally the properties of an evolution system are ensured by construction, the uniform convergence in \eqref{U-limit} allowing the commutation of the defining limit and the differentiation by time.
\end{proof}

Note that instead of the Lipschitz property for $v$ we could also make the stronger assumption of continuous differentiability $v \in \Cont^1([0,T],$ $\Sigma(L^2 + L^\infty))$ and set $L = \max_{t\in[0,T]} \|\dot{v}(t)\|_{\Sigma}$. This condition appears in the original proof of \citeasnoun{reed-simon-2}. A similar proof strategy is also used in \citeasnoun[ch.~5, Th.~3.1]{pazy} but in the more general setting of a `hyperbolic' evolution equation resting on \citeasnoun{kato1973}. It only requires continuity in time with a remark added addressing the possibility of an even weaker condition on the generator family being $L^1$ in time. Yet this proof has not been studied that thoroughly and we rely on the results above for the derivations in the following section thus sticking to the Lipschitz property for potentials.

\section{Regularity in higher Sobolev norms}
\label{sect-regularity-sobolev}

\begin{xquote}{\citeasnoun{tao-book}}
The analysis of a PDE is a beautiful subject, combining the rigour and technique of modern analysis and geometry with the very concrete real-world intuition of physics and other sciences. Unfortunately, in some presentations of the subject (at least in pure mathematics), the former can obscure the latter, giving the impression of a fearsomly technical and difficult field to work in.
\end{xquote}

\subsection{Motivation from TDDFT}

Our basic equation in TDDFT will later be the \q{divergence of local forces} equation \eqref{eq-div-force-density} involving (weak) spatial partial derivatives of $\psi[v]$ up to 4\textsuperscript{th} order. The task is now to show that this quantity is well-defined, yet up to now only $H^2$ regularity of trajectories $\psi[v]$ (classical solutions) is guaranteed for certain potentials by \autoref{th-schro-dyn-td}. More generally we want to limit ourselves to potentials that allow unique solutions of Schrödinger's equation in the Sobolev space $H^s$ of $s$-times (weakly) differentiable functions where all derivatives are still in $L^2$.
\[
\left.
\begin{array}{rcl}
	\psi_0 &\in& H^s \\
	v &\in& \mathrm{?}
\end{array}
\right\} \Rightarrow \psi([v],t) \in H^s
\]
Note the comparable definition of \emph{$H$-admissible} in \citeasnoun[Def.~6.2.2]{jerome}: A space $X$ is $H$-admissible if $\{\exp(-\i Ht)\}_t$ leaves $X$ invariant and forms a $\Cont^0$ semigroup on $X$. For time-dependent $H(t)$ one needs this property for the time-indexed family of semigroups.

We will first give a short overview about literature on the so-called \q{growth of Sobolev norms} (rather a limit to their growth), then pursue a do-it-yourself approach following the existence proof in \autoref{sect-existence-stepwise-static}.

\subsection{On growth of Sobolev norms}

The first reference to results like this has been found in the neat online review by \citeasnoun{staffilani} on Schrödinger equations with spatially periodic boundary conditions (i.e., on a torus). The work is concerned with semilinear Schrödinger equations of the Gross--Pitaevskii type but the associated proofs make use of (periodic) Strichartz' estimates (cf.~\eqref{strichartz-ineq}) for free evolution. (Which, it seems, need a frequency cut-off.) Then an estimate for the $H^s$ Sobolev norm for solutions to the semilinear equation is given, which already lead us into the desired direction. Staffilani notes:

\begin{quote}
The growth of high Sobolev norms has a physical interpretation in the context of the \emph{Low-to-High frequency cascade}. In other words, we see that $\|u(t)\|_{H^s}$ weighs the higher frequencies more as $s$ becomes larger, and hence its growth gives us a quantitative estimate for how much of the support of $|\hat{u}|^2$ has transferred from the low to the high frequencies. This sort of problem also goes under the name \emph{weak turbulence}.
\end{quote}

Here the phenomenon gets a real physical interpretation, a trace we followed online to \citeasnoun{tao-book}, to \citeasnoun{tao-blog}, and to \citeasnoun{bourgain} who writes on the linear Schrödinger equation with time-dependent (actually a rare sight!) potential. The proof has been recently refined by \citeasnoun{delort} and we cite his result. The potential shall be limited to smooth functions in $t, x$ with bounded derivatives and periodic in $x$ on a Zoll manifold (something periodic like a sphere or torus). Then for any $s>0$ there is a constant $C>0$ such that for all $u_0 \in H^s$ the solution satisfies, for any $t\in \R$,
\[
\|u(t)\|_{H^s} \leq C(1+|t|)\|u_0\|_{H^s}.
\]
Since the \q{smooth} condition is the only real downside in this result, it should be studied if a certain mitigation for particular maximal $s \in \N$ is possible. How this might look like is shown in the next section, where the $H^s$ stability is derived from the existence proof in \autoref{sect-existence-stepwise-static}. We are in a way able to generalise the result of \citeasnoun{delort} by including non-smooth potentials, but on the downside we give a rougher estimate and have regularity only up to a certain $s\in \N$.

\subsection{Sobolev regularity from the stepwise static approximation}

Here Sobolev regularity means the preservation of a $H^s$ property under evolution with certain potentials. Because we use the methods developed in the section before we are actually limited to even orders $s=2m$. The case $s=0$ clearly follows directly from the unitarity of the evolution system.

\begin{theorem}\label{th-sobolev-regularity}
The evolution system generated by $H(t) = H_0 + v(t)$ with $v \in \Lip([0,T], W^{2(m-1),\Sigma})$ preserves $H^{2m}$-regularity for $m \in \N$, i.e., if $\psi_0 \in H^{2m}$ then also $\psi(t) = U(t,0)\psi_0 \in H^{2m}$.
\end{theorem}

\begin{proof}
We need to check if $\|U(t,0) \psi_0\|_{2m,2}$ stays finite thus the modus operandi is very much in the spirit of the proof of \autoref{th-schro-dyn-td} which already settled the case $m=1$. Like in \eqref{eq-schro-dyn-td-U-splitting} for regularity in $D(H_0) \subseteq H^2$ we start by splitting the approximated evolution system into its individual short-interval parts with constant potentials.
\[
U_k(s,0) = U_k\left(s,\frac{\lfloor sk \rfloor}{k}\right) U_k\left(\frac{\lfloor sk \rfloor}{k},\frac{\lfloor sk \rfloor-1}{k}\right) \ldots U_k\left(\frac{1}{k},0\right)
\]
Next we introduce the identity $H^{-m}(t)H^m(t) : H^{2m} \rightarrow H^{2m}$ in front of every such short-time evolution taking $t$ as the lower bound of the corresponding time interval. This means we can swap $H^m(t)$ and the evolution operator with the same potential.
\begin{align*}
U_k(s,0) = \;&H^{-m}\left(\frac{\lfloor sk \rfloor}{k}\right) U_k\left(s,\frac{\lfloor sk \rfloor}{k}\right) H^m\left(\frac{\lfloor sk \rfloor}{k}\right) \\
&H^{-m}\left(\frac{\lfloor sk \rfloor-1}{k}\right) U_k\left(\frac{\lfloor sk \rfloor}{k},\frac{\lfloor sk \rfloor-1}{k}\right) H^m\left(\frac{\lfloor sk \rfloor-1}{k}\right) \cdots\\
&H^{-m}(0) U_k\left(\frac{1}{k},0\right) H^m(0)
\end{align*}
We rewrite the $H^m$-$H^{-m}$ encounters again as in \autoref{lemma-rs-3}.
\begin{align*}
U_k(s,0) = \;&H^{-m}\left(\frac{\lfloor sk \rfloor}{k}\right) U_k\left(s, \frac{\lfloor sk \rfloor}{k} \right) \\
&\prod_{j=1}^{\lfloor sk \rfloor} \left[ K_m\left( \frac{j}{k},\frac{j-1}{k} \right) + \id \right] U_k\left( \frac{j}{k}, \frac{j-1}{k} \right) H^m(0)
\end{align*}
Like the proof of \autoref{th-schro-dyn-td} we assume without loss of generality $H(t) \geq 1$ for all $t \in [0,T]$ and thus \autoref{lemma-rs-3} tells us that the operator $K_m : L^2 \rightarrow L^2$ is bounded by $C_m[v] L /k$. We use \autoref{lemma-rs-1} to get rid of the preceding $H^{-m}$ when evaluating the $H^{2m}$-norm and \autoref{lemma-rs-2} to estimate $H^m(0)$.
\begin{equation}\label{eq-schro-evolut-estimate-2m}
\begin{aligned}
\|U_k(s,0)\psi_0\|_{2m,2} &\lesssim \left\| \prod_{j=1}^{\lfloor sk \rfloor} \left[ K_m\left( \frac{j}{k},\frac{j-1}{k} \right) + \id \right] U_k\left( \frac{j}{k}, \frac{j-1}{k} \right) H^m(0) \psi_0 \right\|_2 \\
&\leq \left[ \frac{C_m[v] L}{k} +1 \right]^{\lfloor sk \rfloor} \|H^m(0) \psi_0\|_2 \\
&\lesssim \Pi_m[v(0)] \left(\left[ \frac{C_m[v] L}{k} +1 \right]^k \right)^{\lfloor sk \rfloor/k} \|\psi_0\|_{2m,2} \\[0.5em]
&\longrightarrow \Pi_m[v(0)] \cdot \e^{C_m[v] L s} \cdot \|\psi_0\|_{2m,2}.
\end{aligned}
\end{equation}
\end{proof}

A similar proof to the one above shows a very interesting smoothing result related to such evolution systems $U[v]$ that is much like the crucial \autoref{lemma-Q} in the successive substitutions method.

\begin{lemma}\label{lemma-Uv-smoothing}
The evolution system $U$ generated by $H(t) = H_0 + v(t)$ with $v \in \Cont^1([0,T], W^{2m,\Sigma})$ defines the following bounded smoothing map.
\begin{align*}
\Cont^1([0,T],H^{2m}) &\longrightarrow \Cont^0([0,T],H^{2(m+1)}) \cap \Cont^1([0,T],H^{2m}) \\
\varphi &\longmapsto \left( t \mapsto \int_0^t U(t,s) \varphi(s) \d s \right)
\end{align*}
\end{lemma}

\begin{proof}
By differentiating with respect to time it is clear that the map is into $\Cont^1([0,T],H^{2m})$ if we can show that the result of the integral is indeed in $H^{2(m+1)}$. The desired continuity in time follows from continuity of the integral and the evolution system. For $U(t,s)$ we use a decomposition as in \eqref{eq-schro-dyn-td-U-splitting}. The time steps are now taken at $t_j$ with $t_K=t$, $t_{K-1}=t-1/k$ equally spaced down to $t_0=t-K/k$ with $K=\lfloor (t-s) k \rfloor$. This means the last time step is between $t_0$ and $s$ and will mostly be shorter than $1/k$. Then the usual trick of introducing a $H^{-1}$-$H$ encounter is applied at the start time $s$.
\begin{equation}\label{eq-U-smooth-splitting}
\begin{aligned}
U(t,s)\varphi(s) &= \lim_{k\rightarrow \infty} U_k(t_K,t_{K-1}) \ldots U_k(t_1,t_0) U_k(t_0,s) \varphi(s) \\
&= \lim_{k\rightarrow \infty} U_k(t_K,t_{K-1}) \ldots U_k(t_1,t_0) H^{-1}(s)U_k(t_0,s) H(s) \varphi(s)
\end{aligned}
\end{equation}
Varying $s$ slightly will not change $K$ for almost all values of $s \in [0,T]$ and we can introduce a partial derivative $\partial_s$.
\begin{align*}
\partial_s \left( H^{-1}(s)U_k(t_0,s) \varphi(s) \right) =\;& \left(\partial_s H^{-1}(s)\right) U_k(t_0,s) \varphi(s) \\
&+ \i H^{-1}(s) U_k(t_0,s) H(s) \varphi(s) \\
&+ H^{-1}(s)U_k(t_0,s)\dot\varphi(s)
\end{align*}
The second term is the one we sought because it appears in \eqref{eq-U-smooth-splitting}. The first term will be studied by evaluating the partial derivative of the inverse Hamiltonian as a difference quotient.
\begin{align*}
\partial_s H^{-1}(s) &= \lim_{h \rightarrow 0}\frac{1}{h}\left( H^{-1}(s+h)-H^{-1}(s)\right) \\
&= \lim_{h \rightarrow 0}\frac{1}{h}H^{-1}(s+h)(H(s)-H(s+h))H^{-1}(s) \\[0.4em]
&= H^{-1}(s) \dot H(s) H^{-1}(s) = H^{-1}(s) \dot v(s) H^{-1}(s)
\end{align*}
This is where the differentiability condition for the potential shows up. Now putting this all back into \eqref{eq-U-smooth-splitting} we have
\begin{align*}
U(t,s)\varphi(s) = &-\i \lim_{k\rightarrow \infty} U_k(t_K,t_{K-1}) \ldots U_k(t_1,t_0) \Big( \partial_s \left( H^{-1}(s)U_k(t_0,s) \varphi(s) \right) \\
&- H^{-1}(s) \dot v(s) H^{-1}(s) U_k(t_0,s) \varphi(s) - H^{-1}(s)U_k(t_0,s)\dot\varphi(s) \Big)
\end{align*}
Now finally we have to evaluate all three terms with respect to the desired $H^{2(m+1)}$-norm and see if we really have a gain in regularity. The $\partial_s$ in the first term is dragged to the front and used to kill the integral, while the $H^{-1}(s) : H^{2m} \rightarrow H^{2(m+1)}$ (\autoref{lemma-rs-1}; note that we assume $H(t)\geq 1$ for all $t \in [0,T]$ without loss of generality as in the proof of \autoref{th-schro-dyn-td}) already does the job. The second and third term include a $H^{-1}(s)$ too, without even having the need to get rid of the integral. The only term in the way is $H^{-1}(s)\dot v(s)$ but it is bounded $H^{2(m+1)} \rightarrow H^{2(m+1)}$ by \autoref{lemma-rs-0}. Finally evaluating the limit works just like in \eqref{eq-schro-evolut-estimate} and boundedness is confirmed.
\end{proof}

Note the trade-off between one order of differentiability in time against two orders in the Sobolev space describing the spatial regularity which is typical for a PDE of order one in time and second order in space---Schrödinger's equation. And indeed the expression
\begin{equation}\label{eq-inhom-se-sol}
f(t) = -\i\int_0^t U(t,s) \varphi(s) \d s
\end{equation}
is easily shown to fulfil the inhomogeneous Schrödinger equation
\[
\i \partial_t f(t) = H(t)f(t) + \varphi(t)
\]
with initial value $f(0)=0$. This relationship is actually called \q{Duha\-mel's principle} that we will meet again in \autoref{sect-full-int-pic} in the context of the interaction picture and even more profoundly in \autoref{sect-duhamel} where the inhomogeneous Schrödinger equation occurs in connection with the variation of quantum trajectories.

The inhomogeneous Schrödinger equation served as a strong indication for us towards the lemma above. It could also be used as a proof strategy, just like in \citeasnoun[ch.~5, Th.~5.3]{pazy}, where it is shown that \eqref{eq-inhom-se-sol} is a strong solution to the inhomogeneous Schrödinger equation with a proof not unlike ours. Note that the conditions there are exactly the same as in our lemma for the lowest regularity level $m=0$.

As it will be customary from \autoref{sect-banach-spaces} onward, we define a set of quantum trajectories, a single trajectory being a map $\psi : [0,T] \rightarrow \H = L^2$. The most natural set for such trajectories is the Banach space $\Cont^0([0,T], L^2)$ endowed with the norm $\|\psi\|_{2,\infty} = \sup_{t \in [0,T]} \|\psi(t)\|_2$. This includes classical as well as generalised solutions to the Schrödinger equation, cf.~the table on page \pageref{tab-solution-types}. The trajectory resulting from an initial state $\psi_0$ under evolution with the Hamiltonian $H[v] = H_0 + v$ will then be noted as $U[v]\psi_0 : t \mapsto U([v],t,0)\psi_0$.

\begin{corollary}\label{cor-schro-evolution-continuous}
Let $\psi_0 \in H^{2(m+1)}$. Then the mapping
\[
U[\cdot]\psi_0 : \Lip([0,T], W^{2m,\Sigma}) \longrightarrow L^\infty([0,T], H^{2m})
\]
from potentials to trajectories is continuous.
\end{corollary}

\begin{proof}
In the case $m=0$ \autoref{th-schro-dyn-konst} (Stone's theorem) tells us that the evolution operator $U[v] : H^2 \rightarrow L^\infty([0,T],L^2)$ is uniformly bounded for all $v$ because it is unitary. This means the extension of $U[v]$ to all $\psi_0 \in L^2$ is also continuous in $v$. This is not true any more for $m\geq 1$ and considerations like below become necessary.\\
We take $v,h \in \Lip([0,T], W^{2m,\Sigma})$ and use the properties of the evolution system with the same trick as in \eqref{eq-U-trick}.
\begin{align*}
\big( &U([v+h],t,0) - U([v],t,0) \big) \psi_0 \\
&= U([v],t,s) U([v+h],s,0) \psi_0 \Big|_{s=0}^t \\
&= \int_0^t \partial_s U([v],t,s) U([v+h],s,0) \psi_0 \d s \\
&= \int_0^t \big( (\partial_s U([v],t,s)) U([v+h],s,0) + U([v],t,s) (\partial_s U([v+h],s,0)) \big) \psi_0 \d s  \\
&= \i \int_0^t U([v],t,s) (H([v],s) - H([v+h],s)) U([v+h],s,0) \psi_0 \d s \\
&= -\i \int_0^t U([v],t,s) h(s) U([v+h],s,0) \psi_0 \d s 
\end{align*}
Now the necessary estimate is performed with the $H^{2m}$-norm of the integrand. One first uses \eqref{eq-schro-evolut-estimate-2m} to get rid of $U([v],t,s)$ in front and then \autoref{lemma-rs-0} for the multiplication with $h(s)$ is applied. Note that we introduce constants depending continuously on $v$ and $h$ in the \q{$\lesssim$} estimates because we omit the exponential and the $\Pi_l$ from \eqref{eq-schro-evolut-estimate-2m} twice.
\begin{equation}\label{eq-evolut-lipschitz}
\begin{aligned}
\big\| \big( U&([v+h],t,0) - U([v],t,0) \big) \psi_0 \big\|_{2m,2} \\
&\leq  \int_0^t \| U([v],t,s) h(s) U([v+h],s,0) \psi_0 \|_{2m,2} \d s \\
&\lesssim \int_0^t \| h(s) U([v+h],s,0) \psi_0 \|_{2m,2} \d s \\
&\lesssim \int_0^t \| h(s) \|_{2m,\Sigma} \cdot \| U([v+h],s,0) \psi_0 \|_{2(m+1),2} \d s \\
&\lesssim \int_0^t \| h(s) \|_{2m,\Sigma} \cdot \| \psi_0 \|_{2(m+1),2} \d s \\[0.5em]
&\leq t \max_{s\in [0,t]}\| h(s) \|_{2m,\Sigma} \cdot \| \psi_0 \|_{2(m+1),2}
\end{aligned}
\end{equation}
For $h \rightarrow 0$ in the $\Lip([0,T], W^{2m,\Sigma})$-topology this expression goes to zero thus showing continuity if $\psi_0 \in H^{2(m+1)}$.
\end{proof}

The above results all hold even if an additional static potential is present that does not fulfil the respective \q{Sobolev--Kato--Lipschitz} restriction but is just from a potential space of Kato perturbations $\Sigma(L^2+L^\infty)$ as in \autoref{sect-kato-peturbations}. Such a potential can always be simply included in the basic building block of the Hamiltonian $H_0$ without changing its domain (\autoref{th-sum-space}) and the proofs all work as given. The most prominent example for such an additional static potential is clearly the usual Coulomb interaction which was shown to be from the required class in \autoref{ex-atomic-hamiltonian}.

Note that \eqref{eq-evolut-lipschitz} has the \emph{form} of a Lipschitz estimate for $\|\psi([v'],t)-\psi([v],t)\|_{2m,2}$ used again in \eqref{eq-estimate-delta-q-2}. It does not prove Lipschitz continuity because additional terms depending on $v,v'$ are suppressed. Yet those terms depend continuously on the given potentials which makes a Lipschitz estimate possible if the potentials are varied only over a compact set.

\section{The successive substitutions method}
\label{sect-yajima}

\subsection{Yajima's set of Banach spaces}
\label{sect-banach-spaces}

Given two Banach spaces $X$ and $Y$ both continuously embedded in a larger Hausdorff topological vector space (this is called an \emph{interpolation couple}, see \citeasnoun{triebel}) we equip their intersection $X \cap Y$ with the canonical norm $\|x\|_{X \cap Y} = \|x\|_X + \|x\|_Y$ and their sum $X+Y$, defined as all possible sums $x+y$ of elements $x \in X, y \in Y$, has norm $\|z\|_{X+Y} = \inf\{ \|x\|_X + \|y\|_Y : x \in X, y \in Y, z=x+y\}$.\footnote{To prove that those are again Banach spaces see \citeasnoun{triebel}.}

We fix the time interval of interest to $[0,T]$ while the whole configuration space for $N$ particles is $\Omega^N = \R^n$ and has dimensionality $n=d\cdot N$. The Hilbert space of wave functions (\emph{state space}) will be denoted by $\H = L^2 = L^2(\Omega^N)$. No special symmetries with respect to the particle positions is assumed and it is furnished with the usual norm $\|\cdot\|_\H = \|\cdot\|_2$. Next we define a Lebesgue space (\emph{trajectory space}) over $[0,T] \times \Omega^N$ by
\[
L^{q,\theta} = \left\{ \varphi \;\middle|\; \|\varphi\|_{q,\theta} = \left( \int_0^T \left( \int_{\Omega^N} |\varphi(t,x)|^q \d x \right)^{\theta/q} \d t \right)^{1/\theta} < \infty \right\} \!\!\!\!\!\!\mod \mathcal{N}
\]
with the null set $\mathcal{N} = \{ \varphi \in L^{q,\theta} \mid \|\varphi\|_{q,\theta}=0 \}$. As a generalisation of an $L^\theta$ space with its values in the Banach space $L^q$ this is called a \emph{Bochner space}. The first superscript $q$ denotes the $L^q$ space in spatial coordinates and $\theta$ the $L^\theta$ space over the (finite) time interval. Latin characters are always used for the space part and Greek ones for time, or simply remember space(1)-time(2) for the ordering. Do not confuse the norm with the double subscript with the Sobolev norm frequently used in the previous sections. $q$ or $\theta = \infty$ are possible and defined in the usual way with the supremum (uniform) norm in time and the essential supremum norm in space. The resulting spaces are denoted with capital, calligraphic letters $\X$ and $\V$, their elements (thought about as trajectories in the spatial Banach space) usually with small Latin or Greek letters. A $\varphi \in \X$ evaluated at a time $t \in [0,T]$ is written simply $\varphi(t)$ and is thus a scalar function over the configuration space $\Omega^N$. If they are only elements of a purely spatial space, like an initial state, we write them as $\varphi_0$ or similar.

We proceed with the definition of important Banach spaces for the evolution of a quantum state following \citeasnoun{yajima}.

\begin{definition}[Banach space of quantum trajectories]\label{def-X}
Let the principal indices for the Banach space $\X$ be $2 \leq q \leq \infty$ and $2 < \theta \leq \infty$ with their dual exponents $q' = q/(q-1), \theta' = \theta/(\theta-1)$ and therefore fulfilling $1 \leq q' \leq 2$ and $1 \leq \theta' < 2$ as well as the typical Hölder relations $1/q + 1/q' = 1$ and $1/\theta + 1/\theta' = 1$. Let further $2/\theta = n(1/2 - 1/q)$ which implies $q < 2n/(n-2)$ for $n \geq 3$. We define $\X$ and its topological dual $\X'$ by
\begin{align*}
\X &= \Cont^0([0,T], \H) \cap L^{q,\theta} \mtext{and} \\
\X' &= L^{2,1} + L^{q',\theta'}.
\end{align*}
\end{definition}

The special relation between the exponents $q,\theta$ of this Banach space is called \emph{Schrödinger-admissible} in \citeasnoun{dancona} where the conditions were slightly widened to $\theta \geq 2$ with the choice $(n,\theta,q) = (2,2,\infty)$ ruled out. $\X$ contains the usual state space of quantum mechanics $\Cont^0([0,T], \H)$ which is equipped with the supremum norm in its time variable, making it a Banach space for compact time intervals $[0,T]$ because of the uniform convergence enforced by the supremum norm. The norms of $\X$ and $\X'$ are
\begin{align*}
&\|\varphi\|_\X = \|\varphi\|_{2,\infty} + \|\varphi\|_{q,\theta} \mtext{and} \\
&\|\varphi\|_{\X'} = \inf\{ \|\varphi_1\|_{2,1} + \|\varphi_2\|_{q',\theta'} \mid \varphi_1 \in L^{2,1}, \varphi_2 \in L^{q',\theta'}, \varphi=\varphi_1+\varphi_2\}.
\end{align*}

Those trajectory spaces are accompanied by the corresponding Banach spaces for potentials that guarantee the stability of evolution operators, i.e., Schrödinger trajectories lying in the space $\X$. That this indeed holds true will be finally proved in \autoref{Uv-bounded}.

\begin{definition}[Banach space of potentials]\label{def-V}
Related to $\X$ we define $\V$ demanding of its indices $p \geq 1,\alpha \geq 1,\beta > 1$ that $0 \leq 1/\alpha < 1-2/\theta$ and $1/p = 1-2/q$.
\[
\V = L^{p,\alpha} + L^{\infty,\beta}
\]
\end{definition}

The condition on $p$ actually guarantees finite potential energy at almost all times for $v(t) \in L^p$ and a state $\psi(t) \in L^q$ because $n(t) \in L^{q/2}$ and $1/p + 2/q = 1$ means $v(t)n(t) \in L^1$. But note that because of the condition on $q$ in \autoref{def-X} this set of inequalities demands $p > \frac{n}{2}$ for all $n \geq 3$ which implies $p \rightarrow \infty$ for very large particle numbers which rules out Coulombic singular potentials as already noted in \autoref{sect-overview-td-hamiltonian}. Still we allow for a more than ``physical'' set of potentials, including the usual one-particle external scalar potentials, symmetric two-particle interaction but also potentials that act differently on different particle coordinates (thus destroying any assumed Bose or Fermi symmetry) or include more than two points.

\subsection{It's all about the inequalities}

Our set of Banach spaces will be supplemented by an armoury of powerful inequalities between the different norms (which go by the dazzling names of Minkowski, Hölder, Cauchy, Schwarz, Bunyakovsky, Young, Hardy, Littlewood, Sobolev, Kato, Poincaré, Strichartz etc., plus combinations of them), proving embeddings between certain spaces and the boundedness of the operator $Q$ from \citeasnoun{yajima} which will be defined soon. $C, C_0, C_1 > 0$ and so on are being used as constants, sometimes indices are added to highlight dependencies.\footnote{These sections do not use the \q{$\lesssim$} notation because they have been compiled prior to its introduction into this thesis.} The spatial domain is $\Omega^N = \R^n$ throughout this section and all dependencies. In a bounded domain $\Omega$ inequalities like Strichartz' (\autoref{strichartz-ineq}) might yield different estimates or even fail to hold.\footnote{The later considerations regarding a fixed-point proof for TDDFT in \autoref{ch-fp} have $\Omega$ bounded as a principal condition and thus the results of this \autoref{sect-yajima} and those that built on it are not directly  portable. Thus we will later rely more on Schrödinger solutions formulated in the stepwise static approximation method from \autoref{sect-stepwise-static} and properties derived from that approach.} We denote the free propagator acting on $\H$ by $U_0(t) = \exp(-\i H_0 t)$ with the free Hamiltonian $H_0 = -\onehalf\Delta$. The first inequality presented here is due to \citeasnoun{kato1973} and is cited after \citeasnoun{yajima}.

\begin{theorem}[Kato inequality]\label{kato-ineq}
Let the exponents $q,q',\theta$ be like in \autoref{def-X}. Then for every $\psi_0 \in L^{q'}$ it holds
\[ \| U_0(t) \psi_0 \|_q \leq (2\pi |t|)^{-2/\theta} \|\psi_0\|_{q'} \]
and $(t \mapsto U_0(t) \psi_0) \in \Cont^0(\R \setminus \{0\}, L^q)$.
\end{theorem}

This respects the expected conservation of probability equality (unitarity of $U_0$) for $q=q'=2$, i.e., $\|U_0(t)\psi_0\|_2 = \|\psi_0\|_2$. As the Kato inequality introduces a singular factor of the type $|t|^{-2/\theta}$, estimates involving singular integrals become important.

\begin{definition}[Riesz potential]\label{def-riesz}
For $f \in \Lloc(\Omega^N)$ and $0 < \alpha < n$ we define
\[
(I_\alpha f)(x) = \frac{1}{C_\alpha} \int_{\R^n} \frac{f(y)}{|x-y|^{n-\alpha}} \d y, \quad C_\alpha = \pi^{n/2} 2^\alpha \frac{\Gamma(\alpha/2)}{\Gamma((n-\alpha)/2)}.
\]
\end{definition}

Note that the Riesz potential can be regarded as the inverse of a fractional power of the Laplace operator $I_\alpha = (-\Delta)^{-\alpha/2}$. For the operator $I_\alpha$ the following inequality holds, giving an estimate for $I_\alpha f$ under the $s$-norm by the $r$-norm of $f$ where the indices $r$ and $s$ are related by $\alpha$. This also was the starting point for the Sobolev embedding theorem in Sobolev's original proof. \cite{wiki-sobolev}

\begin{theorem}[Hardy--Littlewood--Sobolev inequality]\label{hls-ineq}
For $0 < \alpha < n$, $1 \leq r < n/\alpha$, $1/s = 1/r - \alpha/n$ there is a constant $C$ depending only on $r$ such that for all $f \in L^r$ we have
\[
\|I_\alpha f\|_s \leq C \|f\|_r.
\]
\end{theorem}

The next important inequality is in its original form due to \citeasnoun[originally for the wave equation]{strichartz} and proves a remarkable smoothing effect of the free Schrödinger propagator in the sense that it improves $L^q$-smoothness ($q$ like in \autoref{def-X}) for almost all times $0 < t \leq T$. The endpoint $(q,\theta) = (2n/(n-2),2)$ for $n \geq 3$ is excluded in our definition but specialised \q{endpoint Strichartz estimates} are available. \cite{keel-tao} Strichartz' inequality is extensively used for uniqueness results for Schrödinger equations, also and especially in the semilinear case. \cite{cazenave} The version for solutions to the free Schrödinger equation used here is from \citeasnoun{ginibre-velo} and is cited after \citeasnoun{yajima2}. Its proof is also part of the proof for unique Schrödinger solutions in \citeasnoun{yajima} and thus it can be seen as a corollary to \autoref{lemma-Q} here. To achieve a more complete picture about the properties of the free evolution operator we still note it now. We write $U_0 \psi_0 = (t \mapsto U_0(t)\psi_0)$ for the state trajectory to the initial state $\psi_0 \in \H$.

\begin{theorem}[Strichartz inequality]\label{strichartz-ineq}
Let the exponents $q,\theta$ be as in \autoref{def-X}. Then there is a constant $C_0$ such that for every $\psi_0 \in \H$ it holds
\[ \| U_0 \psi_0 \|_{q,\theta} \leq C_0 \|\psi_0\|_2. \]
\end{theorem}

In \citeasnoun{dancona} it is stressed that such Strichartz estimates seem to be more general than the basic smoothing estimates of $L^{q'}$-$L^q$-type like \autoref{kato-ineq} (Kato inequality) as they can be found for generalisations of Schrödinger equations as well, like with non-linearities or time-dependent potentials (which we study here), where the $L^{q'}$-$L^q$-estimates fail to hold. This means for us that we can rely on Strichartz estimates also in the case of additional potentials as in \autoref{Cv-strichartz}, which will become important when studying $\psi[v]$ under variations of the potential $v$ in \autoref{sect-func-diff-successive}.

Following Theorem 1 of \citeasnoun{burq} such Strichartz estimates hold for evolution under the influence of certain time-independent one-particle singular external potentials like a point-dipole. Note that the conditions on the indices demand $q = 2$ for large dimensionality, like it would be the case for a many-body system.

Note that for this inequality to hold the exponents of the $L^p$ spaces in space, $q$, and time, $\theta$, must be related in a dimensionality-dependent way, $2/\theta = n(1/2 - 1/q)$ to be exact. Sharp estimates, i.e., definite values for the constant $C_0$, have only recently been established and seem to be limited to small dimensionality until today (see for example \citeasnoun{hundertmark}).

We stress again that the spatial domain here is always $\R^n$, Strichartz estimates for other, probably bounded domains are available, although not in this general form. The case of compact manifolds is treated for example in \citeasnoun{burq-2004} where a loss of derivatives occurs on the right hand side of the inequality. Instead of the $L^2$-norm the norm of a (fractional) Sobolev space gives the upper bound then. Such more advanced Strichartz estimates will not be considered here.

These different inequalities involving $U_0$ ask for a \emph{polymorphic} definition of the symbol, where the right operational use is determined from the context, i.e., the object on which the operator is applied. This means one has to be more cautious but on the other side is able to gain a deeper insight into how the objects of the theory are linked. We readily collect three different uses of the symbol $U_0$.

Combining \autoref{strichartz-ineq} (Strichartz inequality) and $U_0(t)$ unitary we can easily derive the boundedness of an operator
\begin{align*}
U_0 : \H &\longrightarrow \X \\
\psi_0 &\longmapsto U_0 \psi_0 = (t \mapsto U_0(t)\psi_0),
\end{align*}
with the bound
\begin{equation}\label{U0-bounded}
\| U_0 \psi_0 \|_\X = \| U_0 \psi_0 \|_{2,\infty} + \| U_0 \psi_0 \|_{q,\theta} \leq (1+C_0) \|\psi_0\|_2.
\end{equation}
By \autoref{kato-ineq} (Kato inequality) the free propagator $U_0(t)$ at a time $t \neq 0$ can be seen as a bounded operator
\[ U_0(t) : L^{q'} \longrightarrow L^q \]
and of course equally well for all $t \in \R$ as the usual
\[ U_0(t) : L^2 \longrightarrow L^2. \]
And finally to achieve maximal flexibility another operational interpretation of $U_0$ will be introduced where it operates on a time-dependent expression.
\begin{align*}
U_0 \varphi &= (t \mapsto U_0(t)\varphi(t)) \\
U_0^* \varphi &= (t \mapsto U_0(-t)\varphi(t))
\end{align*}
This makes it into a pointwise acting operator with respect to the time variable, well defined on $\Cont^0(\R, \H)$. The compound operator defined later in this section are bounded on $\X$ as well, see \autoref{Q-bounded}. This notation has a great significance related to the interaction picture because the transformed quantity $\tilde\psi=U_0^{*}\psi_0$ will be shown to fulfil a reduced version of the Schrödinger equation, namely the Tomonaga--Schwinger equation \eqref{schwinger-tomonaga}.

The usual Hölder inequality for $L^p$ spaces can easily be adapted to an $L^{m,\mu}$ space.

\begin{theorem}[Hölder inequality]\label{hoelder-ineq}
Let $1/m = 1/r + 1/s$ and $1/\mu = 1/\rho + 1/\sigma$. Then
\[ \|f g\|_{m,\mu} \leq \|f\|_{r,\rho} \|g\|_{s,\sigma}. \]
\end{theorem}

The next lemma is a close relative to \autoref{lemma-sum-space-inequality} and \autoref{lemma-rs-0} but now for Banach spaces that already include time-dependence.

\begin{lemma}\label{lemma-mult-op}
A multiplication operator $v \in \V$ is a bounded operator $\X \rightarrow \X'$ and fulfils
\[
\|v \varphi\|_{\X'} \leq T^* \|v\|_\V \|\varphi\|_\X
\]
with $T^* = \max\{ T^{1-1/\beta}, T^{1 - 2/\theta-1/\alpha} \}$.
\end{lemma}

\begin{proof}
We remember the partitioning $v = v_1 + v_2$ with $v_1 \in L^{p,\alpha}, v_2 \in L^{\infty,\beta}$ given by the norm of $\V$ and use Hölder's inequality for each part of $v\varphi = v_1\varphi + v_2\varphi$. To get the final result we need to change the time indices of the norms to bigger values, which is possible with the simple relation (for arbitrary $m,\gamma,\rho$ and $\rho>\gamma$ using Hölder too)
\[
\|f\|_{m,\gamma} = \|1 \cdot f\|_{m,\gamma} \leq \|1\|_{\infty,\gamma\rho/(\rho-\gamma)} \|f\|_{m,\rho} = T^{1/\gamma-1/\rho} \|f\|_{m,\rho}.
\]
For the $L^{q',\theta'}$-part of $\X$ we have with $1/q'-1/q=1/p$
\[
\|v_1\varphi\|_{q',\theta'} \leq \|v_1\|_{p,\theta\theta'/(\theta-\theta')} \|\varphi\|_{q,\theta} \leq T^{1-2/\theta-1/\alpha} \|v_1\|_{p,\alpha} \|\varphi\|_{q,\theta}
\]
and for the $L^{2,1}$-part 
\begin{equation}\label{eq-lemma-mult-op}
\|v_2\varphi\|_{2,1} \leq \|v_2\|_{\infty,1} \|\varphi\|_{2,\infty} \leq T^{1-1/\beta} \|v_2\|_{\infty,\beta} \|\varphi\|_{2,\infty}.
\end{equation}
The right hand side of the lemma's statement clearly includes those two estimates which concludes the proof.
\end{proof}

\begin{definition}
The simple linear integral operator $S$ is defined as
\[ (S\varphi)(t) = \int_0^t \varphi(s) \d s. \]
\end{definition}

\begin{definition}\label{def-Qv}
We define the linear integral operator $Q_v$ as
\[ Q_v = -\i U_0 S U_0^* v. \]
\end{definition}

Note that these two definitions are similar but slightly different from $S$ and $Q$ in our primary reference \cite{yajima}. This kind of transformation of operators with $U_0$ and $U_0^*$ reminds of the so-called interaction picture discussed in the next section. Written out explicitly we have
\begin{align*}
Q_v : \X &\longrightarrow \X \\
\varphi &\longmapsto \left( t \mapsto -\i\int_0^t U_0(t-s) v(s) \varphi(s) \d s \right).
\end{align*}
This operator will be shown to be bounded on $\X$ if $v \in \V$ in \autoref{Q-bounded} below. The next lemma expresses the crucial smoothing property of the free evolution. This is needed after multiplying with a potential $v$ which casts the operand outside of $\X$ and into $\X'$, visible in \autoref{lemma-mult-op}. Thus boundedness of $Q_v$ as an operator $\X \rightarrow \X$ is possible.

\begin{lemma}\label{lemma-Q}
The $Q = -\i U_0 S U_0^*$ part of $Q_v$ is a bounded operator $\X' \rightarrow \X$ with operator norm $\|Q\|=C_Q$.
\end{lemma}

\begin{proof}
Remember the expressions for the $\X$- and $\X'$-norms
\begin{align*}
&\|\varphi\|_\X = \|\varphi\|_{2,\infty} + \|\varphi\|_{q,\theta} \\
&\|\varphi\|_{\X'} = \inf\{ \|\varphi_1\|_{2,1} + \|\varphi_2\|_{q',\theta'} \mid \varphi_1 \in L^{2,1}, \varphi_2 \in L^{q',\theta'}, \varphi=\varphi_1+\varphi_2\}
\end{align*}
which means that in order to be bounded the operator $Q : \X' \rightarrow \X$ has to fulfil all four inequalities
\begin{subequations}\begin{align}
\|Q \varphi\|_{2,\infty} &\leq C_1 \|\varphi\|_{2,1} \label{bounded-a} \\
\|Q \varphi\|_{2,\infty} &\leq C_2 \|\varphi\|_{q',\theta'} \label{bounded-b} \\
\|Q \varphi\|_{q,\theta} &\leq C_3 \|\varphi\|_{2,1} \label{bounded-c} \\
\|Q \varphi\|_{q,\theta} &\leq C_4 \|\varphi\|_{q',\theta'}. \label{bounded-d}
\end{align}\end{subequations}
We start from below with \eqref{bounded-d} by drawing the spatial $q$-norm into the integral arising from $S$ by using the triangle inequality (Minkowski's integral inequality). We estimate further by changing the upper bound of the integral to $T$ as the integrand is surely positive now. In the second step one estimates the $q$-norm by a $q'$-norm by applying \autoref{kato-ineq} (Kato inequality).
\begin{align}\label{bounded-d-estimate}
\|U_0 S U_0^* \varphi\|_{q,\theta} &= \left\| t \mapsto \left\| \int_0^t U_0(t-s)\varphi(s)\d s \right\|_q \right\|_\theta \nonumber\\
&\leq \left\| t \mapsto \int_0^T \left\| U_0(t-s)\varphi(s) \right\|_q \d s \right\|_\theta \\
&\leq (2\pi)^{-2/\theta} \left\| t \mapsto \int_0^T (t-s)^{-2/\theta} \left\| \varphi(s) \right\|_{q'} \d s \right\|_\theta \nonumber
\end{align}
The singular time integral can be treated by \autoref{hls-ineq} (Hardy--Littlewood--Sobolev inequality) where the dimensionality is $n=1$, $\alpha = 1 - 2/\theta > 0$ (ruling out $\theta=2$ as stated before), $s = \theta$ and thus $r = \theta'$. The constant from this estimate will be combined with the $2\pi$ factor to $C_4$ right away.
\begin{align*}
\|U_0 S U_0^* \varphi\|_{q,\theta} &\leq C_4 \left\| t \mapsto \left\| \varphi(t) \right\|_{q'} \right\|_{\theta'} = C_4 \|\varphi\|_{q',\theta'}
\end{align*}
Next we show \eqref{bounded-b} by rewriting the 2-norm as an inner product.
\begin{align*}
\|U_0 S U_0^* \varphi\|_{2,\infty}^2 &= \sup_{t \in [0,T]} \int_{\Omega^N} \left| \int_0^t U_0(t-s)\varphi(s)\d s \right|^2 \d x \\
&= \sup_{t \in [0,T]} \left\langle \int_0^t U_0(t-s)\varphi(s)\d s, \int_0^t U_0(t-r)\varphi(r)\d r \right\rangle \\
&= \sup_{t \in [0,T]} \int_0^t \left\langle \varphi(s), \int_0^t U_0(s-t) U_0(t-r)\varphi(r) \d r\right\rangle \d s
\end{align*}
We apply the concatenation property of the free evolution operator $U_0(s-t) U_0(t-r) = U_0(s-r)$, rewrite the inner product as a spatial $L^1$-norm and estimate it with Hölder's inequality.
\begin{align*}
\|U_0 S U_0^* \varphi\|_{2,\infty}^2 &\leq \sup_{t \in [0,T]} \int_0^t \left\|\varphi^*(s) \cdot \int_0^t U_0(s-r)\varphi(r) \d r\right\|_1 \d s\\
&\leq \sup_{t \in [0,T]} \int_0^t \| \varphi(s) \|_{q'} \cdot \left\| \int_0^t U_0(s-r)\varphi(r) \d r \right\|_q \d s \\
&\leq \sup_{t \in [0,T]} \int_0^t \| \varphi(s) \|_{q'} \cdot \int_0^t \left\| U_0(s-r)\varphi(r) \right\|_q \d r \d s
\end{align*}
As all the integrands are positive we can estimate by changing the upper bound of the integral to $T$ and get rid of the supremum. Finally the outer time integral will be again written as an $L^1$-norm and Hölder becomes applicable once more. The final steps are the same as in \eqref{bounded-d-estimate}.
\begin{align*}
\|U_0 S U_0^* \varphi\|_{2,\infty}^2 &\leq \left\| s \mapsto \| \varphi(s) \|_{q'} \cdot \int_0^T \left\| U_0(s-r)\varphi(r) \right\|_q \d r \right\|_1 \\
&\leq \| \varphi \|_{q',\theta'} \cdot  \left\| s \mapsto \int_0^T \left\| U_0(s-r)\varphi(r) \right\|_q \d r \right\|_\theta \\[0.7em]
&\leq C_2^2 \| \varphi \|_{q',\theta'}^2
\end{align*}
Equation \eqref{bounded-c} will be proved with a duality argument between $L^{q,\theta}$ and its topological dual $L^{q',\theta'}$ and the dual pair $L^{2,\infty}$-$L^{2,1}$ likewise, both linked by the temporal-spatial inner product $(\cdot,\cdot)$. We take $\varphi, \psi \in L^{q',\theta'}$ with $\psi$ as ``test function'' adjoined to $U_0 S U_0^* \varphi \in L^{q,\theta}$(this is result \eqref{bounded-d}).
\begin{align*}
\left( \psi,U_0 S U_0^* \varphi \right) &= \int_0^T \left\langle \psi(t),\int_0^t U_0(t-s)\varphi(s)\d s \right\rangle \d t \\
&= \int_0^T \!\!\d t \int_0^t \!\!\d s \left\langle U_0(s-t)\psi(t),\varphi(s) \right\rangle
\end{align*}
In the next step we interchange the integrals by changing the integration intervals to $s \in [0,T]$ and $t \in [s,T]$. When taking the absolute value this can be estimated by the $L^{1,1}$-norm and this by Hölder's inequality by the norms of the dual spaces $L^{2,\infty}$ and $L^{2,1}$. Finally the estimate of $U_0 S U_0^* \psi$ in the $L^{2,\infty}$-norm has already been done with \eqref{bounded-b}, only with a slightly different interval of integration which will not change the result, and we get back to $L^{q',\theta'}$.
\begin{align*}
|\left( \psi,U_0 S U_0^* \varphi \right)| &= \left| \int_0^T \left\langle \int_s^T U_0(s-t)\psi(t) \d t,\varphi(s) \right\rangle \d s \right| \\
&\leq \int_0^T \int_{\Omega^N} \left| \int_s^T U_0(s-t)\psi(t) \d t \cdot \varphi(s) \right| \d x \d s \\
&\leq \left\|t \mapsto \int_s^T U_0(s-t)\psi(t) \d t\right\|_{2,\infty} \cdot \|\varphi\|_{2,1} \\
&\leq C_2 \|\psi\|_{q',\theta'} \cdot \|\varphi\|_{2,1}
\end{align*}
Saturation of the Hölder inequality gives us for special $\psi \in L^{q',\theta'}$
\[
|\left( \psi,U_0 S U_0^* \varphi \right)| = \|\psi\|_{q',\theta'} \cdot \|U_0 S U_0^* \varphi\|_{q,\theta}
\]
and thus by setting $C_3 = C_2$ the desired \eqref{bounded-c}
\[
\|U_0 S U_0^*\varphi\|_{q,\theta} \leq C_3 \|\varphi\|_{2,1}.
\]
A very similar argument can be used to examine $(\varphi,U_0\psi)$ with static $\psi \in \H$ and prove \autoref{strichartz-ineq} (Strichartz inequality). The final inequality \eqref{bounded-a} is fairly obvious even for $C_1=1$ if we use standard estimates and unitarity of $U_0$ on $L^2(\Omega^N)$.
\begin{align*}
\|U_0 S U_0^* \varphi\|_{2,\infty} &= \sup_{t \in [0,T]} \left\| \int_0^t U_0(t-s)\varphi(s)\d s \right\|_2 \\
&\leq \int_0^T \left\|U_0(t-s)\varphi(s) \right\|_2 \d s \\
&= \int_0^T \left\|\varphi(s) \right\|_2 \d s = \|\varphi\|_{2,1}
\end{align*}
With $C_Q = \max\{ C_1+C_3, C_2+C_4 \}$ this concludes the proof of the lemma.
\end{proof}

\begin{corollary}\label{Q-bounded}
$Q_v: \X \rightarrow \X$ is bounded with operator norm $\|Q_v\| \leq C_Q T^* \|v\|_\V$ where $T^* = \max\{ T^{1-1/\beta}, T^{1-n/(2p)-1/\alpha} \}$. This also proves that applying $Q_v$ more than once is of order larger than $v$ in the sense that $Q_v^k \in \lilo(\|v\|_\V)$ as $v \rightarrow 0$ for $k \geq 2$.
\end{corollary}

For the precise definition of the small Landau symbol $\lilo$ see \autoref{def-small-o} in the next chapter.

\begin{proof}
The operator $Q_v = -\i U_0 S U_0^* v = Q v$ is a concatenation of two bounded operators
\[
\X \overarrow{v\cdot} \X' \overarrow{Q} \X.
\]
Thus we have with \autoref{lemma-mult-op} and \autoref{lemma-Q} that
\[
\|Q_v\| = \|Q\| \cdot \|v\cdot\| \leq C_Q T^* \|v\|_\V.
\]
The relation $2/\theta = n(1/2-1/q)$ of the exponents of $\X$ from \autoref{def-X} combined with the assumptions $1/\alpha < 1 -2/\theta$ and $1/p = 1 -2/q$ from \autoref{def-V} and the duality of the exponents yields
\[
\frac{1}{\alpha} < 1-\frac{n}{2}\left(1-\frac{2}{q}\right) = 1-\frac{n}{2p}.
\]
This can be used to rewrite the exponent in $T^*$ and the proof is done.
\end{proof}

The relation $0 \leq 1/\alpha < 1-n/(2p)$ is generalised to ``$\leq$'' in \citeasnoun{dancona} for bounded time intervals which seems to relate to the more general $\theta \geq 2$ in their work.

Unfortunately the class of allowed potentials decreases considerably for large $n$ because of $n < 2p$ (this is still a condition in Theorem 1.1 of \citeasnoun{dancona}). As we have already seen in \autoref{sect-overview-td-hamiltonian} this rules out Coulombic singular potentials already for $n \geq 6$ thus only allowing the usual molecular theory for one single electron. Nevertheless this fits more or less to a usual hypothesis for static potentials $v \in L^{n/2}(\R^n)+L^\infty(\R^n)$ if $n \geq 3$, see for example \citeasnoun[11.3]{lieb-loss}.

\subsection{The interaction picture and existence of Schrö\-dinger solutions}
\label{sect-full-int-pic}

Let $H_0$ be the free Hamiltonian with dense domain $D(H_0) \subset \mathcal{H}$ and $v$ a real, time-dependent potential acting as multiplication operator such that for a fixed initial state $\psi_0 = \psi(0) \in D(H_0)$ a strong solution to the Schrödinger equation \eqref{standard-se} below exists. We write $H[v] = H_0 + v$ for the combined, time-dependent Hamiltonian and usually mark the $v$-dependency at the solution as well.
\begin{equation}\label{standard-se}
 \i \partial_t \psi[v] = H[v]\psi[v]
\end{equation}

The time evolution for the free dynamics alone is $U_0(t) = \exp(-\i H_0 t)$, describing the free evolution of a state if a time span $t$ passes. We will switch over to the $H_0$-interaction picture (a form of the Dirac intermediate picture) with trajectory $\tilde{\psi}$ by the substitution
\[
\psi = U_0 \tilde{\psi}.
\]
Putting this into \eqref{standard-se} yields the Tomonaga--Schwinger equation.
\begin{equation}\label{schwinger-tomonaga}
\i\partial_t \tilde{\psi} = \tilde{v} \tilde{\psi}, \quad \tilde{v} = U_0^* v U_0
\end{equation}
Integration of this equation over the time interval $[0,t]$ leaves us with the following recursive integral equation for solutions to \eqref{schwinger-tomonaga} which is an almost equivalent formulation of the problem. Here we will be mainly interested in solutions to the mild version of Tomonaga--Schwinger.
\begin{equation}\label{schwinger-tomonaga-int}
\tilde{\psi}(t) = \psi_0 - \i\int_0^t \tilde{v}(s) \tilde{\psi}(s) \d s
\end{equation}
Iterating this equation gives the (not yet time-ordered) Dyson series, but we may also arrive there by writing the equation above equivalently as a recursive expression for $\psi([v],t)$. This means we transform it back to the original Schrödinger picture, which leaves us with an integral equation of Volterra type that we call the ``mild'' Schrödinger equation. Here the operator $Q_v$ already studied in \autoref{Q-bounded} occurs.
\begin{equation}\label{schroequ-mild}
\psi([v],t) = U_0(t)\psi_0 - \i\int_0^t U_0(t-s) v(s) \psi([v],s) \d s = (U_0\psi_0 +Q_v\psi[v])(t)
\end{equation}
Note that the term \q{mild solution} here is in contrast to its use in \autoref{def-mild-solution} but still relates to the use in \citeasnoun[Def.~12.15]{renardy-rogers} and \citeasnoun[Def.~IV.2.3]{pazy} noted before. There it specifies solutions to an integral equation that one gets from a \q{variation of constants formula} or \q{Duhamel's formula} for Cauchy problems with inhomogeneity.
\[
\partial_t u(t) = A u(t) + f(t)
\]
Let $g(t) = T(-t)u(t)$, $T(t) = \e^{tA}$ the generated semigroup, just like in the interaction picture, then
\begin{align*}
\partial_t g(t) &= -A T(-t) u(t) + T(-t) \dot{u}(t) \\
&= -A T(-t) u(t) + T(-t) (A u(t) + f(t)) = T(-t) f(t).
\end{align*}
Integration yields
\[
g(t)=g(0) + \int_0^t T(-s)f(s) \d s
\]
and thus for the original solution
\[
u(t) = T(t)g(t) = T(t)u_0 + \int_0^t T(t-s)f(s) \d s.
\]
If we take the solution-dependent part $v(t)\psi(t)$ in \eqref{standard-se} as `inhomogeneity' we just arrive at \eqref{schroequ-mild}. Thus we adopt this terminology which is also used for example in \citeasnoun{faou} and \citeasnoun{guerrero}, whereas \citeasnoun{yajima} and \citeasnoun{dancona} simply use the term \q{integral equation}.

Solving \eqref{schroequ-mild} for $\psi[v]$ yields a Neumann series called \q{Dyson--Phil\-lips expansion} in \citeasnoun{cycon} that is equivalent to the Dyson series mentioned before. The Neumann series is just the operator generalisation of a geometric series. We write it as an equation not to determine $\psi([v],t)$ at a given time instant but for the whole trajectory $\psi[v]:t \mapsto \psi([v],t)$ within $\X$.
\begin{equation}\label{neumann-series}
\psi[v] = (\id - Q_v)^{-1} U_0 \psi_0 = \sum_{k=0}^\infty Q_v^k U_0\psi_0
\end{equation}
This series actually converges for $\psi_0 \in \H$ if the Strichartz estimate gives $U_0 \psi_0 \in \X$ and $T$ short enough such that $\|Q_v\| \leq C_Q T^* \|v\|_\V < 1$. The uniqueness of solutions to the Schrödinger equation for longer time intervals is still guaranteed by a continuation procedure (cf.~the proof of \autoref{Uv-bounded}). This result was already found by \citeasnoun[Th.~6.2]{phillips} in the case of a strongly continuously differentiable $v\cdot : [0,\infty) \rightarrow B(\H)$.

The result can be used to define an evolution operator considering a potential $v \in \V$ written $\psi([v],t) = U([v],t,s)\psi([v],s)$ with start time $s$ and end time $t$ which in \citeasnoun{yajima} is shown to fulfil the usual properties of evolution systems. Note that $U([0],t,s) = U_0(t-s)$ is just the free evolution. Analogously to $U_0$ we define the evolution under a potential $v \in \V$ as a mapping $U[v]$ from initial states to trajectories.
\begin{align*}
U[v] : \H &\longrightarrow \X \\
\psi_0 &\longmapsto U[v]\psi_0 = (t \mapsto U([v],t,0)\psi_0))
\end{align*}
From \eqref{neumann-series} we have the series representation of this operator.
\begin{equation}\label{def-Uv}
U[v] = \sum_{k=0}^\infty Q_v^k U_0 \in (\id + Q_v + o(\|v\|_\V)) U_0 \mtext{as} v \rightarrow 0
\end{equation}
We have now set the stage to formulate a theorem for unique solutions to the Schrödinger equation with potential $v \in \V$ by showing boundedness of the constructed evolution operator $U[v]$. (cf.~\citeasnoun{yajima}, Cor.~1.2 (3))

\begin{theorem}\label{Uv-bounded}
For arbitrary albeit finite $T>0$ and $v \in \V$ we get $U[v] : \H \rightarrow \X$ as a bounded operator solving the mild Schrödinger equation
\[
(U[v]\psi_0)(t) = U_0(t)\psi_0 - \i\int_0^t U_0(t-s) v(s) (U[v]\psi_0)(s) \d s
\]
for all $\psi_0 \in \H$.
\end{theorem}

\begin{proof}
Take $M \in \N$ large enough and define $\tau = T/M$ and $\tau^* = \max\{ \tau^{1-1/\beta}$, $\tau^{1-2/\theta-1/\alpha} \}$ as in \autoref{Q-bounded} such that $C_Q \tau^* \|v\|_\V < 1$. (Note that the exponents in the definition of $\tau^*$ are both positive due to \autoref{def-V} and thus its value decreases with bigger $M$.) The time interval $[0,T]$ will now be partitioned into a finite number of subintervals $I_1 = [0,\tau], I_2 = [\tau,2\tau], \ldots, I_M=[(M-1)\tau,T]$. If we limit ourselves to the first interval we get an estimate for the evolved initial state by applying \eqref{neumann-series} and \autoref{Q-bounded} together with an obvious bound for the $\X$-norm restricted to the first time interval $\|\varphi\|_{\X | I_1} \leq \|\varphi\|_\X$.
\begin{align*}
\|U[v] \psi_0\|_{\X | I_1} &\leq \sum_{k=0}^\infty \|Q_v^k U_0 \psi_0\|_{\X|I_1} \\
&\leq \sum_{k=0}^\infty (C_Q \tau^*\|v\|_\V)^k \, \|U_0 \psi_0\|_{\X} \\
&= \frac{1}{1-C_Q \tau^*\|v\|_\V} \, \|U_0 \psi_0\|_{\X}
\end{align*}
In assuming $M$ large enough we have convergence of the geometric series involved because $\tau^*$ can be made arbitrarily small. Putting in $\|U_0 \psi_0\|_\X \leq (1+C_0)\|\psi_0\|_2$ from \eqref{U0-bounded} we get an inequality
\[
\|U[v] \psi_0\|_{\X | I_1} \leq \frac{1+C_0}{1-C_Q \tau^*\|v\|_\V} \|\psi_0\|_2.
\]
By taking $U([v],\tau,0)\psi_0$ as a new initial state and shifting the potential $v$ by $\tau$ in time we can proceed inductively until all of $[0,T]$ is spanned. To get an estimate for the $\X$-norm we divide the temporal integral of the $L^{q,\theta}$-norm into the subintervals $I_1, \ldots, I_M$. As any new initial state respects $\|U([v],t,0)\psi_0\|_2 = \|\psi_0\|_2$ due to unitarity, each of them can be estimated separately by the inequality above.
\begin{align*}
\|U[v] \psi_0\|_{\X} &= \|U[v] \psi_0\|_{2,\infty} + \|U[v] \psi_0\|_{q,\theta} \\
&= \|\psi_0\|_2 + \left( \int_0^T \|U([v],t,0)\psi_0\|_q^\theta \d t \right)^{1/\theta} \\
&= \|\psi_0\|_2 + \left( \int_0^\tau \|U([v],t,0)\psi_0\|_q^\theta \d t + \int_\tau^{2\tau} \ldots \right)^{1/\theta} \\
&\leq \|\psi_0\|_2 + \left( \sum_{m=1}^M\|U[v] \psi_0\|_{\X | I_m}^\theta \right)^{1/\theta} \\
&\leq \left(1+M^{1/\theta} \frac{1+C_0}{1-C_Q \tau^*\|v\|_\V}\right) \|\psi_0\|_2
\end{align*}
Thus boundedness is proved. $M$ can be chosen in a unique and optimal way to minimise the expression preceding $\|\psi_0\|_2$. This can be seen as the fraction is monotonically decreasing in $M$ starting with $\infty$ at $M$ such that $C_Q \tau^*\|v\|_\V = 1$. On the other side $M^{1/\theta}$ increases monotonically to $\infty$ and the whole expression must attain a global minimum in between.
\end{proof}

This result shows existence and uniqueness (it is the limit of a contracting fixed-point scheme defined by \eqref{schroequ-mild}) of solutions to the Schrödinger equation with a potential $v \in \V$ and is therefore of great significance. A more direct and simpler Strichartz-like estimate immediately relating to a fixed-point technique is due to \citeasnoun{dancona}.

\begin{corollary}\label{Cv-strichartz}
For arbitrary albeit finite $T>0$ and $v \in \V$ (in certain cases $T \rightarrow \infty$ becomes feasible) we have the following Strichartz estimate for solutions to the mild Schrödinger equation.
\[
\|\psi[v]\|_{q,\theta} \leq C_v \|\psi_0\|_2
\]
We get the $v$-dependent Strichartz constant $C_v=2 M^{1/\theta} (1+C_0)$. For the definition of $M(v) \in \N$ note the details in the beginning of the proof.
\end{corollary}

\begin{proof}
Firstly divide the time interval $[0,T]$ again into a finite number $M$ of subintervals $I_1, \ldots, I_M$. Each subinterval be short enough such that
\[
C_Q |I_m|^*\|v\|_{\V | I_m} \leq \onehalf
\]
with the length of the interval transformed to $|I_m|^*$ as in \autoref{Q-bounded}. Note that by this division also an infinite time interval $[0,\infty)$ gets feasible for potentials decaying fast enough in time. Now take the recursive formula \eqref{schroequ-mild} and define a map $\Phi : \X \rightarrow \X$
\[
\Phi(\psi) = U_0 \psi_0 + Q_v \psi.
\]
A fixed point of this map would be a solution to the mild Schrödinger equation. We have by \autoref{strichartz-ineq} and \autoref{Q-bounded} the following estimate if we limit ourselves to any of the subintervals.
\[
\|\Phi(\psi)\|_{\X | I_m} \leq (1+C_0) \|\psi_0\|_2 + \onehalf \|\psi\|_{\X | I_m}.
\]
Now $\Phi$ clearly defines a contraction mapping and the unique fixed point $\psi = \Phi(\psi) \in \X|_{I_m}$ fulfils $\|\psi\|_{\X | I_m} \leq 2(1+C_0) \|\psi_0\|_2$. The final step is to concatenate all of these estimates to the full time interval as in the proof of \autoref{Uv-bounded} above.
\begin{align*}
\|\psi\|_{q,\theta} &= \left( \sum_{m=1}^M \int_{I_m} \|\psi(t)\|_q^\theta \d t \right)^{1/\theta} \\
&\leq \left( \sum_{m=1}^M (2(1+C_0) \|\psi_0\|_2)^\theta \right)^{1/\theta} \\[0.5em]
&\leq 2 M^{1/\theta} (1+C_0) \|\psi_0\|_2
\end{align*}
This Strichartz estimate needs no direct reference to the constant $C_Q$ but it becomes equal to the one in the proof of \autoref{Uv-bounded} if $C_Q |I_m|^*\|v\|_{\V} = \onehalf$ and $[0,T]$ is equipartitioned into subintervals of length $|I_m| = \tau$.
\end{proof}

The cases where an unbounded time interval $[0,\infty)$ is possible reflect the situation of a quickly decreasing potential as $t \rightarrow \infty$. This is similar to the setting of scattering theory where the evolution after a collision process is dominated by the free Hamiltonian again.

\chapter{Functional Differentiability}
\label{ch-diff}

\section{Basics of variational calculus}

The principal setting is that of two Banach spaces $(X,\|\cdot\|_X)$ and $(Y,\|\cdot\|_Y)$ respectively (although for the definition of the Gâteaux derivative locally convex topological vector spaces are sufficient) with a map $f : X \rightarrow Y$ under investigation. Some authors like \citeasnoun{blanchard-bruening-2} draw a line between the term ``Gâteaux differential" (not necessarily continuous nor linear) and ``Gâteaux derivative", an approach not followed here as we only need the latter and speak of functions being Gâteaux differentiable.

\subsection{Gâteaux and Fréchet differentiability}

\begin{definition}\label{def-gateaux}
$f$ is \textbf{Gâteaux differentiable} at $x_0 \in X$ if there exists a linear and continuous (bounded) map (the \textbf{Gâteaux derivative} or \textbf{differential} at $x_0$)
\[
\delta_G f(x_0,\cdot) : X \rightarrow Y
\]
such that for all $h \in X$
\[
\lim_{\lambda \rightarrow 0} \left\| \frac{1}{\lambda} \big(f(x_0+\lambda h)-f(x_0)\big) - \delta_G f(x_0,h) \right\|_Y = 0.
\]
\end{definition}

The Gâteaux derivative is therefore the generalisation of the directional derivative and could analogously be defined by a partial derivative.
\[
\delta_G f(x_0,h) = \partial_\lambda f(x_0 + \lambda h) \big|_{\lambda=0}
\]
In contrast to Gâteaux differentiability the Fréchet derivative demands the limit to hold uniformly for all directions and thus can be compared to the (total) differential of a function.

\begin{definition}
$f$ is \textbf{Fréchet differentiable} at $x_0 \in X$ if there exists a linear and continuous (bounded) map (the \textbf{Fréchet derivative} or \textbf{differential} at $x_0$)
\[
\delta_F f(x_0,\cdot) : X \rightarrow Y
\]
such that
\[
\lim_{h \rightarrow 0} \frac{1}{\|h\|_X} \left\| f(x_0+h)-\big(f(x_0) + \delta_F f(x_0,h)\big) \right\|_Y = 0.
\]
If $f$ is Fréchet differentiable for all $x_0 \in U \subseteq X$ open and the map $\delta_F f : U \rightarrow \mathcal{B}(X,Y),\, x \mapsto \delta_F f(x,\cdot)$ is continuous, the function $f$ is called \textbf{continuously differentiable} or \textbf{Fréchet differentiable} on $U$ and we write $f \in \Cont^1(U,Y)$.
\end{definition}

Note that the main difference between those two definitions lies in the form of the limit, which is to be taken from each direction separately for Gâteaux derivatives, whereas it must hold uniformly for every kind of path towards zero in the case of Fréchet derivatives. Note also that the condition of continuity of $\delta_F f$ in its first component need not be the same as boundedness here, as the derivative would in general not be linear in this argument.

\subsection{Results on equivalence of the two notions}

It follows right away from the definitions that Fréchet differentiability at a point implies Gâteaux differentiability. The converse needs the additional assumption of continuity with respect to $x_0$. To show this we employ the fundamental theorem of calculus for integrals in Banach spaces given below. The utilised integral is that of Riemann generalised to Banach spaces.

\begin{theorem}[fundamental theorem of calculus]\label{fund-theorem-calculus}
\hfill\\
\cite[Cor.~34.1]{blanchard-bruening-2}\\
Let $[a,b]$ be a finite interval and $\gamma : [a,b] \rightarrow X$ a continuous function. For arbitrary $e \in X$ define a function $g : [a,b] \rightarrow X$ by
\[
g(t) = e + \int_a^t \gamma(s) \d s.
\]
Then $g$ is continuously differentiable and $g'(t) = \gamma(t)$ for all $t \in [a,b]$. Thus for $a \leq c < d \leq b$
\[
g(d)-g(c) = \int_c^d g'(t) \d t.
\]
\end{theorem}

\begin{lemma}\label{lemma-frechet}\cite[Lemma 34.3]{blanchard-bruening-2}\\
If $f$ is Gâteaux differentiable at all points of some neighbourhood $U \subseteq X$ of $x_0$ and $x \mapsto \delta_G f(x,\cdot) \in \mathcal{B}(X,Y)$ is continuous on $U$ then for all $h \in X$
\[
\delta_G f(x_0,h) = \delta_F f(x_0,h).
\]
By definition this also means that $f \in \Cont^1(U,Y)$.
\end{lemma}

\begin{proof}
Choose $\varepsilon > 0$ small enough such that $B_\varepsilon(x_0) \subset U$ then for $h \in B_\varepsilon(0)$ $g(\lambda) = f(x_0 + \lambda h)$ is continuously differentiable for all $\lambda \in [0,1]$ and $g'(\lambda) = \delta_G f(x_0 + \lambda h,h)$. By the fundamental theorem of calculus (\autoref{fund-theorem-calculus}) we have
\[
g(1)-g(0) = \int_0^1 g'(\lambda) \d\lambda = \int_0^1 \delta_G f(x_0 + \lambda h,h) \d\lambda
\]
and thus
\begin{align*}
f(x_0 + h) &- f(x_0) - \delta_G f(x_0,h) = g(1)-g(0) - \delta_G f(x_0,h) \\
&= \int_0^1 \big( \delta_G f(x_0 + \lambda h,h) - \delta_G f(x_0,h) \big) \d\lambda.
\end{align*}
This can be estimated by
\begin{align*}
\|f(x_0 + h) &- f(x_0) - \delta_G f(x_0,h)\|_Y \\[0.3em]
&\leq \sup_{\lambda \in [0,1]} \|\delta_G f(x_0 + \lambda h,\cdot) - \delta_G f(x_0,\cdot)\|_{\mathcal{B}(X,Y)}\, \|h\|_X.
\end{align*}
Continuity of $\delta_G f$ in its first component tells us that by dividing by $\|h\|_X$ and taking the limit $h \rightarrow 0$ the right hand side above goes to zero. This proves $\delta_F f(x_0,h) = \delta_G f(x_0,h)$ for all $h \in B_\varepsilon(0)$ and therefore for all $h \in X$ by linearity.
\end{proof}

\begin{proof}
\emph{Alternative proof after \citeasnoun[Prop.~A.3]{andrews-hopper}.}\\
Again choose $\varepsilon > 0$ small enough such that $B_\varepsilon(x_0) \subset U$ and let $g(\lambda)=f(x_0+\lambda h)-f(x_0)-\lambda \delta_G f(x_0,h)$ for $h \in B_\varepsilon(0)$. Note that this map is continuous in $\lambda$ implied by Gâteaux differentiability. The mean value inequality then tells us
\[
\|g(1)-g(0)\|_Y \leq \sup_{\lambda \in [0,1]} \|g'(\lambda)\|_Y.
\]
By putting in $g(0)=0$ and $g(1)=f(x_0+h)-f(x_0)-\delta_G f(x_0,h)$ as well as $g'(\lambda) = \delta_G f(x_0 + \lambda h,h)$ we have
\begin{align*}
\|f(x_0+h)-f(x_0)&-\delta_G f(x_0,h)\|_Y \leq \sup_{\lambda \in [0,1]} \|\delta_G f(x_0 + \lambda h,h)\|_Y \\
&\leq \sup_{\lambda \in [0,1]} \|\delta_G f(x_0 + \lambda h,\cdot)\|_{\mathcal{B}(X,Y)}\, \|h\|_X
\end{align*}
and from here on we can argue exactly as above.
\end{proof}

If the domain $X$ is finite dimensional another proof of equivalence is feasible that demands Lipschitz continuity of the map as an additional condition.

\begin{lemma}\label{lemma-frechet-2}\cite[Prop.~A.4]{andrews-hopper}\\
If $X$ finite-dimensional and $f$ is Gâteaux differentiable at $x_0 \in U \subseteq X$ open as well as Lipschitz-continuous on $U$ then for all $h \in X$
\[
\delta_G f(x_0,h) = \delta_F f(x_0,h).
\]
\end{lemma}

\begin{proof}
Take an arbitrary bounded neighbourhood $U_0 \subseteq X$ of $0 \in X$ such that $U_0+x_0 \subseteq U$ and note that for any $\varepsilon > 0$ it has a finite cover $\{ B_\varepsilon(u_i) \mid u_i \in U_0, i=1,\ldots,N \}$ because $\dim X < \infty$. Gâteaux differentiability at $x_0$ then implies that we find a $\delta>0$ such that for $|\lambda|<\delta$ and all $i \in \{1,\ldots,N\}$
\[
\|f(x_0 + \lambda u_i ) - f(x_0) - \lambda \delta_G f(x_0,u_i)\|_Y < \varepsilon |\lambda|.
\]
For an arbitrary $u \in U_0$ take a corresponding index $i \in \{1,\ldots,N\}$ such that $\|u-u_i\|_X<\varepsilon$. This makes the following estimate possible by using the triangle inequality twice.
\begin{align*}
\|f(x_0 + \lambda u ) &- f(x_0) - \lambda \delta_G f(x_0,u)\|_Y \\
\leq\;& \|f(x_0 + \lambda u)-f(x_0 + \lambda u_i)\|_Y \\
&+ \|f(x_0 + \lambda u_i)-f(x_0)-\lambda \delta_G f(x_0,u_i)\|_Y \\
&+ \|\lambda \delta_G f(x_0,u_i-u)\|_Y \\
\leq\;& (L + 1 + \|\delta_G f(x_0,\cdot)\|_{\mathcal{B}(X,Y)}) \,\varepsilon|\lambda|
\end{align*}
The first term yields the Lipschitz constant, the second originates from the estimate before, and the third is just the norm of the Gâteaux derivative. The right hand side can be made arbitrarily small for any $u \in U_0$ and hence $f$ is Fréchet differentiable at $x_0$.
\end{proof}

Let us conclude with an example taken for a function in finite dimensional space that is Gâteaux differentiable but fails to be differentiable in the Fréchet sense. Let $f:\R^2\rightarrow \R$ be $f(0,0)=0$ and $f(x,y)=x^3 y/(x^6+y^2)$ away from the origin. We get $\delta_G f=0$ at the origin so any directional derivative is zero. But going along the non-straight curve $(t,t^3)$ gives $f(t,t^3)=\onehalf$ if $t \neq 0$ and produces a jump at the origin, thus the function cannot be Fréchet differentiable.

\subsection{More results for Fréchet differentiable maps}

In the case of a continuously embedded subspace $\tilde{X} \hookrightarrow X$ (see \autoref{def-cont-comp-embedded}), i.e., $\tilde{X} \subseteq X$ and there is a $c>0$ such that for all $x \in \tilde{X}$ the estimate $\|x\|_X \leq c\|x\|_{\tilde{X}}$ holds, the restricted function $f|_{\tilde{X}}: \tilde{X} \rightarrow Y$ naturally keeps its Gâteaux differentiability because there is no reference to the $X$-norm in the definition. Yet it also keeps Fréchet differentiability which is shown in an easy proof formulated in the setting used above.

\begin{lemma}
If $f \in \Cont^1(U,Y)$ (Fréchet differentiable) then the restriction $f|_{\tilde{X}} \in \Cont^1(U \cap \tilde{X},Y)$ naturally with respect to the $\tilde{X}$-norm of a continuously embedded subspace $\tilde{X} \hookrightarrow X$.
\end{lemma}

\begin{proof}
Let $i: \tilde{X} \hookrightarrow X$ be the continuous inclusion map. Then for an open subset $U \subseteq X$ the preimage $i^{-1}(U) = U \cap \tilde{X} = \tilde{U}$ is again open. With the estimate on the norms the Fréchet condition easily transfers to $\tilde{X}$.
\begin{align*}
&\lim_{h \rightarrow 0} \frac{1}{\|h\|_{\tilde{X}}} \left\| f(x_0+h)-\big(f(x_0) + \delta_F f(x_0,h)\big) \right\|_Y \\
\leq c &\lim_{h \rightarrow 0} \frac{1}{\|h\|_X} \left\| f(x_0+h)-\big(f(x_0) + \delta_F f(x_0,h)\big) \right\|_Y = 0
\end{align*}
This is true for all $x_0 \in \tilde{U} \subseteq U$.
\end{proof}

Fréchet differentiability turns out to be the concept of choice if it comes to a generalisation of the fundamental theorem of calculus to the calculus of variations because it guarantees continuity in its first argument, the one including the integration variable. 

\begin{corollary}[fundamental theorem of the calculus of variations]\label{cor-fund-th-variations}
\hfill\\
Let $f \in \Cont^1(U,Y)$ (Fréchet differentiable), $x_0, x_1 \in U \subseteq X$, and the straight line $\{ x_0 + \lambda (x_1-x_0) \mid \lambda \in [0,1] \} \subset U$. Then it holds
\[
f(x_1)-f(x_0) = \int_0^1 \delta_F f(x_0 + \lambda (x_1-x_0), x_1-x_0) \d \lambda.
\]
\end{corollary}

\begin{proof}
Define $\gamma(\lambda) = \delta_F f(x_0 + \lambda (x_1-x_0), x_1-x_0)$ and note $\gamma : [0,1] \rightarrow Y$ continuous because of $f \in \Cont^1(U,Y)$ (note that Gâteaux differentiability does not need to fulfil this). We clearly have $\partial_\lambda f (x_0 + \lambda (x_1-x_0)) = \gamma(\lambda)$ from the definition of Gâteaux differentiability and thus by the fundamental theorem of calculus (\autoref{fund-theorem-calculus})
\[
f(x_1)-f(x_0) = \int_0^1 \gamma(\lambda) \d \lambda.
\]
This extends the usual fundamental theorem of calculus to the calculus of variations in Banach spaces.
\end{proof}

\subsection{The Landau symbols}

In dealing with the asymptotic behaviour of maps on Banach spaces, as in the definition of Gâteaux and Fréchet derivatives with a limit procedure, the following notation is frequently employed.

\begin{definition}[small Landau symbol]\label{def-small-o}
Let $f:U \rightarrow Y, g:U \rightarrow \R$ be defined in a neighbourhood of $x \in U \subset X$. We write $f(h) \in \lilo(g(h))$ as $h \rightarrow x$ if and only if
\[
\lim_{h \rightarrow x} \frac{\|f(h)\|_Y}{|g(h)|} = 0.
\]
\end{definition}

Verbally this could be stated as ``$f$ is negligible compared to $g$ close to $x$" and one can apply the concept for a shorter definition of Gâteaux and Fréchet differentiability, i.e.,
\begin{align}
\label{eq-gateaux-order}
 f(x_0+\varepsilon h)-\big(f(x_0) + \varepsilon\, \delta_G f(x_0,h)\big) \in \lilo(\varepsilon) &\mtext{as} \varepsilon \rightarrow 0, \\
\label{eq-frechet-order}
 f(x_0+h)-\big(f(x_0) + \delta_F f(x_0,h)\big) \in \lilo(\|h\|_X) &\mtext{as} h \rightarrow 0.
\end{align}
For the sake of completeness we will also include the definition of the big Landau symbol which states ``$f$ is bounded up to a constant factor by $g$ close to $x$".

\begin{definition}[big Landau symbol]\label{def-big-o}
Let $f:U \rightarrow Y, g:U \rightarrow \R$ be defined in a neighbourhood of $x \in U \subset X$. We write $f(h) \in \bigO(g(h))$ as $h \rightarrow x$ if and only if
\[
\limsup_{h \rightarrow x} \frac{\|f(h)\|_Y}{|g(h)|} < \infty.
\]
\end{definition}

We already introduced the notation of functional dependence by square brackets, e.g.~$\psi[v]$. From now on functional derivatives will always be in the Fréchet sense, thus only the symbol $\delta$ is used, e.g.~$\delta\psi[v;w]$. The semicolon between the arguments illustrates the different standing of $v$ and $w$, continuity in $v$ but additionally linearity in $w$.

\section[Variation of trajectories and observable quantities]{Variation of trajectories and\\observable quantities}
\label{sect-variation-trajectories}

As we are concerned with variations of potentials here we imagine ourselves being located at the solution $\psi[v]$ to the Schrödinger equation with potential $v$ which gets perturbed by $w$ and thus a variation arises. If a first order expression exists it shall be denoted by $\delta \psi[v;w]$ or $\delta \psi([v;w],t)$ if evaluated at a given time, the Gâteaux derivative. If some $\V$ is the space of time-dependent potentials and $\X \subseteq \Cont^0([0,T],\H)$ the corresponding space of quantum trajectories in the Hilbert space then naturally the question arises, if $\psi[\cdot] : \V \rightarrow \X$ is Fréchet differentiable on $\V$ or some subset $\U \subset \V$ for a fixed initial state $\psi_0$. This will also guarantee $\delta\psi[v;w] \in \X$ next to $\psi[v] \in \X$. We try to answer the question in the two different settings of stepwise static approximations (\autoref{sect-stepwise-static}) and successive substitutions (\autoref{sect-yajima}) in the sections \ref{sect-func-diff-stepwise} and \ref{sect-func-diff-successive} respectively.

But what if those variation should be taken over to observable quantities like the expectation value of a self-adjoint operator? We then have to make sure that the respective domain of the operator lies in the time-slices of $\X$ to make the operator applicable. Everything is thus easy in the case of bounded operators treated in \autoref{sect-func-diff-bounded} and the evaluation of the variation leads to the well-known Kubo formula \eqref{eq-kubo}. The case of unbounded observables is discussed in \autoref{sect-func-diff-unbounded}.

In \autoref{sect-variation-density} the investigation is extended to the one-particle density which cannot be given as the expectation value of a proper observable. A study of equations that hold for $\delta \psi$ follows in the final three sections and leads to an alternative proof of existence for strong solutions to Schrödinger's equation and a rough energy estimate.

\subsection{In the stepwise static approximation}
\label{sect-func-diff-stepwise}

Taking the definition of the Gâteaux derivative for a solution $\psi[v] = U[v] \psi_0$, where the evolution system is thought to be constructed as a limit procedure as in \autoref{th-schro-dyn-td} which allows for Hamiltonians $H(t) = H_0 + v(t)$ with $v : [0,T] \rightarrow \Sigma(L^2+L^\infty)$ Lipschitz continuous, one has
\begin{equation}\label{eq-psi-func-diff}
\delta \psi[v;w] = \lim_{\lambda \rightarrow 0} \frac{1}{\lambda} \big(U[v+\lambda w]-U[v]\big) \psi_0.
\end{equation}
To show Gâteaux differentiability one can very much proceed as in showing continuity in \autoref{cor-schro-evolution-continuous} but a higher degree of regularity is needed. Fréchet differentiability then follows as well.

\begin{theorem}\label{th-rs-frechet}
Let $\psi_0 \in H^{2(m+2)}$. Then the trajectory
\[
U[\cdot]\psi_0 : \Lip([0,T], W^{2(m+1),\Sigma}) \longrightarrow L^\infty([0,T],H^{2m})
\]
defined by \autoref{th-schro-dyn-td} is Fréchet-differentiable on the whole given potential space.
\end{theorem}

\begin{proof}
We start with $\delta\psi$ from \eqref{eq-psi-func-diff} above and first show that this expression is indeed a Gâteaux derivative, i.e., linear and continuous in $w$ for all $v$. Following the proof of \autoref{cor-schro-evolution-continuous} we get
\[
\delta \psi([v;w],t) = -\i \lim_{\lambda \rightarrow 0} \int_0^t U([v],t,s) w(s) U([v+\lambda w],s,0) \psi_0 \d s.
\]
Drawing the limit taken with respect to the $H^{2m}$-topology into the integral (the evolution system is uniformly continuous in time since we only consider a finite time interval), through the other quantities, and using continuity from \autoref{cor-schro-evolution-continuous} this expression is simply
\begin{equation}\label{eq-gateaux-deriv-1}
\delta \psi([v;w],t) = -\i \int_0^t U([v],t,s) w(s) U([v],s,0) \psi_0 \d s.
\end{equation}
Yet here one has to be very careful if the conditions for continuity are really fulfilled, because the evolution system needs to operate on a wave function in $H^{2(m+2)}$ to have a continuous outcome in $H^{2(m+1)}$ as demanded by the limit procedure because the multiplication by $w(s)$ is then $H^{2(m+1)} \rightarrow H^{2m}$. But with $v,w \in \Lip([0,T], W^{2(m+1),\Sigma}) \subset \Lip([0,T], W^{2m,\Sigma})$ this is indeed fulfilled. Linearity in $w$ is now clear and continuity follows like in \autoref{cor-schro-evolution-continuous}. This setting is sufficient for Gâteaux differentiability. To have Fréchet differentiability we employ \autoref{lemma-frechet}, the task is therefore to show continuity in $v$. We use the expression above to get the following for all $v,w,h \in \Lip([0,T], W^{2(m+1),\Sigma})$.
\begin{align*}
\delta \psi&([v+h;w],t) - \delta \psi([v;w],t) \\
=& -\i \int_0^t \big( U([v+h],t,s) w(s) U([v+h],s,0) - U([v],t,s) w(s) U([v],s,0) \big) \psi_0 \d s \\
=& -\i \int_0^t \big( (U([v+h],t,s)-U([v],t,s)) w(s) U([v+h],s,0) \\
&+ U([v],t,s) w(s) (U([v+h],s,0)-U([v],s,0)) \big) \psi_0 \d s
\end{align*}
Taking the $H^{2m}$-norm we have an estimate
\begin{align*}
\|\delta \psi&([v+h;w],t) - \delta \psi([v;w],t)\|_{2m,2} \\
\lesssim& \int_0^t \Big( \| (U([v+h],t,s)-U([v],t,s)) w(s) U([v+h],s,0) \psi_0 \|_{2m,2} \\
&+ \Pi_m[v(s)] \cdot \e^{C_m[v] L (t-s)} \cdot \| w(s) (U([v+h],s,0)-U([v],s,0)) \psi_0 \|_{2m,2} \Big) \d s
\end{align*}
where \eqref{eq-schro-evolut-estimate-2m} was already used in the second term. The first term in the integral goes to zero for $h \rightarrow 0$ because of continuity of $U[\cdot]$ applied to wave functions in $H^{2(m+1)}$ like before, the second one as well after the usual estimate with \autoref{lemma-rs-0} and using continuity of $U[\cdot]$ on $\psi_0 \in H^{2(m+2)}$ and potentials in $\Lip([0,T], W^{2(m+1),\Sigma})$.
\end{proof}

We can collect the regularity results of \autoref{th-sobolev-regularity}, \autoref{cor-schro-evolution-continuous}, and finally \autoref{th-rs-frechet} above in the following beautiful array valid for a fixed $\psi_0 \in H^{2m}$.
\begin{equation}\label{eq-rs-regularity}
\begin{aligned}
&U[\cdot]\psi_0 :\\ &\Lip([0,T], W^{2(m-1),\Sigma}) \longrightarrow \left\{
\begin{array}{lll} 
	L^\infty([0,T],H^{2m}) && \mbox{for}\; m \geq 1 \\
	L^\infty([0,T],H^{2(m-1)}) & \mbox{is}\; \Cont^0 & \mbox{for}\; m \geq 1 \\
		L^\infty([0,T],H^{2(m-2)}) & \mbox{is}\; \Cont^1 & \mbox{for}\; m \geq 2
\end{array}
\right.
\end{aligned}
\end{equation}
While the first line just means the mapping is well defined into the space of trajectories with $H^{2m}$-regularity it gets continuous and even Fréchet differentiable with respect to the potential in trajectory spaces of lesser regularity with corresponding coarser topologies. This also allows a clear conjecture to $\Cont^2$ and beyond. Finally we will derive an estimate for $\delta\psi$ with respect to the Sobolev norm where the aim is not on an explicit bound but only one in terms of the potential variation.

\begin{corollary}\label{cor-estimate-delta-psi}
For $v,w \in \Lip([0,T], W^{2m,\Sigma})$, which implies Gâteaux differentiability in $L^\infty([0,T],H^{2m})$, the following estimate holds.
\[
\sup_{t\in [0,T]} \|\delta\psi([v;w],t)\|_{2m,2} \leq T c[v]\max_{t\in [0,T]}\|w(t)\|_{2m,\Sigma} \cdot \|\psi_0\|_{2(m+1),2}
\]
\end{corollary}

\begin{proof}
We start from \eqref{eq-gateaux-deriv-1} and estimate with the help of \eqref{eq-schro-evolut-estimate-2m} from the \q{growth of Sobolev norms} and \autoref{lemma-rs-0} for the multiplication with $w$. The constants involved in the \q{$\lesssim$} relation thus depend on $v$.
\begin{align*}
\sup_{t\in [0,T]}\|\delta\psi([v;w],t)\|_{2m,2} &\leq \int_0^T \|U([v],t,s) w(s) U([v],s,0) \psi_0\|_{2m,2} \d s \\
&\lesssim \int_0^T \|w(s) U([v],s,0) \psi_0\|_{2m,2} \d s \\
&\lesssim \int_0^T \|w(s)\|_{2m,\Sigma} \cdot \|U([v],s,0) \psi_0\|_{2(m+1),2} \d s \\[0.5em]
&\lesssim T \max_{t\in [0,T]}\|w(t)\|_{2m,\Sigma} \cdot \|\psi_0\|_{2(m+1),2}
\end{align*}
\end{proof}

One would generally think that the smoothing map from \autoref{lemma-Uv-smoothing} could be used beneficially to lessen the restrictions for Fréchet differentiability of $\psi[\cdot]$ but the necessary reliance on continuity of $U[\cdot]$ and thus \autoref{cor-schro-evolution-continuous} destroyed this hope. Nevertheless an alternative estimate to the corollary above can be given by employing the smoothing map. Note that this estimate misses the explicit linear dependence on the length $T$ of the time interval.

\begin{corollary}\label{cor-estimate-delta-psi-2}
For $v,w \in \Cont^1([0,T], W^{2(m-1),\Sigma})$, $m\geq 1$, and additionally $v \in \Lip([0,T], W^{2m,\Sigma})$ the following estimate holds as well.
\begin{align*}
\sup_{t\in [0,T]}&\|\delta\psi([v;w],t)\|_{2m,2} \\
&\leq c'[v] \max_{t\in [0,T]}\left( \|w(t)\|_{2(m-1),\Sigma} + \|\dot w(t)\|_{2(m-1),\Sigma}\right) \cdot \|\psi_0\|_{2(m+1),2}
\end{align*}
\end{corollary}

\begin{proof}
Again we start from \eqref{eq-gateaux-deriv-1} but then use the estimate from the bounded smoothing map from \autoref{lemma-Uv-smoothing}. This makes us gain two orders of regularity in space but only by inclusion of the norm of the time derivative. Like before we use \eqref{eq-schro-evolut-estimate-2m} from the \q{growth of Sobolev norms} and \autoref{lemma-rs-0} for the multiplication with $w$. The constants involved in the \q{$\lesssim$} relation thus depend on $v$ again. We will first argue without the supremum over time.
\begin{align*}
\|\delta\psi([v;w],t)\|_{2m,2} \lesssim\;& \|w(t) U([v],t,0) \psi_0\|_{2(m-1),2} \\
&+ \|\partial_t w(t) U([v],t,0) \psi_0\|_{2(m-1),2} \\
\lesssim\;&\|w(t)\|_{2(m-1),\Sigma} \cdot \|U([v],t,0)\psi_0\|_{2m,2} \\
&+\|\dot w(t)\|_{2(m-1),\Sigma} \cdot \|U([v],t,0) \psi_0\|_{2m,2} \\
&+\|w(t)\|_{2(m-1),\Sigma} \cdot \|H([v],t) U([v],t,0) \psi_0\|_{2m,2} \\
\lesssim\;&\|w(t)\|_{2(m-1),\Sigma} \cdot \|\psi_0\|_{2m,2} \\
&+\|\dot w(t)\|_{2(m-1),\Sigma} \cdot \|\psi_0\|_{2m,2} \\
&+\|w(t)\|_{2(m-1),\Sigma} \cdot \|U([v],t,0) \psi_0\|_{2(m+1),2} \\
\lesssim\;&\left( \|w(t)\|_{2(m-1),\Sigma} + \|\dot w(t)\|_{2(m-1),\Sigma}\right) \cdot \|\psi_0\|_{2(m+1),2}
\end{align*}
Take two arbitrary functions $f,g:[0,T] \rightarrow \R_{\geq 0}$ then it is clear that $\max f$, $\max g \leq \max (f+g)$, thus $\onehalf (\max f + \max g) \leq \max (f+g)$ and of course $\max (f+g) \leq \max f + \max g$. Thus considering norms as above the sum of the maxima is equivalent to the maximum of the sum and we show the desired outcome.
\end{proof}

\subsection{In the successive substitutions method}
\label{sect-func-diff-successive}

The results of this and some of the following sections have been published in \citeasnoun{tddft-func-diff}. $\X$ and $\V$ will from now on be exactly the spaces defined in \autoref{sect-banach-spaces} again.

Now consider two potentials, $v$ and its perturbation $v+w$, with their corresponding Schrödinger equations.
\begin{align*}
 \i \partial_t \psi[v] &= H[v]\psi[v] \\
 \i \partial_t \psi[v+w] &= H[v+w]\psi[v+w]
\end{align*}
We will include $v$ into our substitution of $\psi$ such that only $w$ remains in the Tomonaga--Schwinger equation corresponding to the second line above. To this end we will go over to the $H[v]$-interaction picture with state vector $\hat{\psi}$ (now with a hat) by the substitution
\[
\psi(t) = U([v],t,0) \hat{\psi}(t).
\]
Just like before in \eqref{schwinger-tomonaga} we get the Tomonaga--Schwinger equation and its integral version.
\[
\i\partial_t \hat{\psi} = \hat{w} \hat{\psi}, \quad \hat{w}(t) = U([v],0,t) w(t) U([v],t,0)
\]
\begin{equation}\label{schroequ-mild-pertub}
\hat{\psi}(t) = \psi_0 - \i\int_0^t \hat{w}(s) \hat{\psi}(s) \d s
\end{equation}

Note that in the case of the Schrödinger equation with only potential $v$, that is $w=0$, this implies the identity $\hat{\psi}([v],t) = \psi_0$. To get an idea for the expression of $\delta\psi$ starting from Tomonaga--Schwinger the difference $\hat\psi([v+w],t) - \hat\psi([v],t)=\hat\psi([v+w],t) - \psi_0$ is calculated using \eqref{schroequ-mild-pertub} recursively.
\begin{align*}
\hat\psi([v+w],t)-\psi_0 &= - \i \int_0^t \hat{w}(s) \, \hat\psi([v+w],s) \d s \\
&= - \i \int_0^t \hat{w}(s) \left( \psi_0 - \i \int_0^s \hat{w}(r) \hat\psi([v+w],r) \d r \right) \!\d s
\end{align*}
With this expression it is easy to take the corresponding Gâteaux limit to get a first order approximation.
\begin{equation}\label{gateaux-deriv-int}
\delta \hat\psi([v;w],t) = \lim_{\lambda \rightarrow 0} \frac{1}{\lambda} \big(\hat\psi([v+\lambda w],t)-\psi_0\big) =  - \i \int_0^t \hat{w}(s) \psi_0 \d s
\end{equation}
Transformed back to the Schrödinger picture we have the same as in \eqref{eq-gateaux-deriv-1}
\begin{equation}\label{eq-gateaux-deriv-2}
\delta\psi([v;w],t) = - \i \int_0^t U([v],t,s) \, w(s) \, \psi([v],s) \d s,
\end{equation}
specifically varied at the origin and using \autoref{def-Qv} of $Q_w$ this reduces to $\delta\psi[0;w] = Q_w U_0 \psi_0$. We can also define the variation of the evolution operator $U[v]$ acting on $\psi_0$ which is equivalent to \eqref{eq-gateaux-deriv-2} above.
\begin{equation}\label{eq-U-gateaux-deriv}
\delta U([v;w],t,0) = - \i \int_0^t U([v],t,s) \, w(s) \, U([v],s,0) \d s
\end{equation}

To prove that this is actually a Fréchet derivative we will use the merits of operational notation and calculate the difference $\psi[v+\varepsilon w]-\psi[v]$ in the Schrödinger picture by directly putting in the Neumann series \eqref{neumann-series} in a shorthand notation. The only shortcoming is that this expression only holds for times $T > 0$ small enough such that $C_Q T^* \|v\|_\V < 1$ and we have to keep a good eye on this condition of which we will get rid in \autoref{Psi-frechet-2} by a concatenation of short enough time intervals.

\begin{theorem}\label{Psi-frechet}
Let $\psi_0 \in \H$, $\U \subset \V$ bounded and open, and $T>0$ short enough such that $C_Q T^* \|v\|_\V < 1$ for all $v \in \U$. Then the unique solution to the mild Schrödinger equation is Fréchet differentiable on $\U$, i.e., $\psi \in \Cont^1(\U,\X)$. Likewise we have the variation of the evolution operator $\delta U : \U \times \V \rightarrow \mathcal{B}(\H, \X)$.
\end{theorem}

\begin{proof}
We use the shorthand notation $R_v = (\id - Q_v)^{-1}$ as this operator is closely related to the resolvent of $Q_v$. Because of the limitation to potentials $v \in \U$ we have convergence of the Neumann series in \eqref{neumann-series} which means boundedness of $R_v$. Due to $Q_{v+w} = Q_v+Q_w$ the resolvent identity
\begin{equation}\label{eq-resolvent-identity}
R_{v+w} = R_v (\id + Q_w R_{v+w})
\end{equation}
holds.\footnote{We are indebted to a journal referee for pointing out this quicker way of showing the desired result.} Thus inserting recursively we get from \eqref{neumann-series} the difference
\[
\psi[v+\varepsilon w]-\psi[v] = R_v Q_{\varepsilon w} R_{v+\varepsilon w}U_0\psi_0 = \sum_{k=1}^\infty (R_v Q_{\varepsilon w})^k R_v U_0\psi_0.
\]
This series coverges for fixed $v,w$ and small enough $\varepsilon$. We use again linearity $Q_{\varepsilon w} = \varepsilon Q_{w}$ for $\varepsilon \in \R$ and the Gâteaux limit follows immediately.
\begin{equation}\label{Psi-frechet-expression}
\delta \psi[v;w] = \lim_{\varepsilon \rightarrow 0} \frac{1}{\varepsilon} (\psi[v+\varepsilon w]-\psi[v]) = R_v Q_w R_v U_0\psi_0
\end{equation}
Continuity (and linearity) of the above form of $\delta \psi$ in its second argument is readily established by continuity (and linearity) of $Q_w$ in $w$. This proves Gâteaux differentiability. If we additionally show $v \mapsto \delta\psi[v,\cdot]$ continuous as a mapping $\U \rightarrow \mathcal{B}(\V,\X)$ then \autoref{lemma-frechet} implies Fréchet differentiability. This is certainly true if $\lim_{h \rightarrow 0}\|\delta\psi[v+h;w] - \delta\psi[v;w]\|_\X = 0$ for all $w \in \V$. We show this by using expression \eqref{Psi-frechet-expression} for $\delta\psi$ and the resolvent identity \eqref{eq-resolvent-identity} once more.
\begin{align*}
\delta\psi[v+h;w] - \delta\psi[v;w] &= (R_{v+h} Q_w R_{v+h} - R_v Q_w R_v) U_0\psi_0 \\
&= \sum_{(j,k) \neq (0,0)}^\infty (R_v Q_h)^j R_v Q_w (R_v Q_h)^k R_v U_0\psi_0
\end{align*}
Again those sums will converge for small enough $h \in \V$ and the expression is well defined. As there is at least one $Q_h$ contained in every term and $\|Q_h\| \leq C_Q T^* \|h\|_\V$ the whole expression goes to $0$ as $h \rightarrow 0$. This makes $\psi : \V \rightarrow \X$ Fréchet differentiable on $\U$.
\end{proof}

Note particularly that if we want to widen $\U \subset \V$ to an open ball with radius $R$ of allowed potentials this means the time bound $T$ limited by $T^* < (C_Q R)^{-1}$ gets smaller and vice versa. By dividing the time interval in sufficiently short subintervals with individual evolution operators we can circumvent this limitation as shown by the following corollary.

\begin{corollary}\label{Psi-frechet-2}
For arbitrary albeit finite $T>0$ and $\psi_0 \in \H$ the unique solution to the mild Schrödinger equation is Fréchet differentiable on all of $\V$, i.e., $\psi \in \Cont^1(\V,\X)$.
\end{corollary}

\begin{proof}
We use the way $U[v]$ has been put together by expressions as in \eqref{def-Uv}, each one for a short enough time interval such that convergence is guaranteed. This means take $M \in \N$ large enough and define $\tau = T/M$ and $\tau^*$ as in \autoref{Uv-bounded} such that $C_Q \tau^* \|v\|_\V < 1$ for a fixed $v \in \V$. We also use the same partition into subintervals $I_1 = [0,\tau], I_2 = [\tau,2\tau], \ldots, I_M=[(M-1)\tau,T]$. Imagine for the time being $M=2$ is large enough, later we will generalise this case. Now we have
\[
(U[v]\psi_0)(t) = \left\{
\begin{array}{ll}
 U([v], t, 0) \psi_0 & \mbox{for}\; t \in I_1 \\
 U([v], t, \tau) \, U([v], \tau, 0) \psi_0 & \mbox{for}\; t \in I_2.
\end{array} \right.
\]
Each of the individual evolution operators is well defined due to the limitation to a sufficiently short time interval. Also their variations \eqref{eq-U-gateaux-deriv} are well defined, proved in \autoref{Psi-frechet}, one just needs to shift the potentials accordingly in time like in the proof of \autoref{Uv-bounded} to have the $Q_v$ and $Q_w$ operators acting correctly as the integrals therein always start at $t=0$. To determine $\delta U[v;w]$ we put in the expansion $U[v+w] \in U[v] + \delta U[v;w] + \lilo(\|w\|_\V)$ as $w \rightarrow 0$ for all evolutions.
\[
(U[v+w]\psi_0)(t) \in \left\{
\begin{array}{ll}
 U([v], t, 0) \psi_0 + \delta U([v;w], t, 0) \psi_0 + \lilo(\|w\|_\V) & \mbox{for}\; t \in I_1 \\[0.5em]
  U([v], t, \tau) \, U([v], \tau, 0) \psi_0 \\
+ \, \delta U([v;w], t, \tau) \, U([v], \tau, 0) \psi_0 \\
+ \, U([v], t, \tau) \, \delta U([v;w], \tau, 0) \psi_0 & \mbox{for}\; t \in I_2 \\
+ \, \delta U([v;w], t, \tau) \, \delta U([v;w], \tau, 0) \psi_0 \\
+ \, \lilo(\|w\|_\V)
\end{array} \right.
\]
The quadratic $\delta U$ term is of order $\lilo(\|w\|_\V)$ as $w \rightarrow 0$ as well and can therefore be neglected in the whole $\delta U[v;w]$ expression. We show this with the boundedness of $\delta U$ in its second argument from \autoref{Psi-frechet}, introducing a bound $C>0$. Further we employ the obvious estimate $\|\varphi(t)\|_2 \leq \|\varphi\|_{\X|I_m} \leq \|\varphi\|_\X$ for $t \in I_m$.
\begin{align*}
\|\delta U([v;w], \cdot, \tau) \, \delta U([v;w], \tau, 0) \psi_0\|_{\X | I_2} &\leq C \|w\|_\V \|\delta U([v;w], \tau, 0) \psi_0\|_2 \\
&\leq C  \|w\|_\V \|\delta U[v;w] \psi_0\|_{\X|I_1} \\
&\leq C^2 \|w\|_\V^2 \|\psi_0\|_2
\end{align*}
The extension to $M>2$ is straightforward and gives us the following product rule for $\delta \psi[v;w]$ at time $t \in I_m$.
\begin{align*}
\delta \psi([v;w],t) &= (\delta U[v;w]\psi_0)(t) \\
&= \delta U([v;w], t, (m-1)\tau) \ldots U([v], 2\tau,\tau) \, U([v], \tau,0) \psi_0 \\
&+ \cdots \\
&+ U([v], t, (m-1)\tau) \ldots \delta U([v;w], 2\tau,\tau) \, U([v], \tau,0) \psi_0 \\
&+ U([v], t, (m-1)\tau) \ldots U([v], 2\tau,\tau) \, \delta U([v;w], \tau,0) \psi_0
\end{align*}
The conditions of linearity and continuity needed for Fréchet differentiability can be directly transferred from \autoref{Psi-frechet}, as we add only finitely many terms.
\end{proof}

Additionally to the already established estimates for the Fréchet derivative from the stepwise static approximation \autoref{cor-estimate-delta-psi} and \autoref{cor-estimate-delta-psi-2} we can give the following estimate with the successive substitutions method.

\begin{corollary}\label{Psi-frechet-estimate}
For $v,w \in \V$ we have the following estimate for the Fréchet derivative.
\[
\sup_{t\in [0,T]}\|\delta\psi([v;w],t)\|_{2} \leq (1+C_v)^2 T^* \|w\|_\V \|\psi_0\|_2
\]
\end{corollary}

\begin{proof}
We start with the definition of the Fréchet derivative using the $H[v]$-interaction picture as in \eqref{gateaux-deriv-int} and by applying Minkowski's inequality. The transformation with the evolution operator $U([v],t,0)$ does not affect the $L^2$-norm, so we have $\sup_{t\in [0,T]}\|\delta\psi([v;w],t)\|_{2} = \sup_{t\in [0,T]}\|\delta\hat\psi([v;w],t)\|_{2}$.
\begin{align*}
\sup_{t\in [0,T]}\|\delta\psi([v;w],t)\|_{2} &= \sup_{t\in[0,T]} \left\| \int_0^t \hat{w}(s) \psi_0 \d s\right\|_2 \\
&\leq \int_0^T \|\hat{w}(s) \psi_0\|_2 \d s = \|\hat{w} \psi_0\|_{2,1}
\end{align*}
Next we apply the topological duality of $L^{2,\infty}$-$L^{2,1}$ with the time-space inner product $(\cdot,\cdot)$ to saturate the Hölder inequality with a special $\varphi \in L^{2,\infty} \subset \X$. Note that we write $\|\cdot\|_{2,\infty}$ and $\|\cdot\|_{2,1}$ for the associated norms now, not to be confused with the Sobolev norms.
\begin{equation}\label{estimate-duality-1}
|(\varphi,\hat w \psi_0)| = \|\varphi\|_{2,\infty} \cdot \|\hat w\psi_0\|_{2,1}
\end{equation}
Similarly we get by $\X$-$\X'$ duality and Hölder's inequality after substituting back the transformed $\hat w$ and moving one $U[v]$ to the left side of the inner product
\begin{equation}\label{estimate-duality-2}
|(\varphi,\hat w \psi_0)| = |(U[v]\varphi,w \psi[v])| \leq \|U[v]\varphi\|_\X \cdot \|w\psi[v]\|_{\X'}.
\end{equation}
Our main aim will be to get an estimate for the right hand side of \eqref{estimate-duality-2} which in return yields an inequality for $\|\delta\psi[v;w]\|_{2,\infty}$ over \eqref{estimate-duality-1}.
First we considers the term $\|U[v]\varphi\|_\X$ which has to be treated carefully, because it involves the time-dependent evolution of an also time-dependent trajectory, i.e., $t \mapsto U([v],t,0)\varphi(t)$. But we easily have
\[
\|U([v],t,0)\varphi(t)\|_q \leq \sup_{s \in [0,T]} \|U([v],t,0)\varphi(s)\|_q
\]
and thus
\[
\|U[v]\varphi\|_\X = \|\varphi\|_{2,\infty} + \|U[v]\varphi\|_{q,\theta} \leq \|\varphi\|_{2,\infty} + \sup_{s \in [0,T]} \|U[v]\varphi(s)\|_{q,\theta}.
\]
The Strichartz estimate from \autoref{Cv-strichartz} gives us
\[
\|U[v]\varphi(s)\|_{q,\theta} \leq C_v \|\varphi(s)\|_2
\]
and we have in combination
\begin{equation}\label{estimate-strichartz-3}
\|U[v]\varphi\|_\X \leq \|\varphi\|_{2,\infty} + C_v \sup_{s \in [0,T]} \|\varphi(s)\|_2 = (1+C_v) \|\varphi\|_{2,\infty}.
\end{equation}
The final term is $\|w\psi[v]\|_{\X'}$ from \eqref{estimate-duality-2} which is treated with \autoref{lemma-mult-op} for estimating the action of the multiplication operator $w$ and then a second time with the Strichartz estimate from \autoref{Cv-strichartz}.
\begin{equation}\label{estimate-strichartz-4}
\|w\psi[v]\|_{\X'} \leq T^* \|w\|_\V \|\psi[v]\|_\X \leq T^* \|w\|_\V \cdot (1+C_v)\|\psi_0\|_2
\end{equation}
We are now able to put \eqref{estimate-duality-1} and \eqref{estimate-duality-2} together with the estimates \eqref{estimate-strichartz-3} (where $\|\varphi\|_{2,\infty}$ cancels out) and \eqref{estimate-strichartz-4} above to state the inequality of the corollary.
\end{proof}

\subsection{Variation of bounded observable quantities}
\label{sect-func-diff-bounded}

We want to investigate the functional differentiability of the expectation value of observables. Consider the expectation value of a time-independent, self-adjoint, bounded operator $A : \H \rightarrow \H$ for a fixed initial state $\psi_0$ at time $t \in [0,T]$.
\[
\langle A \rangle_{[v]}(t) = \langle \psi([v],t) ,A \psi([v],t) \rangle
\]
Using the product rule for functional variations of potentials and switching to the $H[v]$-interaction picture once more we get the following from \eqref{gateaux-deriv-int} and $\hat\psi([v],t) = \psi_0$. (Note: The inner product is antilinear in the first component.)
\begin{equation}\label{eq-kubo}
\begin{aligned}
\delta \langle A \rangle_{[v;w]}(t) &= \langle \delta \psi([v;w],t), A \psi([v],t) \rangle + c.c. \\[0.5em]
&= \langle \delta \hat\psi([v;w],t), \hat A(t) \hat\psi([v],t) \rangle + c.c. \\
&= \i \int_0^t \langle \hat w(s) \psi_0, \hat A(t) \psi_0 \rangle \d s + c.c. \\
&= \i\int_0^t \langle [\hat{w}(s), \hat{A}(t)] \rangle_0 \d s
\end{aligned}
\end{equation}
This is exactly the Kubo formula of first order perturbations central to linear response theory touched upon in \autoref{sect-lrt}. Note especially that $\hat A(t)$ gets time-dependent because of the $H[v]$-interaction picture transformation with $U([v],t,0)$. We have the following estimate for the original $\delta A$ and arbitrary, finite times $T>0$ using the CSB inequality and the operator norm of $A$.
\begin{align*}
| \delta \langle A \rangle_{[v;w]}(t) | &\leq \sup_{t \in [0,T]} | \delta \langle A \rangle_{[v;w]}(t) | \\
&\leq 2\,\sup_{t \in [0,T]} | \langle \delta \psi([v;w],t), A \psi([v],t) \rangle |  \\
&\leq 2 \,\sup_{t \in [0,T]} \|\delta \psi([v;w],t)\|_2 \cdot \|A\psi([v],t) \|_2 \\
&\leq 2 \sup_{t\in [0,T]} \|\delta\psi([v;w],t)\|_{2,\infty} \cdot \|A\| \cdot \| \psi_0 \|_2
\end{align*}
At this point one can employ the estimate from \autoref{Psi-frechet-estimate} or alternatively from \autoref{cor-estimate-delta-psi} or \autoref{cor-estimate-delta-psi-2} with $m=0$ to finally achieve an estimate in terms of the involved potentials.

\subsection{Variation of unbounded observable quantities}
\label{sect-func-diff-unbounded}

To make similar estimates as in the section above for an unbounded operator $A$ we have to demand $\psi([v],t), \delta \psi([v;w],t) \in D(A)$ for all times $t \in [0,T]$. This means the potentials have to be from an appropriate Banach space to stabilise the trajectories within $D(A)$ as well as allowing Fréchet differentiation with respect to a trajectory space including $D(A)$ in the time slices. Typically those domains $D(A)$ will be Sobolev spaces.

In the stepwise static approximations setting we can guarantee such a regularity of Sobolev class right away if we limit ourselves to suitable Sobolev--Kato--Lipschitz spaces for the potentials as was demonstrated in \autoref{th-rs-frechet}. In the successive substitutions setting \autoref{prop-stabilising-set} gives a possible hint and a special result for $H^2$-regularity is presented in \autoref{def-XH2} and \autoref{def-VH2} below.

\begin{proposition}\label{prop-stabilising-set}
Let $\X_X$ be a reduced trajectory space $\X_X \subset \X$ that includes the time slices $X \subset \H$, i.e., $\varphi \in \X_X$ has $\varphi(t) \in X$ for all $t \in [0,T]$, and $\V_X \subset \V$ a subspace. Further let $U_0 : X \rightarrow \X_X$, i.e., the free evolution stabilises $X$, and $Q_v : \X_X \rightarrow \X_X$ be bounded by $\|Q_v\| \leq \varepsilon \|v\|_{\V_X}$ for all $v \in \V_X$, $\varepsilon > 0$ arbitrarily small for sufficiently short time $T>0$. Then $\psi[v],\delta\psi[v;w] \in \X_X$ if $v,w \in \V_X$ and $\psi_0 \in X$.
\end{proposition}

The proof would just be a repetition of the proofs in \autoref{sect-full-int-pic} and \autoref{sect-func-diff-successive} above with $\X,\V$ exchanged for $\X_X,\V_X$. This is possible because the assumptions of the proposition readily replace the central \autoref{Q-bounded}. Note that we expect such spaces $\X_X$ to be accompanied by a formal dual $\X'_X \supset \V_X \cdot \X_X$ to allow for estimates as in \autoref{lemma-mult-op}.

One important case is the Laplace operator with $X=D(\Delta) = H^2$ if $\Omega=\R^d$, a Sobolev space which is also treated in \citeasnoun{yajima} but with an approach slightly different to the one outlined in the proposition above. Such a trajectory surely has finite kinetic energy, associated with the state space $H^1$, at all times. We will just repeat the definition of the corresponding spaces from \citeasnoun{yajima} to be able to compare them to spaces with similar properties derived in \autoref{th-strong-dynamics}.

\begin{definition}\label{def-XH2}
A Banach space of quantum trajectories with $H^2$-regularity at every time instant is given by
\begin{align*}
\X_{H^2} = \{ \varphi \in \Cont^0([0,T],H^2) \mid \partial_t \varphi \in \X \}, \\
\|\varphi\|_{\X_{H^2}} = \sup_{t \in [0,T]} \|\varphi(t)\|_{H^2} + \|\partial_t \varphi\|_\X.
\end{align*}
\end{definition}

The corresponding set of allowed potentials for stability of the state space $\X_{H^2}$ under evolution by the Schrödinger equation is given as follows.

\begin{definition}\label{def-VH2}
The conditions on the exponents for $\V_{H^2}$ read as follows (note that $p, \alpha$ are still linked to $q,\theta$ by \autoref{def-V}).
\begin{align*}
&\tilde{p} = \max\{ p,2 \} \\
&p_1 = 2np/(n+4p) &\mathrm{if}\; n \geq 5 \\
&p_1 > 2p/(p+1) &\mathrm{if}\; n = 4 \\
&p_1 = 2p/(p+1) &\mathrm{if}\; n \leq 3 \\
&\alpha_1 > 4p/(4p-n)
\end{align*}
The corresponding Banach space is then defined as
\[
\V_{H^2} = \{ v \in \Cont^0([0,T],L^{\tilde{p}}) + \Cont^0([0,T],L^\infty) \mid \partial_t v \in L^{p_1,\alpha_1} + L^{\infty,\beta} \}.
\]
\end{definition}

The proof uses an approximation technique for potentials and trajectories and thus does not need to rely on all the assumptions of \autoref{prop-stabilising-set}. Since the space $\X_{H^2}$ which takes the role of $\X$ as the space of trajectories incorporates the Laplace domain $H^2$ the necessary conditions for a well-defined variation $\delta \langle \Delta \rangle$ are met.

The main reason why such regularity results have been originally studied by us is the special structure of the internal forces term that will appear in the main equation of TDDFT, see \autoref{sect-internal-forces} for a definition. It includes $4^\mathrm{th}$ order spatial derivatives of the wave function. Like the density if formulated as an expectation value it is \emph{not} an operator but rather an operator-valued distribution. In any case we must restrict the potentials in such a way that the quantum trajectory lies in the domain of the observable under consideration. Further we may want to use the functional variation of this internal forces term and thus also $\delta \psi$ has to be in the domain of the observable and allow for $4^\mathrm{th}$ order spatial derivatives. The best result at hand to guarantee this is \autoref{th-rs-frechet} and the whole discussion is carried out in \autoref{sect-q-mapping-frechet}.

To this end the following lemma showing commutativity of Fréchet derivative $\delta$ and spatial (weak) derivative $\partial$ in any direction will later prove valuable.

\begin{lemma}\label{lemma-permut-delta}
For $\psi_0 \in H^1$ and $v,w$ from a space that guarantees $\psi([v],t)$, $\delta\psi([v;w],t) \in H^1$ for all $t \in [0,T]$ it holds $\partial \delta \psi[v;w] = \delta \partial \psi[v;w]$ where we naturally define the spatial derivative $\partial\psi[v] = (\partial\psi)[v] = \partial(\psi[v])$.
\end{lemma}

\begin{proof}
Take $\varphi \in \mathcal{D}(\Omega^N)$ a test function and rewrite the functional derivative in $\delta \partial \psi[v;w]$ as a limit. This expression is taken at each time $t \in [0,T]$ with the time variable suppressed.
\begin{align*}
\langle \varphi,\delta \partial \psi[v;w] \rangle &= \langle \varphi, \lim_{\lambda \rightarrow 0}\frac{1}{\lambda} (\partial\psi[v+\lambda w]-\partial\psi[v]) \rangle \\
&= \lim_{\lambda \rightarrow 0}\frac{1}{\lambda} \langle \varphi,\partial\psi[v+\lambda w]-\partial\psi[v] \rangle \\
&= -\lim_{\lambda \rightarrow 0}\frac{1}{\lambda} \langle \partial\varphi,\psi[v+\lambda w]-\psi[v] \rangle \\
&= - \langle \partial\varphi,\lim_{\lambda \rightarrow 0}\frac{1}{\lambda} (\psi[v+\lambda w]-\psi[v]) \rangle = \langle \varphi, \partial\delta \psi[v;w] \rangle
\end{align*}
The integration by parts without boundary terms is due to the vanishing of the test function at the border. The last limit above is actually the weak limit but it is equivalent to the limit in the Banach space topology since this exists by assumption of Fréchet differentiability. We are now able to identify  $\partial \delta \psi[v;w] = \delta \partial \psi[v;w]$.
\end{proof}

\subsection{Variation of the density}
\label{sect-variation-density}

Another important quantity though no proper observable is the (one-particle) density. We adopt the notation $x=x_1$, $\bar{x} = (x_2,\ldots,x_N)$, $\ushort{x} = (x,\bar{x}) = (x_1,\ldots,x_N)$, $\bar{\Omega}=\Omega^{N-1}$, and $\nabla$ only acting on the first particle position $x_1$. For spatially (anti-)symmetric trajectories $\psi[v] \in \Cont^0([0,T], \H) \supset \X$ the density is defined as
\[
n([v],t,x) = N \int_{\bar{\Omega}}  \d \bar{x}\, |\psi([v],t,x,\bar{x})|^2.
\]
Within our framework it is now natural to ask for the Fréchet derivative $\delta n[v;w]$. The necessary continuity property automatically translates from $\delta\psi$ because all involved operations are continuous themselves. As in \eqref{eq-kubo} we get
\[
\delta n([v;w],t,x) = N \int_{\bar{\Omega}} \d \bar{x}\, \psi^*([v],t,x,\bar{x}) \delta\psi([v;w],t,x,\bar{x}) + c.c.
\]
An estimate can now be easily established with \autoref{Psi-frechet-estimate} or alternatively with \autoref{cor-estimate-delta-psi} or \autoref{cor-estimate-delta-psi-2} setting $m=0$.
\begin{align*}
\sup_{t \in [0,T]} \|\delta n([v;w],t)\|_1 &= \sup_{t \in [0,T]} \int_{\Omega} \d x \left| \delta n([v,w],t,x) \right| \\
&\leq 2 N \sup_{t \in [0,T]} \langle |\psi([v],t)|, |\delta\psi([v;w],t)| \rangle \\
&\leq 2 N \sup_{t \in [0,T]} \|\delta\psi([v;w],t)\|_2 \cdot \|\psi_0\|_2
\end{align*}

To establish a more explicit connection to physics and standard non-equilibrium density-response theory \cite{stefanucci-vanleeuwen} we
consider only symmetric one-body perturbations $w \in \V$ of the form $w(t,\ushort{x})=\sum_{k=1}^{N}w(t,x_k)$ (using the same symbol twice). Furthermore we adopt the usual tacit assumption that the unitary evolution operator $U([v],t,s)$ can be represented by an integral transformation with an integral kernel (the so-called propagator) of the form $U([v],t,x,\bar{x};s,y,\bar{y})$. Then the functional derivative can be rewritten as
\[
 \delta n([v;w],t,x) = \int_{0}^{t} \d s \int_{\Omega} \d y \, \chi([v],t,x;s,y) \, w(s,y),
\]
where the kernel is defined by 
\begin{equation}\label{eq-response-kernel}
\begin{aligned}
& \chi([v],t,x;s,y) \\
& = -\i N^2 \! \int\limits_{\bar\Omega^{2}} \d \bar{x}  \d \bar{y}\, \psi^*([v],t,x,\bar{x}) U([v],t,x,\bar{x};s,y,\bar{y}) \psi([v],s,y,\bar{y}) + c.c.
\end{aligned}
\end{equation}
This linear-response kernel is further studied in \autoref{sect-lrt}.

The potential-density mapping $n[v]$ obviously plays a fundamental role in density functional theory where the core problem is to show its invertability in an adequate Banach space setting. This leads us to the question how its range, the set of all so-called $v$-representable densities, $n[\V]$ looks like. We showed that $n[\cdot]$ is a Fréchet differentiable mapping from the Banach space $\V$ to $\Cont^0([0,T],L^1(\Omega))$, the base space of densities. If for some potential $v$ we have $\delta n[v,\cdot]$ as a bounded linear isomorphism then $n[\cdot]$ is locally invertible at $v$ by the inverse function theorem. Such a proof of invertability is given in \citeasnoun[3.2]{van-leeuwen2} for switch-on potentials that are Laplace-transformable by showing $\ker \delta n[v,\cdot] = \{ w \mid w(t,\ushort{x}) = C(t) \}$, the equivalence class of zero in the framework of the Runge--Gross theorem (\autoref{runge-gross-th}).

\subsection{An inhomogeneous Schrödinger equation for $\delta \psi$}
\label{sect-inhom-schro-derivative}

Here we show that one can derive an inhomogeneous Schrödinger equation for $\delta\psi[v; w]$ with initial value $0$. We might then employ a theorem for inhomogeneous initial value problems in the hyperbolic case from \citeasnoun{pazy} to show $\delta \psi([v;w],t) \in H^2$ for special classes of potentials, thus giving one more alternative approach towards the regularity results in \autoref{th-rs-frechet} and \autoref{prop-stabilising-set}.

We use the implicit definition of the functional derivative $\delta\psi$ \eqref{eq-gateaux-order} and write 
\[
\psi[v + \varepsilon w] \in \psi[v] + \varepsilon\,\delta\psi[v; w] + o(\varepsilon),
\]
to put this into the Schrödinger equation
\[
\i \partial_t \psi[v + \varepsilon w] = H[v + \varepsilon w] \psi[v + \varepsilon w]
\]
while we split the Hamiltonian with $H[v + \varepsilon w] = H[v] + \varepsilon w(t)$. Thus we get the ordinary Schrödinger equation for $\psi[v]$ which cancels itself and we are left with the following.
\[
\i \partial_t \delta\psi[v; w] \in H[v] \delta\psi[v; w] + w\psi[v] + o(\varepsilon)
\]
In the (strong) limit $\varepsilon \rightarrow 0$ an inhomogeneous Schrödinger equation for $\delta\psi[v; w]$ remains.
\begin{equation}\label{inhom-schro-eq}
\i \partial_t \delta\psi[v; w] = H[v] \delta\psi[v; w] + w \psi[v]
\end{equation}

But pay attention that the such defined functional derivative is not automatically continuous in the first argument, thus not a Fréchet derivative but rather only a Gâteaux derivative. To ensure Fréchet differentiability we can still resort to \autoref{th-rs-frechet}. By using \citeasnoun[ch.~5, Th.~5.3]{pazy} already mentioned after \autoref{lemma-Uv-smoothing} we really have the desired $\delta\psi([v; w],t)$ as a (unique) strong solution to the equation above given once more by \eqref{eq-gateaux-deriv-1}. The critical condition is that the inhomogeneity $w \psi[v]$ is in $\Cont^1([0,T],\H)$. A similar situation can be found in \citeasnoun[VIII.5]{lions-book} but with an $L^2$ condition in time.

Let us see what the critical condition tells us in the setting of the stepwise static approximation.
\[
\partial_t (w \psi[v])  = \dot w  \psi[v]  + w \dot \psi[v] \in \Cont^0([0,T],\H)
\]
Now we need $\psi([v],t) \in H^4$ to still have $\dot\psi([v],t) \in H^2$ by the Schrödinger equation. To achieve this degree of regularity for the state at all times we have to demand $v \in \Lip([0,T],W^{2,\Sigma})$ by \autoref{th-sobolev-regularity}. If further $w \in \Cont^1([0,T],\Sigma)$ the conditions for the existence of a unique solution to \eqref{inhom-schro-eq} are met by using the estimate \autoref{lemma-sum-space-inequality}. This shows the equation has a strong solution $\delta\psi([v; w],t) \in H^2$ and can thus be taken as an alternative regularity result for Gâteaux derivatives but with an additional $\Cont^1$ condition on the potential. Such a condition already showed up in a similar case noted at the end of \autoref{sect-existence-stepwise-static} and in the estimate for the functional derivative \autoref{cor-estimate-delta-psi-2}.

Next we study \eqref{inhom-schro-eq} in the stepwise static approximation method, that is the setting of \citeasnoun{yajima}. Here the encompassing space for the Schrödinger equation is $\X'=L^{2,1} + L^{q',\theta'}$ from \autoref{def-X} and in an analogy to the criterion above we want to test
\[
\partial_t(w \psi[v]) = \dot w \psi[v] + w \dot \psi[v] \in \X'.
\]
We assume $v$ such that $\psi([v],t) \in H^2$ and $\dot \psi[v] \in \X$, demanding exactly the space $\V_{H^2}$ from \autoref{def-VH2}. The estimate for the second term $\|w \dot \psi[v]\|_{\X'}$ is readily given by \autoref{lemma-mult-op} and $\V_{H^2} \subset \V$. The first term will be broken up using $\dot w = \dot w_1 + \dot w_2$ and estimated with the extended Hölder inequality \autoref{hoelder-ineq}. The first part  gets further estimated as in \eqref{eq-lemma-mult-op} with $\beta>1$, the second part with \autoref{lemma-ab-inequ}.
\begin{align*}
\|\dot w_1 \psi[v]\|_{2,1} &\leq \|\dot w_1\|_{\infty,1} \|\psi[v]\|_{2,\infty} \lesssim \|\dot w_1\|_{\infty,\beta} \|\psi[v]\|_{2,\infty} \\
\|\dot w_2 \psi[v]\|_{q',\theta'} &\leq \|\dot w_2\|_{q',\theta'} \, \|\psi[v]\|_{\infty,\infty} \\
&\lesssim \|\dot w_2\|_{q',\theta'} \sup_{t \in [0,T]}(\|\Delta \psi([v],t)\|_2 + \|\psi([v],t)\|_2)
\end{align*}

Note that \autoref{lemma-ab-inequ} only holds in dimensionality $n \leq 3$, but it is not our aim to derive a general regularity result here, just to show consistency. If we now derive expressions in terms of $n$ and $p$ using the relations in \autoref{def-X} and \autoref{def-V} for the indices $q',\theta'$ we wondrously arrive exactly at the indices from \autoref{def-VH2} of $\V_{H_2}$.
\begin{align*}
q' &= \frac{2p}{1+p} = p_1 \\
\theta' &= \frac{4p}{4p-n} < \alpha_1
\end{align*}
Note that the last index $\alpha_1$ can always be chosen larger because the time interval is assumed to be bounded and thus $L^{\alpha_1}([0,T]) \subset L^{\theta'}([0,T])$. This can be taken as a demonstration of consistency in the established framework.

\subsection{A relation to Duhamel's principle}
\label{sect-duhamel}

In trying to find an iterative solution to \eqref{inhom-schro-eq} we can follow exactly the scheme developed in \autoref{sect-full-int-pic}. First we set $\psi = U_0 \tilde{\psi}, \delta\psi = U_0 \delta\tilde{\psi}, \tilde{v} = U_0^* v U_0, \tilde{w} = U_0^* w U_0$ and get
\[
\i \partial_t \delta\tilde{\psi}[v; w] = \tilde{v} \delta\tilde{\psi}[v; w] + \tilde{w} \tilde{\psi}[v]
\]
Now integration over the time interval $[0,t]$ recognising $\delta\psi([v;w],0)=0$ and switching back to the Schrödinger picture yields an equation analogous to the mild Schrödinger equation \eqref{schroequ-mild}.
\begin{equation}\label{schroequ-mild-2}
\begin{aligned}
\delta\psi([v;w],t) &= -\i \int_0^t U_0(t-s) (v(s)\delta\psi([v;w],s) + w(s)\psi([v],s)) \d s \\
&= (Q_v \delta\psi[v;w])(t) + (Q_w \psi[v])(t)
\end{aligned}
\end{equation}
In comparing the above equation with the usual mild Schrödinger equation
\[
\psi([v],t) = (Q_v\psi[v])(t)+(U_0\psi_0)(t)
\]
we note that the solution $\delta\psi[v;w]$ to a Schrödinger equation with inhomogeneity $w\psi[v]$ could equally well be interpreted as a Cauchy problem for the homogeneous Schrödinger equation with `initial data' $U_0^* Q_w \psi[v]$. Because of the time dependency this expression is no real candidate for an initial wave function but the procedure reminds of Duhamel's principle for solving inhomogeneous linear evolution equations that we already met in \autoref{sect-full-int-pic}. This states that the solution to \eqref{inhom-schro-eq} is given by
\[
\delta\psi([v; w],t) = -\i \int_0^t U([v],t,s) w(s) \psi([v],s) \d s,
\]
and actually this is exactly the first expression we derived for the variational derivative in equation \eqref{eq-gateaux-deriv-1}. To get a series representation for the variation we can invert equation \eqref{schroequ-mild-2} and put in the Neumann series as in \eqref{neumann-series}.
\begin{align*}
\delta\psi[v;w] &= (\id - Q_v)^{-1} Q_w \psi[v] = (\id - Q_v)^{-1} Q_w (\id - Q_v)^{-1} U_0 \psi_0 \\
&= \sum_{k,l=0}^\infty Q_v^k Q_w Q_v^l U_0 \psi_0
\end{align*}
By just rewriting the double sum we arrive at the expression from the proof of \autoref{Psi-frechet} thereby closing the circle. This means, in reverse, we have also found exact conditions on solubility of inhomogeneous Schrödinger equations of the kind of \eqref{inhom-schro-eq} in the previous sections.

\subsection{Energy estimates from trajectory variations and another existence proof}
\label{sect-energy-estimates}

The goal of this section is to derive estimates on the expectation value of the Hamiltonian. When directly substituted with Schrödinger's equation we get
\[
\langle \psi([v],t), H([v],t) \psi([v],t)\rangle = \langle \psi([v],t), \i\dot \psi([v],t)\rangle
\]
which by no means can be considered to be conserved because of the varying potential $v$. A direct estimate for $\dot{\psi}[v]$ will lead to variational derivatives if we write out the time derivative explicitly and transform it such that a variation with a potential appears.
\[
\dot{\psi}([v],t) = \lim_{h \rightarrow 0} \frac{1}{h} (U([v],t+h,0) - U([v],t,0))\psi_0.
\]
By defining a time-shift operator $T_h v(t) = v(t+h)$, $h \in \R$, we can rewrite the evolution operator $U([v],t+h,0)$ by evolving first over the short interval $[0,h]$ and then again over $[0,t]$ but with a shifted potential $T_h v$.
\[
U([v],t+h,0) = U([T_h v],t,0) U([v],h,0)
\]
After introducing the null term
\[
-U([v],t,0) U([v],h,0) + U([v],t,0) U([v],h,0)
\]
we group the expression inside the limit in the following fashion.
\begin{align*}
\dot{\psi}([v],t) = \lim_{h \rightarrow 0} \frac{1}{h} \big((&U([T_h v],t,0) - U([v],t,0)) U([v],h,0) \\ + &U([v],t,0) (U([v],h,0) -\id)\big)\psi_0
\end{align*}
We expand $T_h v$ into $v + h\frac{1}{h}(T_h v - v)$ in the first term to establish an expression alike the Gâteaux difference quotient. Assuming Fréchet differentiability we have uniform convergence and we might separately evaluate the limits.
\begin{align*}
&\lim_{h \rightarrow 0} \frac{1}{h}(T_h v - v) = \dot{v} \\[0.3em]
&\lim_{h \rightarrow 0} \frac{1}{h} (U([v + h \dot{v}],t,0) - U([v],t,0)) = \delta U([v;\dot{v}],t,0) \\[0.6em]
&\lim_{h \rightarrow 0} U([v],h,0) = \id \\[0.4em]
&\lim_{h \rightarrow 0} \frac{1}{h} (U([v],h,0) -\id) = \partial_t U([v],t,0) \big|_{t=0} = -\i H([v],0)
\end{align*}
This gives us the following variant of the Schrödinger equation
\begin{equation}\label{strange-se}
\i \dot{\psi}([v],t) = H([v],t) \psi([v],t) = \i \delta \psi([v;\dot{v}],t) + U([v],t,0) H([v],0) \psi_0,
\end{equation}
or more elegantly by suppressing the time variable and writing down an equation for whole trajectories
\begin{equation*}
\i \dot{\psi}[v] = H[v] \psi[v] = \i \delta \psi[v;\dot{v}] + U[v] H[v(0)] \psi_0.
\end{equation*}

This has some immediate consequences for the question of strong solutions to the Schrödinger equation. If the right hand side of the equation is well defined, which we will show to be the case in \autoref{th-strong-dynamics} below for $v, \dot{v} \in \V$ and $\psi_0 \in H^2$ in the successive substitutions setting (of course the same strategy could also be followed in the stepwise static approximations setting), we also have a definite time-derivative along the trajectory and thus a $H^2$ trajectory. Note particularly how this seems to yield another (similar, yet simpler) Banach space of potentials guaranteeing $H^2$-regularity like that of \citeasnoun{yajima} given in \autoref{def-VH2}. Such a $\Cont^1$ condition for the potential already appeared in the classic work of \citeasnoun[Th.~6.2]{phillips} and is repeated similarly in \citeasnoun[Lemma 3.2]{yajima}, both times for bounded operators $v\cdot$. To make the statement precise we need a small technical lemma, which is a slight variation of the one given in \citeasnoun{hiriart}.

\begin{lemma}\label{lemma-Lr}
For $1 \leq p < r < q \leq \infty$ and any $\Omega \subseteq \R^n$ measurable we have
\[
L^r(\Omega) \subset L^p(\Omega) + L^q(\Omega).
\]
\end{lemma}

\begin{proof}
Let $f \in L^r(\Omega)$ and define the sets $X = \{ x\in\Omega : |f(x)|>1 \}$ and $X^c$ its complement which are well-defined up to null sets. By aid of the characteristic function of these sets we decompose $f = f_1 + f_2$ with $f_1 = f \cdot \1_X$ and $f_2 = f \cdot \1_{X^c}$. We now prove $f_1 \in L^p(\Omega)$.
\[
\|f_1\|_p^p = \int_X |f(x)|^p \d x = \int_X \underbrace{|f(x)|^{p-r}}_{\leq 1} |f(x)|^r \d x \leq \|f\|_r^r < \infty
\]
Clearly $f_2 \in L^q(\Omega)$ can be proved just analogously for $q < \infty$. The case $q = \infty$ is even simpler because of $\esssup{x\in\Omega} |f_2(x)| = \esssup{x\in X^c} |f(x)|$ $\leq 1$. Note that for $|\Omega|<\infty$ we have the result $L^r(\Omega) \subset L^p(\Omega)$ by Hölder's inequality which implies the result above.
\end{proof}

\begin{theorem}\label{th-strong-dynamics}
For $v, \dot{v} \in \V$, $p \geq 2$ in the definition of $\V$, and $\psi_0 \in H^2$ we get a strong solution to Schrödinger's equation up to an arbitrary albeit finite time $T$.
\end{theorem}

\begin{proof}
By the equation derived above
\[
\i \dot{\psi}[v] = \i \delta \psi[v;\dot{v}] + U[v] H[v(0)] \psi_0.
\]
we have a strong solution to the Schrödinger equation if the right hand side is well-defined and continuous, making $\psi[v]$ continuously differentiable in time. \autoref{Psi-frechet-2} tells us we get a well-defined Fréchet derivative $\delta \psi[v;\dot{v}]$ if $v,\dot{v} \in \V$. The last term includes $(H_0 + v(0))\psi_0$ with $v(0) \in L^p + L^\infty$. The assumption $p \geq 2$ is very reasonable because of $p > \frac{n}{2}$ for dimensionality $n \geq 3$ like explained in \autoref{sect-banach-spaces}. Now $L^p \subset L^2 + L^\infty$ due to \autoref{lemma-Lr} and thus $v(0) \in L^2 + L^\infty$. Now by virtue of \autoref{th-kato} (Kato's theorem) $H_0 + v(0)$ is self-adjoint on $H^2$ and thus the whole right hand side is an element in $L^2$. $U([v],t,0)$ being strongly continuous in time guarantees continuity.
\end{proof}

Note that this is not a circular reasoning in the successive substitutions setting because the condition $v \in \V$ only proves existence of generalised solutions that are made strong with the additional requirement that $\dot v \in \V$, similar to the condition put forward in \autoref{def-VH2} or \autoref{sect-inhom-schro-derivative}.

By using the time-differentiation rules for $U([v],t,s)$ with arbitrary initial time $s$ and directly relating the unitary evolutions independent of the initial state $\psi_0$ an alternative, elegant form of the equation \eqref{strange-se} above would be
\[
(\partial_t - \partial_s) U([v],t,s) = \delta U([v;\dot{v}],t,s).
\]

An energy estimate at time $t$ can be derived from the expectation value of the Hamiltonian together with \eqref{strange-se}.
\begin{align*}
E(t) &= \langle \psi([v],t), H([v],t) \psi([v],t) \rangle \\
&= \i \langle \psi([v],t), \delta \psi([v;\dot{v}],t) \rangle
+ \langle \psi([v],t), U([v],t,0) H([v],0) \psi_0 \rangle
\end{align*}
If we use $\langle \psi([v],t), U([v],t,0) H([v],0) \psi_0 \rangle = \langle \psi_0, H([v],0) \psi_0 \rangle = E(0)$ and move this term to the other side, we just have the energy difference at times $t$ and 0 respectively. We get an estimate almost directly from the CSB inequality and \autoref{Psi-frechet-estimate} (which we will use for the moment) or alternatively from \autoref{cor-estimate-delta-psi}  or \autoref{cor-estimate-delta-psi-2} with $m=0$.
\begin{align*}
E(t) - E(0) &= \i \langle \psi([v],t), \delta \psi([v;\dot{v}],t) \rangle \\[0.5em]
\sup_{t \in [0,T]} |E(t) - E(0)|& \leq \sup_{t \in [0,T]} \|\psi([v],t)\|_2 \cdot \|\delta \psi([v;\dot{v}],t)\|_2 \\
&\leq (1+C_{v})^2 T^* \|\dot{v}\|_\V \|\psi_0\|_2^2
\end{align*}
If we desire an estimate for the kinetic part $\|\nabla \psi([v],t)\|_2^2$ in a simple form, we additionally demand $v(0) = v(t) = 0$ (switch-on-off potential), which implies $E(0) = \|\nabla \psi_0\|_2^2, E(t) = \|\nabla \psi([v],t)\|_2^2$. 
\[
\|\nabla \psi([v],t)\|_2^2 \leq \|\nabla \psi_0\|_2^2 + (1+C_{v})^2 T^* \|\dot{v}\|_\V \|\psi_0\|_2^2
\]
Note that this estimates scales linearly with $\|\dot{v}\|_\V$ which could be interpreted as the linear dependency of energy transmission on the frequency of an electric potential like in the Planck--Einstein relation $E = h\nu$. A direct estimate in the other direction is given by the Hardy--Sobolev inequality for $n\geq 3$.
\[
\|\nabla \psi\|_2 \geq \frac{n-2}{2} \| x \mapsto \psi(x)/x \|_2
\]

\section{Linear response theory}
\label{sect-lrt}

This section was compiled following notes from Heiko Appel, themselves based on a seminar given by Robert van Leeuwen at the University of Würzburg in July 1999.

\subsection{Smooth approximations to the density operator}

It is customary in many-particle theory to derive the one-particle density as the expectation value of a \q{density operator} defined in the framework of second quantisation that is formally treated like an operator on a Hilbert space and can thus be inserted in formulas like Kubo's \eqref{eq-kubo}. Yet such a density is really an operator valued distribution and although the formal treatment as a self-adjoint operator yields useful results, it is not mathematically well-defined and certain manipulations can easily lead to dead ends. We will try to develop a rigorous formulation here by studying the test function $\eta_0 \in \Cont_0^\infty(\Omega)$ symmetrically centred around the origin. It is further assumed non-negative and normalised.
\[
\int_\Omega \eta_0(x) \d x = 1
\]
The limit of monotonously decreasing support is thus an approximation to the delta distribution. We define the shifted version $\eta_x(y) = \eta_0(y-x) = \eta_y(x)$ centred around $x \in \Omega$. The inner product $\langle \eta_x,f \rangle$ is now an arbitrarily good approximation for $f(x)$ if $f:\Omega \rightarrow \R$ is continuous and $\supp \eta_x \subset \Omega$ ($x$ not too close to the border of $\Omega$ if it is assumed bounded or anywhere if $\Omega$ is considered to be periodically tesselated). Further this quantity as a function in $x$ is the convolution with a mollifier $(\eta_0 * f)(x) = \langle \eta_x,f \rangle$ defined for $\Omega = \R^d$ that makes $\eta_0 * f$ a smooth function. The same notation will be used if $\eta_x$ is applied as a multiplication operator to a many-particle state and we identify $(\eta_x \psi)(\ushort{y}) = (\eta_x(y_1)+\ldots+\eta_x(y_N)) \psi(\ushort{y}) = \eta_x(\ushort{y}) \psi(\ushort{y})$. Such an $\eta_x$ is a smooth approximation to the density operator defined by delta distributions.

\subsection{Deriving the density response}

Take now $\eta_x$ as a (bounded) multiplication operator and insert it into the Kubo formula \eqref{eq-kubo}. The resulting quantity gives an approximation to the density response at $x \in \Omega$ and at time $t$ if the potential is changed from $v$ by adding $w$. Fréchet differentiability is assumed here.
\[
\delta \langle \eta_x \rangle_{[v;w]}(t) = \i\int_0^t \langle [\hat{w}(s), \hat{\eta}_x(t)] \rangle_0 \d s
\]
Remember that the involved expectation value is taken with respect to the initial state $\psi_0 \in \H$ because we are in the $H[v]$-interaction picture. In the next step we will also smooth out the variation of the potential $w$, effectively to remove it from the central body of the equation. This is only possible if such a perturbation $w$ is assumed to consist of equal one-particle potentials only, for which we will apply the same notation $w : [0,T] \times \Omega \rightarrow \R$ for simplicity.
\[
w(t,\ushort{x}) = \sum_{i=1}^N w(t,x_i)
\]
Usually this step is done without smoothing for a continuous $w$ with the delta distribution in the form of the \q{density operator}. To avoid such distributions we will resort again to the same $\eta_x$ and define the smooth approximate
\[
w_\eta(t,x) = \langle \eta_x,w(t) \rangle = \int_\Omega \eta_x(y)w(t,y) \d y = \int_\Omega \eta_y(x)w(t,y) \d y.
\]
If put into the Kubo formula where the transformation to the $H[v]$-interaction picture occurs, only the $x$ dependent part is affected and by change of integration order $w$ can be moved outside the central commutator. This happens $N$ times identically because of the symmetry of all quantities involved regarding permutations of the particle positions $x_1,\ldots,x_N$ but this is already captured by $\eta_y(\ushort{z})$ being a sum of $N$ identical one-particle multiplication operators $\eta_y$.
\begin{align}\label{eta-x-lin-resp}
\delta \langle \eta_x \rangle_{[v;w_\eta]}(t) &= \i \int_0^t \d s \int_{\Omega^N} \d \ushort{z} \,\psi_0^*(\ushort{z}) [\hat{w}_\eta(s,\ushort{z}), \hat{\eta}_x(t,\ushort{z})] \psi_0(\ushort{z}) \nonumber\\
&= \i \int_0^t \d s \int_\Omega \d y \int_{\Omega^N} \d \ushort{z} \,\psi_0^*(\ushort{z}) [\hat{\eta}_y(s,\ushort{z}), \hat{\eta}_x(t,\ushort{z})] \psi_0(\ushort{z}) w(s,y) \nonumber\\
&= \int_0^\infty \d s \int_\Omega \d y\, \i \theta(t-s) \langle [\hat{\eta}_y(s), \hat{\eta}_x(t)] \rangle_0\, w(s,y) 
\end{align}
In the last step the expectation value for the initial state is reintroduced and the time integral is taken to infinity and cut off with a Heaviside function in the integrand. At this point a problem arises, if the test functions $\eta_x$ are replaced by delta distributions to give the precise variation of the density instead of its smoothed version. The problem is that the (point-wise) multiplication of distributions is not generally well-defined. There is even a two page no-go result by \citeasnoun{schwartz1954} that shows that a multiplication algebra with differentiation including the distributions is not possible if one wants to keep the point-wise multiplication in the usual algebra of $\Cont^k$ functions. Yet it is still possible if one only keeps the point-wise multiplication of smooth functions, worked out to a full theory by \citeasnoun{colombeau} specifically in the context of quantum field theory. But the link to distributions in this framework is then not unique, there may be different Colombeau functions all representing the delta distribution related to different limiting procedures.

The integration kernel in \eqref{eta-x-lin-resp} now defines the so-called \emph{retarded linear density response function}, sometimes also called \emph{correlation function} or \emph{susceptibility}.
\begin{equation}\label{def-lin-resp-function}
\chi(t,x;s,y) = -\i \theta(t-s) \langle [\hat{\eta}_x(t), \hat{\eta}_y(s)] \rangle_0
\end{equation}
This quantity has already been derived in a different form in \eqref{eq-response-kernel}. In physics literature it is often denoted as the variational derivative of the one-particle density in the form
\[
\chi(t,x;s,y) = \frac{\delta n(t,x)}{\delta w(s,y)},
\]
whereas here we derived it only as an approximation to it, provided it is well-defined. The qualifier \q{retarded} relates to the presence of $\theta(t-s)$ that only allows influences previous to the time $t$ and arises naturally in such initial value problems. In \citeasnoun{fetter-walecka} it is noted that the given equation for $\delta\langle \eta_x \rangle$ \q{typifies a general result that the linear response of an operator to an external perturbation is expressible as the space-time integral of a suitable retarded correlation function.}

\subsection{Density response from the spectral measure of the Hamiltonian}

To get a more explicit result one restricts $v$ to time-independent potentials also including interactions, perturbed by some still time-dependent one-particle potential $w$. This allows the unitary evolution one-parameter group with self-adjoint generator $H=H[v]$ to be expressed as the spectral representation \cite[Cor.~28.1]{blanchard-bruening-2}
\[
U(t) = \int \e^{-\i \varepsilon t} \d E_{H}(\varepsilon).
\]
The $\d E_H(\varepsilon)$ is called the spectral measure of $H$ and is a projection-valued measure projecting onto the spectral subspaces of $H$, more specifically onto an eigenspace if $\varepsilon$ is an eigenvalue. One also gets analogous expressions for the resolution of identity and the Hamiltonian itself.
\begin{align*}
\id &= \int \d E_{H}(\varepsilon)\\
H &= \int \varepsilon \d E_{H}(\varepsilon)
\end{align*}
Additionally one starts in the ground state of the unperturbed system with Hamiltonian $H$ at time $t=0$. That such a ground state exists as an eigenstate associated with the infimum of the spectrum of $H$ is tacitly assumed. The ground state thus evolves according to $U(t)\psi_0 = \exp(-\i \varepsilon_0 t)\psi_0$, $\varepsilon_0 = \inf \sigma(H)$. Now switching back to the Schrödinger picture the response function can be evaluated as
\begin{align*}
\chi(t&,x;s,y) \\
=& -\i \theta(t-s) \langle \psi_0, \left( U(-t) \eta_x U(t-s) \eta_y U(s) - U(-s) \eta_y U(s-t) \eta_x U(t) \right)\psi_0 \rangle \\[0.7em]
=& -\i \theta(t-s) \left( \langle U(t)\psi_0, \eta_x U(t-s) \eta_y U(s) \psi_0 \rangle - \langle U(s) \psi_0, \eta_y U(s-t) \eta_x U(t) \psi_0 \rangle \right) \\[0.5em]
=& -\i \theta(t-s) \left( \e^{\i \varepsilon_0 (t-s)} \langle \psi_0, \eta_x U(t-s) \eta_y \psi_0 \rangle - \e^{-\i \varepsilon_0 (t-s)} \langle \psi_0, \eta_y U(s-t) \eta_x \psi_0 \rangle \right) \\
=& -\i \theta(t-s) \int \left( \e^{-\i (\varepsilon-\varepsilon_0) (t-s)} \langle \eta_x\psi_0, \d E_{H}(\varepsilon) \eta_y\psi_0 \rangle \right. \\
&\left. - \e^{\i (\varepsilon-\varepsilon_0) (t-s)} \langle \eta_y\psi_0, \d E_{H}(\varepsilon) \eta_x \psi_0 \rangle \right).
\end{align*}
In the last step the spectral representation of the unitary evolution $U(t-s)$ and its inverse $U(s-t)$ were inserted. One observes that the result does not depend on $t,s$ independently but only on their difference $t-s$. If spatially homogeneous systems are studied, the same is true for the coordinates $x,y$. We finally write for the response function
\begin{equation}\label{lin-resp-function}
\chi(\tau;x,y) = -\i \theta(\tau) \int \e^{-\i \Omega_\varepsilon \tau} \d\mu_{x,y}(\varepsilon) + c.c.
\end{equation}
where we used the abbreviations
\begin{align*}
\tau &= t-s\\
\Omega_\varepsilon &= \varepsilon-\varepsilon_0 \mtext{(the excitation energies)}\\
\d \mu_{x,y}(\varepsilon) &= \langle \eta_x \psi_0, \d E_{H}(\varepsilon) \eta_y\psi_0 \rangle\\
\d \mu_{x,y}^*(\varepsilon) &= \langle \eta_y \psi_0, \d E_{H}(\varepsilon) \eta_x\psi_0 \rangle.
\end{align*}
The last two terms describe the complex measure that enters the Lebesgue--Stieltjes integration of the spectral decomposition. If the spectrum of the unperturbed Hamiltonian $H$ is fully discrete, the spectral decomposition reduces to a sum over all eigenvalues $\varepsilon_k$ starting with the ground-state energy $\varepsilon_0$. Defining the excitation energies $\Omega_k = \varepsilon_k-\varepsilon_0$ and corresponding eigenstates $\psi_k$ (counting multiples), the response function reduces to
\begin{equation}\label{eq-resp-function-discrete}
\chi(\tau;x,y) = -\i \theta(\tau) \sum_k \e^{-\i \Omega_k \tau} \langle \eta_x \psi_0, \psi_k \rangle \langle \psi_k, \eta_y\psi_0 \rangle + c.c.
\end{equation}

\subsection{Relation to the Lebesgue decomposition theorem}

The general response function \eqref{lin-resp-function} can also be partitioned such that the discrete sum from \eqref{eq-resp-function-discrete} is one of three possible parts arising out of the spectral integration using the Lebesgue decomposition theorem \cite[19.61]{hewitt-stromberg} for the complex regular Borel measure $\mu = \mu_{x,y}$. The unique decomposition into three parts is written as
\[
\mu = \mu_a + \mu_s + \mu_p.
\]
Here $\mu_a$ is \emph{absolutely continuous} with respect to the Lebesgue measure $\lambda$, noted $\mu_a \ll \lambda$, which means that for all $A \subset \R$ with $\lambda(A)=0$ it follows $\mu_a(A)=0$ (one also says $\mu_a$ is dominated by $\lambda$). The term \q{absolutely continuous} already hints that such a measure is linked to the Lebesgue measure by integration (see also \autoref{sect-lip-abs-cont}) and indeed $\mu_a \ll \lambda$ guarantees the existence of a Radon--Nikodym derivative $f=\d \mu_a/\d \lambda$ such that
\[
\mu_a(A) = \int_A f \d\lambda(A).
\]
The remaining parts that are not absolutely continuous $\mu_s$ and $\mu_p$ are called singular towards $\lambda$, noted $\mu_s,\mu_p \perp \lambda$. The part that is still continuous (though not absolutely continuous) is called the \emph{singular continuous} measure $\mu_s$. It rarely occurs in a physical context, one example being the Cantor distribution that takes a non-zero value only on the points of the Cantor set. The associated cumulative function is known as \q{Devil's staircase} and is continuous and monotonously increasing but with a vanishing derivative almost everywhere. The final \emph{pure point} measure $\mu_p$ takes values only at countably many points, the cumulative function thus exhibiting discontinuous jumps at those locations. Its prototype and general building block is of course the Dirac delta measure.

\subsection{The response function in frequency representation}

We will proceed transforming the response function \eqref{lin-resp-function} to get a frequency representation. The following integral representation of the included Heaviside function will be beneficial in transforming all terms including $\tau$ to exponentials.
\[
\theta(\tau) = \lim_{\gamma \searrow 0} \frac{\i}{2\pi} \int_{-\infty}^\infty \frac{\e^{-\i \omega \tau}}{\omega + \i \gamma} \d \omega
\]
If we put this into \eqref{lin-resp-function} we are almost directly led to the Fourier transform of the response function.
\begin{align*}
\chi(\tau;x,y) &= \int_{-\infty}^\infty \d \omega \lim_{\gamma \searrow 0} \frac{1}{2\pi} \int \d\mu_{x,y}(\varepsilon) \frac{\e^{-\i (\omega + \Omega_\varepsilon) \tau}}{\omega + \i \gamma} + c.c. \\
&= \frac{1}{2\pi} \int_{-\infty}^\infty \d \omega\, \e^{-\i \omega \tau} \lim_{\gamma \searrow 0} \int \frac{\d\mu_{x,y}(\varepsilon)}{\omega - \Omega_\varepsilon + \i \gamma} + c.c.
\end{align*}
After the substitution $\omega \rightarrow \omega - \Omega_\varepsilon$ we clearly have a Fourier transform in the first term and an inverse Fourier transform in the complex conjugate term. To straighten that we substitute $\omega \rightarrow -\omega$ only in the complex conjugate term. Then both terms are Fourier transformed from the frequency to the time domain.
\begin{align*}
\chi(\tau;x,y) &= \mathcal{F}_{\omega \rightarrow \tau} \lim_{\gamma \searrow 0} \left(\int \frac{\d\mu_{x,y}(\varepsilon)}{\omega - \Omega_\varepsilon + \i \gamma} - \int \frac{\d\mu_{x,y}^*(\varepsilon)}{\omega + \Omega_\varepsilon + \i \gamma} \right)
\end{align*}
(Note that this is somewhat unusual, as in general the inverse Fourier transform goes from the frequency to the time domain. Also the factor $1/(2\pi)$ is more commonly included in the inverse transformation, not in the forward one. In any case this transposition just changes a sign in the argument.)

It is now just a matter of reading off the Fourier transform of the response function.
\begin{equation}\label{eq-resp-function-omega}
\tilde{\chi}(\omega;x,y) = \lim_{\gamma \searrow 0} \left(\int \frac{\d\mu_{x,y}(\varepsilon)}{\omega - \Omega_\varepsilon + \i \gamma} - \int \frac{\d\mu_{x,y}^*(\varepsilon)}{\omega + \Omega_\varepsilon + \i \gamma} \right)
\end{equation}

Let us recapitulate what the response function describes physically. It has been defined as the term \eqref{def-lin-resp-function} in the density response \eqref{eta-x-lin-resp},
\[
\delta \langle \eta_x \rangle_{[v;w_\eta]}(t) = \int_0^\infty \d s \int_\Omega \d y\, \chi(t-s;x,y) w(s,y).
\]
As such $\chi$ captures the answer of the (approximate) density at position $x$ to a kick by a change of potential $w$ at point $y$ with a time delay of $t-s$. Now its associated Fourier transform $\tilde\chi$ embodies the ability of the system to oscillate with a certain frequency after such a kick. If the system exhibits resonant modes the response function will have a peak at that frequency. The expression for $\tilde\chi$ \eqref{eq-resp-function-omega} shows such peaks at $\omega = \pm \Omega_k$ if $\varepsilon_k$ is an eigenvalue of the unperturbed Hamiltonian. In the case of only discrete spectrum we have the Fourier transform of \eqref{eq-resp-function-discrete}
\[
\tilde{\chi}(\omega;x,y) = \lim_{\gamma \searrow 0} \sum_k \left(\frac{\langle \eta_x \psi_0, \psi_k \rangle \langle \psi_k, \eta_y\psi_0 \rangle}{\omega - \Omega_k + \i \gamma} - \frac{\langle \eta_y \psi_0, \psi_k \rangle \langle \psi_k, \eta_x\psi_0 \rangle}{\omega + \Omega_k + \i \gamma} \right).
\]
This form of the response function is called Lehmann representation \cite{lehmann}. Taken as a function in $\omega \in \C$ the poles approach the respective $\Omega_k$ from the lower complex plane, thus $\tilde\chi$ is a complex analytic (holomorphic) function on the whole upper complex plane. This makes it accessible to powerful tools from complex analysis for further study. At the end of the day one could also try to take the limit to Dirac functions for $\eta_x,\eta_y$ to get the exact density response. The expression above is also the starting point for applications of TDDFT in spectroscopy, to calculate the absorption spectra of molecules up to several thousand atoms \cite{andrade,adamo-jacquemin}. The TDDFT approximation of $\tilde\chi$ uses the Kohn--Sham orbitals (see \autoref{sect-KS}) to determine the response function of the non-interacting system and the (approximated) exchange-correlation kernel $f_\mathrm{xc} = \delta \vxc[n] / \delta n$  evaluated at the ground state density to transform it to the interacting case.

\chapter{Introduction to Density Functional Theory}
\label{ch-dft}

\begin{xquote}{\citeasnoun{dirac1929}}
    The general theory of quantum mechanics is now almost complete, the imperfections that still remain being in connection with the exact fitting in of the theory with relativity ideas. [...] The underlying physical laws necessary for the mathematical theory of a large part of physics and the whole of chemistry are thus completely known, and the difficulty is only that the exact application of these laws leads to equations much too complicated to be soluble. It therefore becomes desirable that approximate practical methods of applying quantum mechanics should be developed, which can lead to an explanation of the main features of complex atomic systems without too much computation.
\end{xquote}

\begin{xquote}{Max Jammer in \emph{The Conceptual Development of Quantum Mechanics} (1966) quoted after \citeasnoun{hughes}}
In spite of its high sounding name [...] quantum theory, and especially the quantum theory of polyelectronic systems, prior to 1925 was, from the methodological point of view, a lamentable hodgepodge of hypotheses, principles, theorems, and computational recipes rather than a logical consistent theory.
\end{xquote}

\section{Precursors of DFT}

\subsection{Motivation}
\label{sect-dft-motivation}

Density functional theory and its time-dependent counterpart is exactly such an approximate practical method Dirac envisioned in the above citation. It is devised to approximately solve Coulombic many-particle---i.e., chemical---systems with far less computational effort than by solving Schrödinger's equation as a whole. A brief history of Dirac's program is put forward by Walter Kohn in his \citeasnoun{kohn-nobel} whose course we want to follow in these introducing sections.

One of the first almost complete studies to calculate binding energy and atomic distance of $H_2$ as the most basic molecule in the ground state was undertaken in \citeasnoun{james-coolidge}. They used the ansatz
\[
\psi = \varphi(x_1, x_2) \otimes (\chi_+ \otimes \chi_- - \chi_- \otimes \chi_+)
\]
with the spinorial part in the antisymmetric spin singlet state and $\varphi$ a general, normalised function depending on the two electron coordinates symmetric under their interchange and respecting the molecule's symmetry. $\varphi$ is then expanded in terms of $M$ parameters $p_1, \ldots, p_M$, in the case of James and Coolidge $M \leq 13$. The energy expectation value $\langle H \rangle_\psi$ is then minimised with respect to this parameter space, with a very convincing outcome.

But such an approach is only feasible with two electrons, for the number of parameters grows exponentially with the number of particles. Six electrons already demand a parameter space of $10^9$ dimensions. A limit in terms of computational power is thus soon reached and Kohn writes in his \citeasnoun{kohn-nobel} of an \q{exponential wall} that we run against.

A further issue lies in the approximation of a multi-particle state $\psi$ by some $\tilde{\psi}$. If the error for \emph{one} particle is $|\langle \tilde{\psi},\psi \rangle| = 1-\varepsilon$, then in the case of $N$ particles we get $|\langle \tilde{\psi}, \psi \rangle| = (1-\varepsilon)^N \approx
\e^{-N \varepsilon} \approx 0.37$ if $\varepsilon N \approx 1$ and an accuracy much, much lower if the particle number $N$ increases.

Kohn thus puts forward the following \q{provocative statement}, referring to an old paper of his former teacher \citeasnoun{van-vleck} in which he addressed possible errors due to large particle numbers in Heisenberg's theory of ferromagnetism:
\begin{quote}
In general the many-electron wavefunction [...] for a system of $N$ electrons is not a legitimate scientific concept, when $N \geq N_0$, where $N_0 \approx 10^3$.
\end{quote}

Kohn justifies this with his prerequisites for a \q{legitimate scientific concept} not fulfilled by a multi-particle wave function:
\begin{itemize}
    \item It can be calculated with sufficient accuracy.
    \item It can be measured and recorded with sufficient accuracy.
\end{itemize}

This puts it close to a phenomenological concept in the sense of \citeasnoun{cartwright} whereas the wave function is part of a simulacrum account, as already discussed in the introductory \autoref{sect-epistem-disclaimer}. Note that one may also read this as a challenge to epistemological treatments of quantum physics, especially when looking for any kind of quantum ontology.

Contrary to the wave function other physically or chemically interesting quantities, such as total energy or one-particle density, do not involve such limitations. If one would therefore be able to predict physical properties of a system just by analysing the one-particle density
\begin{equation}\label{def-n}
n(x) = N \int |\psi(x, \bar x)|^2 \d \bar x
\end{equation}
the \q{exponential wall} melts away. If and how this is possible for ground-state systems is the essence of DFT. But prior to that a direct predecessor of DFT shall be presented.

\subsection{Thomas--Fermi theory}
\label{sect-thomas-fermi-theory}

This first semi-classical theory that tries to describe the electronic structure of matter just in terms of its electronic density was developed 1927 just shortly after the formulation of Schrödinger's equation. Herein, the approximated total energy $E$, its kinetic part $T$, as well as the electron-electron interaction approximated by the Hartree term $\VH$ and external potential $\Vext$, are functionals of the density $n$ alone.
\[
E[n] = T[n] + \VH[n] + \Vext[n]
\]
The kinetic part $T[n]$ is approximated by considering a non-inter\-acting homogeneous electron gas with uniform density $n=N/V$. Assuming a three-dimensional box of side length $L$, volume $V=L^3$, and periodic boundary conditions the possible wave lengths in any direction are $\lambda_i = L/i$ with $i \in \N_0$. Applying the de Broglie relation every two states (because of the spin degree of freedom) therefore occupy a small cube of side length $2\pi\hbar/L$ in momentum-space. Filling up the whole Fermi sphere with $N$ electrons this gives us the Fermi momentum $p_f$ as its radius.
\[
\frac{4}{3} \pi p_f^3 = \frac{N}{2} \left(\frac{2\pi\hbar}{L}\right)^3
\]
Solving for $p_f$ gives
\begin{equation}\label{eq-fermi-radius}
p_f = \hbar \left( \frac{3\pi^2 N}{V} \right)^{\frac{1}{3}}
\end{equation}
and if we insert this into the formula for the classical kinetic energy this corresponds to a Fermi energy of
\[
T_f(N) = \frac{p_f^2}{2m} = \frac{\hbar^2}{2m} \left(\frac{3\pi^2 N}{V}\right)^{\frac{2}{3}}.
\]
The average kinetic energy of an electron is now given by the mean value
\[
\bar{T} = \frac{1}{N} \int_0^N T_f(N') \d N' = \frac{3}{5} T_f(N).
\]
Now $N/V$ in the formula of the Fermi energy clearly equals the mean electronic density. This will be substituted by the local density $n(x)$ which gives us a total kinetic energy of
\[
T[n] = \frac{3}{5} \frac{\hbar^2}{2m} \left(3\pi^2\right)^{\frac{2}{3}} \cdot \! \int n(x)^{\frac{5}{3}} \d x
\]
and consequently a kinetic energy density of
\[
\tau([n],x) = \frac{3}{5} \frac{\hbar^2}{2m} \left(3\pi^2n(x)\right)^{\frac{2}{3}}.
\]
This is the first occurrence of the principle idea of DFT, the approximation of electronic properties (here the density of the kinetic energy) as functionals of the one-particle density alone. This happens only \emph{locally} and under the assumption of a homogeneous electron gas. Such daring approximations fall under the name of \emph{local-density approximations} (LDA) and they are again highlighted in \autoref{sect-tfdt-lda}.

Assuming an external scalar potential $v(x)$ the remaining parts are given by
\begin{align}\label{hartree-term}
\VH[n] &= \frac{1}{2} \int \frac{n(x) n(y)}{|x-y|} \d x \d y \quad\mbox{and} \\
\Vext[n] &= \int v(x) n(x) \d x. \nonumber
\end{align}
Subsequently $E[n]$ is minimised under the constraint $N = \int n(x) \d x$. The Lagrange multiplier occuring in this procedure will then be the so-called chemical potential.


This simple approximation is not able to accommodate for atomic shells or molecular bonds and is thus not applicable in many cases. Nevertheless the calculated electronic density $n(x)$ from Thomas--Fermi theory can be useful as a starting point for an iterative procedure in DFT. Using the Thomas--Fermi energy as a lower bound for the complete quantum mechanical Hamiltonian interesting results regarding the stability of matter have been derived by \citeasnoun{lieb-thirring}.

The restriction to only scalar external potentials thus ignoring all magnetic effects will be employed throughout this work and is actually justified in many cases as for example \citeasnoun{bohm-pines-1} put it: \q{These magnetic interactions are weaker than the corresponding coulomb interactions by a factor of approximately $v^2/c^2$ and, consequently, are not usually of great physical interest.}

\subsection{Reduced density matrices}
\label{sect-rdm}

We already saw that to get a computable system it is all about reduction. But before jumping right to it as with Thomas--Fermi theory and throwing away almost everything that seems to be typical for quantum mechanics, one might try and see how far one is allowed to reduce complexity and still retain the full and exact information about at least some properties of the system. The benchmark will be of course the expectation value of the electronic Hamiltonian with external scalar potential $v$ and Coulombic interaction.
\begin{equation}\label{eq-mol-hamiltonian}
H = T + \Vee + \Vext = -\frac{1}{2}\sum_{k=1}^N\Delta_k  + \sum_{j<k} \frac{1}{|x_j-x_k|} + \sum_{k=1}^N v(x_k)
\end{equation}
In this setting only one- ($T$ and $\Vext$) and two-particle ($\Vee$) operators are present, so a reduced form of the wave function that still gives exact results for their expectation values would have to include at least two different particle coordinates.

Let the $N$-particle wave function, this time including spin, be
\[
\psi(\ushort{x}\ushort{s})=\psi(x_1 s_1,\ldots,x_N s_N),
\]
with $s_i \in \{0,1\}$ for spin-$\onehalf$ fermions. Then the associated general spinless \emph{$p$\textsuperscript{th} order reduced density matrix} ($p$-RDM), $p \in \{1,\ldots,N\}$, is given by
\begin{align*}
\rho_{(p)}(x_1,\ldots,x_p,x_1',&\ldots,x_p') \\
= {N \choose p} \sum_{\ushort{s}} \int &\psi(x_1 s_1,\ldots,x_p s_p,x_{p+1} s_{p+1},\ldots,x_N s_N) \\
&\psi^*(x_1' s_1,\ldots,x_p' s_p,x_{p+1} s_{p+1},\ldots,x_N s_N) \d x_{p+1}\ldots \d x_{N}
\end{align*}
and especially for $p=1$
\[
\rho_{(1)}(x,x') = N \sum_{\ushort{s}} \int \psi(x s_1,x_2 s_2,\ldots,x_N s_N) \psi^*(x' s_1,x_2 s_2,\ldots,x_N s_N) \d\bar{x}.
\]

In the case of a local one-particle operator $A$ that always acts as a symmetrised multiplication operator $\sum_{k=1}^N a(x_k)$ in spatial representation we can compute the expectation value as
\begin{equation}\label{one-point-exp-value}
\begin{aligned}
\langle A \rangle_\psi &= \sum_{k=1}^N \sum_{\ushort{s}} \int \psi^*(\ushort{x} \ushort{s}) a(x_k) \psi(\ushort{x} \ushort{s}) \d \ushort{x} \\
&= N \sum_{\ushort{s}} \int \psi(x s_1,\ldots,x_N s_N) \psi^*(x s_1,\ldots,x_N s_N) a(x) \d x \d \bar{x} \\
&= \int \rho_{(1)}(x,x) a(x) \d x = \int n(x) a(x) \d x.
\qquad \left( \rho_{(1)}(x,x) \equiv n(x) \right)
\end{aligned}
\end{equation}
Of course the wave function is assumed to be fully (anti-)symmetric to make the exchange of coordinates feasible and so it is possible to apply $a(x)$ depending only on the first particle coordinates. Note that for such local operations only the diagonal part of the 1-RDM matters. A non-local effect could be given by an $A$ acting as an integral operator with kernel $a(x,x')$ instead of an multiplication operator.
\[
A \psi(\ushort{x} \ushort{s}) = \sum_{k=1}^N \int a(x_k,x_k') \psi(\ldots x_k's_k \ldots) \d x_k'
\]
The expectation value then becomes
\begin{align*}
\langle A \rangle_\psi &= \sum_{k=1}^N \sum_{\ushort{s}} \int \psi^*(\ushort{x} \ushort{s}) \int a(x_k,x_k') \psi(\ldots x_k' \ldots) \d x_k' \d \ushort{x} \\
&= N \sum_{\ushort{s}} \int \psi^*(x s_1,\ldots,x_N s_N) \int \psi(x' s_1,\ldots,x_N s_N) a(x,x') \d x' \d x \d \bar{x} \\
&= \int \rho_{(1)}(x',x) a(x,x') \d x' \d x.
\end{align*}
A more relevant non-local (though acting only in a neighbourhood) one-particle operator is the Laplace operator $\Delta = \sum_{k=1}^N \Delta_k$ with the obvious expectation value
\begin{equation}\label{laplace-exp-value}
\langle \Delta \rangle_\psi = N \sum_{\ushort{s}} \int \psi^*(\ushort{x} \ushort{s}) \Delta_1 \psi(\ushort{x} \ushort{s}) \d \ushort{x} = \int \Delta_x \rho_{(1)}(x,x') \Big|_{x'=x} \d x.
\end{equation}

For similar two-particle operators $B$ acting as $\sum_{j<k} b(x_j,x_k)$ we apply the 2-RDM respectively. As we can of course go downwards in the hierarchy from 2-RDM to 1-RDM
\begin{equation}\label{rdm-hierarchy}
\rho_{(1)}(x,x') = \frac{2}{N-1} \int \rho_{(2)}(x,x_2,x',x_2) \d x_2
\end{equation}
and to the one-particle density with just the diagonal of the 1-RDM $n(x) = \rho_{(1)}(x,x)$, all we need for an exact evaluation of the potential terms in the usual electronic Hamiltonians including full electron-electron interaction is the 2-RDM. The possibility to vary over the (convex) set of all such 2-RDMs to get a minimiser for the total energy and thus the ground state is intriguing, but the set of all possible expressions that look like a 2-RDM but are not associated to a real $N$-particle wave function is by far too large and would lead to much to low ground-state energies. The problem of finding necessary and sufficient conditions for a 2-RDM to be \q{$N$-representable} was set forward by \citeasnoun{coleman1963} but until today it is only solved partially (see \citeasnoun{mazziotti2012} for a recent result on $N$-representability with mixed states). This formidable problem is also known as \emph{Coulson's challenge} \cite{coleman-yukalov} and following Dudley Herschbach in a \citeasnoun{uoc2006} \q{that quest has been a `holy grail' of theoretical chemistry for more than 50 years} which gets compared to a football game involving the Chicago Bears in the paragraph following the citation.

It is fully resolved only for ensemble $N$-representability of fermionic 1-RDMs where the only necessary condition is indeed for the so-called natural spin orbitals \cite{loewdin-1955}, the eigenfunctions of the 1-RDM, to be at most singly occupied, i.e., having eigenvalues $0 \leq n_i \leq 1$, which is just the basic Pauli exclusion principle. For pure states more restrictions in the form of linear inequalities involving the $n_i$ follow from the antisymmetry of the wave function, thus generalising the Pauli exclusion principle. \cite{altunbulak-klyachko}

We come back to the energy expectation value using \eqref{one-point-exp-value} with substitution of $n(x) = \rho_{(1)}(x,x)$, \eqref{laplace-exp-value} for the Laplacian, and the expectation value of a two-particle operator evaluated with $\rho_{(2)}$ along the diagonal.
\begin{equation}\label{rdm-energy-exp}
\begin{aligned}
\langle H \rangle_\psi = &-\int \Delta_x \rho_{(1)}(x,x') \Big|_{x'=x} \d x \\
&+ \int v(x)n(x) \d x \\
&+ \int \frac{\rho_{(2)}(x_1,x_2,x_1,x_2)}{|x_1-x_2|} \d x_1 \d x_2
\end{aligned}
\end{equation}
We want to consider again the classical interaction energy in the form of the Hartree term \eqref{hartree-term} as part of the total electron-electron interaction.
\[
\rho_{(2)}(x_1,x_2,x_1,x_2) = \onehalf n(x_1) n(x_2) (1+h(x_1,x_2))
\]
This must be seen as a definition of the unknown \emph{pair correlation function} $h(x_1,x_2)$ that captures non-classical effects. Inserting the equation above into our hierarchy rule \eqref{rdm-hierarchy} yields
\[
\rho_{(1)}(x_1,x_1) = n(x_1) = \frac{1}{N-1} \int n(x_1) n(x_2) (1+h(x_1,x_2)) \d x_2
\]
and by using the normalization of $n(x)$ to $N$ the condition
\[
\int n(x_2) h(x_1,x_2) \d x_2 = -1
\]
for all $x_1$ with $n(x_1) > 0$. For fixed $x_1$ the term $n_\mathrm{xc}(x_1,x_2) = n(x_2) h(x_1,x_2)$ is like a charged particle with sign opposite to that of the electron and thus also called the \emph{exchange-correlation hole}. Using this notation the expectation value of $\Vee$ becomes
\[
\langle \Vee \rangle_\psi = \VH[n] + \frac{1}{2} \int \frac{n(x_1) n_\mathrm{xc}(x_1,x_2)}{|x_1-x_2|} \d x_1 \d x_2.
\]
All this and more on this topic can be found for example in \citeasnoun[2.3ff]{parr-yang}.

\subsection{Thomas--Fermi--Dirac theory and the local-density approximation}
\label{sect-tfdt-lda}

\begin{xquote}{Walter Kohn, \citeasnoun{kohn-nobel}}
The LDA, obviously exact for a uniform electron gas, was a priori expected to be useful only for densities varying slowly on the scales of the local Fermi wavelength $\lambda_F$ and TF wavelength, $\lambda_{TF}$. In atomic systems these conditions are rarely well satisfied and very often seriously violated. Nevertheless the LDA has been found to give extremely useful results for most applications. [...] Experience has shown that the LDA gives ionization energies of atoms, dissociation energies of molecules and cohesive energies with a fair accuracy of typically 10--20\%. However the LDA gives bond-lengths and thus the geometries of molecules and solids typically with an astonishing accuracy of $\sim 1\%$.
\end{xquote}

\begin{xquote}{Angel Rubio (24\textsuperscript{th} February 2015)}
When LDA was proposed, people said this is the most crazy, irrelevant, and useless theory; we are where we are now because it worked.
\end{xquote}

This refinement of Thomas--Fermi theory was proposed by \citeasnoun{dirac1930} and included a better approximation of the electron-electron interaction going beyond the Hartree term by using the full 1-RDM instead. We follow again \citeasnoun[6.1]{parr-yang} in its presentation. The aim is to include the additional so-called \emph{exchange} effects, originating from the Pauli exclusion principle for fermions that demands antisymmetric wave functions. To account for that, the setting is restricted to single Slater determinants involving the orthonormal spin orbitals $\phi_i(x_j s_j)$, each a tensor product of a spatial orbital $\varphi_{\lceil i/2 \rceil}(x_j) \in L^2(\R^3)$ and a spinor $\chi_\pm \in \C^2$ that fully capture this specific antisymmetry. The wave function of the ground state is thus approximated by
\[
\psi(\ushort{x} \ushort{s}) = \frac{1}{\sqrt{N!}} \sum_{\sigma \in S_N} \sign(\sigma) \phi_{\sigma(1)}(x_1 s_1) \ldots \phi_{\sigma(N)}(x_N s_N)
\]
with the permutation group $S_N$ of $N$ elements. For $N$ even, which shall be implied from now on, this form already assumes equal distribution of spins up and down.

The 1-RDM for such a system, all spin degrees of freedom summed up and all orbitals with particle positions $x_2,\ldots,x_N$ integrated out, gives the \emph{Fock--Dirac density matrix} in spinless form.
\[
\rho_{(1)}(x,x') = \sum_{i=1}^N \sum_s \phi_i(x s) \phi_i^*(x' s)
\]
Now the 2-RDM can of course be cast in an analogous form and comes out as a determinant of 1-RDMs. This is also true for all higher $p$-RDMs in this special case of single Slater determinants as wave functions. The extra $\onehalf$ in the second term below comes from the mixing of spin coordinates between the particles with coordinates $x_1s_1$ and $x_2's_2$ (or $x_1's_1$ and $x_2s_2$ respectively) where only half of them are not perpendicular due to their spin coordinates.
\begin{equation}\label{eq-2rdm-slater}
\rho_{(2)}(x_1,x_2,x_1',x_2') = \onehalf \left(\rho_{(1)}(x_1,x_1')\rho_{(1)}(x_2,x_2') - \onehalf\rho_{(1)}(x_1,x_2')\rho_{(1)}(x_2,x_1')\right)
\end{equation}
If we put this expression into the interaction term of the expectation value for the Hamiltonian \eqref{rdm-energy-exp} we are only left with the diagonal term $\rho_{(2)}(x_1,x_2,x_1,x_2) = \onehalf (n(x_1)n(x_2) - \onehalf|\rho_{(1)}(x_1,x_2)|^2)$. The first part just yields the Hartree term while the second is the starting point for another application of the homogeneous electron gas approximation to get an expression for the exchange energy.
\begin{align*}
\langle\Vee\rangle_\psi = \int \frac{\rho_{(2)}(x_1,x_2,x_1,x_2)}{|x_1-x_2|} \d x_1 \d x_2 &= \overbrace{\frac{1}{2}\int \frac{n(x_1)n(x_2)}{|x_1-x_2|} \d x_1 \d x_2}^{\VH[n]} \\[0.2cm]
&- \underbrace{\frac{1}{4}\int \frac{|\rho_{(1)}(x_1,x_2)|^2}{|x_1-x_2|} \d x_1 \d x_2}_{\Vx[n]}
\end{align*}
Again we take the uniform density $n=N/V$ in a box of volume $V=L^3$ with periodic boundary conditions. The spin orbitals are assumed to be plane waves $\varphi_p(x) = \exp(\i p \cdot x/\hbar)/\sqrt{V}$ with momenta $p = (n_x,n_y,n_z) 2\pi\hbar / L$ each one doubly filled starting with the lowest energy $n_x,n_y,n_z=0,\pm 1, \pm 2,\ldots$ up to the Fermi surface $|p| = p_f$ (this is called the \q{closed shell} assumptions and here again an even number of electrons is necessary). We get for the 1-RDM in this special setting
\[
\rho_{(1)}(x_1,x_2) = 2 \sum_p \varphi_p(x_1) \varphi_p^*(x_2) = \frac{2}{V} \sum_p \e^{\i p \cdot (x_1-x_2)/\hbar}
\]
where the factor 2 originates from the double occupation of every orbital and the sum over $p$ is considered over all occupied orbitals. We want to pass over to an integral over the wavenumber vectors $k= p/\hbar =(n_x,n_y,n_z)2\pi/L \in \R^3$ with step size $2\pi/L \rightarrow 0$.
\begin{align*}
\rho_{(1)}(x_1,x_2) &= \frac{1}{4\pi^3} \int_{|k|<k_f} \e^{\i k \cdot (x_1-x_2)} \d k \\
&= \frac{1}{2\pi^2} \int_0^{k_f} r^2 \d r \int_0^\pi \sin \theta  \,\e^{\i r |x_1-x_2| \cos \theta } \d \theta
\end{align*}
The integral over the sphere with radius $k_f = p_f/\hbar$ has been transformed to spherical coordinates $(r,\theta,\varphi)$, the $z$-axis pointing into direction $x_1-x_2$, with the azimuthal angle already integrated out to $2\pi$. The integral gets evaluated as
\[
\rho_{(1)}(x_1,x_2) = \frac{\sin(k_f |x_1-x_2|) - k_f |x_1-x_2| \cos(k_f |x_1-x_2|)}{\pi^2 |x_1-x_2|^3}.
\]
The Fermi radius can again be determined as in \eqref{eq-fermi-radius} and by substituting an inhomogeneous density $n(x)=\rho_{(1)}(x,x)$ for $N/V$. It further seems natural to choose $x=(x_1+x_2)/2$.
\[
k_f(x) = \left(3\pi^2n(x)\right)^\frac{1}{3}
\]
Then with the substitution $t = k_f(x) |x_1-x_2|$
\[
\rho_{(1)}(x_1,x_2) = 3n(x) \frac{\sin t - t \cos t}{t^3} 
\]
we can use this expression to finally evaluate the exchange energy
\begin{align*}
\Vx[n] &= \frac{1}{4} \int \frac{|\rho_{(1)}(x_1,x_2)|^2}{|x_1-x_2|} \d x_1 \d x_2 \\
&= \frac{9}{2} \int k_f(x) n(x)^2 \frac{(\sin t - t \cos t)^2}{t^7} \d x_1 \d x_2.
\end{align*}
By substituting coordinates to $x = (x_1 + x_2)/2, x_- = x_1 - x_2$ and introducing polar coordinates for $x_-$ with radius $|x_-| = t/k_f(x)$ we get
\[
\Vx[n] = 9\pi \int k_f(x)^{-2} n(x)^2 \d x \int_0^\infty \frac{(\sin t - t \cos t)^2}{t^5} \d t.
\]
The $t$-integral is evaluated as $\frac{1}{4}$ and re-inserting the expression for $k_f(x)$ we have
\[
\Vx[n] = \frac{3}{4} \left(\frac{3}{\pi}\right)^\frac{1}{3} \int n(x)^\frac{4}{3} \d x.
\]
Putting together all terms we have the following energy functional where the exchange effects add a negative correction to the Hartree term.
\begin{equation}\label{eq-energy-lda}
E[n] = T[n] + \VH[n] - \Vx[n] + \Vext[n]
\end{equation}
This is an approximation using a single Slater determinant for the state and \emph{locally} the spin orbitals of a homogeneous electron gas like in Thomas--Fermi theory. Such a scheme is thus known as a \emph{local-density approximation} (LDA), approximations that are central to DFT. The full energy functional \eqref{eq-energy-lda} considers the kinetic energy in the LDA approximation, the classical Coulomb interaction in form of the Hartree term, the exchange correction, and the external potential. All effects that are still missing are collected under the the label of additional \q{correlation} effects. LDA-type approximations to those have an even more heuristic flavour and are derived in an intention to capture the statistical correlation not only for the same spin component as in \eqref{eq-2rdm-slater} but also for opposite spin components. To achieve this an even more involved form of the Slater determinant has to be assumed, where the orbitals for one spin component themselves perturbatively depend on the coordinates of electrons with opposite spin component. \cite{wigner1934}

\section{Foundations of DFT}

For a general introduction to DFT see for example \citeasnoun{dreizler-gross}, \citeasnoun{eschrig}, and \citeasnoun{capelle}, where the last two are available freely online.

\subsection{The Hohenberg--Kohn theorem}
\label{sect-hk-theorem}

Thomas--Fermi theory and its variants still employ a very rough approximation of the kinetic energy term as a density functional and thus we expect peculiar features of quantum mechanics to be lost completely. The DFT approach tries to circumvent this problem.

In the following let the Hamiltonian $H = T + \Vee + \Vext$ of a quantum-mechanical many-electron system be given by three parts as in \eqref{eq-mol-hamiltonian}: the kinetic part $T$, the electron-electron interaction $\Vee$ and the external potential $\Vext$ defined by the one-particle potential $v$ as a multiplication operator.
\[
\Vext = \sum_{k=1}^{N} v(x_{k})
\]
The expectation value of the external potential thus is
\[
\langle \psi, \Vext \psi \rangle = \int v(x)  n(x)  \d x = \langle v,n \rangle.
\]
Note that this quantity amounts to a dual pairing of $v$ and $n$ which is important as soon as one considers appropriate function spaces. The basic principle of density functional theory is the following celebrated theorem due to \citeasnoun{hohenberg-kohn} that shows existence of a density-potential mapping for ground states. We do not give precise premisses or a mathematically rigorous proof here but rather present it in the form usually found in literature, thus the more modest label \q{conjecture}. Hints towards the lack of rigour can be found in the provided formal proof.

\begin{conjecture}[Hohenberg--Kohn theorem]\label{hk-th}
The one-particle density $n$ of a bounded many-electron system in its non-degenerate ground state uniquely determines the external potential $v$ modulo an additive constant.
\end{conjecture}

\begin{proof}[Formal proof]
Let $\psi_1$ be the non-degenerate ground-state solution to the time-inde\-pendent Schrödinger equation $H_1 \psi_1 = E_1 \psi_1$ with $H_1 = T + \Vee + V_1$. It holds
\[
E_1 = \langle \psi_1, H_1 \psi_1 \rangle = \langle \psi_1, (T+\Vee)
\psi_1 \rangle + \int v_1(x)  n(x)  \d x.
\]
We proceed by reductio ad absurdum. Assume there is a second potential $v_2 \neq v_1 + const$, $H_2=T+\Vee+V_2$ which leads to a (possibly degenerate) ground-state solution to $H_2 \psi_2 = E_2 \psi_2$ having the same density $n(x)$. It follows that $\psi_2 \neq \e^{\i \alpha}\psi_1$ because else $\psi_1$ would solve $H_2 \psi_1 = E_2 \psi_1$ too and subtraction of the two Schrödinger equations would give us $v_1(x) - v_2(x) = const$, already contradicting the assumptions. The necessary condition $\psi_1 \neq 0$ for almost all $(x_1, \ldots, x_N)$ for this argument is assumed to hold for a class of \q{resonable} potentials, where among other things no regions are surrounded by infinite potential barriers.\footnote{A precise treatment of such a \q{unique continuation property} was given in \autoref{sect-ucp}. A recent approach that puts the condition on the one-particle density $n$ instead is given by \citeasnoun{lammert}.} The next step lies in analysing the expectation value of $H_1$ under $\psi_2$ and vice versa. $\psi_2$ could never minimise the energy because it is different from any complex-multiple of the non-degenerate ground state as argued above. Therefore we have
\[
E_1 < \langle \psi_2, H_1 \psi_2 \rangle = E_2 + \int (v_1(x) -
v_2(x)) n(x) \d x
\]
and analogously (but with a \q{$\leq$} sign because of the possible degeneracy)
\[
E_2 \leq \langle \psi_1, H_2 \psi_1 \rangle = E_1 + \int (v_2(x) -
v_1(x)) n(x) \d x.
\]
Addition of both inequalities leads to $E_1+E_2 < E_2+E_1$ and thus a contradiction to our original assumption of different potentials resulting in the same ground state is found.
\end{proof}

If one takes for granted that the kinetic term and the electron-electron interaction is identical in all systems under consideration, the following corollary shows that it is possible to define an injective mapping $n \mapsto \psi$ modulo a global phase factor.

\begin{corollary}\label{cor-HK}
By the Hohenberg--Kohn theorem (\autoref{hk-th}) the density $n$ of a non-degen\-erate ground state uniquely determines the external potential $v$ (modulo an additive constant) and of course the particle number $N$. Thus it defines the Hamiltonian $H$ of the system and by solving the associated Schrödinger equation the respective ground state $\psi$ (modulo a global phase factor) and from that all other properties of the quantum-mechanical system.
\end{corollary}

This means that in principle it suffices to know the electronic density of the ground state to calculate all other quantities of interest. Another important consequence leads into the converse direction: Because $n$ particularly fixes the total energy of the system, the density of the ground state itself can be determined by a variation principle. This conclusion is exemplified in the following section.

\subsection{The Hohenberg--Kohn variation principle}

\begin{definition}
A density $n \geq 0$ normalised to $N = \int n(x) \d x$ is called \textbf{$v$-representable} if and only if there exists a scalar potential $v$ such that $n$ is the corresponding ground-state density. One distinguishes \textbf{interacting} and \textbf{non-interacting $v$-representable}, depending on whether one adds a Coulomb interaction term in the Hamiltonian or not. The set of interacting $v$-representable densities is denoted as $\mathcal{N}$. 
\end{definition}

\begin{definition}
Let $\psi[n]$ be defined by the mapping $n \mapsto \psi$ from \autoref{cor-HK} for non-degenerate ground-state densities then the \textbf{HK functional} is defined as
\[
F_\mathrm{HK}[n] = \langle \psi[n], (T+\Vee) \psi[n] \rangle
\]
and consequently gives the kinetic and interaction part of the total ground-state energy.
\end{definition}

It should be noted that $F_\mathrm{HK}$ is identical in all physical systems under consideration, like atoms, molecules and solid state bodies. A hypothetical exact functional $F_\mathrm{HK}$ that yields this energy directly and only from the density would thus render all efforts to solve Schrödinger's equation for ground state energies superfluous and simplify quantum mechanics radically. In practice approximations are used because clearly the mapping $n \mapsto \psi$ is not analytically accessible except in very simple cases like that of a one-dimensional quantum ring, see \citeasnoun{ruggenthaler-2013}. One popular approximation has already been derived in \eqref{eq-energy-lda} (just subtract the $\Vext$) with LDA.

Because for Hamiltonians as considered here the smallest expectation value for the total energy is always given in an eigenstate---the ground state---(this is known as the Rayleigh--Ritz principle) one can formulate the following variation principle for determining the ground-state density.

\begin{corollary}[HK variation principle]
The ground-state density $n$ for a given external potential $v$ minimises the functional $E[n] = F_\mathrm{HK}[n] +
\int\! v(x) n(x) \!\d x$ varied over all $v$-representable densities.
\[
E_0 = E[n] = \min_{n' \in \mathcal{N}} E[n']
\]
\end{corollary}

This has obvious benefits: Instead of varying over a huge parameter space covering as many $\psi$ as possible as in the method of James and Coolidge mentioned in \autoref{sect-dft-motivation} `only' all possible densities have to be considered. This calls up the issue of $N$-representability already mentioned regarding 2-RDMs in \autoref{sect-rdm}. A density that was found as the minimiser of the variation principle should be related to an antisymmetric $N$-particle wave function via \eqref{def-n}, else it is worthless in the context of quantum theory. Luckily in this case the problem is fully resolved with the so-called \citeasnoun{harriman} construction that gives a corresponding Slater determinant for any conceivable density.

But it is also a highly non-trivial task to give the set of all $v$-representable densities $\mathcal{N}$ which is needed for the variation principle given above. This problem is only fully resolved in the case of \emph{ensemble} $v$-representability, meaning that the density can be written as a linear combination of the densities of possibly degenerate ground states, and a system additionally either formulated on a lattice \cite{chayes-1985} or in a coarse-grained version by averaging over regular cells \cite{lammert-2006}. Yet by resorting to the Levy--Lieb constrained-search functional the problem of $v$-representability can be avoided, the functional is defined on a clearly defined domain of possible $N$-representable densities but not convex. \cite{levy-1979} A convex functional is beneficial because it has a unique minimiser, something achieved by going further to Lieb's convex-conjugate functional. \cite{lieb-1983} Finally this is made functionally differentiable by employing Moreau--Yosida regularisation as given by \citeasnoun{kvaal-2014}.

How this variation principle formulated with one of the mentioned functionals can be applied in practice will be explained by means of the self-consistent Kohn--Sham equations. Yet prior to that its precursor as a self-consistent method shall be briefly discussed.

\subsection{The self-consistent Hartree--Fock method}

\begin{xquote}{\citeasnoun{wigner1938}}
The first question which immediately arises is: Why is it that the very simplest form of the theory gives such satisfactory results for many---indeed most---questions?
\end{xquote}

In the same year as Thomas--Fermi theory was published, just one year after the formulation of the quantum mechanics of Schrödinger, Hartree proposed a method to self-consistently solve many-particle quantum systems using the same simplified interaction term $\VH$ (already indexed with \q{H} for \q{Hartree} before) as in Thomas--Fermi theory but now as an operator on wave functions in the form of a mean field. As a major difference the kinetic term gets not simplified to a density functional but keeps its original form. The resulting Hamiltonian $H=T+\VH+\Vext$ thus consists of the three parts
\begin{equation}\label{hartree-hamiltonian}
T = -\frac{1}{2}\sum_{k=1}^N\Delta_k, \quad \VH = \frac{1}{2} \sum_{k=1}^N \int \frac{n(x)}{|x_k -
x|} \d x, \quad \Vext = \sum_{k=1}^N v(x_k).
\end{equation}
All three parts are sums of terms only depending on the individual particle positions, thus the whole Schrödinger equation $H \psi = E \psi$ is diagonalised and decouples into $N$ separate and identical one-particle Schrödinger equations on three-dimensional configuration space.
\[
\left(-\frac{1}{2} \Delta_k + \frac{1}{2} \int \frac{n(x)}{|x_k - x|} \d x + v(x_k)\right)
\varphi_k(x_k) = \varepsilon_k \varphi_k(x_k).
\]
The ground-state solution is then built together from the $N$ lowest orbitals $\varphi_k$ as a simple tensor product.
\[
\psi(\ushort{x}) = \varphi_1(x_1) \ldots \varphi_N(x_N)
\]
This \q{Hartree method} respects the early form of the Pauli exclusion principle in that every orbital contains one electron only, but not the later formulation demanding an antisymmetric wave function. The density and energy of the whole system can then be calculated as
\[
n(x) = \sum_{k=1}^N |\varphi_k(x)|^2, \quad E =
\sum_{k=1}^N \varepsilon_k.
\]
Now the desired consistency kicks in because of course one expects this density $n$ to be the same that entered the Hartree term above. The following iterative scheme is thus provided:

\begin{enumerate}\itemsep0em
    \item Start with a guessed $n$, e.g.~from Thomas--Fermi theory.
    \item Solve the one-particle Schrödinger equation to get $\{\varphi_k\}_{k = 1, \ldots, N}$.
    \item Calculate $n$ from that and start anew.
\end{enumerate}

To account for the antisymmetry of the wave function a Slater-determinant ansatz is used that leads to an additional exchange term like in Thomas--Fermi--Dirac theory and the method is then known under the names of Hartree and Fock. A major advantage over Thomas--Fermi like theories is the more accurate consideration of the kinetic energy, a feature also shared by the Kohn--Sham equations discussed next. Another shared feature is the inclusion of the interaction (and exchange) effects in one effective potential embedded in a self-consistent scheme. These effects only get approximated in Hartree--Fock theory, whereas the path to the Kohn--Sham equations lies in questioning whether all such interactions can be represented \emph{exactly} by an effective external potential.

\subsection{The self-consistent Kohn--Sham equations}
\label{sect-KS}

The Hamiltonian of the Hartree method described above \eqref{hartree-hamiltonian} had the form of a Schrödinger Hamiltonian without explicit interaction between the electrons, but approximated those effects by a mean field averaged over all electrons. The DFT description of such a non-interacting system (traditionally labelled by the index $s$) is given by minimisation of the HK functional
\[
E_s[n] = T_s[n] + \int v(x) n(x) \d x
\]
with
\[
T_s[n] = \langle \psi_s[n], T \psi_s[n] \rangle
\]
and $\psi_s$ the mapping from a given density to wave functions by solving a non-interacting Schrödinger equation. Variation goes over all non-interacting $v$-representable densities and one ends up with the exact ground-state density $n_s(x)$.

A critical question is now if for each interacting $v$-representable density $n$ an auxiliary potential $\vKS[n]$ can be found such that exactly the same density is reproduced by a non-interacting system. This hypothetical setting is called the \emph{Kohn--Sham (KS) system}. \cite{kohn-sham} More concisely one could ask if the sets of all interacting and non-interacting $v$-representable densities are equal. But even if one allows ensembles of states represented by density matrices (\emph{ensemble representability}) instead of only pure states this question remains open. Yet it can be shown that the sets of interacting and non-interacting ensemble $v$-representable densities are sufficiently similar, i.e., they are dense in each other. \cite[Th.~16]{van-leeuwen3} In the time-dependent case the extended Runge--Gross theorem (\autoref{ext-runge-gross-th}) yields a positive answer for time-analytic potentials.

If one extends the domain of $T_s$ to interacting $v$-representable densities then the HK variation principle from above can be formulated with $\vKS[n]$. This assumes functional differentiability of $E[n]$ with respect to the density in order to define this effective potential yielding a density from an interacting system.
\[
E[n] = T_s[n] + \int \vKS([n],x) n(x) \d x
\]
As in the Hartree--Fock method the effective potential consists of an interaction part, the external potential still present from the interacting system, plus---in contrast to the Hartree--Fock method---a correction term, the so-called \q{exchange-correlation potential}.
\[
\vKS([n],x) = \int \frac{n(y)}{|x - y|} \d y + v(x) +
\vxc([n],x)
\]
This $\vxc[n]$ is actually defined by the equation above to make it hold exactly, because $\vKS[n]$ follows uniquely from $n$ by the HK theorem. Like in the Hartree--Fock method the one-particle orbitals can be computed by solving the one-particle Schrödinger equation which leads to the same self-consistent scheme.
\begin{equation}\label{eq-ti-ks-se}
\left(-\onehalf \Delta_k + \vKS([n],x_k)\right) \varphi_k(x_k) = \varepsilon_k
\varphi_k(x_k).
\end{equation}
But the original problem has such been only shifted because the exact form of the exchange-correlation potential can only be approximated. Different approaches for that exist that \q{come from clever guesses, experience or are systematically constructed} \cite[ch.~13]{ullrich}, yet all are already a sophistication of the Hartree--Fock method. Following Kohn this is the place where actual physics enters DFT, as he told in his \citeasnoun{kohn-nobel}: \q{These approximations reflect the physics of electronic structure and come from outside DFT.}

\section{Time-dependent DFT}
\label{sect-tddft}

\begin{xquote}{Ilya Tokatly in \citeasnoun{marques}}
One of the most important conceptual achievements of DFT is the possibility to formulate the many-body problem in a form of a closed theory that contains only a restricted set of basic variables, such as the density [\ldots] In classical physics, a theory of this type is known for more than two centuries. This is classical hydrodynamics. In fact, the Runge-Gross mapping theorem in TDDFT proves the existence of the exact quantum hydrodynamics.
\end{xquote}

The straightforward extension of the principles of DFT to a time-dependent regime with a central theorem corresponding to that of Hohenberg--Kohn was laid out by \citeasnoun{runge-gross}. Many enhancements apply to more involved settings, including time-dependent current DFT with external vector-potential \cite{vignale-2004}, a multicomponent setup which respects nuclear motions such as phonon modes \cite[ch.~12]{marques} or more recent a generalisation to photonic degrees of freedom in cavity quantum electrodynamics \cite{ruggenthaler-qed-1,ruggenthaler-qed-2,flick-2015}.

A general overview of TDDFT is given in \citeasnoun{ullrich}, a nice non-expert survey in \citeasnoun{ullrich-2014}, and special topics are included in \citeasnoun{marques} which features parts of our work in chapter 9.

\subsection{Basics of the Runge--Gross theorem}
\label{sect-basics-rg}

The state of an $N$-particle system, usually assumed to consist of electrons, is written as $\psi(t,\ushort{x})$ and is usually assumed to be spatially (anti-)symmetric as well as normalised to 1 in $\H = L^2(\Omega^N)$. The evolution of such a state is determined by the time-dependent Schrödinger equation
\[
\i \partial_t \psi = H \psi
\]
with the self-adjoint Hamiltonian $H = T + W + V$ being constituted by a kinetic part
\[
T = -\frac{1}{2} \sum_{i=1}^N \Delta_i,
\]
an interaction part, which can---but need not---be determined by Coulombic repulsion
\[
W = \frac{1}{2} \sum_{\substack{i,j=1 \\ i \neq j}}^N w(x_i-x_j),\quad w(x) = \frac{1}{|x|},
\]
and an energy term from an external, real scalar potential $v$
\begin{equation}\label{op-ext-pot}
V = \sum_{i=1}^N v(t,x_i).
\end{equation}

The wave function $\psi$ determines the one-particle density just like in the time-independent case \eqref{def-n}.
\[
n(t,x) = N \int |\psi(t,x, \bar x)|^2 \d \bar x
\]
Additionally the one-particle current density with $\nabla$ only acting on the first particle position $x_1=x$ (because of the assumed symmetry it does not matter on which particle position it acts or which other coordinates get integrated out) can be defined as
\begin{equation}\label{j-def}
j(t,x) = N \int \Im\{\psi^*(t,x,\bar x) \nabla \psi(t,x,\bar x)\}  \d \bar x.
\end{equation}
Even in a stationary state this current can be non-zero, as for instance in the case of hydrogen eigenstates with magnetic quantum number $m \neq 0$ that have an azimuthal current. But its divergence $\nabla\cdot j$ is always zero in the case of an eigenstate because from the time-dependent Schrödinger equation one can easily derive the continuity equation
\begin{equation}\label{cont-eq}
\partial_t n = -\nabla\cdot j
\end{equation}
and thus $\nabla\cdot j = -\partial_t n = 0$.

Now it is exactly this change of density in time that we will be interested in, more specifically how it is steered by an external potential. Studying this will later lead to \eqref{eq-def-q}, the central equation of our approach to TDDFT. Yet for the first result on a possible mapping from (time-dependent) densities to potentials we do not need to go exactly that far.

In the following theorems we need the notion of \emph{analyticity} for functions defined on the real line, a concept we are already concerned with since the early \autoref{sect-note-non-analyticity}. It describes a function which has a representation as a Taylor series in a neighbourhood around every point of its domain. The \emph{identity theorem} for analytic functions tells us that if two such analytic functions are identical on some open subset or more generally on a set with accumulation point, they are identical everywhere. That means in turn that if the Taylor coefficients of two analytic functions are all identical at a point, the two functions are again identical everywhere. Further we say that a function is \emph{analytic at a point} if its Taylor series at this point converges in a neighbourhood of this point. The maximal radius of convergence then always hits a pole of the function, either on the real axis, or---more subtly---somewhere in the complex plain and thus invisible if we only consider real arguments. We are now prepared to state the famous Runge--Gross theorem \cite{runge-gross} if only as a conjecture.

\begin{conjecture}[Runge--Gross theorem]\label{runge-gross-th}
Given the initial wave function $\psi_0$ of an $N$-particle system there is a one-to-one mapping between densities and external potentials up to a merely time-dependent function. The potential is assumed to be time-analytic while the initial density $n_0$ must tend to zero at the border of $\Omega$ and fulfils $n_0 > 0$ almost everywhere. Both the potential and $\psi_0$ additionally must be sufficiently smooth in space.
\end{conjecture}

Note that by virtue of the Hohenberg--Kohn theorem (\autoref{hk-th}) in the case of a non-degenerate ground state the $\psi_0$ is uniquely determined by the one-particle density $n$ and thus can be removed as a separate prerequisite in the theorem.

Later we will change the condition that the initial density $n_0$ goes to zero at the border of $\Omega$ over to the potentials. This is primarily due to mathematical considerations and erases the need to speak of a one-to-one mapping modulo a merely time-dependent part for the potentials. Anyway the adding of an exclusively time-dependent potential will only change the phase of the wave function and thereby has no effect on observables. From the Runge--Gross theorem (\autoref{runge-gross-th}) follows the very essence of time-dependent density functional theory.

\begin{corollary}\label{runge-gross-cor}
Since the density determines the potential from which we can calculate the complete wave function by solving the Schrödinger equation, any observable is a functional of the one-particle density and the initial state $\psi_0$.
\end{corollary}

Note that dependence is not only on the density $n$ at a given time instant but also on all previous times back to the fixed initial state. So bear in mind that in general a functional notation $f[n]$ includes time-dependence of $n$, it is history-dependent and thus has a \q{memory} effect, else the notation would be $f[n(t)]$.

\subsection{The classical Runge--Gross proof}
\label{sect-runge-gross-proof}

Conveniently the $k$-th time-derivative of a sufficiently smooth $f(t)$ at time $t=0$, i.e., the $k$-th Taylor coefficient at $0$, will be noted shortened as $f^{(k)} = \partial_t^k f(0)$. This way the evaluation of $f$ at $t=0$ can be denoted as $f^{(0)}$ but we will more commonly write this as $f_0 = f(0)$.

\begin{proof}[Formal proof of \autoref{runge-gross-th}]
By going through the tedious task of solving the Schrödinger equation at least in thought we have a map $v \mapsto n$ and we now have to prove that it is injective. First we show that starting from the same initial state $\psi_0 = \psi'_0$ two potentials $v, v'$ that allow a sufficiently regular solution of the time-dependent Schrödinger equation (see \autoref{sect-regularity-sobolev}) and differ by more than a merely time-dependent function necessarily lead to different current densities $j \neq j'$. As the Hamiltonians of the two systems only differ in the potential, most terms cancel when subtracting the time-derivatives of the currents at time $t=0$. Physically this describes the difference of the force densities.
\begin{align*}
\partial_t (j &- j') |_{t=0} = j^{(1)}-j'^{(1)}= \\
&= N \left. \int \Im\{ (\partial_t \psi) \nabla \psi^* + \psi \nabla \partial_t \psi^* - (\partial_t \psi') \nabla \psi'^* - \psi' \nabla \partial_t \psi'^* \}\d\bar{x} \,\right|_{t=0} \\
&= N \left. \int \Re\{ V\psi \nabla \psi^* - \psi \nabla (V\psi^*) - V'\psi' \nabla \psi'^* + \psi' \nabla (V'\psi'^*) \} \d\bar{x}\,\right|_{t=0} \\
&= -N \int |\psi_0|^2 \nabla (v_0 - v'_0) \d\bar{x}\\ 
&= -n_0 \nabla(v_0 - v'_0)
\end{align*}
(Note that from the sum in \eqref{op-ext-pot} only the first term is left, because $\nabla$ only affects the first particle position.)\\
We can repeat this procedure and find the Taylor coefficients of $j-j'$ at $t=0$ analogously \q{after some straightforward algebra} \cite{runge-gross} that is omitted in most presentations as well as here. See \citeasnoun{fournais-2016} for a more elaborated mathematical treatment.
\begin{equation}\label{j-taylor-coeff}
j^{(k+1)} - j'^{(k+1)} =-n_0 \nabla (v^{(k)} - v'^{(k)})
\end{equation}
Now there must be a $k$ where the term $v^{(k)} - v'^{(k)}$ is $x$-dependent. Else for all $k$ the difference would be identical for all $x$ and therefore it would hold $v(t,x) - v'(t,x) = c(t)$, which was excluded by assumption.\footnote{Note that it is not sufficient for the original potential to be Taylor-expandable only at $t=0$ because maybe the $x$-dependent difference between $v$ and $v'$ only appears after the radius of convergence of this Taylor series. The original formulation of \citeasnoun{runge-gross} states that $v$ \q{can be expanded into a Taylor series with respect to the time coordinate around $t=t_0$.} This does not unambiguously imply that this series has an infinite radius of convergence. If so, then $v$ would be everywhere holomorphic (entire) and thus also analytic, whereas analyticity alone would be sufficient for the desired identity theorem to hold. In \citeasnoun{maitra-2002} the condition on the potential is again relaxed to only \emph{piecewise} time-analytic functions.} Thus the right hand side of \eqref{j-taylor-coeff} does not vanish for some $k$, which tells us the currents $j,j'$ cannot be equal. This already leads us to the conclusion that there must be a one-to-one correspondence $v \leftrightarrow j$.\\
If we apply $\nabla\cdot$ to \eqref{j-taylor-coeff} and use the continuity equation \eqref{cont-eq} we can take this result over to the density difference $\rho = n-n'$.
\begin{equation}\label{n-taylor-coeff}
\rho^{(k+2)} = n^{(k+2)} - n'^{(k+2)} = \nabla\cdot \left( n_0 \nabla (v^{(k)} - v'^{(k)})\right)
\end{equation}
Arguing as above there must be some $k$ where $f(x) = v^{(k)}(x) - v'^{(k)}(x) \neq const$. To proceed we form the $L^2(\Omega)$ inner product with $f$.
\[
\langle f, \rho^{(k+2)} \rangle = \langle f, \nabla\cdot (n_0 \nabla f) \rangle
\]
The initial density $n_0$ vanishes at the border of $\Omega$ so we can apply integration by parts without a boundary term to get
\begin{equation}\label{rg-scalar-prod}
\langle f, \rho^{(k+2)} \rangle = -\langle \nabla f, n_0 \nabla f \rangle.
\end{equation}
Because of $f(x) \neq const$ it holds $\nabla f \neq 0$ and together with $n_0 > 0$ almost everywhere\footnote{Instead of this condition \citeasnoun{runge-gross} assumed $n_0$ to be \q{reasonably well behaved} that is they exclude the possibility that at every point $\nabla f = 0$ \emph{or} $n_0 = 0$, although misleadingly they state it just the other way round: \q{[W]e merely have to exclude that the initial density vanishes in precisely those subregions of space where $[f]=const$, [\ldots]}.} it clearly follows that $\langle \nabla f, n_0 \nabla f \rangle > 0$. Thus at least one Taylor coefficient of $\rho=n-n'$ is non-zero and we have different densities.
\end{proof}

Instead of $n_0 > 0$ one could demand $\langle \cdot, n_0 \cdot \rangle$ being an inner product on an adequate weighted Hilbert space. Such a more general line of thought is followed in our later considerations in \autoref{sect-weighted-soblev}. This condition of course relates again to the unique continuation property discussed in \autoref{sect-ucp} if the initial state is considered as an eigenstate.

\subsection{The Kohn--Sham scheme and the extended Runge--Gross theorem}
\label{sect-ks-rg}

\autoref{runge-gross-cor} tells us that the initial state $\psi_0$ and the density $n$ in time already fix all possibly time-dependent observables. It seems now practical in calculating properties of a many-particle quantum system to switch from an interacting system to one without interactions, setting $W=0$, with exactly the same density $n$ and a suitable initial state $\phi_0$ which therefore yields the very same observables. Just like in the time-independent case this auxiliary system is called the \emph{Kohn--Sham system}. One gets it by choosing just the right external potential $\vKS$, of course different from the \q{real} one, that simulates all interaction effects and is unique up to a merely time-dependent function. That this is indeed possible is---if not proved---at least hinted at by the Runge--Gross theorem. What is still necessary to show is that the given density from an interacting system is reproducible by the potential $\vKS$ applied to a non-interacting one. This property is usually called \emph{non-interacting $v$-representability}: Is an interacting $v$-representable density also non-interacting $v$-representable?

The main advantage in actually calculating observables is that the Schrödinger equation can then be factorised into $N$ independent Schrö\-din\-ger equations each limited to the one-particle sector just like in \eqref{eq-ti-ks-se} in the time-independent case.
\[
\i \partial_t \varphi_i = \left( -\onehalf \Delta + \vKS[n,\phi_0,\psi_0] \right) \varphi_i
\]
A mathematical treatment of the Schrödinger equation above  using semigroup theory was conducted by \citeasnoun{jerome-2015} with solutions in $\Cont([0,T],H_0^1(\Omega))$ taking $\vKS[n]$ in the typical form usually applied in approximations. Note that also the density of the system is now a simple sum of orbital densities.
\[
n(t,x) = \sum_{i=1}^N |\varphi_i(t,x)|^2
\]
In most cases the KS initial state $\phi_0$ is taken as a Slater determinant of single-particle orbitals $\varphi_i(0,x)$. The KS potential $\vKS$ is defined to incorporate the previous external potential $v$ from the interacting system, the usual Hartree potential $\vH$ to at least come up for some effects of the interaction, and the missing part, the so-called exchange-correlation potential $\vxc$, which is again actually \emph{defined} this way.
\[
\vKS[n,\phi_0,\psi_0] = v + \vH[n] + \vxc[n,\phi_0,\psi_0]
\]
\[
\vH([n],t,x) = (n(t) * w)(x) = \int n(t,x') w(x-x') \d x'
\]
Of course the exchange-correlation potential is not known exactly but different approximations exist to accommodate for different cases. One remaining subtlety is the dependence of the KS potential on the initial states. In cases where $\psi_0$ and $\phi_0$ are non-degenerate ground states one can get rid of this dependency by invoking the Hohenberg--Kohn theorem (\autoref{hk-th}) because the density alone already encodes all the information needed.

Bear in mind that the exchange-correlation potential $\vxc$ in general not only depends on the density $n(t)$ at a given time, but also on all previous times which directly follows from the note given after \autoref{runge-gross-cor}. If $\vxc([n],t)$ is approximated by a functional $\vxc[n(t)]$ without time-history this is called an \emph{adiabatic} approximation.

We will now aim at answering the question of non-interacting $v$-representability posed at the start of the section affirmatively, the following conjecture and its formal proof are due to \citeasnoun{van-leeuwen}. An attempt of a further mathematical discussion giving expressions for the involved domains of potentials and densities was made in \citeasnoun{tddft1} and really kicked off the studies that led to the present work.

\begin{conjecture}[extended Runge--Gross theorem]\label{ext-runge-gross-th}
Given a system with an initial wave function $\psi_0$ evolving under the interaction potential $w$ and external potential $v$ it is possible to choose a unique $v'$ and an initial wave function $\psi'_0$ such that this system exactly reproduces the same density $n$ for a different interaction $w'$. Here the density $n$ is assumed to be time-analytic and sufficiently smooth in space.
\end{conjecture}

We want to consider two special cases of auxiliary external potentials $w'$ right away. First the case from the Kohn--Sham scheme where we put $w'=0$, i.e., switching off all interactions between the electrons. By the theorem above we can conclude that there is always an external potential $v' = \vKS$ to give a non-interacting model for a fixed density. This is exactly the answer to the raised question: Under the given conditions (and ignoring the  concerns raised below) an interacting $v$-representable system is always non-interacting $v$-representable.

If we further consider $w=w'$ and $\psi_0=\psi'_0$ the theorem exactly reproduces the main conclusion from the Runge--Gross theorem (\autoref{runge-gross-th}) and is therefore aptly named an extension of it.

\begin{proof}[Formal proof]
Similar to the proof of \autoref{runge-gross-th} we start by calculating the force density $\partial_t j$ for a given system directly from the definition of the current \eqref{j-def}.
\[
\partial_t j = N \int \Im\{(\partial_t\psi^*) \nabla \psi + \psi^* \nabla \partial_t\psi\} \d \bar{x}
\]
Next we put in the Schrödinger equation for the time-derivatives of $\psi$.
\begin{align*}
\partial_t j &= N \int \Re\{(H\psi^*) \nabla \psi - \psi^* \nabla H\psi\} \d \bar{x} \\
&= N \int \Re\{(H\psi^*) \nabla \psi - \psi^* H \nabla\psi\ - \psi^* (\nabla H)\psi\} \d \bar{x}
\end{align*}
If we use that for $H = T+W+V$ given in \autoref{sect-basics-rg} both $W$ and $V$ are real multiplicative operators and can be drawn before $\psi^*$ the first two terms in the expression given above just form $2\i$ times the imaginary part of $W\psi^* \nabla \psi$ and $V\psi^* \nabla \psi$ respectively and will thereby vanish when taking the real part. Furthermore $T$ is not explicitly spatially-dependent which makes the last term with $(\nabla T)$ vanish too. Like in the proof of \autoref{runge-gross-th} we can write $\nabla V = \nabla v$ because $\nabla$ only affects the first particle position.
\begin{equation}\label{eq-force-density-1}
\partial_t j = N \int \left( \Re\{(T\psi^*) \nabla \psi - \psi^* T \nabla\psi\} - (\nabla W) |\psi|^2 - (\nabla v) |\psi|^2 \right)  \d \bar{x}
\end{equation}
The external potential term is simply integrated to $-n \nabla v$. We then apply the divergence operator $\nabla\cdot$ to get $\nabla\cdot \partial_t j = -\partial_t^2 n$ with \eqref{cont-eq} on the left hand side. The remaining integral we just collect into an expression $q$ that includes all the internal forces.
\begin{equation}\label{eq-def-q}
q = -N \nabla\cdot \int \left( \Re\{(T\psi^*) \nabla \psi - \psi^* T \nabla\psi\} - (\nabla W) |\psi|^2 \right) \d \bar{x}
\end{equation}
The equation in this form can be named the \q{divergence of force density equation} (more commonly as \q{divergence of local forces} or very modestly \q{fundamental equation of TDDFT}). It directly relates the external potential and the density and will therefore play a crucial role in all further investigations.
\begin{equation}\label{eq-div-force-density}
\boxed{
-\nabla\cdot (n \nabla v) = q - \partial_t^2 n
}
\end{equation}
The initial wave function $\psi'_0$ of the auxiliary system must be chosen such that it yields $n_0$ and $n^{(1)}$ correctly as we are trying to reproduce this given density. The latter also implies that the initial momenta of both systems are the same, else only an infinite force could prevent the two systems from evolving differently. \cite{van-leeuwen2} The higher orders of the Taylor expansion of $n$ are matched automatically as we can relate them by \eqref{eq-div-force-density} to the lesser orders.\\
We are now ready to show that we can construct a $v'$ for a second system with different interaction such that $n=n'$. As $n$ is assumed to be analytic it is sufficient to assume $n^{(k)}=n'^{(k)}$ for all $k$. In the same fashion it is sufficient to get the Taylor coefficients $v'^{(k)}$ do determine $v'$. In the zeroth order it is enough to set $t=0$ in \eqref{eq-div-force-density} for the primed system.
\begin{equation}\label{v-coeff-0}
-\nabla\cdot (n_0 \nabla v'_0) = q'_0 - n^{(2)}
\end{equation}
The $q'_0$ is known because we can calculate it as an expectation value of the initial wave function $\psi'_0$. Given equation \eqref{v-coeff-0} we solve it for $v'_0(x) = v'(0,x)$.\footnote{That this is indeed possible is by no means clear at this stage. Indeed there is a theory about the possibility of unique solutions to such Sturm--Liouville problems which makes up a considerable part of \autoref{ch-fp}.}\\
The higher terms now follow if we differentiate \eqref{eq-div-force-density} $k$ times to get a recursion relation between density coefficients and potential coefficients.
\begin{equation}\label{v-coeff-recursion}
-\nabla\cdot (n_0 \nabla v'^{(k)}) = \sum_{l=0}^{k-1} \left( \begin{array}{c} k \\ l \end{array} \right) \nabla\cdot (n^{(k-l)} \nabla v'^{(l)}) + q'^{(k)} - n^{(k+2)}
\end{equation}
Here the term $q'^{(k)}$ can be calculated by successive applications of the Heisenberg equation for the second system to the expectation value of an operator involving an alternative formulation of \eqref{eq-def-q} as noted below above \eqref{eq-tensor-q}. This is indeed possible because the Hamiltonian for the primed system at time $t=0$ is already defined with $v'_0$ from \eqref{v-coeff-0}. In \eqref{v-coeff-recursion} we now solve for $v'^{(k)}$ for all $k \geq 1$. This gives us all orders of a supposed Taylor expansion of the desired external potential $v'$ for the auxiliary system.
\end{proof}

Note a serious flaw in the last argument, which is already cautiously indicated by the use of the word \q{supposed}. Because at this point we do not know for sure if this series has a particularly large or even non-zero radius of convergence. The typical way of thinking seems to be that a non-analytic potential can only produce a non-analytic density---like the cusps in the density from singularities in the potential mentioned in \autoref{sect-analyticity-eigenstates} for eigenstates---and thus, as the density is assumed to be analytic, the potential must be as well. This is not in conflict with the example for non-analytic wave functions from analytic initial states under free evolution in \autoref{sect-kovalevskaya}. But still the whole reliance on analyticity of the substantial quantities is considered unfavourable. This was already pointed out \autoref{sect-note-non-analyticity} from a more epistemological viewpoint and will be criticised again in the following section on limitations of the given proofs.

\subsection{Limitations of the Runge--Gross theorems}

The obvious restriction of the proofs given above is the limitation to time-analytic potentials and/or densities that might even prove critical in the extended theorem as just mentioned above. Real time-analyticity seems unnatural if one considers that by this restriction a physical quantity is fixed for all times if one knows its value and that of all derivatives at a single instant. Especially if applied to a quantity like the external potential which is often thought of being controlled by an experimentalist with her own free will from the outside, this fact seems a bit bewildering.

Additionally in \citeasnoun{maitra-2010} the concern is raised, that non-ana\-lytic\-ities of the initial wave function or the potential in space may sequentially cause non-ana\-lytic\-ities in time. This is demonstrated with a few examples in one spatial dimension with resulting non-time-analytic densities. Only very recently \citeasnoun{fournais-2016} scrutinised the classical Runge--Gross proof even more and gave expressions for possible classes of initial states and external potentials that turn out to be quite restrictive if singular interactions are considered.

One pursuit to relax the condition of analyticity for the external potential was given by \citeasnoun{van-leeuwen2}. The Runge--Gross theorem is therein proved for Laplace transformable switch-on potentials and systems initially in the ground state. Note that some density variations produced by a perturbation of the external potential are ruled out, because they can never be created by finite perturbations. An example of such a density variation would be one which is zero on an area of space of positive measure. This limitation already showed up in the classical Runge--Gross proof and will occur again in our later considerations.

A different approach to non-analytic potentials is pursued in \autoref{sect-dipole} by restricting the potentials to a rather special type where the time-dependent part is always of the form of a dipole field. But the main path to a generalisation will be shown in the next big chapter on a fixed-point variant of the Runge--Gross proof.

The other major condition in the classical Runge--Gross theorem (\autoref{runge-gross-th}) is towards the initial density being strictly positive almost everywhere and zero at the boundary. The zero boundary condition is clearly assumed to hold for any finite system under consideration, the positivity condition is also seen relatively uncritical as was already noted below the proof of \autoref{runge-gross-th}.

\subsection{Reformulation of the internal forces term}
\label{sect-internal-forces}

Because it will be relevant for later sections we want to devote some more time on reformulating the $q$-term defined in \eqref{eq-def-q}. It actually captures the whole non-local properties of the potential and does not only depend on the density at one time-instant but also on all previous ones. This is expressed by the dependence on the wave function $\psi[n]$ and reflects the so-called memory effects. For the purpose of a de facto calculation of a potential corresponding to a given density, this implicit dependence will be circumvented by a fixed-point procedure in the following chapter.

We have defined 
\[
q = -N \nabla\cdot \int_{\bar{\Omega}}  \left( \Re\{(T\psi^*) \nabla \psi - \psi^* T \nabla\psi\} - (\nabla W) |\psi|^2 \right) \d \bar{x}
\]
and reintroduce
\[
T = -\frac{1}{2} \sum_{i=1}^N \Delta_i, \quad
W = \frac{1}{2} \sum_{\substack{i,j=1 \\ i \neq j}}^N w(x_i-x_j).
\]
In the first term $(T\psi^*) \nabla \psi$ all the parts of $T$ with $i\geq 2$ can be moved over to $\nabla \psi$ invoking integration by parts and thus cancel with the following term. The remaining part is $-\onehalf \Delta$ acting on $x=x_1$ because it is not integrated over in the definition of $q$. In $\nabla W$ all terms not involving $x=x_1$ vanish anyway and if we put in the additional but uncontroversial symmetry assumption $w(x)=w(-x)$, the formal statement of Newton's third law, the leftover is equal to
\[
\nabla W = \sum_{j=2}^N \nabla w(x-x_j) = \sum_{j=2}^N \nabla_j w(x-x_j).
\]
Note that the gradient is now with respect to the particle position $x_j$ instead of $x=x_1$ which is clearly equivalent because of the symmetric character of the interaction potential. This means we can use integration by parts as well and move the gradient over to $|\psi|^2$.
\begin{equation*}\label{eq-q-term-1}
q = -N \nabla\cdot \int_{\bar{\Omega}} \left( -\onehalf \Re\{(\Delta\psi^*) \nabla \psi - \psi^* \Delta \nabla\psi\} + \sum_{j=2}^N w(x-x_j) \nabla_j|\psi|^2 \right) \d \bar{x}
\end{equation*}
In the final step we apply the outside divergence to the individual terms. In taking the real part the terms $(\nabla \Delta \psi^*)\cdot \nabla \psi - (\nabla \psi^*) \cdot\Delta \nabla \psi$ vanish. In the term involving the interaction potential the divergence is again taken as $\nabla_j$ and brought over to $\nabla_j|\psi|^2$. We thus have
\begin{equation}\label{eq-q-term-2}
q = N \int_{\bar{\Omega}} \left( \onehalf |\Delta\psi|^2 - \onehalf\Re\{ \psi^* \Delta^2 \psi\} + \sum_{j=2}^N w(x-x_j) (\nabla_j-\nabla) \cdot \nabla_j|\psi|^2 \right) \d \bar{x}.
\end{equation}
We will prove that we can give conditions on $\psi$ and $w$ such that $q \in L^1(\Omega)$ which is somehow expected for an expectation value, but also $q \in L^2(\Omega)$ which will be an important ingredient in \autoref{sect-appl-lax-milgram}. Note the additional assumption $w(x)=w(-x)$ which was needed to derive the special form \eqref{eq-q-term-2} that serves as a basis in both proofs.

\begin{lemma}\label{lemma-q-L1}
If $\Omega \subseteq \R^3$, $\psi \in H^4(\Omega^N)$, and $w \in L^2(\R^3) + L^\infty(\R^3)$ with $w(x)=w(-x)$ then $q \in L^1(\Omega)$.
\end{lemma}

\begin{proof}
The condition $\psi \in H^4(\Omega^N)$ clearly makes the terms $|\Delta\psi|^2$, $\psi^* \Delta^2 \psi \in L^1(\Omega^N)$, thus integrated over they are still $L^1(\Omega)$. The interaction term is a little bit more involved. We test the $L^1$-condition on the individual terms, i.e., we want to prove
\[
\int_{\Omega^N}  |w(x-x_j)| \cdot |(\nabla_j-\nabla) \cdot \nabla_j|\psi|^2| \d \ushort{x} < \infty.
\]
The term with $|\psi|^2 = \psi^* \psi$ and derivatives up to second order is split into terms of the sort $D^\alpha \psi^* \cdot D^\beta \psi$ with $|\alpha|,|\beta| \leq 2$. This means we can take the integral as an $L^2(\Omega^N)$ inner product.
\[
\langle |D^\alpha \psi|, |w(x-x_j)| \cdot |D^\beta \psi| \rangle \leq \|D^\alpha \psi\|_2 \cdot \|w(x-x_j) \cdot D^\beta \psi\|_2
\]
The first norm is clearly finite with $\psi \in H^2(\Omega^N) \subset H^4(\Omega^N)$ while the second can be treated like in the proof of \autoref{th-kato} where $D^\beta \psi \in H^2(\Omega^N)$ follows from $\psi \in H^4(\Omega^N)$ again. That the theorem was only formulated for $\R^3$ is not a real issue after a variable change to $x_- = x-x_j, x_+ = x+x_j$.
\end{proof}

In the later treatment of the full TDDFT proof the $L^1$-regularity is found to be insufficient and the Hilbert space structure of $L^2$ becomes necessary. And indeed it is possible to prove $L^2$-regularity in exactly the same setting as above. That this is true was not clear at all to us, although \citeasnoun{maitra-2010} already wrote that \q{the restriction to locally square-integrable right-hand sides (meaning that the integral of the square of the right-hand side over a finite region is finite) is generally satisfied for physical wave functions.} This statement is accompanied with a reference to our \citeasnoun{tddft1} where $L^2$-regularity is actually still stated as a precondition. This issue seems to be resolved with the following lemma.

\begin{lemma}\label{lemma-q-L2-2}
If $\Omega \subseteq \R^3$, $\psi \in H^4(\Omega^N)$, and $w \in L^2(\R^3) + L^\infty(\R^3)$ with $w(x)=w(-x)$ then $q \in L^2(\Omega)$.
\end{lemma}

\begin{proof}
We define the Hilbert space $L^2(\bar\Omega)$ with reduced configuration space including the particle positions $x_2, \ldots, x_N$ with inner product $\langle \cdot,\cdot \rangle'$ and induced norm $\|\cdot\|'$. Functions with this configuration space have the missing particle position $x=x_1$ as an additional parameter noted as an index. The first two terms to check in the expression for $q$ as in \eqref{eq-q-term-2} are of the kind $\langle \Delta\psi,\Delta\psi \rangle'_x$ and $\langle \psi,\Delta^2\psi \rangle'_x$. Note that $\Delta$ is not self-adjoint with respect to this altered Hilbert space because we cannot use integration by parts in the $x$ coordinates. We start with the second term and use CSB.
\[
|\langle \psi,\Delta^2\psi \rangle'_x| \leq \|\psi\|'_x \|\Delta^2\psi\|'_x
\]
Now as a matter of fact $N \|\psi\|'^2_x = n(x) \in W^{4,1}(\Omega)$ and thus with the Sobolev embedding \eqref{eq-sobolev-embedding} is bounded. This means when integrating over $x$ we can move the supremum (maximum) of $n(x)$ in front of the integral and we are done using $\psi \in H^4(\Omega^N)$.
\begin{align*}
\int_\Omega |\langle \psi,\Delta^2\psi \rangle'_x|^2 \d x &\leq \frac{1}{N} \sup_{x \in \Omega} n(x) \int_\Omega \|\Delta^2\psi\|'^2_x \d x \\
&= \frac{1}{N} \sup_{x \in \Omega} n(x) \cdot \|\Delta^2 \psi\|_2^2 < \infty
\end{align*}
The first term is
\[
\langle \Delta\psi,\Delta\psi \rangle'_x = \| \Delta\psi \|'^2_x
\]
and similar to before $\|\Delta\psi\|'^2_x \in W^{2,1}(\Omega)$ which equals $\|\Delta\psi\|'_x \in W^{2,2}(\Omega)$ and again by Sobolev embedding \eqref{eq-sobolev-embedding} this is bounded. Thus testing $L^2(\Omega)$ we have
\[
\int_\Omega \| \Delta\psi \|'^4_x \d x \leq \left(\sup_{x \in \Omega} \|\Delta\psi\|'_x\right)^2 \cdot \|\Delta \psi\|_2^2 < \infty.
\]
Finally for the interaction term in \eqref{eq-q-term-2} we use a notation like in \autoref{lemma-q-L1} and want to show that for $|\alpha|,|\beta| \leq 2$
\[
\left|\langle D^\alpha \psi, w(x-x_j) \cdot D^\beta \psi \rangle'_x\right| \leq \|D^\alpha \psi\|'_x \cdot \|w(x-x_j) \cdot D^\beta \psi\|'_x
\]
is in $L^2(\Omega)$. Like before $\|D^\alpha \psi\|'_x$ is bounded thus
\[
\int_{\Omega} \d x \, \left|\langle D^\alpha \psi, w(x-x_j) \cdot D^\beta \psi \rangle'_x\right|^2 \leq \left(\sup_{x \in \Omega} \|D^\alpha \psi\|'_x \right)^2 \cdot \|w(x-x_j) \cdot D^\beta \psi\|_2^2.
\]
The last $L^2$-norm to be finite is shown just like in \autoref{lemma-q-L1}.
\end{proof}

This proof strategy was successful because the Sobolev embedding into a space of bounded functions is possible in the low dimensionality of the configuration space of $q$ as a derived quantity, but never in the full configuration space of the wave function for an arbitrary particle number. That the wave function actually \emph{is} always a bounded function was historically considered self-evident and was removed only later by the Hilbert space setting of $L^2$ wave functions. \citeasnoun{schroedinger-2} considered it \q{the almost obvious demand that, as a physical quantity, the function $\psi$ must be single-valued, finite, and continuous throughout configuration space.}

Another form of $q$ is in tensor notation, a version of it was already derived by \citeasnoun{martin-schwinger}. The expression is usually reformulated as the expectation values of operator-valued tensor quantities, then given in second quantisation (\citeasnoun[9.2]{marques}; \citeasnoun{van-leeuwen}), which will be avoided here.\footnote{It is customary to use the field operators of second quantisation in many-body theory even though one is not directly involved with field theory. They adhere to the correct Fermi or Bose symmetry and arbitrary particle number without any additional considerations and sometimes allow for a more efficient notation. Yet it should be noted that the annihilation operator on the continuous Fock space is not closeable. If we take a single-particle wave function $\langle x|\psi \rangle=\exp(-\lambda|x|^2)$ and let $\lambda\rightarrow \infty$ then $|\psi\rangle\rightarrow 0$, but $\hat\psi(x)|\psi\rangle=|0\rangle$, the vaccum state, which stays constant. The adjoint creation `operator' $\hat\psi^\dagger(x)$ is really an operator-valued distribution because formally it creates particles with wave function $\psi(x')=\delta(x-x')$. \cite[II.1.3.2(7) and II.1.3.14(8)]{thirring}}
\begin{equation}\label{eq-tensor-q}
q = \partial_l\partial_k T_{kl} + \partial_l W_l
\end{equation}
Here summation over multiple indices is implied. In the definition of the tensor quantities we will for once use the common but inaccurate notation of the delta distribution as the \q{Dirac function} $\delta(x)$. The first expectation value is that of the momentum-stress tensor.
\[
T_{kl}(t,x) = \left( \frac{1}{2} \sum_{i=1}^N \langle \partial^i_k \psi, \delta(x - x_i) \partial^i_l \psi \rangle - \frac{1}{8} \partial_{k}\partial_{l} n(t,x) \right) + (k \leftrightarrow l)
\]
The derivatives $\partial^i_k$ are meant to act only on the $k$-th component of the particle position $x_i$ with $k$ only ranging from 1 to 3, whereas $\partial_k$ acts on the $k$-th component of $x$. The interaction term will be expressed as well as an expectation value.
\[
W_l(t,x) = \langle \psi, \sum_{i=1}^N \delta (x-x_i) \sum_{\substack{j=1 \\ j \neq i}}^N \partial_{l} w(|x-x_j|) \psi \rangle
\]
This expression could be recast into an order-2 tensor following an idea in \citeasnoun{puff-gillis} also used and explained in \citeasnoun{tokatly}. Yet another form of $q$ is found in literature that derives itself from the Heisenberg equation. If we think about $\hat{n}, \hat{\jmath}$ as operators (problems related to that notion are discussed in \autoref{sect-lrt}) then the following form of \eqref{eq-div-force-density} can be derived for their expectation values $\langle \hat n\rangle = n, \langle \hat \jmath\rangle = j$ under some time-dependent state $\psi(t)$.
\[
\partial_t j = \partial_t \langle \hat \jmath \rangle = \i \langle [H,\hat \jmath]\rangle = \i \langle [T+W,\hat \jmath]\rangle - n \nabla v
\]
As before the divergence of that equation yields
\[
-\partial_t^2 n = \i \nabla\cdot \langle [T+W,\hat \jmath]\rangle - \nabla\cdot (n \nabla v)
\]
and thus the $q$-term is the expectation value of a commutator
\[
q[v] = -\i \nabla\cdot \langle [T+W,\hat \jmath] \rangle.
\]

\subsection{Examination of the second law term}

Now that a result for $q \in L^2$ is at hand with \autoref{lemma-q-L2-2} we want to complement this with a similar result for $\partial_t^2 n \in L^2$. We call this part of \eqref{eq-div-force-density} the \q{second law term} because it is the response of the material density to internal and external forces according to Newton's second law $F=ma$. In opposition to the results given for the internal forces term we do this with respect to an initial state $\psi_0$ and a given potential (interaction and external) $v$ that determines a quantum trajectory $\psi[v]$ by Schrödinger's equation.

\begin{lemma}\label{lemma-dt2n-L2}
If $\psi_0 \in H^4(\Omega^N)$ and $v \in \Lip([0,T],W^{2,\Sigma})$ then $\partial_t^2 n([v],t)$ $\in L^2(\Omega)$ for almost all $t \in [0,T]$.
\end{lemma}

\begin{proof}
First note that Sobolev regularity of the initial state is conserved as shown by \autoref{th-sobolev-regularity}, i.e., $\psi([v],t) \in H^4(\Omega^N)$. We then start with the definition of the density \eqref{def-n} and by differentiating it twice with respect to the time parameter putting in the Schrödinger equation with Hamiltonian $H[v]$.
\begin{align*}
\partial_t^2 n &= N \partial_t^2 \int |\psi|^2 \d \bar x = -\i N \partial_t \int \psi^* H[v] \psi \d \bar x + c.c. \\
&= -\i N \int \left(\i (H[v] \psi^*) H[v] \psi + \psi^* \dot{v}H[v] \psi -\i \psi^* H[v]^2 \psi \right) \d \bar x + c.c.
\end{align*}
Now we have three different parts that we each need to estimate in the style of \autoref{lemma-q-L2-2} using the Hilbert space $L^2(\bar\Omega)$ with reduced configuration space.
\[
\langle H[v] \psi, H[v] \psi \rangle'_x = \|H[v]\psi\|'^2_x
\]
We have $\|H[v]\psi\|'_x \in W^{2,2}(\Omega)$ due to \autoref{lemma-rs-2}, wich solves the case just like in the proof of \autoref{lemma-q-L2-2}, but for $H[v]$ instead of $\Delta$. The third term follows analogously. For the part involving $\dot{v}$ we write
\[
\langle \dot{v} \psi, H[v] \psi \rangle'_x \leq \|\dot{v}\psi\|'_x \cdot \|H[v]\psi\|'_x.
\]
Again $\|H[v]\psi\|'_x$ can be dismissed as bounded and the $L^2$-condition then demands $\dot{v}\psi \in L^2$.  Now by assumption $v : [0,T] \rightarrow W^{2,\Sigma}$ is Lipschitz-continuous thus also of bounded variation. Such functions are always differentiable almost everywhere if the codomain is a Banach space with the Radon--Nikod\'{y}m property (this is a Banach space extension of Rademacher's theorem). In fact this is just one of the many equivalent characterisations of this property and was previously called the Gelfand--Fréchet property. Now it is known that every reflexive Banach space has the desired Radon--Nikod\'{y}m property. \cite[p.~30, p.~2, p.~5]{diestel-uhl-1976} While we already showed in \autoref{sect-lebesgue-dual} that $L^\infty$ is not reflexive, the Kato perturbations $L^2+L^\infty$ that constitute the potential space are, because the $L^2$ contributions eliminate distributional values in the dual that thus remains $L^1\cap L^2$. It follows $\dot{v} \in W^{2,\Sigma} \subset \Sigma$ and $\dot{v}\psi \in L^2$ by \autoref{lemma-sum-space-inequality}.
\end{proof}

\subsection{Dipole laser-matter interaction as a special case}
\label{sect-dipole}

As a special example we consider potentials with an arbitrary but only space-dependent part $w(x)$ (now $w$ does not denote the interaction potential any more) combined with a dipole interaction $x \cdot E(t)$. This corresponds to a system illuminated by a laser beam in dipole approximation. \cite[ch.~4 (4.5)]{marques}
\begin{equation}\label{eq-potential-laser}
v(t,x) = w(x) + x \cdot E(t)
\end{equation}
Without loss of generality we can assume $E(0) = 0$. For these potentials we can weaken the condition of time-analyticity in the Runge--Gross theorem (\autoref{runge-gross-th}). Note however that this form of a potential cannot fulfil the later requirement of zero boundary conditions raised in \autoref{sect-hs-potentials}. Instead we will rely on the density $n$ and the wave function $\psi$ having zero boundary conditions, a slight deviation from the original assumptions but seemingly unproblematic in nature. Similar to the original theorem this will only yield uniqueness modulo a constant $c_0$.

The results presented here were first published in \citeasnoun{tddft2}. Please note that they relate to the Runge--Gross Theorem and not to its \emph{extended} version.

\begin{proof}[Proof of \autoref{runge-gross-th} for laser-matter interaction]
Assume we have two potentials $v, v'$ of the given type leading to the same density $n$. As we consider two systems with the same initial state $\psi_0$ the quantity $q$ from \eqref{eq-def-q} at time $t=0$ will the same for both systems, i.e., $q_0=q'_0$. As we assume the same density we derive from \eqref{eq-div-force-density} at $t=0$ that
\[
n_0 \nabla(v_0 - v'_0) = 0.
\]
Arguing like in the original proof by using $n_0>0$ almost everywhere this tells us that $v_0 = v'_0 + c_0$. For potentials of the special type \eqref{eq-potential-laser} this is equal to $w(x) = w'(x) + c_0$, thus
\begin{equation}\label{eq-pot-w-diff}
w(x)-w'(x) = c_0.
\end{equation}
We now want to consider the net force exerted on the system, first with an arbitrary potential $v$. For this we integrate the force density \eqref{eq-force-density-1} given in tensor notation as in \eqref{eq-tensor-q} over the whole domain $\Omega$.
\begin{equation}\label{eq-force-int}
F = \int_\Omega \partial_t j = -\int_\Omega (n \nabla v + \partial_k T_{kl} + W_l) \d x
\end{equation}
As shown in \citeasnoun{tokatly} the internal forces term $\partial_k T_{kl} + W_l$ can be written in divergence form altogether.
\[
\partial_k T_{kl} + W_l = \partial_k \Pi_{kl}
\]
Inserting this into \eqref{eq-force-int} and using the Gauss--Ostrogradsky theorem this is equal to evaluating $\Pi_{kl}$ only at the border of $\Omega$ where it is zero because of the $\psi=0$ border condition. Hence the total force is equal to
\begin{equation}\label{eq-force-int2}
F = \int_\Omega \partial_t j = -\int_\Omega n \nabla v \d x
\end{equation}
which just subsumes the fact that the net internal forces cancel each other out. Another way of calculating this quantity is by analogue to Newton's second axiom $F=m a$.
\begin{equation}\label{eq-force-int3}
F = \partial_t^2 \int_\Omega x \, n \d x = \int_\Omega x \, \partial_t^2 n \d x
\end{equation}
This formula can be derived by inserting \eqref{eq-div-force-density}, using the Gauss--Ostrogradsky theorem another time on $q$, then integrating by parts which just gives us \eqref{eq-force-int2} again.
\[
\int_\Omega x \, \partial_t^2 n \d x = \int_\Omega x \, \nabla\cdot(n\nabla v) \d x = -\int_\Omega n\nabla v \d x
\]
The good thing about the representation \eqref{eq-force-int3} is that for two systems with different potentials $v, v'$ but the same density $n$ it immediately becomes obvious that $F - F' = 0$. Via \eqref{eq-force-int2} this difference yields
\[
F - F' = -\int_\Omega n \nabla (v-v') \d x = 0.
\]
Restricting to our special class of potentials and remembering \eqref{eq-pot-w-diff} we get in dimensionality $d=\nabla \cdot x$
\[
F - F' = -\int_\Omega n \nabla (x\cdot(E-E')) \d x = -d\int_\Omega n (E-E') \d x = -Nd(E-E') = 0.
\]
This makes $E(t) = E'(t)$, thus $v(t,x)=v'(t,x)+c_0$.
\end{proof}

The proof is easily generalised to a broader class of potentials. Imagine $v(t,x) = w(x) + g(x) \cdot E(t)$ with a vector part $g(x)$ which obeys $\partial_k g(x) > 0$ or $<0$ respectively for all $k$ and $x \in \Omega$. Then the last step of the proof can be executed in just the same way, the only additional requirement on the potential is of course such that a solution to Schrödinger's equation indeed exists.

\chapter{A Fixed-Point Proof of the Runge--Gross Theorem}
\label{ch-fp}

\section{Weak solutions to a fixed-point scheme}

\subsection{Construction of the fixed-point scheme}

Given an initial wave function $\psi_0$ and an open and bounded space-domain $\Omega \subset \R^d$ we want to prove that for any one-particle density $n : [0,T] \times \Omega \rightarrow \R_{\geq 0}$ from a certain class there is a unique solution $v : [0,T] \times \Omega \rightarrow \R$ to the non-linear equation
\begin{equation}
\label{sl-nonlinear-2}
 -\nabla\cdot (n \nabla v) = q[v] - \partial_t^2 n.
\end{equation}
already known from \eqref{eq-div-force-density} in the proof of the extended Runge--Gross theorem (\autoref{ext-runge-gross-th}). Note that this equation is supposed to hold at all time instants $t \in [0,T]$ and that the argument in $q[v]$ is not the potential at this respective time instant $t$ but includes the whole time-dependent $v$ up to this time which in turn gives a solution $\psi([v],t)$ to Schrödinger's equation at time $t$ that is taken to compute $q([v],t)$ from \eqref{eq-def-q}. Boundedness of $\Omega$ is essential for the important results of \autoref{sect-embedding-theorems} to hold that will further be needed for the application of the Lax--Milgram theorem (\autoref{lax-milgram}) in \autoref{sect-appl-lax-milgram}.

The main idea is to choose any starting potential $v_0$ (at all times) and calculate $q[v_0]$. Next we use this as a right hand side in \eqref{sl-nonlinear-2} to determine $v_1, v_2$ etc.~in an iterative scheme consisting of a partial differential equation of generalised Sturm--Liouville type, also called \q{divergence form}.
\begin{equation}
\label{sl-iteration-2}
 -\nabla\cdot (n \nabla v_{i+1}) = q[v_i] - \partial_t^2 n
\end{equation}
Now it is argued in \citeasnoun{tddft3}, where this construction was originally suggested to provide an alternative route to the extended Runge--Gross theorem (\autoref{ext-runge-gross-th}), that this iteration has a unique fixed point $v$ solving \eqref{sl-nonlinear-2} in some adequately chosen Banach space $(V, \|\cdot\|_V)$ including time-dependence. This would be guaranteed by 
the Banach fixed-point theorem
if the mapping $\mathcal{F} : v_i \mapsto v_{i+1}$ defined by \eqref{sl-iteration-2} is well-defined on 
$V$ as a contraction map, i.e., there is a $\xi \in (0,1)$ for which for all $v,w \in V$
\begin{equation}
\label{F-contraction}
\| \mathcal{F}[v] - \mathcal{F}[w] \|_V \leq \xi \| v-w \|_V.
\end{equation}
After introducing a proper Banach space $(W, \|\cdot\|_W)$ for the right hand side of \eqref{sl-iteration-2} we will split this problem into two questions. Firstly, we need the (linear) solution operator of the Sturm--Liouville equation \eqref{sl-iteration-2} mapping $W \rightarrow V$ to be bounded, where linearity allows us to cancel the $\partial_t^2 n$ term.
\begin{equation}
\begin{aligned}
\label{F-contraction-inequ-1}
\| (-\nabla\cdot (n\nabla))^{-1}q[v] &- (-\nabla (n\nabla))^{-1}q[w] \|_V \\
&= \| (-\nabla\cdot (n\nabla))^{-1}(q[v]-q[w]) \|_V \\
&\leq \xi_1 \| q[v] - q[w] \|_W
\end{aligned}
\end{equation}
Secondly the mapping $q : V \rightarrow W$, the internal forces from a given initial state under the influence of an external field, must allow for a similar estimate. This question is further studied in \autoref{sect-q-mapping-frechet} and \autoref{sect-q-mapping-lipschitz}.
\begin{equation}
\label{F-contraction-inequ-2}
\| q[v] - q[w] \|_W \leq \xi_2 \| v-w \|_V
\end{equation}
Obviously the inequalities \eqref{F-contraction-inequ-1} and \eqref{F-contraction-inequ-2} can be combined to get \eqref{F-contraction} with $\mathcal{F}[v] = (-\nabla (n\nabla))^{-1} (q[v]-\partial_t^2 n)$ and $\xi = \xi_1 \xi_2$. Note that only $\xi \in (0,1)$ is needed, this needs not to hold true for $\xi_1$ and $\xi_2$ individually.

Is is easy to show yet interesting to note that the first step from $v_0$ to $v_1 = \mathcal{F}[v_0]$ for a contraction mapping $\mathcal{F}$ with Lipschitz constant $\xi \in (0,1)$ already yields an estimate on how far the final fixed point $v$ lies away from the starting point $v_0$.
\begin{equation}\label{fixed-point-distance}
\begin{aligned}
\|v-v_0\| &= \lim_{k \rightarrow \infty} \|v_k-v_0\| \leq \sum_{k=0}^\infty \|v_{k+1}-v_{k}\| \\
&\leq \sum_{k=0}^\infty \xi^k \|v_1-v_0\| = \frac{1}{1-\xi} \|v_1-v_0\|
\end{aligned}
\end{equation}
By this estimate all points of the converging sequence $\{v_k\}$ including the limit are within the closed ball $\overline{B_r(v_0)}$ width radius $r = (1-\xi)^{-1} \|v_1-v_0\|$.

\subsection{Relation to a non-linear Schrödinger equation}

Following \citeasnoun{maitra-2010} we might think of $q$ at time $t$ from \eqref{sl-nonlinear-2} not as a functional of $v$ (at all previous times) but depending on a quantum state $\psi(t)$ (at a single time) that comes from propagation under $v$. Then \eqref{sl-nonlinear-2} is equivalent to the following coupled equations.
\begin{align}
-\nabla\cdot (n(t) \nabla v(t)) &= q[\psi(t)] - \partial_t^2 n(t) \label{eq-nonlin-se-1}\\
\i \partial_t \psi &= H[v] \psi \label{eq-nonlin-se-2}
\end{align}
The first equation has the the benefit of eliminating the dependence on previous times from \eqref{sl-nonlinear-2}, a fact highlighted by adding explicit time dependence to all quantities. Solving \eqref{eq-nonlin-se-1} gives $v=v[n,\psi]$ and if we put this potential into \eqref{eq-nonlin-se-2} we have a non-linear Schrödinger equation equivalent to the coupled equations above.
\[
\i \partial_t \psi = H[v[n,\psi]] \psi
\]
Showing existence and uniqueness to this equation is thus equivalent to the original problem if we further demand that the resulting $\psi$ gives the right density $n$ (this is the $\rho$-problem, see \autoref{sect-rho-problem}). Now a strategy in proving this would be again a fixed-point scheme defined by the now linear Schrödinger equation
\[
\i \partial_t \psi_{i+1} = H[v[n,\psi_i]] \psi_{i+1}.
\]
The relation to the sequence $v_i$ defined by \eqref{sl-iteration-2} is straightforward. If we start with $\psi_1 = \psi[v_1]$ as the solution to \eqref{eq-nonlin-se-2} with a given potential $v_1$ then by \eqref{eq-nonlin-se-1} $v_2 = v[n,\psi_1]$ etc., consequently $\psi_i = \psi[v_i]$ for all indices $i \geq 2$. This, first noted in \citeasnoun{tddft-review}, makes the fixed-point procedure go side-by-side with approaches employing the non-linear version of Schrödinger's equation.

\subsection{The dreaded $\rho$-problem}
\label{sect-rho-problem}

If we assume existence of a fixed-point iteration converging to some $v$, does this really guarantee that by propagation of the initial state $\psi_0$ we get the prescribed density $n$? Let $n[v]$ be the density from propagation with $v$ that clearly has to fulfil \eqref{sl-nonlinear-2} because this equation is a direct consequence of Schrödinger's equation. Now if we take the difference $\rho = n[v]-n$ this quantity has to obey an equation where the $q$-term in cancelled out.
\begin{equation}\label{eq-rho}
\partial_t^2 \rho = \nabla\cdot (\rho \nabla v)
\end{equation}
The prescribed density has to fit to the initial state which means at time $t=0$ (no propagation yet)
\begin{equation}\label{eq-rho-initial-condition}
\begin{aligned}
n(0) &= n([v],0) \quad\mtext{and}\\
\partial_t n(0) &= -\nabla\cdot j([v],0) = \partial_t n([v],0)
\end{aligned}
\end{equation}
or equivalently $\rho(0) = \partial_t \rho(0) = 0$. We can find even more conditions on $\rho$ to further restrict the set of possible solutions by noting that both $n$ and $n[v]$ are normalised to the particle number $N$. Thus for all times it holds
\begin{equation}\label{eq-rho-condition}
\int \rho(t,x) \d x = 0.
\end{equation}
Now \eqref{eq-rho} together with initial conditions \eqref{eq-rho-initial-condition} looks like a good candidate for an initial value problem of a linear evolution equation with unique zero solution, thus telling us that indeed the resulting potential $v$ gives the prescribed density through propagation. It was assumed in \citeasnoun{tddft3} that this is \emph{trivially} true and written in \citeasnoun{ruggenthaler-2012} that it \emph{should} be true but we became more careful by now.

Reliable results are available for analytic $v$. As already noted in \autoref{sect-kovalevskaya}, the classical result of \citeasnoun{kovalevskaya} (usually called the Cauchy--Kowalevski theorem) states that for analytic (in time and space) potentials $v$ there is a unique solution within the class of analytic functions which thus must be the zero solution. This leaves open the possibility of non-analytic solutions, a gap that is readily closed by Holmgren's theorem. It shows that for $v$ analytic the solution \emph{must} be in the analytic class thus leaving zero as the only possible solution. Yet here we are interested in more general potentials, so the quest is not over.

The reason we are sceptical can be found in uniqueness counterexamples like that of \citeasnoun{tychonoff} for the heat equation that is $\Cont^\infty$ and has zero initial value but is still not zero for $t>0$. In that case uniqueness can be re-established, as shown also by \citeasnoun{tychonoff}, if the solution is not allowed to grow too large in space. A simple counterexample in our situation would be to take $\Omega=\R$ and $\nabla v=t^{-2}$ for $t>0$. A non-zero solution with zero initial values is $\rho = t^2 \e^{2x}$, but it grows fast in $x$ and does not fulfil the condition \eqref{eq-rho-condition}.

Strategies tried to find a uniqueness proof include the transformation to a system of first order in time and apply the method of characteristics and results from semigroup theory, or considering a weak version of the PDE and use energy methods to find conserved quantities and estimates for them. Currently the problem must be marked as still open.

\subsection{Weak solutions to the Sturm--Liouville equation}
\label{sect-weak-sl}

We rewrite \eqref{sl-iteration-2} in the more general form
\begin{equation}
\label{sl-general}
 -\nabla\cdot (n \nabla v) = \zeta
\end{equation}
and pose the questions of existence and uniqueness of a solution $v \in V$ depending on $n$ and $\zeta \in W$ as well as on possible conditions on the spatial domain $\Omega$. This will provide us in its course with explicit ideas for the Banach spaces $V,W$ and finally an estimate like \eqref{F-contraction-inequ-1}.

Note the similarity of this equation to a classical problem of electrostatics, namely determining the potential $v$ for a given charge distribution $\rho$ under variable relative permittivity of the material $\varepsilon_r \geq 1$.
\[
 - \nabla\cdot (\varepsilon_r \nabla v) = \frac{\rho}{\varepsilon_0}
\]
The case of $\varepsilon_r = 1$ corresponds to a perfect insulator like vacuum, the case of unbounded $\varepsilon_r$ to a perfect conductor. The homogeneous version with $\rho=0$ but non-zero boundary terms for $v$ is called the \emph{Calderón problem}. \cite{salo-2008}

The underlying domain of the functions in $V$ (like the Banach spaces of quantum trajectories etc.) still incorporates the time interval $[0,T]$ and therefore \eqref{sl-general} has to hold true for all $t \in [0,T]$. This would make it possible, at least in principle, to have different spaces $V_t = \{v(t) \mid v \in V \}$ as time-cuts of $V$. Such a fibration can also be realised in the case of a time depending space domain $\Omega_t$. We then have $v_t \in V_t$ as mappings $v_t : \Omega_t \rightarrow \R$. Although such constructions might have practical benefits we will try to keep it simple here and stick to temporally uniform spaces.

In the next step we generalise the derivatives in \eqref{sl-general} to their weak form. \autoref{weak-derivative} extends naturally to $\nabla$ and will be used to consider so-called \q{weak solutions} of \eqref{sl-general}, which will accordingly be defined by adjoining an arbitrary $\varphi \in \Cont^\infty_0$.
\[
-\langle \varphi, \nabla\cdot (n \nabla v) \rangle = \langle \varphi,\zeta \rangle
\]
We use integration by parts with vanishing boundary term (because $\varphi$ is zero at the boundary) to get a bilinear form $Q$ on $\Cont^\infty_0 \times V$ by
\begin{equation}
\label{Q-def}
Q(\varphi,v) = \langle \nabla \varphi, n \nabla v \rangle = \langle \varphi,\zeta \rangle.
\end{equation}
Note that the right hand side $\zeta$ acting on $\varphi$ can now be generalised to a distribution, i.e., a continuous linear functional on $\Cont^\infty_0$. We denote this again as $\zeta$ to make it clear that this distribution will be identified with $\zeta \in W$ from \eqref{sl-general} if we set $\zeta(\varphi) = \langle \varphi,\zeta \rangle$. But our definition of $Q$ would be fully symmetric if the two spaces of its definition coincide. So how about widening the first space from test functions by including all limits of Cauchy sequences under the $V$-norm (completion), eventually getting to all of $V$. Accordingly the space for $\zeta$ gets shrunk down to the dual of $V$. In fact a solution to this new problem, if one could possibly be found, would still be a weak solution, only in a stricter sense. So by now our task is to find a unique $v \in V$ which for all $u \in V$ fulfils
\begin{equation}
\label{Q-general}
Q(u,v) = \zeta(u)
\end{equation}
which then calls for an application of \autoref{lax-milgram} (Lax--Milgram).

An important note should be directed towards an alternative construction on periodic domains. Imagine $\Omega \subset \R^d$ bounded but shaped in such a way that it can be tessellated to cover all of $\R^d$, typically in the form of a rectangular box or cube. Now the same construction is possible if instead of zero boundary conditions for $\varphi$ periodic ones are requested for all involved quantities. This will again lead to a vanishing boundary term and makes the same construction as follows possible.

\section{Potentials from weighted Sobolev spaces}
\label{sect-weighted-soblev}

\subsection{Construction of the Hilbert space of potentials}
\label{sect-hs-potentials}

We have seen in the preceding section that our question for unique weak solutions to \eqref{sl-general} can be generalised to the form of \eqref{Q-general}. An answer is given by \autoref{lax-milgram} (Lax--Milgram) which applies if the congruent time-cuts $U$ of $V$ are Hilbert spaces with a norm that makes $Q$ coercive as well as continuous (see \autoref{def-coercive-continuous}). The use of the $L^2$ inner product from before suggests $U \subset L^2(\Omega)$. Observing that $\langle\cdot,\cdot\rangle$ and $Q$ can be naturally combined to a bilinear form
\begin{equation}
\label{weighted-sp}
\langle u,v \rangle_U = \langle u,v \rangle + Q(u,v) = \langle u,v \rangle + \langle \nabla u, n \nabla v \rangle,
\end{equation}
we might ask if this yields an adequate inner product to define our Hilbert space $U$, i.e., we have to check if
\[
\|u\|_U = \sqrt{\langle u,u \rangle_U} = \sqrt{\langle u,u \rangle + \langle \nabla u, n \nabla u \rangle}
\]
is actually a norm. One obvious restriction is $n \geq 0$ but this is certainly true for a one-particle density. A second restriction is the necessity that  the derivative $\nabla u$ is well-defined, if only in a weak sense. Because the norm of $U$ is constructed in such a way that $\|u\|_2 \leq \|u\|_U$ we have the supposed property $U \subset L^2(\Omega)$ as a continuous embedding. Thus we are led to the following chain of inclusions allowing elements of $U$ to be weakly differentiated according to \autoref{weak-derivative}.
\[
U \subset L^2(\Omega) \subset L^2_\mathrm{loc}(\Omega) \subset \Lloc(\Omega)
\]
If $n(x)=1$ then $\|\cdot\|_U$ is just the norm of the Sobolev space $W^{1,2}(\Omega) = H^1(\Omega)$ and written as $\|\cdot\|_{1,2}$ so we adopt the notation $H^1(\Omega,n)$ for the complete normed space equipped with $\|\cdot\|_U = \|\cdot\|_{1,2,n}$ and call it a ``weighted Sobolev space''. \cite{adams,kufner-opic} The inner product will from thereon also be noted as $\langle\cdot,\cdot\rangle_{1,2,n}$.

Let us get our notation of the different norms involved straight by defining them all for general $p \in [1,\infty)$ by integrals as in \eqref{eq-sobolev-norm}.
\begin{align*}
\|u\|_p &= \left( \int_\Omega |u|^p \d x \right)^{\frac{1}{p}} \\
\|u\|_{1,p} &= \left( \|u\|_p^p + \left\|\nabla u\right\|_p^p \right)^{\frac{1}{p}} \\
\|u\|_{p,n} &= \left\|u \, n^{\frac{1}{p}}\right\|_p = \left( \int_\Omega |u|^p \,n \d x \right)^{\frac{1}{p}} \\
\|u\|_{1,p,n} &= \left( \|u\|_p^p + \left\|\nabla u\right\|_{p,n}^p \right)^{\frac{1}{p}}
\end{align*}
Note the use of the weighting function $n$ only in the second term of the definition of $\|\cdot\|_{1,p,n}$. If the first index of $\|\cdot\|_{1,p}$ or $\|\cdot\|_{1,p,n}$ is an integer larger than 1 then higher derivatives are also included but this is not needed here.

Now continuity of $Q$ as a condition in \autoref{lax-milgram} (Lax--Milgram) is easy to prove.
\begin{equation}\label{Q-continuous}
\begin{aligned}
|Q(u,v)| &= \left|\langle \nabla u,n \nabla v \rangle\right| \\
&= \left|\left\langle \sqrt{n}\, \nabla u,\sqrt{n}\, \nabla v \right\rangle\right| \\
&\leq \left\| \sqrt{n}\, \nabla u \right\|_2 \cdot \left\| \sqrt{n}\, \nabla v \right\|_2 \\
&= \left\| \nabla u \right\|_{2,n} \cdot \left\| \nabla v \right\|_{2,n} \\
&\leq \left\| u \right\|_{1,2,n} \cdot \left\| v \right\|_{1,2,n}
\end{aligned}
\end{equation}

We still must not forget the restriction to functions which vanish at the border of $\Omega$ in order to justify the integration by parts used to derive \eqref{Q-def} in the first place. This is of course not generally true for elements of $H^1(\Omega,n)$ but can be met if one adopts the usual definition of $H_0^1(\Omega)$ (see \autoref{sect-zero-boundary}) to weighted Sobolev spaces. In that we take the space of infinitely differentiable functions on $\Omega$ with compact support $\Cont^\infty_0(\Omega)$ and form the completion under our weighted Sobolev norm $\|\cdot\|_{1,2,n}$. This of course is only possible if all $\varphi \in \Cont^\infty_0(\Omega)$ have finite $\|\varphi\|_{1,2,n}$ in the first place, but this is surely true if $n \in \Lloc(\Omega)$, a restriction easily met by a density. The resulting space $U = H_0^1(\Omega,n)$ equipped with inner product \eqref{weighted-sp} is then complete and thus a full-fledged Hilbert space of functions which vanish at the border of $\Omega$. This will be the main space of our further investigations.

Let us also define the dual of this Hilbert space $H^{-1}(\Omega,n) = H_0^1(\Omega,n)'$, i.e., the space of linear continuous functionals on $H_0^1(\Omega,n)$ (see \autoref{sect-sobolev-dual}). $\zeta$ from \eqref{Q-general} is thought to be an element of this space.

Nevertheless for the coercivity of $Q$ with respect to the Hilbert space $H_0^1(\Omega,n)$ to hold, we need additional assumptions on the weighting function $n$. We will first examine stronger restrictions on $n$ that make $-\nabla\cdot(n\nabla\cdot)$ an elliptic partial differential operator.

\subsection{The simple elliptic case}
\label{sect-elliptic-case}

The following case was considered in \citeasnoun{tddft1}. We start by noting that the assumptions of \autoref{lax-milgram} (Lax--Milgram) on $Q$ are easily met if we can impose the following restriction on $n$. Let $m, M > 0$ and for almost all $x \in \Omega$
\begin{equation*}\label{n-elliptic-restriction}
m \leq n(x) \leq M.
\end{equation*}
In this case the norms $\|\cdot\|_{1,2,n}$ and $\|\cdot\|_{1,2}$ are equivalent because of
\[
m \langle \nabla u, \nabla u \rangle \leq \langle \nabla u, n \nabla u \rangle \leq M \langle \nabla u, \nabla u \rangle
\]
and therefore we have $H_0^1(\Omega,n) = H_0^1(\Omega)$. Coercivity can now be proved by means of the Poincar\'{e} inequality (cf.~\citeasnoun[6.30]{adams}, in this form it is sometimes also called ``Friedrichs' inequality'').

\begin{theorem}[Poincar\'{e} inequality]\label{th-poincare}
If the domain $\Omega \subset \R^n$ is contained in a strip of finite width $a$ then for all $u \in W^{1,q}_{0}(\Omega)$
\[
 \|u\|_q \leq \frac{a}{q^{1/q}} \|\nabla u\|_q.
\]
\end{theorem}

\begin{proof}
We assume without loss of generality that $\Omega$ is contained in the strip $\{x \in \R^n \mid 0 < x_1 < a\}$. Since $\Cont_0^\infty(\Omega)$ is dense in $W^{1,q}_{0}(\Omega)$ we just have to show the estimate for $u \in \Cont_0^\infty(\Omega)$. One uses the Hölder inequality to get
\[
|u(x_1,x')| = \left|\int_0^{x_1} 1 \cdot \partial_t u(t,x') \d t \right| \leq x_1^{1-1/q} \cdot \left(\int_0^{x_1} |\partial_t u(t,x')|^q \d t\right)^{1/q}
\]
where we abbreviated $x' = (x_2,\ldots,x_n)$ as individual coordinates, not particle positions as usual. The final estimate now follows.
\[
\|u\|_q^q \leq \int_{\R^{n-1}} \d x' \int_0^a x_1^{q-1} \d x_1 \int_0^a |\partial_t u(t,x')|^q \d t \leq \frac{a^q}{q} \|\nabla u\|_q^q
\]
If integration leads outside of the domain $\Omega$ we imagine all functions to be continued by zero.
\end{proof}

Obviously the condition on $\Omega$ is always fulfilled for bounded domains. Applied to the case $q=2$ this gives us with $\lambda = a/\sqrt{2}$, the constant from the Poincaré inequality,
\begin{align*}
 Q(u,u) &= \langle \nabla u, n \nabla u \rangle \,\geq m \|\nabla u\|_2^2 \mtext{and} \\
 \lambda^2 Q(u,u) &\geq \lambda^2 m \|\nabla u\|_2^2 \geq m \|u\|_2^2.
\end{align*}
Combination of these results yields
\[
 Q(u,u) \geq \frac{m}{1+\lambda^2} \left( \|u\|_2^2 + \|\nabla u\|_2^2
 \right) = \frac{m}{1+\lambda^2} \| u \|_{1,2}^2.
\]
We thus have established the coercivity of $Q$. Continuity can be shown directly by the boundedness of $n$ from above.
\[
 | Q(u,v) | = |\langle \nabla u, n \nabla u \rangle | 
  \leq  M \|\nabla u\|_2 \cdot \|\nabla v\|_2 \leq
 M \|u\|_{1,2} \cdot \|v\|_{1,2} < \infty
\]
These restrictions on $n$ also imply that the differential operator defined by the left hand side of \eqref{sl-general} is \emph{strongly elliptic} which corresponds to coercivity.

If $n$ can be assumed to be continuous on the closed domain $\overline{\Omega}$ then $n$ also attains its extremal values on $\overline{\Omega}$ and the restriction \eqref{n-elliptic-restriction} reduces to the form
\begin{equation*}
 0 < n < \infty \quad \mbox{on}\; \overline{\Omega}.
\end{equation*}

However $n>0$ is quite a strong restriction on the one-particle density which will typically be zero at the border of $\Omega$. We try to weaken it in the following sections. In a periodic setting there are no borders, so the condition is naturally less restrictive though not absent. The second part of the restriction $n < \infty$ on the other side will be usually met for densities coming from differentiable wave functions by virtue of the Sobolev embedding \eqref{eq-sobolev-embedding}, like already employed in the proof of \autoref{lemma-q-L2-2}.

\subsection{Embedding theorems for weighted Sobolev spaces}
\label{sect-embedding-theorems}

Trouble starts if $n$ is not gapped away from zero everywhere. Clearly \eqref{sl-general} is harder to invert in regions where $n$ is small and it can have no unique solution when $n=0$ on a set of positive measure. Such problems could be avoided if one knows that simultaneously to $n \rightarrow 0$ the inhomogeneity $\zeta \rightarrow 0$ but we will try here to prove full coercivity for a fixed class of one-particle densities $n$. In this we follow a strategy largely outlined in \citeasnoun[Example 1.3]{drabek}. The idea is to continuously embed the weighted Sobolev space into a non-weighted one and further into an ordinary $L^2$ space.

\begin{lemma}\label{embedding1} Let $\Omega$ be bounded, $p > q \geq 1$ and the weighting function $n$ such that $n^{-s} \in L^1(\Omega)$ for $s = \frac{q}{p-q}$. Then we have an embedding
\[
W^{1,p}(\Omega, n) \hookrightarrow W^{1,q}(\Omega).
\]
\end{lemma}

\begin{proof}
Using Hölder's inequality with $\frac{1}{p'} + \frac{1}{q'} = \frac{q}{p} + \frac{p-q}{p} = 1$ we derive
\begin{align*}
\left\| \nabla u \right\|_q &= \left\| \, |\nabla u|^q \right\|_1^{\frac{1}{q}} = \left\| \left( |\nabla u|^q n^{\frac{q}{p}} \right) n^{-\frac{q}{p}} \right\|_1^{\frac{1}{q}} \\
&\leq \left( \left\| |\nabla u|^q n^{\frac{q}{p}} \right\|_{\frac{p}{q}} \left\| n^{-\frac{q}{p}} \right\|_{\frac{p}{p-q}}\right)^{\frac{1}{q}} = \left\| |\nabla u|^p n \right\|_1^{\frac{1}{p}} \left\| n^{-s} \right\|_1^{\frac{p-q}{pq}}
\end{align*}
and thus
\begin{equation}\label{inequ-q-p,n}
\| \nabla u \|_q \leq c \| \nabla u \|_{p,n}.
\end{equation}
Now we can easily establish the inclusion, considering that $L^p(\Omega) \hookrightarrow L^q(\Omega)$ for $\Omega$ bounded.
\end{proof}

The following lemma is part of the Rellich--Kondrachov theorem, a collection of embeddings of Sobolev spaces that are \emph{compact}, and it will not be proved here. The full-fledged theorem in all its generality can be found in \citeasnoun[6.3]{adams}, a presentation more adapted to the situation treated here is \citeasnoun[Th.~D.4.1]{blanchard-bruening}.

\begin{lemma}\label{embedding2}
Let $\Omega$ be bounded with dimension $d$ and satisfy the cone condition. Then we have a compact embedding
\[
W^{m,q}(\Omega) \hookrightarrow\hookrightarrow L^r(\Omega)
\]
for $1 \geq \frac{1}{r} > \frac{1}{q} - \frac{m}{d}$, provided $m \geq 1$ and $m q < d$.
\end{lemma}

These embeddings hold naturally also for $W^{1,p}_0(\Omega, n)$ and $W^{m,q}_0(\Omega)$ respectively. In the case of \autoref{embedding2} we can even drop the cone condition if considering $W^{m,q}_0(\Omega)$. This Lemma is in complete accordance with the rough estimates stated in \autoref{sect-lebesgue-sobolev-spaces} that the $W^{m,q}$-norm measures $AN^m V^{1/q}$. Using the well-known relation between the volume of the support in the space domain $V$ and the frequency bandwidth $N$ usually called \q{uncertainty principle} $V \gtrsim N^{-d}$ we get a lower estimate $A V^{1/q-m/d}$ for the $W^{m,q}$-norm that represents an $L^{r}$-norm (ignoring the substituted frequency bandwidth) with $1/r = 1/q-m/d$. This is exactly the limiting case of the Rellich--Kondrachov result.

\begin{theorem}\label{embedding3}
Let $\Omega$ be bounded with dimension $d$ and the weighting function $n$ such that $n^{-s} \in L^1(\Omega)$ for an $s > \frac{d}{2}$, then we have a compact embedding
\[
H_0^1(\Omega, n) \hookrightarrow\hookrightarrow L^2(\Omega).
\]
\end{theorem}

\begin{proof}
We start by considering the smallest necessary $s$ in \autoref{embedding1}, as $n^{-s} \in L^1(\Omega)$ for larger $s$ is only harder to achieve. From $s = \frac{q}{p-q}$ we get $q = \frac{2s}{s+1}$ for $p=2$. Now we choose $m=1$ and $r=2$ in \autoref{embedding2} to get the following chain of embeddings.
\[
H_0^1(\Omega, n) = W_0^{1,2}(\Omega, n) \hookrightarrow W_0^{1,q}(\Omega) \hookrightarrow\hookrightarrow L^2(\Omega)
\]
The condition on $s$ then reads $\frac{1}{2} > \frac{s+1}{2s} - \frac{1}{d}$ which transforms to $s > \frac{d}{2}$. The further condition $mq<d$ transforms to $(2-d)s<d$ and is thus only a restriction for $d=1$ reading $s<1$. Since a density $n$ fulfilling $n^{-s} \in L^1(\Omega)$ is still in $L^1(\Omega)$ for any smaller power than $s$ and $\Omega$ bounded, this is not a true restriction.
\end{proof}

Considerable time after publishing these findings in \citeasnoun{tddft4} we learned about a very similar study of \citeasnoun{caldiroli-musina-2000} to guarantee unique solutions to \eqref{sl-general} even in the case of semi-linear $\zeta$ depending on $v$. Yet the semi-linear $\zeta$ could not already include the non-linear $q[v]$ because this quantity depends on all previous times and so something like the fixed-point procedure is still needed. Their main assumption on $n$ is that if $z \in \Omega$ is a zero of $n$ then it must decrease more slowly than $|x-z|^\alpha$, $\alpha \in (0,2]$, in a neighbourhood of $z$. The problem is \q{subcritial} or \q{compact}, which exactly corresponds to our compactness result above, if $\alpha < 2$. Higher $\alpha$ correspond to a quicker decrease, so lower $\alpha$ put a tighter condition on the function at zeros. The formalisation of this condition is for all $z \in \overline{\Omega}$
\begin{equation}\label{eq-caldiroli-musina}
\liminf_{x \rightarrow z} \frac{n(x)}{|x-z|^\alpha} > 0,
\end{equation}
or equivalently with the small Landau symbol of \autoref{def-small-o}
\[
n(x) \notin o(|x-z|^\alpha) \mtext{as} x \rightarrow z.
\]
To see this take for example $z=0$, so a critical density decreases like $|x|^\alpha = r^\alpha$. The condition $n^{-s} \in L^1(\Omega)$ then demands for some small radius $R>0$
\[
\int_0^R r^{-\alpha s} r^{d-1} \d r < \infty
\]
and thus $d > \alpha s$. But we have $s > \frac{d}{2}$ and thus $\alpha < 2$. Our main result of this section \autoref{embedding3} is Corollary 2.6 in \citeasnoun{caldiroli-musina-2000}. If we look for the largest possible exponents $r$ to make the embedding \autoref{embedding2} possible we arrive exactly at the number $2^*_\alpha = 2d/(d-2+\alpha)$ used in Proposition 2.5 of \citeasnoun{caldiroli-musina-2000}. They give it as a generalisation to the inequality of \citeasnoun{caffarelli-kohn-nirenberg} where $n(x)=|x|^\alpha$ which is in turn a generalisation of the famous Hardy--Sobolev inequality.\footnote{The Heisenberg uncertainty inequality can be seen as a consequence of the Hardy--Sobolev inequality, see \citeasnoun{aermark}.}
\[
\left(\frac{d-2}{2}\right)^2 \int_\Omega \frac{u(x)^2}{|x|^2} \d x \leq \int_\Omega |\nabla u(x)|^2 \d x
\]

As an interesting remark \citeasnoun{caldiroli-musina-2000} note that if \eqref{eq-caldiroli-musina} holds for an $\alpha < 2$ then $n$ cannot be of class $\Cont^2$. To see this consider $f:\R \rightarrow \R_{\geq 0}$ and its second derivative at a point $x$ where $f(x)=0$
\[
f''(x) = \lim_{h \searrow 0} \frac{f(x+h)+f(x-h)}{h^2} = \lim_{h \searrow 0} \frac{f(x+h)+f(x-h)}{h^\alpha} \cdot \frac{1}{h^{2-\alpha}}.
\]
Note that if such an $f$ obeys \eqref{eq-caldiroli-musina} with $\alpha < 2$ then the first fraction is strictly positive while the second one clearly diverges. The fact that a density $n$ with zeros (including the borders) and of sufficient regularity cannot simultaneously fulfil the condition \eqref{eq-caldiroli-musina} suggests that this condition is actually too restrictive. Our condition $n^{-s} \in L^1(\Omega)$ is less restrictive because \citeasnoun{caldiroli-musina-2000} derive in Remark 2.3 that theirs leads to a necessarily finite number of zeros in $\overline{\Omega}$, whereas $n^{-s} \in L^1(\Omega)$ clearly allows for an infinite number of zeros too. An example would be $\Omega = \{ x \in \R^2 \mid |x|<1 \}$ the unit disc, $n(x)=(1-r^2)^{1/3}$, which has zeros all along the border of $\Omega$ but fulfils $n^{-2} \in L^1(\Omega)$.

\subsection{A proof of coercivity}
\label{sect-coercivity}

\begin{lemma}\label{Q-coercive}
Let $\Omega$ be bounded with dimension $d$ and the weighting function $n$ such that $n^{-s} \in L^1(\Omega)$ for an $s > \frac{d}{2}$. Then the bilinear form $Q(u,v) = \langle \nabla u, n \nabla v \rangle$ on $H_0^1(\Omega,n)$ is coercive.
\end{lemma}

\begin{proof}
To prove coercivity we need to show $Q(u,u) \geq c \| u \|^2_{1,2,n}$ for all $u \in H_0^1(\Omega,n)$. We start by observing $Q(u,u) = \| \nabla u \|_{2,n}^2$ through the special construction of this norm. An application of the preceding embedding theorems yields more or less the desired inequality. First we use the result \eqref{inequ-q-p,n} in the proof of \autoref{embedding1} for the case $p=2$ and $q,s$ appropriate. We therefore have to demand $n^{-s} \in L^1(\Omega)$.
\begin{equation}\label{coercivity-res1}
\| \nabla u \|_q \leq c_1 \| \nabla u \|_{2,n}
\end{equation}
Next we apply the Poincaré inequality (\autoref{th-poincare}) $\|u\|_q \leq c \|\nabla u\|_q$ to the left hand side. It immediately follows
\begin{equation}\label{coercivity-res2}
\| u \|_q \leq c_2 \| \nabla u \|_{2,n}.
\end{equation}
This clearly means that $\| \nabla \cdot \|_{2,n}$ is not only a seminorm, but a norm, since $\| \nabla u \|_{2,n} > 0$ for all $u \neq 0$. If we combine results \eqref{coercivity-res1} and \eqref{coercivity-res2} after taking $(\cdot)^q$ we get
\[
\| u \|_q^q +  \| \nabla u \|_q^q \leq \left( c_1^q + c_2^q \right) \| \nabla u \|_{2,n}^q
\]
which leads us straight to the norm of $W^{1,q}(\Omega)$ and
\[
\|u\|_{1,q} \leq c_3 \| \nabla u \|_{2,n}.
\]
Now we make use of \autoref{embedding2}, which tells us in the case $m = 1$ and $r = 2$ under the restrictions $\frac{1}{2} > \frac{1}{q}-\frac{1}{d}$ and $q<d$
\[
\| u \|_2 \leq c_4 \| u \|_{1,q}
\]
and therefore the Hardy-type inequality \cite{opic-kufner-hardy}
\begin{equation}\label{coercivity-res3}
\| u \|_2 \leq c_5 \| \nabla u \|_{2,n}
\end{equation}
holds. The restrictions above automatically hold with $q=\frac{2s}{s+1}$ and $s > \frac{d}{2}$ as in \autoref{embedding3}. It is now easy to arrive at the desired inequality by squaring \eqref{coercivity-res3} and adding another $\| \nabla u \|_{2,n}^2$.
\[
\| u \|_{1,2,n}^2 \leq \left( c_5^2 + 1 \right) \| \nabla u \|_{2,n}^2
\]
\end{proof}

As an interesting side effect we can now easily show that $Q$ delivers a norm equivalent to the standard norm of $H_0^1(\Omega,n)$.

\begin{corollary}\label{cor-energy-norm}
The norm defined by $u \mapsto \sqrt{Q(u,u)}$ on $H_0^1(\Omega,n)$ is equivalent to $\|\cdot\|_{1,2,n}$ for $\Omega$ bounded with dimension $d$ and $n$ such that $n^{-s} \in L^1(\Omega)$ for an $s > \frac{d}{2}$.
\end{corollary}

\begin{proof}
The first inequality is already given by the last line in the proof above. We just have to add another $\|u\|_2^2$ to arrive at
\[
c \|u\|_{1,2,n}^2 \leq \| \nabla u \|_{2,n}^2 = Q(u,u) \leq \|u\|_2^2 + \| \nabla u \|_{2,n}^2 = \|u\|_{1,2,n}^2
\]
and to conclude the proof.
\end{proof}

\subsection{Application of the Lax--Milgram theorem}
\label{sect-appl-lax-milgram}

The results of the previous sections now culminate in a theorem about the existence and uniqueness of weak solutions to the Sturm--Liouville type equation \eqref{sl-general} as well as in an estimate needed for the contraction \eqref{F-contraction-inequ-1}.

\begin{theorem}\label{sl-solutions}
Let $\Omega$ be bounded with dimension $d$, the weighting function $n$ such that $n^{-s} \in L^1(\Omega)$ for an $s > \frac{d}{2}$, and $\zeta \in H^{-1}(\Omega,n)$.\footnote{Note that the restriction $s > \frac{d}{2}$ was wrongly put $s=2$ in our theorem given in \citeasnoun{tddft4}. Still everything is correct in the typical case $d=3$ where the smallest applicable integer value is indeed $s=2$.} Then there is a unique $v \in H_0^1(\Omega,n)$ such that for all $u \in H_0^1(\Omega,n)$
\[
\langle \nabla u, n \nabla v \rangle = \langle u,\zeta \rangle.
\]
Furthermore this solution $v$ is bounded by the given datum $\zeta$, more precisely
\[
\|v\|_{1,2,n} \leq \|\zeta\|_{H^{-1}(\Omega,n)}.
\]
\end{theorem}

\begin{proof}
The bilinear form $Q$ on the Hilbert space $H_0^1(\Omega,n)$ defined by $Q(u,v) = \langle \nabla u, n \nabla v \rangle$ was shown to be continuous and coercive in \eqref{Q-continuous} and \autoref{Q-coercive} respectively. Therefore \autoref{lax-milgram} (Lax--Milgram) becomes applicable and the proof is done with a coerciveness constant found to be $c\leq 1$ in \autoref{cor-energy-norm}.
\end{proof}

We thus laid the groundwork for inequality \eqref{F-contraction-inequ-1} with $\xi_1 = 1$ as this is nothing else but the inequality in the theorem above. Closely related we have the following theorem as a solution to the general eigenvalue problem for $Q$.

\begin{theorem}\label{th-sl-eigenbasis}
Given the bilinear form $Q(u,v) = \langle \nabla u, n \nabla v \rangle$ on $H_0^1(\Omega,n)$ under the conditions of \autoref{sl-solutions} there is a monotone increasing sequence $(\lambda_m)_{m \in \mathbb{N}}$ of eigenvalues
\begin{equation}
	0 < \lambda_1 \leq \lambda_2 \leq \lambda_m \stackrel{m \rightarrow \infty}{\longrightarrow} \infty
\end{equation}
and an orthonormal basis $\{e_m\}_{m \in \mathbb{N}} \subset H_0^1(\Omega,n)$ of $L^2(\Omega)$ such that for all $u \in H_0^1(\Omega,n)$ and all $m \in \mathbb{N}$
\begin{equation}
	Q(u,e_m) = \lambda_m \langle u,e_m \rangle.
\end{equation}
\end{theorem}

\begin{proof}
By \autoref{embedding3} we have a compact embedding $H_0^1(\Omega, n) \hookrightarrow\hookrightarrow L^2(\Omega)$ and this makes \autoref{th-lax-milgram-eigenvalue} applicable which yields just the given proposition.
\end{proof}

If we want to apply these theorems to guarantee solutions to the iteration step \eqref{sl-iteration-2} in our fixed-point scheme, we have to make sure that $\zeta = q[v_i]-\partial_t^2 n$ is indeed a regular distribution in the dual $H^{-1}(\Omega,n)$ of our Hilbert space $H_0^1(\Omega,n)$. Less generally we might demand $\zeta \in L^2(\Omega)$ because of the inclusions
\[
H_0^1(\Omega,n) \hookrightarrow L^2(\Omega) \hookrightarrow H^{-1}(\Omega,n)
\]
stemming from \autoref{embedding3} and the duality relation. Of course this poses additional restrictions on $v_i$ which we tried to grasp already with \autoref{lemma-q-L2-2} that indeed shows $q[v_i] \in L^2(\Omega)$ if the potential guarantees $H^4$-regularity of the trajectory (cf.~\autoref{th-sobolev-regularity}). Similarly \autoref{lemma-dt2n-L2} studied conditions for $\partial_t^2 n \in L^2(\Omega)$ to hold.

If we test the density term $\partial_t^2 n \in H^{-1}(\Omega,n)$, a conversion with the continuity equation $\partial_t n = -\nabla\cdot j$ presents itself. Adjoined with an arbitrary $u \in H_0^1(\Omega,n)$ we get
\[
\langle u, \partial_t^2 n \rangle = -\langle u, \nabla\cdot \partial_t j \rangle = \langle \nabla u, \partial_t j \rangle = \left\langle \sqrt{n} \nabla u, n^{-\frac{1}{2}} \partial_t j \right\rangle.
\]
Now $\sqrt{n} \nabla u \in L^2(\Omega)$ because of $u \in H_0^1(\Omega, n)$, so the question remains if this is also true for $n^{-\frac{1}{2}} \partial_t j$. This means we have to demand a finite force integral
\begin{equation}\label{eq-finite-force}
\int \frac{|\partial_t j|^2}{n}\d x < \infty.
\end{equation}
Note a certain similarity of this term to the so-called Weizsäcker term from time-independent DFT
\begin{equation}\label{eq-weizsaecker}
\int |\nabla \sqrt{n}|^2\d x = \frac{1}{4} \int \frac{|\nabla n|^2}{n}\d x
\end{equation}
which is indeed finite for all wave functions with finite kinetic energy following \autoref{lemma-lieb-1983} below.

\subsection{Stronger solutions to the Sturm--Liouville equation}
\label{sect-stronger-sol-sl}

Next we want to give conditions for \q{stronger} solutions to the actual Sturm--Liouville equation \eqref{sl-general} and not only to its weak counterpart \eqref{Q-general}. The natural space for the inhomogeneity for those seems to be $\zeta \in L^2(\Omega) \subset H^{-1}(\Omega,n)$ which was already partly used in the last section. The following lemma restricts the density $n$ such that effectively $H_0^1(\Omega,n) = H_0^1(\Omega)$, unfortunately further reducing the range of applicability.

\begin{corollary}\label{cor-strong-lm-sol}
Under the conditions of \autoref{sl-solutions} in $d\leq 3$ but adding $\zeta \in L^2(\Omega)$ and $n \in W^{4,1}(\Omega)$ bounded from below $n(x)\geq m > 0$ for all $x \in \Omega$, we have a unique strong solution $v \in H_0^1(\Omega) \cap H^2(\Omega)$ to the equation $-\nabla\cdot(n\nabla v) = \zeta$ that is furthermore bounded by
\[
\|v\|_{2,2} \leq m^{-1} \left(\frac{a^2}{2} + 1\right) \left(1 + \sqrt{m}^{-1} \|\nabla n\|_\infty\right) \|\zeta\|_2,
\]
where $a$ is the diameter of the thinnest stripe that contains the domain $\Omega$.
\end{corollary}

\begin{proof}
We start by noting that the Sobolev embedding $W^{4,1}(\Omega) \subset W^{3,1}(\Omega) \subset L^\infty(\Omega)$ implies the existence of a bound $M>0$ such that $n(x) \leq M$ for all $x \in \Omega$ and thus $H_0^1(\Omega,n) = H_0^1(\Omega)$ as in \autoref{sect-elliptic-case}. Because every solution is also a weak solution naturally $v \in H_0^1(\Omega)$. The remaining property of the solution can be derived by application of the product rule
\[
\nabla\cdot(n\nabla v) = \nabla n \cdot \nabla v + n \Delta v = -\zeta
\]
that yields the estimate
\[
m\|\Delta v\|_2 \leq \|n \Delta v\|_2 \leq \|\zeta\|_2 + \|\nabla n\|_\infty \cdot \|\nabla v\|_2.
\]
The $\|\nabla v\|_2$ in the last term is already estimated by \autoref{sl-solutions} as a weak solution and by $\|\zeta\|_2$ due to the embedding $L^2(\Omega) \hookrightarrow H^{-1}(\Omega,n)=H^{-1}(\Omega)$.
\[
\sqrt{m}\|\nabla v\|_2 \leq \|\sqrt{n}\nabla v\|_2 \leq \|v\|_{1,2,n} \leq \|\zeta\|_{H^{-1}} \leq \|\zeta\|_2
\]
Now $\nabla n \in W^{3,1}(\Omega) \subset L^\infty(\Omega)$ so we can infer
\begin{equation}\label{eq-laplace-v-estimate}
\|\Delta v\|_2 \leq m^{-1}\left(1 + \sqrt{m}^{-1} \|\nabla n\|_\infty\right) \|\zeta\|_2
\end{equation}
To get an estimate in the $H^2$-norm we proceed like in \autoref{th-sobolev-norm-laplace}. The Sobolev norm is given as
\[
\|v\|_{2,2}^2 = \|v\|_2^2 + \sum_{|\alpha| = 1} \|D^\alpha u\|_2^2 + \sum_{|\beta| = 2} \|D^\beta u\|_2^2 = \|v\|_2^2 + \|\nabla v\|_2^2 + \|\Delta v\|_2^2
\]
if we additionally use \eqref{eq-sum-laplace-norm}. With the Poincar\'{e} inequality (\autoref{th-poincare}) used several times we then have an estimate purely in terms of $\|\Delta v\|_2$.
\[
\|v\|_{2,2}^2 \leq \left(\frac{a^4}{4} + \frac{a^2}{2} + 1\right) \|\Delta v\|_2^2 \leq \left(\frac{a^2}{2} + 1\right)^2 \|\Delta v\|_2^2
\]
If we combine the above with \eqref{eq-laplace-v-estimate} we finally get the desired estimate for the $H^2$-norm of $v$.
\end{proof}

If we want to use $\zeta \in L^1$ rather than $L^2$, where a suitable result for the $q$-term is available as well with \autoref{lemma-q-L1}, we can show an according corollary.

\begin{corollary}\label{cor-strong-lm-sol-2}
Under the conditions of \autoref{sl-solutions} but adding $\zeta \in L^1(\Omega)$ and $\sqrt{n} \in H^1(\Omega)$, we have a unique strong solution $v \in H_0^1(\Omega,n)$ fulfilling $n\Delta v \in L^1(\Omega)$ to the equation $-\nabla\cdot(n\nabla v) = \zeta$.
\end{corollary}

\begin{proof}
For a weak solution it holds naturally $v \in H_0^1(\Omega,n)$ and the expression $\nabla\cdot(n\nabla v) = \nabla n \cdot \nabla v + n \Delta v$ has to make sense in $L^1(\Omega)$. The first term can be divided into
\[
\nabla n \cdot \nabla v = 2 \nabla \sqrt{n} \cdot \sqrt{n} \nabla v
\]
where $\sqrt{n} \nabla v \in L^2(\Omega)$ from $v \in H_0^1(\Omega,n)$ and $\nabla \sqrt{n} \in L^2(\Omega)$ as assumed. Thus $n \Delta v = \nabla\cdot(n\nabla v) - \nabla n \cdot \nabla v \in L^1(\Omega)$.
\end{proof}

Both results fit to the setting where the potential is from a Sobolev--Kato space $W^{2,\Sigma}$ (see \autoref{def-sobolev-kato}) and thus $\Delta v \in L^2+L^\infty$. \autoref{cor-strong-lm-sol} is conceptually similar to that of \q{boundary regularity} in \citeasnoun[6.3.2]{evans} for partial differential operators in divergence form with $\Cont^1$-coefficients that take on the role of $n$. It states that any weak solution is also strong, i.e., in $H^2$. Here too we get a solution that fulfils $\Delta v \in L^2$.

In \autoref{cor-strong-lm-sol-2} the relation to Sobolev--Kato spaces is different and assuming a potential from $W^{2,\Sigma}$ rather opens another way towards the desired outcome. We added $\sqrt{n} \in H^1(\Omega)$ as an ingredient, which means $n \in L^3(\Omega)$ by the Sobolev embedding $H^1 \subset L^6$ \cite[4.12, I.C]{adams} and thus $n \in L^3 \cap L^1 \subset L^2 \cap L^1$ following an idea of \citeasnoun{lieb-1983}. Combining this with $\Delta v \in L^2+L^\infty$ we get the desired $n\Delta v \in L^1$.

As already mentioned after \eqref{eq-weizsaecker} $\sqrt{n} \in H^1(\Omega)$ is a natural choice because it always holds for wave functions with finite kinetic energy by the following lemma. Note that this time the usual condition $\Omega$ bounded is not needed.

\begin{lemma}\label{lemma-lieb-1983}\cite[Th.~1.1]{lieb-1983}\\
If for the wave function it holds $\psi \in H^1(\Omega)$ then $\sqrt{n} \in H^1(\Omega)$ and further $\|\nabla\sqrt{n}\|_2 \leq \sqrt{N} \|\nabla\psi\|_2$.
\end{lemma}

\begin{proof}
$\sqrt{n} \in L^2(\Omega)$ because $\int_\Omega n \d x = N$ is just the normalisation of the wave function. Now with the usual expression for $n$ and $\nabla$ acting only on $x\in \R^d$
\[
\nabla n(x) = 2 N \int_{\bar{\Omega}} \d \bar{x}\, \Re\left\{ \psi(x,\bar{x})^* \nabla \psi(x,\bar{x}) \right\}.
\]
Then with the CSB inequality for a Hilbert space $L^2(\bar\Omega)$ with reduced configuration space including the particle positions $x_2, \ldots, x_N$ just like in the proof of \autoref{lemma-q-L2-2}
\begin{align*}
(\nabla n(x))^2 &\leq 4N^2 \left( \int_{\bar{\Omega}} \d \bar{x}\, |\psi(x,\bar{x})| \cdot |\nabla \psi(x,\bar{x})| \right)^2 \\
&\leq 4N^2 \int_{\bar{\Omega}} \d \bar{x}\, |\psi(x,\bar{x})|^2 \cdot \int_{\bar{\Omega}} \d \bar{x}\, |\nabla\psi(x,\bar{x})|^2 \\
&= 4N n(x) \int_{\bar{\Omega}} \d \bar{x}\, |\nabla\psi(x,\bar{x})|^2
\end{align*}
and thus
\[
\int_\Omega (\nabla\sqrt{n(x)})^2 \d x = \frac{1}{4}\int_\Omega (\nabla n(x))^2 n(x)^{-1} \d x \leq N \|\nabla\psi\|_2^2.
\]
\end{proof}

Note that in \autoref{cor-strong-lm-sol-2} above the condition $\zeta \in H^{-1}(\Omega,n)$ is not replaced, so really $\zeta \in L^1(\Omega)\cap H^{-1}(\Omega,n)$. One idea to instead widen the class of possible inhomogeneities $\zeta$ in \eqref{Q-def} for weak solutions would be to let go of the Lax--Milgram theorem and switch to Lax--Milgram--Lions (\autoref{lax-milgram-lions}) instead. By adjoining functions $\varphi$ from a more regular normed space like that of Lipschitz functions $W^{1,\infty}(\Omega)$ continuously embedded in $H_0^1(\Omega,n)$ to form the bilinear form $Q(\varphi,v)$ on has more options for the inhomogeneity that can come from the dual space, then clearly including $L^1(\Omega)$. Yet the problem is to prove the condition in the theorem that demands something comparable to coercivity but now for a finer space (see also \citeasnoun[Cor.~III.2.3]{showalter}). Such an estimate will not follow from coercivity on $H_0^1(\Omega,n)$ and we are stuck again.

\section{Special cases}

\subsection{The one-dimensional case}

We have seen already in the proof of \autoref{embedding3} that the case $d=1$ is slightly special. But what weighs much more is that the Sturm--Liouville problem \eqref{sl-general}, now in its classical form $-\partial_x (n \partial_x v) = \zeta$, can be directly integrated to yield a solution. On a domain $\Omega = (a,b)$ we readily get
\[
v(x) = -\int_a^x \frac{1}{n(y)} \int_a^{y} \zeta(z) \d z \d y + c_1 \int_a^x \frac{\d y}{n(y)} + c_2.
\]
The constants $c_1,c_2 \in \R$ spanning the whole space of solutions are determined by the boundary conditions. Such a solution is always bounded by
\begin{align*}
|v(x)| &\leq \int_a^x \frac{\d y}{n(y)} \cdot \int_a^x |\zeta(y)| \d y + |c_1| \int_a^x \frac{\d y}{n(y)} + |c_2| \\
&\leq \|n^{-1}\|_1 \cdot (\|\zeta\|_1 + |c_1|) + |c_2|
\end{align*}
and therefore sufficient conditions for a solution $v \in L^{\infty}(\Omega)$ are $n^{-1} \in L^1(\Omega)$ and $\zeta \in L^1(\Omega)$. But classical Sturm--Liouville theory can go further and includes densities that are in a sense \emph{singular}. The lower endpoint $a$ of the interval $\Omega$ is called \emph{regular} if
\[
\int_a^x \frac{\d y}{n(y)} < \infty
\]
for an arbitrary $x \in \Omega$ and it is called \emph{singular} if the integral diverges. A classical example of a density leading to singular endpoints is the Legendre differential equation for a $\nu \in \N_0$
\[
\partial_x ((1-x^2) \partial_x u) = -\nu(\nu+1)u
\]
on $\Omega=(-1,1)$ where one has $n(x) = 1-x^2$ with $n^{-1} \notin L^1(\Omega)$. But solutions to that equation are well known, the two linearly independent solutions given by the Legendre polynomials and the Legendre functions of the second kind. As eigenfunctions the orthogonal Legendre polynomials yield a basis of $L^2(\Omega)$ like in \autoref{th-sl-eigenbasis} and thus a unique solution to the (self-adjoint) Sturm--Liouville problem is guaranteed if $\zeta$ is orthogonal to the $\nu=0$ eigenspace, the kernel of the differential operator. Note it holds $n^{-1+\varepsilon} \in L^1(\Omega)$ for any $\varepsilon>0$ which nicely fits to the $s=1-\varepsilon>\onehalf$ condition of \autoref{sl-solutions}, so our study seems to be flexible enough to include such cases.

A detailed account on the Sturm--Liouville problem in $d=1$ with respect to the fixed-point proof of TDDFT putting special attention on periodic domains $\Omega=\mathbb{S}^1$ is given in \citeasnoun{ruggenthaler-2012}. As a primer to Sturm--Liouville theory we refer to the book of \citeasnoun{zettl}. Also \citeasnoun{caldiroli-musina-2001} devoted a paper to the one-dimensional singular Sturm--Liouville problem using their condition \eqref{eq-caldiroli-musina} on $n$.

\subsection{The spherical-symmetric case}

Another one-dimensional setting is achieved if we assume all quantities to be fully spherical symmetric in a $d$-dimensional space, writing $v(r), n(r), \zeta(r)$.
\[
-\nabla\cdot(n(r)\nabla v(r)) = \zeta(r)
\]
As we have $\nabla r = \frac{x}{r}$, the unit vector pointing in direction $x$, we derive
\[
-\nabla\cdot(n(r)\nabla u(r)) = -(d-1)\frac{n(r) v'(r)}{r} - (n(r)v'(r))' = \zeta(r).
\]
Substituting $n(r)v'(r)=w(r)$ we get a linear ODE of first order.
\[
w'(r)+(d-1)\frac{w(r)}{r}=-\zeta(r)
\]
We first solve the homogeneous problem.
\begin{align*}
w_0'(r)+(d-1)\frac{w_0(r)}{r} &= 0 \\
(\ln|w_0(r)|)' &= -(d-1)(\ln(r))' \\
w_0(r) &= C_0 r^{-(d-1)}
\end{align*}
Next we find a particular solution $w(r)=C(r)r^{-(d-1)}$ with the variation of constant method. As we may expect singularities at $r=0$ we integrate from $r \rightarrow \infty$ to solve for $C(r)$.
\begin{align*}
C'(r) &= -r^{d-1}\zeta(r) \\
C(r) &= \int_r^\infty s^{d-1}\zeta(s) \d s + C(\infty) \\
w(r) &= r^{-(d-1)} \left( \int_r^\infty s^{d-1}\zeta(s) \d s + C(\infty) \right)
\end{align*}
If we now substitute $w$ back and solve for $v$ we have after integration
\[
v(r) = -\int_r^\infty s^{-(d-1)}n(s)^{-1} \left( \int_{s}^\infty t^{d-1}\zeta(t) \d t + C(\infty) \right) \d s + v(\infty).
\]
Now we can try and see how $\zeta$ looks like in a real quantum setting with $d=3$, the Coulomb potential $v(r) = -r^{-1}$ of a point-like nucleus, and the density of a hydrogen 1s orbital $n(r) = \alpha \e^{-r}$. Clearly such a density goes to zero in all directions and cannot fulfil $n^{-s} \in L^1$ for any $s\geq 0$ so even less for $s>\frac{3}{2}$ like in the conditions of \autoref{sl-solutions}. Still we want to study if the remaining quantities fit into this framework. The Coulombic $v=-r^{-1}$ is not in $H_0^1(\R^3,n)$ because the infinite tails make $v \notin L^2(\R^3)$. Finally we check if $\zeta$ is in $H^{-1}(\R^3,n)$.
\[
\zeta(r) = -\nabla\cdot(n(r)\nabla v(r)) = \nabla\cdot(\alpha \e^{-r} \nabla r^{-1}) =\frac{\alpha\e^{-r}}{r^2}
\]
We have to test $|\langle u,\zeta \rangle| < \infty$ for all $u \in H_0^1(\R^3,n)$. Observe that with $f = -\sqrt{n}/r$
\[
\zeta = \partial_r (\sqrt{n} f) + \sqrt{n} f
\]
which is already like an element of a Sobolev-space dual in the notation explained right after \autoref{th-sobolev-dual-element}.
\begin{align*}
|\langle u,\zeta \rangle| &\leq |\langle u,\partial_r (\sqrt{n} f) \rangle| + |\langle u,\sqrt{n} f \rangle| \\
&= |\langle \sqrt{n} \partial_r u, f \rangle| + |\langle u,\sqrt{n} f \rangle|
\end{align*}
As $u \in H_0^1(\R^3,n)$ we surely have $\sqrt{n} \partial_r u \in L^2(\R^3)$ and $u \in L^2(\R^3)$, so we need to show $f \in L^2(\R^3)$ and $\sqrt{n} f \in L^2(\R^3)$.
\begin{align*}
\int_{\R^3} |f|^2 \d x &= \int_0^\infty \frac{n}{r^2} r^2 \d r = \int_0^\infty \alpha \e^{-r} \d r = \alpha \\
\int_{\R^3} |\sqrt{n}f|^2 \d x &= \int_0^\infty \frac{n^2}{r^2} r^2 \d r = \int_0^\infty \alpha^2 \e^{-2r} \d r = \frac{\alpha^2}{2}
\end{align*}
Therefore $|\langle u,\zeta \rangle| < \infty$ and $\zeta \in H^{-1}(\R^3,n)$. We conclude that the framework for the Sturm--Liouville equation developed previously only fits partly in this setting. The problems arise because $\Omega=\R^3$ is unbounded which was ruled out anyway as a condition in the theorems of \autoref{sect-appl-lax-milgram}.

\subsection{The single particle case}

The case of a single particle in an arbitrary $d$-dimensional domain $\Omega$ is equivalent to any non-interacting $N$-particle description with factorising initial state which makes it possible to separate the Schrödinger equation into single-particle equations. In such a case the potential that leads to a given density $n$ can be exactly calculated which has been employed for two non-interacting particles on a one-dimensional periodic interval $\Omega = \mathbb{S}^1$ in \citeasnoun{ruggenthaler-2013}. We start with a polar representation for the wave function.
\[
\psi = R \e^{\i S}
\]
The functions $R,S$ are both real and the radial component has $R \geq 0$. Then the Schrödinger equation can be cast into a real and an imaginary part.
\begin{align}
\partial_t S &= \frac{1}{2} \frac{\Delta R}{R} - \frac{1}{2} (\nabla S)^2 - v \label{eq-hamilton-jacobi}\\
\partial_t R &= -\nabla R \cdot \nabla S - \frac{1}{2} R \Delta S \label{eq-single-particle-cont-eq}
\end{align}
The first equation has been divided by $R$ which already hides problems connected to $R \approx 0$ which are not addressed further here. Now this radial component is just the square root of the one-particle density $n$ and if we substitute $R=\sqrt{n}$ into the second equation \eqref{eq-single-particle-cont-eq} we get
\[
\partial_t n = -\nabla \cdot (n \nabla S),
\]
just another form of the marvelled Sturm--Liouville equation, here as a variant of the continuity equation. Now the theory developed in the sections before applies seamlessly and after  solving for $S$ with given $n$ we can directly get $v$ from \eqref{eq-hamilton-jacobi}.

We should add that \eqref{eq-hamilton-jacobi} is a special form of the Hamilton--Jacobi equation of classical mechanics where $S$ is the action while the right hand side is the Hamiltonian function (with reversed sign). This equation, as the most \q{continuous} formulation of classical mechanics with wave-like particle description, was also the starting point for Schrödinger in a search for a wave equation leading to the quantisation rules. \cite{schroedinger-1}

\section{Towards a more rigorous Runge--Gross proof}

\subsection{Transformation to a Schrödinger problem}

An interesting transformation of \eqref{sl-general} noted by \citeasnoun{maitra-2010} occurs if we substitute $\sqrt{n} v=w$. In general this is of course not a valid bijection and conditions on $n^{-1}$ as already put forward in the application of weighted Sobolev spaces are necessary. Formally \eqref{sl-general} becomes
\[
-\Delta w + \left( \frac{1}{2} \frac{\Delta n}{n} -  \frac{1}{4}\left(\frac{\nabla n}{n}\right)^2 \right) w = \frac{\zeta}{\sqrt{n}}.
\]
This is now a PDE of standard time-independent Schrödinger type in three dimensions for which a rich theory is readily available. Observe that the operator on the left hand side has an eigenfunction $\sqrt{n}$ with eigenvalue 0 (the ground state) which is easy to check.

A weak solution in the sense of \autoref{sect-weak-sl} with zero boundary conditions would fulfil $w=\sqrt{n}v \in H_0^1(\Omega)$ which almost corresponds to our $v \in H_0^1(\Omega,n)$. We are thus again in the setting of weighted Sobolev spaces and consequently no real solution to the problem of small densities can be given with this strategy. Taking  $w=\sqrt{n}v$ effectively removes information on $v$ in areas of $n \approx 0$ and corresponds to \q{probing the potential through matter} already visible in the weighted space construction.

\subsection{Fréchet estimate of the internal forces mapping}
\label{sect-q-mapping-frechet}

The main idea in the pioneering paper \cite{tddft3} of the fixed-point construction to finally get an estimate of type \eqref{F-contraction-inequ-2} was to assume functional differentiability of $q[v]$. Then by the fundamental theorem of the calculus of variations (\autoref{cor-fund-th-variations}) we have
\[
q[v_2] - q[v_1] = \int_0^1 \delta q[v_1 + \lambda (v_2 - v_1); v_2-v_1] \d \lambda.
\]
The potential variation should not be termed $w$ this time and we stick to $v_2-v_1$ for it, because this is the exclusively internal interaction potential that is thought of being a part of $v$ already. By further assuming that the Fréchet derivative allows for a linear response function as an integration kernel one has like in \eqref{eta-x-lin-resp}
\begin{equation}\label{eq-q-frechet-kernel}
\begin{aligned}
\delta q([v; v_2-v_1],t) = \int_0^\infty \d s \int_\Omega \d y\, &\i \theta(t-s) \langle [\hat{n}_{H[v]}(s,y), \hat{q}_{H[v]}(t,x)] \rangle_0\\
&(v_2(s,y)-v_1(s,y)).
\end{aligned}
\end{equation}
Here we left all caution aside for a moment and treated $\hat{n}_{H[v]}$ and $\hat{q}_{H[v]}$ as valid operators in the $H[v]$-interaction picture. Now exchanging the order of integration we have with $v_\lambda = v_1 + \lambda (v_2 - v_1)$
\begin{align*}
q[v_2] - q[v_1] = \int_0^\infty \d s \int_\Omega \d y\, &\left[ \i\theta(t-s) \int_0^1  \langle [\hat{n}_{H[v_\lambda]}(s,y), \hat{q}_{H[v_\lambda]}(t,x)] \rangle_0 \d \lambda \right]\\
 &\cdot(v_2(s,y)-v_1(s,y)).
\end{align*}
If the integral operator defined by the kernel in square brackets is then uniformly bounded we can derive the sought for estimate. This procedure is of course problematic because we already \emph{assumed} existence of the various mathematical objects and thus marked the starting point for the study of functional differentiability of quantum trajectories and derived quantities in \autoref{ch-diff}. We can state rigorous conditions for the existence of a Fréchet derivative $\delta q$ with the help of \autoref{th-rs-frechet} and can also derive an estimate for $\|\delta q\|_2$ using the form \eqref{eq-q-term-2} again. Note we use \autoref{lemma-permut-delta} to interchange variational and spatial (weak) derivatives.
\begin{equation}\label{eq-q-derivative}
\begin{aligned}
\delta q = N \Re \int_{\bar{\Omega}} \Big(&\Delta\psi^* \Delta\delta\psi - \onehalf ( \psi^* \Delta^2 \delta\psi + \delta\psi^* \Delta^2 \psi ) \\
+& 2 \sum_{j=2}^N w(x-x_j) (\nabla_j-\nabla) \cdot \nabla_j \psi^*\delta\psi \Big) \d \bar{x}.
\end{aligned}
\end{equation}
The estimate can now be derived exactly like in showing $q \in L^2$ in \autoref{lemma-q-L2-2} and the final estimate will be in terms of $\|\delta\psi\|_{4,2}$ given by \autoref{cor-estimate-delta-psi}. Please bear in mind that the necessary setting for the Fréchet derivative to be defined in regularity class $H^4$ is $\psi_0 \in H^8$ and potential space $\Lip([0,T],W^{6,\Sigma})$. Although this really high degree of regularity is needed for the expressions to be well-defined, the final estimate will be made with respect to a smaller degree of regularity.
\[
\sup_{t\in [0,T]} \|\delta\psi([v;v_2-v_1],t)\|_{4,2} \leq T c[v] \cdot \max_{t\in [0,T]}\|v_2(t)-v_1(t)\|_{4,\Sigma} \cdot \|\psi_0\|_{6,2} 
\]
Note that we only assume external one-body potentials and have $\Omega$ bounded so the $W^{4,\Sigma}$-norm defined in \autoref{def-sobolev-kato} is equivalent to a $H^4$-norm. Collecting all spatial derivatives up to 4\textsuperscript{th} order together with estimates for the interaction potential $w$ in a constant temporarily termed $c_w'$ we have the following.
\begin{equation}\label{eq-estimate-delta-q}
\begin{aligned}
\sup_{t\in [0,T]} \|\delta q([v;v_2-v_1],t)\|_2 &\leq c_w' \sup_{t\in [0,T]} \|\psi([v],t)\|_{4,2} \cdot \|\delta\psi([v;v_2-v_1],t)\|_{4,2} \\
&\leq T c[v,w] \cdot \max_{t\in [0,T]}\|v_2(t)-v_1(t)\|_{4,2} \cdot \|\psi_0\|_{6,2}^2
\end{aligned}
\end{equation}
This gives an expression for the desired $\xi_2$ in \eqref{F-contraction-inequ-2} as well as definitive spaces $V=\Lip([0,T],H^4)$ and $W=L^\infty([0,T],L^2)$. To finally get a contraction bound $\xi \in (0,1)$ in \eqref{F-contraction} for potentially big $\xi_1, \xi_2$ in \citeasnoun{tddft3} and later works we resort to a trick related to the Bielecki norm (see \citeasnoun[1.2.6]{zettl} and originally \citeasnoun{bielecki}). This norm is equivalent to the usual Banach-space norm, yet weights time with a decreasing exponential $\exp(-\alpha t), \alpha > 0$. Through the time integral from expressing the Fréchet derivative with an integration kernel in \eqref{eq-q-frechet-kernel} one gets a $\alpha^{-1}$ that can be taken arbitrarily small. But note that the same effect arises if the considered time interval $[0,T]$ in \eqref{eq-estimate-delta-q} is taken sufficiently small. To continue the procedure for later times $t>T$ one then simply repeats it with new initial values.

Uniform boundedness of \eqref{eq-estimate-delta-q} means to have a bounded operator $\delta q$ for all different possible potentials $v$. But even if this set of possible potentials is made bounded by relying on the first-step estimate \eqref{fixed-point-distance} of the fixed-point scheme one still operates in an infinite-dimensional potential space. This means that a continuous function on a bounded set needs not to achieve a finite supremum, this property demands a (pre-) \emph{compact} set, a problem that we already expressed in \citeasnoun{tddft-review}. Stated differently, a bounded, closed set is not necessarily compact in infinite-dimensional metric spaces. A potentially unbounded Fréchet derivative $\delta q$ means that there are potentials $v_1,v_2$ with fixed distance $\|v_2-v_1\|_V=1$ that produce an arbitrarily large local force difference $q[v_2]-q[v_1]$. It was noted in \citeasnoun{tddft3} that \q{it is, however, physically reasonable to assume that the latter cannot happen. A more precise mathematical study of this point is topic of future investigations.} Well, please refer to the chapters above.

\subsection{Lipschitz estimate of the internal forces mapping}
\label{sect-q-mapping-lipschitz}

We later realised that the use of functional differentiability for an estimate of type \eqref{F-contraction-inequ-2} can be circumvented by directly relying on Lipschitz continuity in $v$ of $\psi[v]$ and $q[v]$ respectively. This would indeed follow from \eqref{eq-evolut-lipschitz} if the $v_1,v_2$ dependency hidden in the \q{$\lesssim$} relation yields an absolute bound for all potentials under consideration. This means the problem related to non-compact potential sets remains, but the path over Lipschitz continuity already helps us to exchange the arbitrarily high $\psi_0 \in H^8$, $v \in \Lip([0,T],W^{6,\Sigma})$ requirements by a more modest $\psi_0 \in H^6$, $v \in \Lip([0,T],W^{4,\Sigma})$ that also stands in direct relation to the derived estimate. To give this estimate we start with expression \eqref{eq-q-term-2} for $q$ and two different potentials $v_1,v_2 \in \Lip([0,T],W^{4,\Sigma})$. For the sake of brevity we write $\psi_1=\psi[v_1]$ and $\psi_2=\psi[v_2]$ for the trajectories of the respective potentials. The formula for the finite difference $q[v_2]-q[v_1]$ will be given analogously to the infinitesimal difference $\delta q$ in \eqref{eq-q-derivative}.
\begin{align*}
q[v_2]&-q[v_1] = N \int_{\bar{\Omega}} \Big(\onehalf \big(\Delta(\psi_2^*-\psi_1^*) \Delta\psi_2 + \Delta\psi_1^*\Delta(\psi_2-\psi_1)\big) \\[0.4em]
&- \onehalf\Re \left\{ (\psi_2^*-\psi_1^*) \Delta^2 \psi_2 + \psi_1^* \Delta^2 (\psi_2-\psi_1) \right\} \\
&+ \sum_{j=2}^N w(x-x_j) (\nabla_j-\nabla) \cdot \nabla_j \big( (\psi_2^*-\psi_1^*)\psi_2 + \psi_1^*(\psi_2-\psi_1) \big) \Big) \d \bar{x}
\end{align*}
To estimate $\|q[v_2]-q[v_1]\|_2$ at some given time $t \in [0,T]$ we proceed like in the proof of \autoref{lemma-q-L2-2} by resorting to the Hilbert space $L^2(\bar\Omega)$ with reduced configuration space $\bar\Omega=\Omega^{N-1}$. The terms in the first two rows of the expression above are of the kind $\langle \Delta^k\psi_1,\Delta^{2-k}(\psi_2-\psi_1) \rangle'_x$, $k=0,1,2$, and can be estimated using CSB.
\[
|\langle \Delta^k\psi_1,\Delta^{2-k}(\psi_2-\psi_1) \rangle'_x| \leq \|\Delta^k\psi_1\|'_x \cdot \|\Delta^{2-k}(\psi_2-\psi_1) \|'_x
\]
Here $\|\Delta^k\psi_1\|'^2_x \in W^{2,1}(\Omega)$ (or higher) due to the preserved regularity $\psi_1 \in H^6$ of the initial state. This equals $\|\Delta^k\psi_1\|'_x \in W^{2,2}(\Omega)$ which is bounded by the Sobolev embedding \eqref{eq-sobolev-embedding}. The other terms enters the Sobolev norm $\|\psi_2-\psi_1\|_{4,2}$ when estimating $\|q[v_2]-q[v_1]\|_2$. The interaction term written as a scalar product consists of
\[
\left|\langle D^\alpha \psi_1, w(x-x_j) \cdot D^\beta (\psi_2-\psi_1) \rangle'_x\right| \leq \| D^\alpha \psi_1\|'_x \cdot \|w(x-x_j) \cdot D^\beta (\psi_2-\psi_1)\|'_x
\]
with $|\alpha|,|\beta| \leq 2$. The first part is bounded as before while the second one can be estimated with the usual \autoref{lemma-sum-space-inequality} for any $w \in L^2(\Omega) + L^\infty(\Omega)$ when evaluating the $L^2$-norm over the remaining coordinates $x$. The remaining $\|D^\beta (\psi_2-\psi_1)\|_{2,2}$ enters the Sobolev norm $\|\psi_2-\psi_1\|_{4,2}$ again. Finally this difference is estimated by \eqref{eq-evolut-lipschitz} and the Sobolev norms of $\psi_1,\psi_2$ up to 6th order yielding the respective bounded terms by \eqref{eq-schro-evolut-estimate-2m}.
\begin{equation}\label{eq-estimate-delta-q-2}
\sup_{t\in [0,T]} \|q([v_2],t)-q([v_1],t)\|_2 \leq T c_L[v_1,v_2;w] \cdot \max_{t\in [0,T]}\|v_2(t)-v_1(t)\|_{4,2} \cdot \|\psi_0\|_{6,2}^2
\end{equation}
Note that if one follows the origin of $c_L$ carefully it is shown to depend continuously on $v_1,v_2 \in \Lip([0,T],W^{4,\Sigma})$ (due to \autoref{lemma-rs-2} and \autoref{lemma-rs-3}) and in first order on the interaction potential $w \in L^2(\Omega) + L^\infty(\Omega)$.

\subsection{More on fixed-point theorems}

The Banach fixed-point theorem is by far not the only available method for guaranteeing solutions to equations involving Banach space endomorphisms. There is a second strain of theorems following Brouwer's fixed-point theorem that are non-constructive and only proof existence, not uniqueness. We follow \citeasnoun[8.1.4, 9.2.2]{evans} in the following presentation, where also an application to (quasilinear) PDEs is given.

\begin{theorem}[Brouwer's fixed-point theorem]
\hfill\\
Assume $f : \overline{B_1(0)}$ $\rightarrow \overline{B_1(0)}$ continuous, where $\overline{B_1(0)}$ is the closed unit ball in $\R^n$. Then there exists $x \in \overline{B_1(0)}$ such that $f(x)=x$.
\end{theorem}

The theorem also holds for general bounded, closed, and convex sets in $\R^n$. The \citeasnoun{wiki-brouwer} gives a very interesting account on an intuitive explanation attributed to Brouwer himself: \q{I can formulate this splendid result different, I take a horizontal sheet, and another identical one which I crumple, flatten and place on the other. Then a point of the crumpled sheet is in the same place as on the other sheet.} Note that the crumpled sheet is thought of being flattened keeping all folds and wrinkles. This intuitive picture also includes the potential non-uniqueness of fixed points and honours him as one of the main figures in the \emph{intuitionsm} branch of philosophy of mathematics.

A generalisation of Brouwer's result to Banach spaces $X$ is Schauder's fixed-point theorem where the key assumption of closeness and boundedness translates to compactness.

\begin{theorem}[Schauder's fixed-point theorem]
Suppose $K \subset X$ compact and convex and $f:K \rightarrow K$ continuous. Then there exists $x \in K$ such that $f(x)=x$.
\end{theorem}

The proof uses a finite open cover of $K$ which is guaranteed to exist through compactness and employs Brouwer's fixed-point theorem. The theorem can be transformed to an alternative form with conditions that are sometimes easier to check because one does not have to identify the set $K$.

\begin{theorem}[Schaefer's fixed-point theorem]
Suppose $f : X \rightarrow X$ continuous and compact and assume further the set
\[
\{ x \in X \mid x = \lambda f(x), 0 \leq \lambda \leq 1 \}
\]
is bounded. Then $f$ has a fixed point in $X$.
\end{theorem}

Thus if we have a bound on the fixed points of all possible $\lambda f$ for $0 \leq \lambda \leq 1$, then we also have the existence of a fixed point for $f$. This is linked to the method of \emph{a priori} estimates in PDE theory. \citeasnoun[p.~539]{evans} explains it as the remarkable informal principle that \q{if we can prove appropriate estimates for solutions of a nonlinear PDE, under the assumption that such solutions exist, then in fact these solutions do exist.} The application of such theorems is now through the inverse of linear elliptic operators that are typically smoothing and can allow for compactness.

As always the most basic example is provided through the Laplacian. Because $-\Delta : H_0^1(\Omega) \cap H^2(\Omega) \rightarrow L^2(\Omega)$ has a well-defined inverse, we can define the inverse operator $(-\Delta)^{-1}$ whose image is continuously embedded in $H_0^1(\Omega)$ which in turn is compactly embedded in $L^2(\Omega)$ for bounded $\Omega \subset \R^n$ obeying the cone condition and $n \geq 3$ (Rellich--Kondrachov Theorem, \autoref{embedding2}). This makes $(-\Delta)^{-1}: L^2(\Omega) \rightarrow L^2(\Omega)$ a compact operator.

This raises some hope that a similar technique can be used for the inverse of our Sturm--Liouville problem $(-\nabla\cdot(n\nabla))^{-1}$, the embedding $H_0^1(\Omega,n) \hookrightarrow\hookrightarrow L^2(\Omega)$ already shown to be compact in \autoref{embedding3}. The non-linear term $q[v]$ allows for certain estimates in $\|v\|$, but the growth is rather exponential as seen in \eqref{eq-schro-evolut-estimate-2m} than allowing a Lipschitz estimate like in \citeasnoun[9.2.2, Ex.~2]{evans}.

\subsection{Open issues}

We are now in the position to use the developed techniques and put them into action to find minimal conditions for a possible (non-extended) Runge--Gross theorem. The principle statement of \autoref{runge-gross-th} is:

\emph{Given the initial wave function $\psi_0$ of an $N$-particle system there is a one-to-one mapping between densities and external potentials up to a merely time-dependent function.}

We thus try to show that the mapping $v \mapsto n[v]$ of external one-body potentials from a wisely chosen domain into the set of densities is injective. The interaction part of the potential is considered to be fixed here and will thus not be denoted separately. We will collect all assumptions needed along the way.

Assume $v \neq v'$ but giving the same density $n[v]=n[v']=n$ for a fixed initial state $\psi_0 \in \H$. Take the \q{divergence of force density equation} \eqref{eq-div-force-density} with both potentials and subtract them to get
\[
-\nabla \cdot (n \nabla (v-v')) = q[v]-q[v']
\]
with the $\partial_t^2 n$ term cancelled. \autoref{cor-strong-lm-sol} gives an estimate for the (strong) solution to the corresponding Sturm--Liouville problem and that is
\[
\|v-v'\|_{2,2} \leq C_{\mathrm{LM}} \|q[v]-q[v']\|_2
\]
at all times $t \in [0,T]$ with the time argument still suppressed. The necessary conditions are $\Omega$ bounded with dimension $d \leq 3$ and the density resulting from both $v$ and $v'$ must be bounded from below, i.e., $n(x)\geq m > 0$ for all $x \in \Omega$. This calls for a periodic setting in space that will be assumed from now on. Having $\Omega$ tessellated removes the zero boundary conditions from $H_0^1(\Omega)$ and we have to determine another fitted space such that the bilinear form from the Sturm--Liouville problem is coercive. This will rule out all $v = const$ from the domain of potentials, a fact which does not yet enter into the notation. The final condition $n \in W^{4,1}(\Omega)$ is automatically fulfilled by \autoref{th-sobolev-regularity} if $v,v' \in \Lip([0,T], H^2(\Omega))$ and $\psi_0 \in H^4(\Omega)$. On the other hand the positivity condition for $n[v]$ is not automatically obeyed by such potentials.

In the next step the difference $q[v]-q[v']$ is attacked with the derived Lipschitz estimate \eqref{eq-estimate-delta-q-2}. One important ingredient here is \autoref{lemma-sum-space-inequality} that demands $d\geq 3$ thus fixing the setting to $d=3$, our beloved physical space.
\[
\sup_{t\in [0,T]} \|q([v],t)-q([v'],t)\|_2 \leq T c_L[v,v';w] \cdot \max_{t\in [0,T]}\|v(t)-v'(t)\|_{4,2} \cdot \|\psi_0\|_{6,2}^2
\]
Now putting this together gives
\[
\max_{t\in [0,T]} \|v(t)-v'(t)\|_{2,2} \leq TC_{\mathrm{LM}} c_L[v,v';w] \cdot \max_{t\in [0,T]}\|v(t)-v'(t)\|_{4,2} \cdot \|\psi_0\|_{6,2}^2
\]
where we replaced the supremum by the maximum because the potentials are Lipschitz-continuous anyway. The $T$ in front can now be made arbitrarily small to give a contradiction \emph{if} the norms on both sides would match. Unfortunately they do not and the inequality above cannot be used as an argument for injectivity of the map $v \mapsto n[v]$. Using the alternative estimate \autoref{cor-estimate-delta-psi-2} for the Fréchet derivative of $\psi[v]$ instead is no good either, because it introduces the additional time derivatives of the potential and has no explicit linear $T$-dependence. Thus the dilemma at this point is that the pieces from the two main ingredients of this proof technique, the Sturm--Liouville inversion and the estimation of the internal forces term, do not fit. Please note that this is not a predominant failure of the fixed-point approach, other treatments also suffer from many of the issues discussed here.

Although we believe we were successful in clarifying some major open problems in the construction of a full-fledged mathematical framework for TDDFT, it seems we were even more efficient in raising new questions. An overview of open issues shows that the whole theoretic building is still fragile.

\begin{enumerate}

	\item \emph{The connection problem.} As the principle nuisance this problem shows up just above, where the norms on both sides of the contraction inequality \eqref{F-contraction} finally do not match. This arises because the necessary regularity for the interaction term is higher than the one coming out from the inversion of the Sturm--Liouville operator.
	
	\item \emph{The small densities problem.} Obviously small or zero densities are a potential hazard when inverting the Sturm--Liouville equation \eqref{sl-general} and we were able to give exact conditions on \emph{how small} they can still be. One strategy to make things easier was to choose periodic boundary conditions for the space domain. But it is by no means clear that potentials from the given classes guarantee permitted densities again. Statements in this direction, if available, would relate conceptually to the unique continuation property of eigenstates (\autoref{sect-ucp}).
	
	\item \emph{The time-control problem.} Solving the fixed-point iteration \eqref{sl-iteration-2} at individual time instants yields the desired potential as a limit but does not give any information on its time regularity (Lipschitz-continuity to be more specific) that is needed for regularity of solutions to the Schrödinger equation.
	
	\item \emph{The $\rho$-problem.} Raised in \autoref{sect-rho-problem} this problem consists in the insecurity if the potential from a fixed-point iteration really leads to the original density if put into Schrödinger's equation with the same initial state.

\end{enumerate}

\subsection{Concluding remarks}

Casting a more positive light on the work at hand one can say that it presents the first fully coherent setting for a well-defined \q{divergence of force density equation} \eqref{eq-div-force-density}.
\[
-\nabla\cdot (n \nabla v) = q - \partial_t^2 n
\]
For the equation to hold in $L^2(\Omega)$ this is \autoref{cor-strong-lm-sol} for the left-hand side, what we call the Sturm--Liouville operator, and the necessary Sobolev regularity of the quantum trajectory for $q$. What is needed here is demonstrated in \autoref{lemma-q-L2-2} in conjunction with \autoref{th-sobolev-regularity} for $m=2$. The \q{second law term} $\partial_t^2 n$ was treated similarly in \autoref{lemma-dt2n-L2}.

Adding to that we gave many results regarding invertibility of the Sturm--Liouville type equation in \autoref{sect-appl-lax-milgram} and in \autoref{sect-stronger-sol-sl}. This usually assumed invertibility is also a major ingredient in the original proof of the extended Runge--Gross theorem (\autoref{ext-runge-gross-th}) that is used to show the existence of an auxiliary Kohn--Sham system and thus forms the basis for applications of TDDFT.

We believe that one serious drawback in the current fixed-point proof for Runge--Gross is the detached treatment of individual time instants which does not impose any temporal regularity on the resulting potential $v$. On the other hand the classical Runge--Gross proof (\autoref{sect-runge-gross-proof}) demands a very high temporal regularity (analyticity) from the potential. It seems that the most reasonable level of regularity would lie in between. The whole formulation is somewhat still too much based on the static setting of DFT. But the magnitude of internal interactions can also be controlled by a frequency dependent norm of the potential through \autoref{cor-estimate-delta-psi-2} that also showed up in the energy estimates in \autoref{sect-energy-estimates}. Now any high-frequency scalar potential would physically be linked to a magnetic field through the charged medium by Maxwell's equations. This can be seen as a hint that actually the classical Runge--Gross description is too restrictive and magnetic fields interacting with charge currents should be included. This is the setting of time-dependent current DFT   \cite{vignale-2004} or in a fully field theoretic description the one of the newly established quantum-electrodynamical DFT \cite{ruggenthaler-qed-1,ruggenthaler-qed-2,flick-2015}. One logical step would thus be the application of the fixed-point approach or similar techniques to those branches and trying to find a closed working scheme there.



\end{document}